%% file: main.tex
\documentclass[fontsize=11pt, a4paper,twoside=semi, bibliography=totoc,usegeometry]{scrbook}
\usepackage{geometry}
\usepackage{scrhack} 

\input{packages}

\title{Model-free Methods for Event History Analysis and Efficient Adjustment}

\author{Alexander Mangulad Christgau}

\def\includepapers{1}

\begin{document}

\input{frontmatter}

\mainmatter

\input{introduction}

\chapter{Nonparametric conditional local independence testing}\label{chapter:LCM}

\if\includepapers1
    \setcounter{section}{0}
    \setcounter{supsection}{0}
    \setHyperefPrefix{ch2}
    \renewcommand\thesubsection{\thesection.\arabic{subsection}}
    \renewcommand\theequation{\thesection.\arabic{equation}}
    \renewcommand\thefigure{\thesection.\arabic{figure}}
    \renewcommand\thetable{\thesection.\arabic{table}}
    \renewcommand\thelem{\thesection.\arabic{lem}}
    \renewcommand\theprop{\thesection.\arabic{prop}}
    \renewcommand\thethm{\thesection.\arabic{thm}}
    \renewcommand\thecor{\thesection.\arabic{cor}}
    \renewcommand\thedfn{\thesection.\arabic{dfn}}
    \begin{authorlist}
    \textsc{Alexander Mangulad Christgau, Lasse Petersen and Niels Richard Hansen}
    \end{authorlist}    
    \input{paper_lcm/main_paper}
    \input{paper_lcm/appendix}
\fi

\chapter{Efficient adjustment for complex~covariates: Gaining efficiency with DOPE}\label{chapter:adjustment}

\if\includepapers1
    \setcounter{section}{0}
    \setcounter{supsection}{0}
    \setHyperefPrefix{ch3}
    \renewcommand\thesubsection{\thesection.\arabic{subsection}}
    \renewcommand\theequation{\thesection.\arabic{equation}}
    \renewcommand\thefigure{\thesection.\arabic{figure}}
    \renewcommand\thetable{\thesection.\arabic{table}}
    \renewcommand\thelem{\thesection.\arabic{lem}}
    \renewcommand\theprop{\thesection.\arabic{prop}}
    \renewcommand\thethm{\thesection.\arabic{thm}}
    \renewcommand\thecor{\thesection.\arabic{cor}}
    \renewcommand\thedfn{\thesection.\arabic{dfn}}
    \begin{authorlist}
        \textsc{Alexander Mangulad Christgau and Niels Richard Hansen}
    \end{authorlist}
    \input{paper_dope/main_paper}
    \input{paper_dope/appendix}
\fi

\chapter{Assumption-lean Aalen regression}\label{chapter:ACM}

\if\includepapers1
    \setcounter{section}{0}
    \setcounter{supsection}{0}
    \setHyperefPrefix{ch4}
    \renewcommand\thesubsection{\thesection.\arabic{subsection}}
    \renewcommand\theequation{\thesection.\arabic{equation}}
    \renewcommand\thefigure{\thesection.\arabic{figure}}
    \renewcommand\thetable{\thesection.\arabic{table}}
    \renewcommand\thelem{\thesection.\arabic{lem}}
    \renewcommand\theprop{\thesection.\arabic{prop}}
    \renewcommand\thethm{\thesection.\arabic{thm}}
    \renewcommand\thecor{\thesection.\arabic{cor}}
    \renewcommand\thedfn{\thesection.\arabic{dfn}}
    \begin{authorlist}
        \textsc{Alexander Mangulad Christgau and Niels Richard Hansen}
    \end{authorlist}
    \input{paper_acm/main_paper}
    \input{paper_acm/appendix}

\fi


\bibliographystyle{abbrvnat}
\bibliography{references}

\newpage

\end{document}

%% file: packages.tex
\usepackage[utf8]{inputenc}
\usepackage{amsmath,amssymb,amsthm,amsfonts,mathtools}
\usepackage{caption,subcaption}
\usepackage[inline,shortlabels]{enumitem}
\usepackage{hyperref}
\usepackage[table,usenames,dvipsnames]{xcolor}
\usepackage{natbib}
\usepackage{tikz}

\usepackage{thmtools}

\hypersetup{
  colorlinks=true,
  linkcolor=RoyalBlue,
  citecolor=gray,
  urlcolor=RawSienna,
  hypertexnames=false
}

\usepackage{rotating}

\usepackage{makecell}

\usepackage{bibentry}

\usepackage{titling} 

\usepackage{placeins}

\usepackage{apptools}

\usepackage{quoting}
\newenvironment{abstract}{\begin{quoting}[leftmargin=0.5cm]\begin{center}\textbf{Abstract}\end{center}}{\end{quoting}}

\newenvironment{authorlist}{\begin{quote}}{\end{quote} \vspace{0.75cm}}

\tikzstyle{round-box} = [draw=RoyalBlue, fill=BurntOrange!35, very thick, rectangle, rounded corners, inner sep=10pt, inner ysep=20pt]
\tikzstyle{box-title} = [fill=RoyalBlue, text=white]
\tikzstyle{scm-node} = [circle,inner sep=2pt,draw=black!50,thin]
\usepgflibrary{arrows}
\usetikzlibrary{calc,fadings,decorations.pathreplacing,intersections,positioning}
\usetikzlibrary{decorations.markings}
\usepgflibrary{arrows}

\usepackage[linesnumbered,ruled,vlined]{algorithm2e}
\newcounter{algorithm}

\usepackage{mathbbol}
\DeclareSymbolFontAlphabet{\mathbb}{AMSb}
\DeclareSymbolFontAlphabet{\mathbbl}{bbold}
\newcommand{\one}{\mathbbl{1}}

\input{paper_lcm/packages.tex}
\input{paper_lcm/commands.tex}

\input{paper_dope/packages.tex}

\input{paper_dope/commands.tex}

\input{paper_acm/commands.tex}

\input{commands}
\input{environments}

%% file: paper_lcm/packages.tex
\usepackage{bm}

%% file: paper_lcm/commands.tex
\newcommand{\blambda}{\bm{\lambda}}
\newcommand{\bLambda}{\bm{\Lambda}}
\newcommand{\convUD}{\xrightarrow{\scalebox{0.6}{$\mathcal{D}/\Theta$}}}
\newcommand{\convUDnull}{\xrightarrow{\scalebox{0.6}{$\mathcal{D}/\Theta_0$}}}
\newcommand{\convUP}{\xrightarrow{\scalebox{0.6}{$P/\Theta$}}}
\newcommand{\convUPnull}{\xrightarrow{\scalebox{0.6}{$P/\Theta_0$}}}
\newcommand{\VVV}{W\hspace{-0.7em}W}

\newcommand{\cl}{c\`adl\`ag}
\newcommand{\lc}{c\`agl\`ad}

\newcommand{\vertiii}[1]{{\left\vert\kern-0.25ex\left\vert\kern-0.25ex\left\vert #1 \right\vert\kern-0.25ex\right\vert\kern-0.25ex\right\vert}}

%% file: paper_dope/packages.tex
\usepackage{amssymb}
\usepackage{amsmath}
\usepackage{amsthm}
\usepackage{latexsym}
\usepackage{todonotes}
\usepackage{booktabs}
\usepackage{tikz}
\usepackage{tkz-graph}

\usepackage{mathabx}

%% file: paper_dope/commands.tex
\newcommand{\bw}{\mathbf{w}}
\newcommand{\bW}{\mathbf{W}}
\newcommand{\bz}{\mathbf{z}}
\newcommand{\bZ}{\mathbf{Z}}
\newcommand{\bV}{\mathbb{V}}
\newcommand{\cZ}{\mathcal{Z}}
\newcommand{\ex}{\mathbb{E}}
\newcommand{\avar}{\mathbb{V}_{\! \Delta}}

\newcommand{\ind}{\,\perp \! \! \! \perp\,}
\newcommand{\indP}{\,{\perp \! \! \! \perp}_P\,}
\newcommand{\given}{\,|\,}
\newcommand{\real}{\mathbb{R}}

\newcommand{\pa}[1]{\left(#1\right)}

\DeclareMathOperator{\var}{Var}
\DeclareMathOperator{\cov}{Cov}

%% file: paper_acm/commands.tex
\newcommand{\bh}{\bm{h}}
\newcommand{\bH}{\bm{H}}
\newcommand{\bM}{\bm{M}}

\usepackage{graphicx,scalerel}
\newcommand\wh[1]{\hstretch{0.25}{\hat{\hstretch{4}{#1}}}}
\newcommand{\hatM}{\wh{M}}

\newcommand{\ol}[1]{\mkern 2mu\overline{\mkern-2mu#1\mkern-0mu}\mkern 0mu}
\newcommand{\oZ}{\ol{Z}}
\newcommand{\oX}{\ol{X}}
\newcommand{\tgamma}{\widetilde{\gamma}}

\DeclareMathOperator*{\minimize}{minimize}

%% file: commands.tex
\newcommand{\LCM}{[\hyperref[chapter:LCM]{{\color{black}\textbf{LCM}}}]}
\newcommand{\DOPE}{[\hyperref[chapter:adjustment]{{\color{black}\textbf{DOPE}}}]}
\newcommand{\ACM}{[\hyperref[chapter:ACM]{{\color{black}\textbf{ACM}}}]}

\renewcommand{\thesection}{\arabic{chapter}.\arabic{section}} 

\let\temp\phi
\let\phi\varphi
\let\varphi\temp

\newcommand{\cF}{\mathcal{F}}
\newcommand{\cG}{\mathcal{G}}

\DeclareMathOperator*{\argmin}{arg\,min}

\newcommand{\floor}[1]{\lfloor{#1}\rfloor}
\newcommand{\ceil}[1]{\lceil{#1}\rfloor}

\newcounter{supsection} 
\renewcommand\thesupsection{\arabic{chapter}.\Alph{supsection}} 
\newcommand{\supplementarysection}[1]{%
  \refstepcounter{supsection} 
  \section*{\thesupsection. #1} 
  \addcontentsline{toc}{section}{%
  \texorpdfstring{\protect\makebox[21.5pt][l]{\thesupsection}}{??} #1} 
  \counterwithin*{subsection}{supsection}
  \counterwithin*{equation}{supsection}
  \counterwithin*{figure}{supsection}
  \counterwithin*{table}{supsection}
  \counterwithin*{lemma}{supsection}
  \counterwithin*{proposition}{supsection}
  \renewcommand\thesubsection{\thesupsection.\arabic{subsection}}
  \renewcommand\theequation{\thesupsection.\arabic{equation}}
  \renewcommand\thefigure{\thesupsection.\arabic{figure}} 
  \renewcommand\thetable{\thesupsection.\arabic{table}}
  \renewcommand\thelem{\thesupsection.\arabic{lem}}
  \renewcommand\theprop{\thesupsection.\arabic{prop}}
  \renewcommand\thethm{\thesupsection.\arabic{thm}}
  \renewcommand\thecor{\thesupsection.\arabic{cor}}
  \renewcommand\thedfn{\thesupsection.\arabic{dfn}}
  \setcounter{subsection}{0} 
  \setcounter{equation}{0} 
  \setcounter{figure}{0} 
  \setcounter{table}{0} 
  \setcounter{lem}{0} 
  \setcounter{prop}{0} 
  \setcounter{thm}{0}
  \setcounter{cor}{0}
}

\newcommand{\setHyperefPrefix}[1]{
    \renewcommand*{\theHsection}{#1.\the\value{section}}
    \renewcommand*{\theHtheorem}{#1.\the\value{theorem}}
    \renewcommand*{\theHproposition}{#1.\the\value{proposition}}
    \renewcommand*{\theHcorollary}{#1.\the\value{corollary}}
    \renewcommand*{\theHassumption}{#1.\the\value{assumption}}
    \renewcommand*{\theHsetting}{#1.\the\value{setting}}
    \renewcommand*{\theHexample}{#1.\the\value{example}}
    \renewcommand*{\theHremark}{#1.\the\value{remark}}
    \renewcommand*{\theHdefinition}{#1.\the\value{definition}}
    \renewcommand*{\theHfigure}{#1.\the\value{figure}}
    \renewcommand*{\theHtable}{#1.\the\value{table}}
    \renewcommand*{\theHalgorithm}{#1.\the\value{algorithm}}
    \renewcommand*{\theHequation}{#1.\the\value{equation}}
    \renewcommand*{\theHlemma}{#1.\the\value{lemma}}
}

%% file: environments.tex
\theoremstyle{plain}

\newtheorem{thm}{Theorem}[section]
\newtheorem{prop}[thm]{Proposition}
\newtheorem{lem}[thm]{Lemma}
\newtheorem{cor}[thm]{Corollary}
\newtheorem{example}{Example}

\newtheorem{asm}{Assumption}[section]

\theoremstyle{remark}
\newtheorem{remark}{Remark}
\newtheorem{rmk}[thm]{Remark}

\theoremstyle{definition}
\newtheorem{definition}{Definition}
\newtheorem{dfn}[thm]{Definition}
\newtheorem{exm}[thm]{Example}

\renewenvironment{example}
    {\pushQED{\qed}\renewcommand{\qedsymbol}{$\spadesuit$}\exm}
    {\popQED\endexm}
\renewenvironment{remark}
    {\pushQED{\qed}\renewcommand{\qedsymbol}{$\diamondsuit$}\rmk}
    {\popQED\endrmk}
\renewenvironment{definition}
    {\pushQED{\qed}\renewcommand{\qedsymbol}{$\clubsuit$}\dfn}
    {\popQED\enddfn}

\renewcommand\thmcontinues[1]{Continued}

\xpretocmd{\chapter}{\setcounter{theorem}{0}}{}{}
\xpretocmd{\chapter}{\setcounter{lemma}{0}}{}{}
\xpretocmd{\chapter}{\setcounter{corollary}{0}}{}{}
\xpretocmd{\chapter}{\setcounter{proposition}{0}}{}{}
\xpretocmd{\chapter}{\setcounter{assumption}{0}}{}{}
\xpretocmd{\chapter}{\setcounter{setting}{0}}{}{}
\xpretocmd{\chapter}{\setcounter{example}{0}}{}{}
\xpretocmd{\chapter}{\setcounter{remark}{0}}{}{}
\xpretocmd{\chapter}{\setcounter{definition}{0}}{}{}
\xpretocmd{\chapter}{\setcounter{figure}{0}}{}{}
\xpretocmd{\chapter}{\setcounter{table}{0}}{}{}
\xpretocmd{\chapter}{\setcounter{algorithm}{0}}{}{}
\xpretocmd{\chapter}{\setcounter{equation}{0}}{}{}

\counterwithout{table}{chapter}
\counterwithout{figure}{chapter}
\counterwithout{equation}{chapter}

\AtAppendix{\counterwithin{section}{chapter}}
\AtAppendix{\counterwithin{theorem}{chapter}}
\AtAppendix{\counterwithin{proposition}{chapter}}
\AtAppendix{\counterwithin{corollary}{chapter}}
\AtAppendix{\counterwithin{assumption}{chapter}}
\AtAppendix{\counterwithin{setting}{chapter}}
\AtAppendix{\counterwithin{example}{chapter}}
\AtAppendix{\counterwithin{remark}{chapter}}
\AtAppendix{\counterwithin{definition}{chapter}}
\AtAppendix{\counterwithin{figure}{chapter}}
\AtAppendix{\counterwithin{table}{chapter}}
\AtAppendix{\counterwithin{algorithm}{chapter}}
\AtAppendix{\counterwithin{equation}{chapter}}
\AtAppendix{\counterwithin{lemma}{chapter}}



%% file: frontmatter.tex
\newlength{\drop}
\newcommand*{\titleTMB}{\begingroup
	\drop=0.1\textheight
	\centering
	\vspace*{10\baselineskip}
	{\large\scshape \theauthor}\\
	\vspace{2 cm}
	{\Huge \thetitle }\\[\baselineskip]
    \vspace{1.5 cm}
	{\Large\scshape phd thesis}\\[\baselineskip]\vspace{1cm}
	{\small \scshape this thesis has been submitted to the phd school of \\the faculty of science, University of Copenhagen}
	\vfill
	
	{\large\scshape Department of Mathematical Sciences \\ University of Copenhagen}\\[\baselineskip]
	{\small\scshape August 2024}\\ [\baselineskip]
	
	\vspace*{\drop}
	\endgroup}

\frontmatter
\begin{titlepage}
        \makebox[0pt]{%
        \raisebox{-25cm}[0pt][0pt]{%
        \hspace{396pt}
        \includegraphics{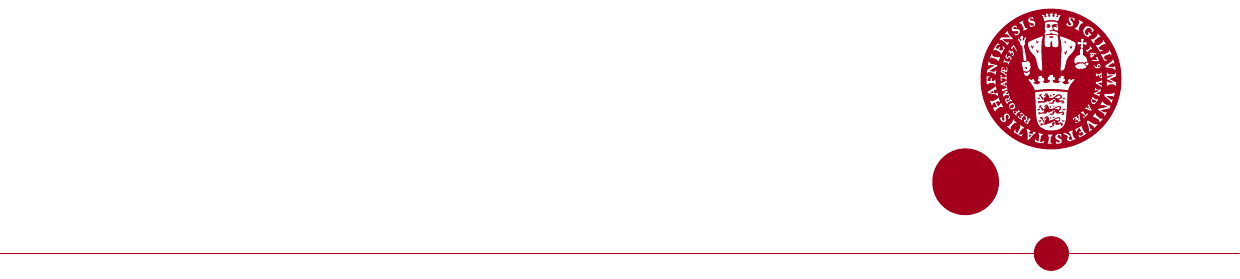}
        }}%
        
	\centering
	\titleTMB
\end{titlepage}

\theauthor \par
\texttt{amc@math.ku.dk}   \par
Department of Mathematical Sciences \par
University of Copenhagen \par
Universitetsparken 5 \par
2100 Copenhagen \par
Denmark

\vspace{2cm}

\begin{minipage}[t]{0.25\linewidth}
    \begin{flushleft}
    	{\textbf{Thesis title:}} \par
    	\, \vspace*{.3cm} \\
            \, \par 
        {\textbf {Supervisor:}} \par 
        \, \vspace*{.3cm} \\     
        {\textbf{Assessment}} \par 
        {\textbf{Committee:}} \vspace*{.3cm}\\ 
        \, \par
        \, \vspace*{.3cm} \\
        \, \par
        \, \vspace*{.3cm} \\
        {\textbf{Date of}} \par
        {\textbf{Submission:}}\vspace*{.3cm}\\
        {\textbf{Date of}} \par
        {\textbf{Defense:}}\vspace*{.3cm}\\
        {\textbf{ISBN:}} \\

    \end{flushleft}
\end{minipage}%
\begin{minipage}[t]{0.75\linewidth}
    \begin{flushleft}
    	\thetitle \par
        \, \vspace*{.3cm} \\

        Professor Niels Richard Hansen \par 
        University of Copenhagen \vspace*{.3cm} \\
        			
        Associate Professor Sebastian Weichwald (chair) \par
        University of Copenhagen \vspace*{.3cm} \\
        
        Professor Ingrid Van Keilegom \par
        KU Leuven \vspace*{.3cm}  \\
        
        Professor Stijn Vansteelandt \par
        Ghent University \vspace*{.3cm}\\
        August 31, \par 
        2024 \vspace*{.3cm}  \\
        November 15, \par
        2024 \vspace*{.3cm}  \\
        978-87-7125-233-0
    \end{flushleft}
\end{minipage}%

\vspace{2cm}
Chapter 1: $\copyright$ Christgau, A. M. \par
Chapter 2: $\copyright$ Institute of Mathematical Statistics. \par
Chapter 3 and 4: $\copyright$ Christgau, A. M. \& Hansen, N. R. 
\vfill

\noindent \textit{This thesis has been submitted to the PhD School of The Faculty of
Science, University of Copenhagen on 31 August 2024. It was supported by the Novo Nordisk Foundation (research grant NNF20OC0062897).}

\newpage
\begin{center}
    \vspace*{\fill}
    \textit{In loving memory of my grandparents, Ole and Sigrid.}
    \vspace*{\fill}    
\end{center}

\newpage
\input{preface_abstract}

\chapter{Contributions and Structure}
Chapter~\ref{chapter:introduction} is an introduction, which briefly motivates model-free statistical inference based on machine learning. 
The introduction is not meant to be a literature review, but rather an effort to build intuition around some key principles that permeate the rest of the thesis.
It is followed by $3$ chapters, each of which corresponds to a paper. For reference within this thesis, we give each paper an acronym, for example \LCM{}. Notation established in each chapter should be considered specific to the chapter where it is introduced. \\ \vspace{0.5 cm}
\nobibliography*
\noindent \textbf{Chapter~\ref{chapter:LCM}} discusses hypothesis testing for time-to-event responses using the \emph{Local Covariance Measure (LCM)} and corresponds to the paper:
\nobibliography*
\begin{enumerate}
    \item[\LCM{}] \citep{christgau2023nonparametric}.
        \bibentry{christgau2023nonparametric}. 
\end{enumerate}
 
 \vspace{0.5 cm}

\noindent \textbf{Chapter}~\ref{chapter:adjustment} formulates a general framework for the method of covariate adjustment, which leads to the \emph{Debiased Outcome-adapted Propensity Estimator (DOPE)}, and corresponds to the paper:
\begin{enumerate}
    \item[\DOPE{}] \citep{christgau2024efficient}. 
        \bibentry{christgau2024efficient}.
\end{enumerate}
    
\vspace{0.5 cm}

\noindent \textbf{Chapter}~\ref{chapter:ACM} introduces a measure of the conditional association between an exposure and a time-to-event, namely the \textit{Aalen Covariance Measure (ACM)}, and corresponds to the following paper in preparation:
\nobibliography*
\begin{enumerate}
    \item[\ACM{}] \citep{christgau2024assumption} 
        \bibentry{christgau2024assumption}. 
\end{enumerate}
\vspace{0.5 cm}

\newpage
During my PhD studies, I also worked on the following two publications, which, however, are not included in this thesis.
\begin{enumerate}
    \item \bibentry{christgau2023moment}. 

    \item \bibentry{bianchi2023generators}
\end{enumerate}

{
\hypersetup{linkcolor=black}
\tableofcontents
}

%% file: preface_abstract.tex
\addcontentsline{toc}{chapter}{Preface}
\chapter*{Preface}
This thesis is the culmination of three wonderful years from September 2021 to August 2024, under the supervision of Professor Niels Richard Hansen at the Department of Mathematical Sciences, University of Copenhagen.
The thesis features three distinct manuscripts on varied, yet related, topics. Each manuscript is a stand-alone scientific contribution and can be read independently. 
Small discrepancies may occur between the contents of a chapter and its associated manuscript. Any typographical or mathematical errors are my own responsibility.

\subsection*{Acknowledgments}
I am, first and foremost, grateful to Niels for giving me the opportunity to pursue a PhD, teaching me how to navigate academia, being an engaged co-author, and providing guidance, input, and feedback when necessary while allowing me to find my own path when appropriate. Most importantly, our collaboration over the past three years has been immensely enjoyable.

I extend my gratitude to my colleagues and now friends at the department. Coming to the office each day has been a delight, and our discussions about statistics and other topics have been truly entertaining, if not inspiring. Thank you for all the memories. In particular, I thank Anton for the numerous conversations, both serious and unserious, about mathematics, statistics, and everything beyond. 

I thank Becca and Jake for their hospitality and enthusiasm during my stay at the University of Chicago. 
I am especially grateful to the Sjursen family for warmly welcoming me into their family and making my time in Chicago even more enjoyable.

A heartfelt thanks goes to all my friends for their encouragement and good times we shared during my studies, with a special mention to my friends from folkeskolen, GHG, 4.~Indre, my study group from KU, and my PhD cohort. To avoid turning this into a lengthy acknowledgment, I will just say: You know who you are; thank you for your friendship. A special thanks goes to Christian, Mikkel, and Rasmus for being amazing roommates and for helping me maintain my sanity throughout my PhD studies.

Finally, but most importantly of all, I want to extend my deepest thanks to my entire family, especially Far and Mor, for their unwavering encouragement and support throughout my life. You have my deepest gratitude and love.
To my grandparents, who never saw the end of this adventure, but whose inquisitive minds and kind hearts shaped me profoundly: this thesis is dedicated to you.

\vspace{0.3cm} \par
\hfill \theauthor \par
\hfill August, 2024 

\newpage
\addcontentsline{toc}{chapter}{Abstract}
\begin{center}
  \normalfont\usekomafont{disposition} \Large Abstract 
\end{center}
This thesis contains a series of independent contributions to mathematical statistics, unified by a model-free perspective. The first chapter elaborates on how a model-free perspective can be used to formulate flexible methods that leverage prediction techniques from machine learning.
Mathematical insights are obtained from concrete examples, and these insights are generalized to principles that permeate the rest of the thesis.

The second chapter studies the concept of local independence, which describes whether the evolution of one stochastic process is directly influenced by another. To test local independence, we define a model-free parameter called the Local Covariance Measure (LCM). We formulate an estimator for the LCM, from which a test of local independence is proposed. We discuss how the size and power of the proposed test can be controlled uniformly and investigate the test in a simulation study.

The third chapter focuses on covariate adjustment, a method used to estimate the effect of a treatment by accounting for observed confounding. We formulate a general framework that facilitates adjustment for any subset of covariate information. We identify the optimal covariate information for adjustment and, based on this, introduce the Debiased Outcome-adapted Propensity Estimator (DOPE) for efficient estimation of treatment effects. An instance of DOPE is implemented using neural networks, and we demonstrate its performance on both simulated and real data.

The fourth and final chapter introduces a model-free measure of the conditional association between an exposure and a time-to-event, which we call the Aalen Covariance Measure (ACM). The ACM
serves as an assumption-lean generalization of the exposure coefficient in the Aalen additive hazards model. We develop a model-free estimation method and show that it is doubly robust, ensuring $\sqrt{n}$-consistency provided that the nuisance functions can be estimated with modest rates. A simulation study demonstrates the use of our estimator in several settings.

\newpage
\vspace{0.3cm}
\begin{center}
  \normalfont\usekomafont{disposition}\Large Sammenfatning
\end{center}
Denne Ph.D.-afhandling indeholder en række selvstændige bidrag inden for matematisk statistik med det til fælles, at de har et model-frit perspektiv. 
Det første afsnit uddyber, hvordan et model-frit perspektiv kan bruges til at formulere fleksible metoder, som benytter prædiktionsteknikker fra maskinlæring. Matematisk indsigt opnås via konkrete eksempler, og denne indsigt generaliseres til principper, der er gennemgående for resten af afhandlingen.

Det andet afsnit omhandler konceptet lokal uafhængighed, hvilket beskriver om udviklingen af en stokastisk proces er direkte påvirket af en anden proces. For at teste hypotesen om lokal uafhængighed definerer vi en model-fri parameter kaldet Local Covariance Measure (LCM). Vi beskriver en estimator af LCM, og baseret på denne foreslår vi nye test af lokal uafhængighed. 
Vi forklarer, hvordan størrelsen og styrken af disse test kan kontrolleres uniformt, og vi udforsker dem i et simulationsstudie.

Det tredje afsnit vedrører justering for kovariater, hvilket er en metode til at estimere effekten af en behandling ved at tage højde for observeret confounding. 
Vi udvikler en generel teori, der kan beskrive justering for enhver delmængde af information i kovariaterne. 
Vi identificerer den optimale information at justere for, og baseret på denne, introducerer vi Debiased Outcome-adapted Propensity Estimatoren (DOPE) for efficient estimering af behandlingseffekter. En version af DOPE implementeres med neurale netværk, og vi demonstrerer dens præstation på både simuleret og virkelig data.

Det fjerde og sidste afsnit introducerer en model-fri størrelse, der måler den betingede sammenhæng mellem en eksponering og en overlevelsestid, som vi kalder Aalen Covariance Measure (ACM). 
Man kan betragte ACM som en fleksibel generalisering af koefficienten for eksponering i en Aalen additive hazard model. 
Vi udvikler en model-fri estimationsmetode og viser at den er dobbelt robust, hvilket garanterer $\sqrt{n}$-konsistens under beskedne betingelser. 
Et simulationsstudie demonstrerer brugen af vores estimator i adskillige scenarier.

%% file: introduction.tex
\chapter{Introduction}\label{chapter:introduction}
This thesis delves into a range of statistical challenges, including event history analysis, statistical efficiency, hypothesis testing, effect estimation, and assumption-lean inference. Despite the variety of topics, a central theme is that the problems and methods are studied from a \textit{model-free} perspective.

To understand what this means, suppose we are given a sample consisting of $n$ observations $Z_1,\ldots,Z_n$, and that our objective is to infer properties about the mechanisms that generated this data. To simplify matters, we may start by assuming that the observations are independent and identically distributed (i.i.d.) according to a distribution $P$, referred to as the \emph{data-generating process (DGP)}. 
Whether this is reasonable assumption depends on the application, but it will be assumed throughout this thesis. 
Under the i.i.d. assumption, our objective can be reformulated in terms of inferring properties about $P$. 
To this end, it is convenient to let $Z$ denote an auxiliary independent observation with distribution $P$.

Before we can understand and appreciate a model-free approach to statistics, we first follow the thought process of Rhonda Fisher, a (fictional) model-based statistician. 

\begin{example}\label{ex:modelbased}
    Rhonda has obtained a sample in which each observation is an independent copy of the template observation $Z = (X,Y)\in \real^d \times \real$. Here $X=(X^1,\ldots,X^d)$ is a $d$-dimensional covariate and $Y$ is a real-valued response variable. 
    She knows that without further assumptions, she can hardly infer anything about the relationship between the covariate and the response. With this consideration, Rhonda begins by positing a linear model:
    \begin{align}\label{eq:linearmodel}
        Y = \alpha + \beta^\top X + \varepsilon,
        \qquad \text{with}
        \quad \varepsilon \sim \mathrm{N}(0,\sigma^2) \text{ and } \varepsilon\ind X,
    \end{align}
    where $\alpha\in \real$, $\beta=(\beta_1,\ldots,\beta_d)\in \real^d$, and $\sigma^2>0$ are unknown parameters. Conveniently, her favorite statistical software can `fit' the linear model \eqref{eq:linearmodel}.
    
    She now recalls that her primary objective is to describe the effect of $X^1$ on $Y$. In view of the linear model \eqref{eq:linearmodel}, she concludes that $\beta_1$ summarizes such an effect.
    Thus she decides to report the summary of the fitted model, which in particular includes an estimate of $\beta_1$ and a valid confidence interval. 
\end{example}

While Rhonda is a caricature, her line of reasoning is ubiquitous and applies to many applications using generalized linear models (GLMs) and Cox models, among other semiparametric models. 
Whereas Rhonda is `lucky' that her model conveniently summarizes her target of interest, others may choose a model with this objective in mind. 

In some applications, for example in physics, a thorough understanding of the data-generating process may justify the belief that the data follow a (semi)parametric model. However, in more complex systems, e.g., those analyzed in the physiological and social sciences, (semi)parametric models are often chosen for simplicity and convenience, rather than being a consequence of principled reasoning.

Because the data-generating process is often more complicated than desired, this leads to a difficult trade-off between model misspecification due to simplicity, and impracticality due to complexity.
Model misspecification may lead to bias in the results of the analysis, 
but a potentially more serious problem is the lack of interpretability of the target parameter for distributions that do not conform to the model. 
Other issues with model-based reasoning were pointed out in the seminal paper by \citet{breiman2001statistical} and are also discussed by \citet{hines2022demystifying}, among others.

 We will proceed to discuss how to define a model-free target parameter. 

\section{Model-free statistical objectives}\label{sec:modelfreeobj}
The importance of explicitly defining a target of interest at the onset of statistical analyses has been increasingly highlighted, particularly within the field of causal inference \citep{van2011targeted}. A starting point is to consider a general collection $\mathcal{P}$ of possible data-generating processes, that is, a collection of probability distributions over the sample space for a single observation. Although causal analyses require additional structure, this thesis will concentrate on observational quantities and their estimation.

We call a function an \textit{estimand} if its domain is $\mathcal{P}$. For example, if each observation is real-valued, a possible estimand is the population average outcome
\begin{align}\label{eq:mean}
    \mu \colon \mathcal{P} \to \real, \qquad    
    \mu(P) \coloneq \int z P(\mathrm{d}z) = \ex_P[Z].
\end{align}
The statistician is responsible for defining an estimand that effectively summarizes the target of interest.

Note that the estimand $\mu$ in \eqref{eq:mean} requires the mild condition that the possible DGPs are integrable. 
Most estimands require similar assumptions to be well-defined, such as the existence of moments, the positivity of certain probabilities, absolute continuity, or regularity conditions on conditional means.
These are properties that can be justified based on knowledge of the underlying DGP.

We say that an estimand is \textit{model-based} if it explicitly requires a parametrization of a component of the DGPs and reduces the dimension of the component.\footnote{This is, admittedly and intentionally, a vague definition.}
If an estimand is not model-based, we say that it is \textit{model-free}, also commonly referred to as nonparametric in other literature.

Returning to the setting of Example~\ref{ex:modelbased}, Rhonda defined her target parameter by restricting $\mathcal{P}$ to the linear Gaussian model
\begin{align*}
    \mathcal{P}^{\text{lin}} \coloneq \{
    P\in \mathcal{P}\given 
    \exists (\alpha,\beta,\sigma^2) &\in \real \times \real^d\times \real_{>0} \colon 
    \text{Equation } \eqref{eq:linearmodel}
    \text{ holds for }
    (X,Y)\sim P
    \}.
\end{align*}
Rhonda's target can be interpreted as the model-based estimand $\tau_{\text{mb}} \colon \mathcal{P}^{\text{lin}} \to \real$ that inverts the parametrization in $\beta_1$ (technically, this requires that $\mathcal{P}$ is initially restricted enough to ensure that $\beta_1$ is \textit{identifiable}). 
Critically, her statistical analysis becomes meaningless if the model is misspecified.
That is, for a distribution $P\notin \mathcal{P}^{\text{lin}}$, the estimand $\tau_{\text{mb}}$ is undefined and it is \textit{a priori} unclear what any estimator of $\tau_{\text{mb}}$, or equivalently, $\beta$, will target. In view of this, model-free estimands have been increasingly advocated over model-based ones.

It turns out, however, that the model-based estimand $\tau_{\text{mb}}$ can be extended to a model-free estimand.
For convenience, we let $X = (D,W)$, where $D\coloneq X^1$ denotes the covariate of interest, and $W \coloneq (X^2,\ldots,X^d)$ represents the other covariates. Then we define the estimand $\tau_{\text{mf}} \colon \mathcal{P} \to \real$ given by
\begin{align}\label{eq:populationRPO}
    \tau_{\text{mf}}(P) 
        \coloneq \frac{\ex_P[(D-\ex_P[D\given W])Y]}{\ex_P[(D-\ex_P[D\given W])^2]}
        = \frac{\ex_P[\cov_P(D,Y\given W)]}{\ex_P[\var_P(D\given W)]},
\end{align}
where we use a subscript $P$ to indicate that an operator is computed under the distribution $P$. 
This estimand unambiguously defines a target parameter without reference to a (semi)parametric model, and the conditions for being well-defined are clear from the expression, unlike the model-based estimand $\tau_{\text{mb}}$.\footnote{Both estimands are well-defined under the moment conditions $\ex_P[|DY|]<\infty$ and $\ex_P[D^2]<\infty$ and the positivity condition $\ex_P[\var_P(D\given W)]>0$.}

To see how $\tau_{\text{mf}}$ is a generalization of $\tau_{\text{mb}}$, consider a distribution $\mathsf{P}\in \mathcal{P}$ that follows the \textit{partially linear model}:
\begin{equation}\label{eq:plm}
    \ex_{\mathsf{P}}[Y\given D, W] = \beta_{\mathsf{P}} D + g_{\mathsf{P}}(W),
\end{equation}
where $\beta_{\mathsf{P}}\in \real$ and where $g_{\mathsf{P}}$ is an integrable function with respect to the distribution of $W$ under $\mathsf{P}$. This semiparametric model generalizes the linear model by allowing a more flexible dependency on the covariates $W$. By applying iterated expectations conditionally on $W$, we obtain that
\begin{align*}
    \tau_{\text{mf}}(\mathsf{P}) = \frac{\ex_{\mathsf{P}}[(D-\ex_{\mathsf{P}}[D\given W])(\beta_{\mathsf{P}} D + g_{\mathsf{P}}(W))]
    }{\ex_{\mathsf{P}}[(D-\ex_{\mathsf{P}}[D\given W])^2]}
    = \beta_{\mathsf{P}} \frac{\ex_{\mathsf{P}}[D^2-(\ex_{\mathsf{P}}[D\given W])^2]}{\ex_{\mathsf{P}}[(D-\ex_{\mathsf{P}}[D\given W])^2]}
    = \beta_{\mathsf{P}}.
\end{align*}
This shows that $\tau_{\text{mf}}$ reduces to the coefficient $\beta_{\mathsf{P}}$ in the (partially) linear model. Thus, we can more generally make sense of estimators that target $\tau_{\text{mf}}$, while maintaining the same interpretation within the (partially) linear model. 

There can exist multiple extensions of a model-based estimand. For example, another possible extension of $\tau_{\text{mb}}$ is the \textit{average partial effect}, defined by
\begin{align}\label{eq:ape}
    \tau_{\text{ape}}(P)
        = \ex_P\left[
        \frac{\partial}{\partial d}
        \ex_P[Y\given D=d,W] \vert_{d=D}
        \right].
\end{align}
Within the partially linear model \eqref{eq:plm}, we observe that the derivative inside the expectation equals $\beta$, and, consequently, $\tau_{\text{ape}}$ is also a model-free extension of $\tau_{\text{mb}}$. 
Arguably, the estimand $\tau_{\text{ape}}$ offers more direct interpretations for causal analyses than $\tau_{\text{mf}}$. However, $\tau_{\text{mf}}$ can also serve as a reasonable summary of conditional association, and estimators of $\tau_{\text{mf}}$ may be more robust and well-behaved than those of $\tau_{\text{ape}}$ \citep{lundborg2023perturbation}.



\section{Estimation}
In this section, we explore how $\sqrt{n}$-consistent estimation of an estimand can be achieved, even when dealing with complex or high-dimensional underlying DGPs. A substantial body of work exists on nonparametric estimation using \emph{influence functions}, with methods such as \emph{one-step estimation} \citep{pfanzagl1985contributions,bickel1998efficient} and \emph{Targeted Maximum Likelihood Estimation (TMLE)} \citep{van2011targeted}. This literature is well summarized in tutorial papers by, e.g., \citet{hines2022demystifying} and \citet{kennedy2022semiparametric}.

While the theory of influence functions is useful for defining an effective estimator, rigorously establishing the distributional properties of the estimator is a task that is typically done through case-specific analysis. In this introduction, we focus on this task through direct and concrete analyses, without invoking influence functions. 

Drawing from \citet{chernozhukov2018}, we begin with a motivating example, emphasizing an estimand-based perspective rather than a semiparametric (score-based) approach. After discussing the example, we will explore how these insights can be extended to more general contexts.

\subsection{An instructive example}\label{sec:motivatingexample}
Consider again the setting where the observed sample $Z_1,\ldots,Z_n$ consists of i.i.d. copies of $Z = (D,W,Y)\in \real\times \real^{d-1} \times \real$ with $Z\sim P \in \mathcal{P}$. Suppose our target of interest is the numerator of \eqref{eq:populationRPO}, i.e., the estimand $\rho \colon \mathcal{P}\to \real$ given by
\begin{align}\label{eq:psiexpression1}
    \rho(P) \coloneq \ex_P[(D-\ex_P[D\given W])Y].
\end{align}
This estimand has been studied for the purpose of quantifying and testing conditional (mean) independence \citep{shah2020hardness,lundborg2023modern}. We shall assume that $D,Y$, and $DY$ are square-integrable under each $P\in \mathcal{P}$, which in particular ensures that $\rho$ is well-defined.

Expression \eqref{eq:psiexpression1} combines two components of the distribution: the expectation operator, $\ex_P[\cdot]$, and the conditional mean function, $m_P(\cdot) \coloneq \ex_P[D\given W=\cdot]$.
According to the central limit theorem (CLT), the expectation operator is pointwise $\sqrt{n}$-approximated by the empirical mean operator, denoted $\mathbb{P}_n[\cdot]$ and given by $\mathbb{P}_n [f(Z)] \coloneq \frac{1}{n}\sum_{i=1}^n f(Z_i)$, for any function $f\in \mathcal{L}^2(P)$.

Distributional quantities that are used to compute the target estimand, but are not of primary interest, such as $m_P$, are called \textit{nuisance parameters}. 
Estimating conditional means can be done using various regression methods, including model-free prediction algorithms from machine learning (ML), which we discuss further in Section~\ref{sec:nuisance}. Assume, for now, that an estimate $\widehat{m}$ of $m_P$ is given. 
This suggests the estimator
\begin{align*}
    \hat{\rho}^{\text{plug-in}}
        \coloneq \mathbb{P}_n[(D-\widehat{m}(W))Y]
        = \frac{1}{n} \sum_{i=1}^n (D_i-\widehat{m}(W_i))Y_i.
\end{align*}
Estimators like $\hat{\rho}^{\text{plug-in}}$, which are derived by substituting nuisance estimates directly into the expression for the estimand, are often referred to as \textit{plug-in} estimators. However, it is important to note that such estimators are not necessarily of the form $\rho(\widehat{P})$ for an estimate $\widehat{P}\in\mathcal{P}$ of the entire distribution $P$, which is a commonly used representation for plug-in estimators.
As we will see shortly, there can be multiple valid expressions for the same estimand, resulting in different plug-in estimators.

To understand the behavior of this estimator, consider the following decomposition
\begin{align*}
    \sqrt{n}(\hat{\rho}^{\text{plug-in}} - \rho(P))
    &= \sqrt{n} \mathbb{P}_n [(D-\widehat{m}(W))Y] - \sqrt{n}\ex_P[(D-m_P(W))Y] \\
    &= \underbrace{\sqrt{n}(\mathbb{P}_n-\ex_P)[(D-m_P(W))Y]}_{\eqcolon U^{\text{plug-in}}}
        + \underbrace{\sqrt{n}\mathbb{P}_n[(m_P(W)-\widehat{m}(W))Y]}_{\eqcolon R^{\text{plug-in}}}.
\end{align*}
The term $U^{\text{plug-in}}$ is an \textit{oracle term} that describes the error of the plug-in estimator based on the true conditional mean $m_P(\cdot)$. The CLT asserts that $U^{\text{plug-in}}$ converges in distribution to a Gaussian distribution with mean zero. 

The term $R^{\text{plug-in}}$ is a sum of $n$ conditionally i.i.d. terms that are controlled by the estimation error of $m_P(\cdot)$.
As we discuss in Section~\ref{sec:nuisance}, we can generally expect the root mean squared error (RMSE) of $\widehat{m}(\cdot)$ to converge at a rate of $n^{-\alpha}$, where $\alpha\in (0,\frac{1}{2})$. 
Consequently, if $\ex_P[Y\given W] \neq 0$, then we can expect $R^{\text{plug-in}}$ to be of stochastic order $n^{\frac{1}{2}-\alpha}$, in which case
$$
    |\sqrt{n}(\hat{\rho}^{\text{plug-in}} - \rho(P))|
    \xrightarrow{P} \infty.
$$
This is, of course, undesirable and makes inference about $\rho(P)$ based on $\hat{\rho}^{\text{plug-in}}$ intractable.

Fortunately, the bias in $R^{\text{plug-in}}$ can be fixed by adding a correction term in the expression. Motivated by the fact that $R^{\text{plug-in}}$ becomes centered when $\ex_P[Y\given W]= 0$, we rewrite the estimand as 
\begin{align}\label{eq:psiexpression2}
    \rho(P) = \ex_P[(D-\ex_P[D\given W])(Y-\ex_P[Y\given W])].
\end{align}
This introduces a new nuisance parameter, $g_P(\cdot)\coloneq \ex_P[Y\given W=\cdot]$, which can be estimated similarly to $m_P(\cdot)$, e.g., using an ML algorithm, yielding another estimate $\widehat{g}(\cdot)$. 
This suggests the following estimator,
\begin{align}\label{eq:rhodouble}
    \hat{\rho}^{\text{double}}
    = \mathbb{P}_n[(D-\widehat{m}(W))(Y-\widehat{g}(W))]
    = \frac{1}{n}\sum_{i=1}^n(D_i-\widehat{m}(W_i))(Y_i-\widehat{g}(W_i)).
\end{align}
To understand why $\hat{\rho}^{\text{double}}$ might be preferable to $\hat{\rho}^{\text{plug-in}}$, we can decompose its estimation error as follows:
\begin{align}\label{eq:gcmdecomp}
    \sqrt{n}(\hat{\rho}^{\text{double}} - \rho(P))
    = U + R_m + R_g  + R_{mg}
\end{align}
where
\begin{align*}
    U &= \sqrt{n}(\mathbb{P}_n - \ex_P)[(D - m_P(W))(Y - g_P(W))], \\
    R_m &= 
        \sqrt{n}\mathbb{P}_n [(m_P(W)-\widehat{m}(W))(Y-g_P(W))], \\
    R_g &= 
        \sqrt{n}\mathbb{P}_n [(D-m_P(W))(g_P(W)-\widehat{g}(W))], \\
    R_{mg} &= 
        \sqrt{n}\mathbb{P}_n [(m_P(W)-\widehat{m}(W))(g_P(W)-\widehat{g}(W))].
\end{align*}
The first term, $U$, is again an oracle term that is asymptotically Gaussian, that is, $U\xrightarrow{d}\mathrm{N}(0,\sigma_P^2)$ with asymptotic variance $\sigma_P^2 \coloneq \var_P((D - m_P(W))(Y - g_P(W)))$. 
We will argue that, under suitable conditions, the error terms $R_m,R_g$, and $R_{mg}$ converge to zero in probability, in which case Slutsky's theorem asserts that
\begin{equation}\label{eq:doubleasymptotics}
     \sqrt{n}(\hat{\rho}^{\text{double}} - \rho(P)) \xrightarrow{d} \mathrm{N}(0,\sigma_P^2).
\end{equation}
There are several strategies for controlling a term like $R_m$, which we elaborate further in Section~\ref{sec:crossterms}. For simplicity, let us assume that $\widehat{m}(\cdot)$ is estimated on an auxiliary sample that is independent from the sample $(Z_i)_{i\in[n]}$ used in $\mathbb{P}_n$, where $[n]\coloneq \{1,\ldots, n\}$. 
Then, the summands in $R_m$ are independent conditionally on $\widehat{m}$ and $(W_i)_{i\in [n]}$, and their conditional means are zero since $\ex[Y_i\given \widehat{m}, (W_j)_{j\in [n]}] = \ex[Y_i\given W_i] = g_P(W_i)$ for $i\in [n]$. 
Note that the first equality crucially relies on $\widehat{m}$ being estimated on an independent sample. 

It follows that the conditional variance of $R_m$ is given by
\begin{align*}
    \var_P(R_m\given \widehat{m},(W_i)_{i\in [n]})
    = \frac{1}{n}\sum_{i=1}^n (\widehat{m}(W_i)-m_P(W_i))^2 \var_P(Y_i \given W_i).
\end{align*}
Thus, if there exists a bound $C>0$ such that $\var_P(Y \given W) \leq C$ almost surely and if $\widehat{m}$ is (out-of-sample) RMSE consistent for $m_P$, then it holds that $\var_P(R_m\given \widehat{m},(W_i)_{i\in [n]})\xrightarrow{P}0$. In this case, an application of Chebyshev's inequality, conditionally on $\widehat{m}$ and $(W_i)_{i\in [n]}$, lets us to conclude that $R_m \xrightarrow{P} 0$.

Since $R_m$ and $R_g$ are symmetric in $(m_P,D)$ and $(g_P,Y)$, the same argument can be applied to show that $R_g \xrightarrow{P} 0$ under analogous conditions on $\widehat{g}$ and $\var_P(D \given W)$.

For the product error term, $R_{mg}$, the Cauchy-Schwarz inequality yields that
\begin{align}\label{eq:producterrorbound}
    R_{mg}
    \leq \sqrt{n} \cdot \sqrt{\mathbb{P}_n[(\widehat{m}(W)-m_P(W))^2]} \cdot 
    \sqrt{\mathbb{P}_n[(\widehat{g}(W)-g_P(W))^2]}.
\end{align}
As a consequence, we see that $R_{mg}\xrightarrow{P}0$ if the product of the RMSE for $\widehat{m}(\cdot)$ and the RMSE for $\widehat{g}(\cdot)$ decays at an order of $o_P(n^{-1/2})$. 
When the rate requirements for nuisance estimators reduce to this condition, we say that the estimator exhibits \textit{rate double robustness}. 
There are several settings where this requirement can be achieved by nonparametric regression estimators, as we elaborate in Section~\ref{sec:nuisance}.

We proceed to discuss, from a more general perspective, conditions under which an estimator may exhibit (rate) double robustness.

\subsection{Double robustness and orthogonality}
\label{sec:orthogonality}
In this section, we generalize some of the insights obtained from the example in the preceding section, and we derive a condition required for rate double robustness. 

Suppose that $\tau \colon \mathcal{P} \to \real$ is an estimand that can be written of the form
\begin{align}\label{eq:estimandtoscore}
    \tau(P) = \frac{\ex_P[\psi(Z, \eta_0(P))]}{\ex_P[\varphi(Z, \eta_0(P))]},
\end{align}
where $\eta_0\colon \mathcal{P} \to \mathcal{E}$ is a nuisance parameter, and where $\varphi, \psi\colon \cZ \times \mathcal{E} \to \real$ are measurable maps, with $\mathcal{Z}$ denoting the sample space for $Z$. We assume that for all $\eta \in \mathcal{E}$ and $P\in \mathcal{P}$, the random variables $\psi(Z,\eta)$ and $\varphi(Z,\eta)$ are integrable under $P$ with $\ex_P[\varphi(Z,\eta)]\neq 0$.
Note that all the estimands considered so far have been of this form.
An estimand may be written of the form \eqref{eq:estimandtoscore} in more than one way, as was highlighted in Section~\ref{sec:motivatingexample}. 


Given a nuisance estimate $\widehat{\eta}\in \mathcal{E}$ of $\eta_0(P)$, we can define a plug-in estimator more generally as:
\begin{equation}\label{eq:pluginestimator}
    \widehat{\tau}(\widehat{\eta})
        =  \frac{\mathbb{P}_n [\psi(Z,\widehat{\eta})]}{\mathbb{P}_n [\varphi(Z,\widehat{\eta})]}
\end{equation}
Both $\hat{\rho}^{\text{plug-in}}$ and $\hat{\rho}^{\text{double}}$ were conceived in this way, so why does $\hat{\rho}^{\text{double}}$ exhibit rate double robustness while $\hat{\rho}^{\text{plug-in}}$ does not?

To answer this question, consider first an estimand that can be written of the form \eqref{eq:estimandtoscore} with $\varphi(\cdot,\cdot)\equiv 1$, and let $\eta_P \coloneq \eta_0(P)$. We can then decompose the general plug-in estimator \eqref{eq:pluginestimator} similarly to the example in Section~\ref{sec:motivatingexample}:
\begin{align*}
    \sqrt{n}(\widehat{\tau}(\widehat{\eta}) -\tau(P))
    = \sqrt{n}(\mathbb{P}_n - \ex_P)[\psi(Z,\eta_P)]
    + \sqrt{n}\mathbb{P}_n[\psi(Z,\widehat{\eta})-\psi(Z,\eta_P)]
\end{align*}
The first term is again an oracle term, whose behavior is governed by the CLT. 

To analyze the second term, we may appeal to a Taylor expansion. 
Assume that $\mathcal{E}$ is a convex set, and assume that the map $r\mapsto \psi(z,\eta_P + r(\eta-\eta_P))$ is two times continuously differentiable for all $(z,\eta)\in \mathcal{Z}\times \mathcal{E}$. Then, a Taylor expansion yields that
\begin{align}\label{eq:T1andT2}
    \sqrt{n}\mathbb{P}_n[\psi(Z,\widehat{\eta})-&\psi(Z,\eta_P)] \\
    &=
    \underbrace{\sqrt{n}\mathbb{P}_n[
        \partial_r \psi(Z,\eta_P + r(\widehat{\eta}- \eta_P) )\vert_{r=0} ]}_{
            \eqcolon T_1}
        + \underbrace{\sqrt{n}\mathbb{P}_n[R_2(Z,\eta_P, \widehat{\eta})]}_{
        \eqcolon T_2},
\end{align}
where $R_2(Z,\eta_P, \widehat{\eta})$ is a second-order remainder term that is left unspecified for now.

As we will discuss in Section~\ref{sec:nuisance}, we can generally expect a nuisance estimator to converge at a rate of $n^{-\alpha}$, where $\alpha\in(0,\frac{1}{2})$. To simplify the analysis, suppose that $\widehat{\eta}-\eta_P = n^{-\alpha} \Delta\eta$ for a fixed $\Delta\eta\in \mathcal{E}$. Although this might seem like an oversimplification, if the plug-in estimator is $\sqrt{n}$-consistent for general nuisance estimates, we can reasonably expect it to be $\sqrt{n}$-consistent under this simplifying assumption.

Then, by the chain rule,
$$
    T_1 = n^{\frac{1}{2}-\alpha}(\mathbb{P}_n-\ex_P)[\partial_r \psi(Z,\eta_P + r\Delta\eta)\vert_{r=0} ]
    + n^{\frac{1}{2}-\alpha}\partial_r \ex_P[ \psi(Z,\eta_P + r\Delta\eta)]\vert_{r=0}
$$
provided we can interchange derivative and expectation.
The first term vanishes in probability, whereas the absolute value of the second term is either diverging to infinity or equal to zero for all $n$.
In conclusion, we can expect $T_1$ to vanish in probability if for all $\eta\in \mathcal{E}$:
\begin{equation}\label{eq:pseudoNeyman}
    \partial_r \ex_P\left[
        \psi(Z,\eta_P + r(\eta- \eta_P) 
    \right]\vert_{r=0}
    =0
\end{equation}
We can interpret the condition \eqref{eq:pseudoNeyman} as the bias of the plug-in estimator not being sensitive to small estimation errors of the nuisance parameter $\eta_0$.

The condition \eqref{eq:pseudoNeyman} is similar to a condition at the core of \emph{Double Machine Learning (DML)} \citep{chernozhukov2018}.
Indeed, we can explicitly translate \eqref{eq:estimandtoscore} into the DML framework by considering $(\tau_P,\eta_P)\coloneq (\tau(P),\eta_0(P))$ as solutions to the \textit{score equation}
\begin{equation}\label{eq:DMLscore}
    \ex_P[S(Z;\theta,\eta))] = 0,
    \qquad \text{where} \quad 
    S(Z;\theta,\eta) \coloneq \theta \cdot \varphi(Z,\eta) - \psi(Z,\eta).
\end{equation}
Definition 2.1 in \citet{chernozhukov2018} states that the score $S$ satisfies the \emph{Neyman orthogonality condition at $(\tau_P,\eta_P)$} if
$
 \partial_r \ex_P\left[
        S(Z;\tau_P, \eta_P + r(\eta-\eta_P)) \right]   
$
exists for all $\eta\in \mathcal{E}$ and $r\in[0,1)$, and if
\begin{equation}
    \partial_r \ex_P\left[
         S(Z;\tau_P, \eta_P + r(\eta-\eta_P)) 
    \right] \vert_{r=0} = 0.
\end{equation}
This is essentially equivalent to \eqref{eq:pseudoNeyman} and 
$
\partial_r \ex_P\left[
\varphi(Z,\eta_P + r(\eta- \eta_P) 
\right]\vert_{r=0}
=0$.
Orthogonality is paramount in order for an estimator $\widehat{\tau}$ of the form \eqref{eq:pluginestimator} to satisfy that
$\sqrt{n}(\widehat{\tau} - \tau_P) \xrightarrow{d} \mathrm{N}(0,\sigma_P^2)$ for some $\sigma_P^2 >0$. 
It is straightforward to verify that the expression in \eqref{eq:psiexpression2} corresponds to a score that satisfies the orthogonality condition, whereas the expression \eqref{eq:psiexpression1} does not, unless $g_P\equiv 0$. Heuristically, this explains why $\widehat{\psi}^{\text{double}}$, which is based on \eqref{eq:psiexpression2}, exhibits rate double robustness. 

Orthogonality alone is not sufficient to ensure asymptotic normality, and \cite{chernozhukov2018} provide additional conditions that ensure this; see their Assumptions~3.1 and 3.2. However, these conditions essentially entail a reformulation of the condition that $T_2\xrightarrow{P}0$, cf. Assumption 3.2(c). 
Obtaining a general bound for $T_2$ with minimal assumptions is a complex task,
so $T_2$ is typically analyzed case-by-case. When the nuisance parameter $\eta_P = (\eta_{P,1},\ldots,\eta_{P,q})$ consists of several nuisance functions, it is often possible to establish a bound in the form:
\begin{align}\label{eq:generalsecondorderbound}
    T_2 \leq O_P\Big(
        \sqrt{n}\sum_{(i,j)\in \mathcal{S}} \|\widehat{\eta}_i - \eta_{P,i}\| \cdot \|\widehat{\eta}_j - \eta_{P,j}\|
    \Big),
\end{align}
for a set $\mathcal{S}\subseteq [q]\times [q]$. 
This represents a general form of rate double robustness, which is met, for example, when $\|\widehat{\eta}_i - \eta_{P,i}\| = o_P(n^{-1/4})$ for each $i\in[q]$.
The bound \eqref{eq:producterrorbound}, for instance, is of this form.

Keeping the orthogonality condition -- or analogous conditions in other frameworks -- in mind when constructing an estimator is crucial. 
This approach is taken in \LCM{}, \ACM{}, and \DOPE{}, though the details of the latter are more intricate. 



\subsection{Estimation of nuisance parameters using machine learning}\label{sec:nuisance}
The estimation procedures discussed up to this point have relied on the ability to estimate the nuisance parameter, and the analysis suggested that an estimation error of order $n^{-1/4}$ would be sufficient. 
Often, the nuisance parameter consists of functions given by conditional expectations, as in Section~\ref{sec:motivatingexample} with $\eta(P)=(m_P,g_P)$.

To discuss estimation of such nuisance functions, consider a generic regression setting where $Z = (X,Y) \in \real^d \times \real$
and suppose that $f_P(\cdot) \coloneq \ex_P[Y\given X=\cdot]$ is the nuisance parameter. 
Assuming that $\ex_P[Y^2]<\infty$, recall that 
$$
    f_P \in \argmin_{f \in \mathcal{F}} \ex_P[(Y-f(X))^2]
$$
for any collection of measurable functions $\mathcal{F}$ that contains $f_P$.
This connects the estimation of the conditional mean function, $f_P$, with the minimization of the mean squared prediction error over a sufficiently rich function class. 

When $f_P$ is known to be linear, this naturally leads to ordinary least squares, and the resulting estimator for $f_P$ converges at a rate of $O_P(n^{-1/2})$ under a few additional conditions.  
Similarly, other classical regression techniques based on semiparametric models that represent $\mathcal{F}$ based on a finite-dimensional representation, such as GLMs, generally yield $\sqrt{n}$-consistent estimators of $f_P$ when the model is correctly specified.

However, contemporary datasets and regression competitions, such as those hosted by Kaggle \citep{bojer2021kaggle}, indicate that ML methods, e.g. random forests, gradient boosting, and neural nets, routinely and significantly outperform classical regression methods in terms of prediction accuracy. 
This suggests that the function classes corresponding to classical methods are not rich enough to capture good approximations of the truth.

Considering larger nonparametric function classes can come with two challenges: it can make the resulting optimization problem difficult and usually leads to overfitting if implemented naively. Remarkable advancements in machine learning have resulted in powerful and efficient tools for handling optimization tasks with few assumptions on the function class. To prevent overfitting, these methods often depend on advanced (implicit) regularization techniques, making them difficult to analyze theoretically.

In view of the above, it is often left as a general assumption that the nuisance parameter can be estimated with a sufficiently fast rate. 
This assumption can be justified based on (i) the empirical observation that machine learning estimators perform well in practice and (ii) by comparing with achievable (minimax) rates for reasonable function classes. 
To elaborate on (ii), suppose, for example, that $f_P$ is known to be differentiable up to order $\floor{s}$, with partial derivatives of order $\floor{s}$ that are $(s-\floor{s})$-Hölder continuous, for some $s>0$. 
In the absence of additional constraints on $f_P$, the optimal uniform estimation rate is $O_P(n^{-\frac{s}{2s+d}})$, which can be attained using local polynomials \citep{Gyorfi2002,Tsybakov2009}. The key insight is not necessarily to use local polynomials, but to understand that effective estimators of $f_P$ can attain an estimation rate, e.g., of order $n^{-1/4}$, when $f_P$ has a smoothness of order $s>d/2$.

 While the discussion above revolves around conditional mean estimation, other estimation tasks might also be of interest. 
 Estimation of the average partial effect \eqref{eq:ape} with DML requires density and derivative estimation \citep{klyne2023average}, and in both \LCM{} and \ACM{} we require estimation of conditional hazards. Such estimation methods are discussed within each manuscript, see for example Section~\ref{sec:estimationofnuisance} in \LCM{}.

\subsection{Controlling cross-terms by cross-fitting}\label{sec:crossterms}
In this section we elaborate the discussion of how to control \textit{cross-terms} such as $R_m$ and $R_g$ from the decomposition \eqref{eq:gcmdecomp} of $\sqrt{n}(\hat{\rho}-\rho(P))$. More generally, in Section~\ref{sec:orthogonality} we encountered the `first-order' term
$$
    T_1 = \sqrt{n}\mathbb{P}_n[
        \partial_r \psi(Z,\eta_P + r(\widehat{\eta}- \eta_P) )\vert_{r=0}].
$$
By simplifying the analysis and considering the deterministic sequence of nuisance estimates $\widehat{\eta} = \eta_P + n^{-\alpha} \Delta \eta$ for a fixed $\Delta\eta \in \mathcal{E}$, we derived that the orthogonality condition \eqref{eq:pseudoNeyman} was essential for achieving $T_1\xrightarrow{P}0$. 
We now discuss how to proceed when $\widehat{\eta}$ is a stochastic estimate, e.g., obtained by machine learning as described in Section~\ref{sec:nuisance}.

When $\widehat{\eta}$ is estimated independently of the sample in $\mathbb{P}_n$, the orthogonality condition \eqref{eq:pseudoNeyman} implies that $T_1$ has mean zero conditionally on $\widehat{\eta}$. 
In many concrete examples, we can then control the conditional variance of $T_1$ by the error in nuisance estimation, in which case the convergence $T_1\xrightarrow{P} 0$ follows from consistency of the nuisance estimator and Chebyshev's inequality. 
For example, this was the case for the estimator $\hat{\rho}^{\text{double}}$ in \eqref{eq:rhodouble}, where it is straightforward to show that 
$T_1 = R_m + R_g$. 

Obtaining an independent nuisance estimate can be achieved by initially splitting the sample into two \textit{folds}, and then using first part for nuisance estimation and the second part for plug-in estimation. 
This comes at the cost of reducing the effective sample size. 
However, it is possible to recover full-sample efficiency by swapping the roles of the two folds and then aggregating the resulting estimates. 
More generally, we can perform \emph{$K$-fold cross-fitting} for any integer $K\geq 2$ as follows:

\begin{itemize}
    \item[i)] Partition the observation indices into $K$ disjoint folds:
    $[n] = I_1 \cup \cdots \cup I_K$.

    \item[ii)] For each $k\in [K]$, use the 
    holdout data $(Z_i)_{i\in [n]\setminus I_k}$ to compute an ML estimate
    $\widehat{\eta}_k$.

    \item[iii)] Compute the plug-in estimates on each fold and aggregate the estimates by
    \begin{equation}\label{eq:DML}
        \check{\tau}  \coloneq \frac{1}{K}\sum_{k=1}^K \underbrace{\frac{\sum_{i\in I_k} \psi(Z_i, \widehat{\eta}_k)}{
        \sum_{i\in I_k} \varphi(Z_i, \widehat{\eta}_k)}}_{\widehat{\tau}_k},
        \qquad \text{or} \qquad
        \tilde{\tau} \coloneq 
        \frac{
        \frac{1}{K}\sum_{k=1}^K \sum_{i\in I_k} \psi(Z_i, \widehat{\eta}_k)}{
        \frac{1}{K}\sum_{k=1}^K \sum_{i\in I_k} \varphi(Z_i, \widehat{\eta}_k)}.
    \end{equation}
\end{itemize}
The estimators $\check{\tau}$ and $\tilde{\tau}$ are known as the \textsc{dml}1 and \textsc{dml}2 estimators, respectively (see Definitions 3.1 and 3.2 in \citet{chernozhukov2018}). The estimators are asymptotically equivalent (under conditions) and they are equal when $\varphi$ is a constant, which was the case in~\eqref{eq:psiexpression2}. In small samples, however, \textsc{dml}2 may perform better due to the increased numerical stability achieved by pooling the estimates before division.

The \textsc{dml}1 estimator is considered in all the manuscripts in this thesis; therefore, we present a simple result for its asymptotics.
\begin{lem}[Asymptotics of \textsc{dml}1 estimator]\label{lem:DML1asymp}
    Let $\check{\tau}$ and $(\widehat{\tau}_k)_{k\in[K]}$ be the estimators defined in \eqref{eq:DML} and let $\sigma_P^2>0$.
    Suppose that each $k\in [K]$,
    \begin{equation*}
        \lim_{n\to \infty}n/|I_k| = K
        \qquad \text{and} \qquad
        \sqrt{|I_k|} \cdot (\widehat{\tau}_k - \tau_P)= U_k + R_k,
    \end{equation*}
    where $U_k$ and $R_k$ are variables satisfying that
    $U_k \xrightarrow{d} \mathrm{N}(0,\sigma_P^2)$
    and
    $R_k \xrightarrow{P} 0$ as $n\to \infty$.   
    Assume also the joint independence statement 
    $
        U_1 \ind \cdots \ind U_K
    $.
    Then it holds that 
    $$
        \sqrt{n}(\check{\tau} - \tau_P) \xrightarrow{d} \mathrm{N}(0,\sigma_P^2)
    $$ 
    as $n\to \infty$.
\end{lem}
\begin{proof}
 The conditions on $I_k, U_k$, and $R_k$ imply that
 \begin{equation*}
    \tilde{U}_k\coloneq 
     \frac{\sqrt{n}}{\sqrt{K|I_k|}} U_k \xrightarrow{d}
     \mathrm{N}(0,\sigma_P^2)
     \qquad \text{and} \qquad
     \tilde{R}_k \coloneq 
     \frac{\sqrt{n}}{\sqrt{K|I_k|}} R_k \xrightarrow{P}
     0.
\end{equation*}
By joint independence, it follows that $(\tilde{U}_1,\ldots, \tilde{U}_k) \xrightarrow{d} \mathrm{N}(0,\sigma_P^2)^{\otimes K}$ \citep[Theorem 2.8(ii)]{billingsley2013convergence}. By continuous mapping and the convolution property of the Gaussian distribution, we conclude that
\begin{align*}
    \sqrt{n}(\check{\tau} - \tau_P) 
    = \frac{1}{\sqrt{K}} \sum_{k=1}^K (\tilde{U}_k + \tilde{R}_k)
    \xrightarrow{d} \mathrm{N}(0,\sigma_P^2)
\end{align*}
as $n\to \infty$.
\end{proof}
In practice, we can apply the lemma with oracle terms $(U_k)$ corresponding to using the oracle nuisance estimate $\eta_P$. In this case, it simplifies the analysis of the \textsc{dml}1 estimator to that of a single sample split estimator $\widehat{\tau}_k$.

Note that the \textsc{dml}1 estimator indeed recovers full-sample efficiency, as it is scaled by $\sqrt{n}$ rather than $\sqrt{|I_k|}$.
In particular, the asymptotic distribution is independent of the number of folds $K$. 
The choice of $K$ will have a practical impact in smaller samples, as it will affect the remainder terms $(R_k)_{k\in [K]}$, which are controlled by the nuisance estimation error. Larger values of $K$ will allocate more data to the estimation of the nuisance parameter and are thus likely to result in smaller error terms. 
However, this also comes at the cost of higher computational demand since the nuisance estimator has to be fitted $K$ times. Choosing $K=4$ or $K=5$ might be reasonable in practice, cf. Remark 3.1. in \citet{chernozhukov2018}, which was also consistent with the findings in the simulation studies conducted for the work in this thesis.

\subsection{Alternatives to cross-fitting}
In the preceding discussion, we have covered how cross-fitting can be used to control the cross-term. This raises the question: 
\begin{equation}\label{eq:JPquestion}
    \textit{Is cross-fitting necessary?}
\end{equation}
There is no simple or general answer to this question, and it is instead typically discussed on a case-specific basis. In \LCM{} and \ACM{}, we found that cross-fitting was necessary, and for \DOPE{} we found that cross-fitting was neither necessary nor significantly detrimental.

To understand how a cross-term might be controlled without cross-fitting, consider again the term $R_m$ from \eqref{eq:gcmdecomp}, but now as a map 
\begin{align*}
    R_m(\bar{m})
    \coloneq \frac{1}{\sqrt{n}} \sum_{i=1}^n (m_P(W_i)-\bar{m}(W_i))(Y_i-g_P(W_i))
\end{align*}
which takes as input any measurable function $\bar{m}\colon \real^{d-1}\to \real$. Assuming for simplicity that $|Y|\leq C$ for some constant $C>0$, the CLT implies that for each $\bar{m}\in \mathcal{L}^2(P_W)$,
\begin{equation}\label{eq:Rmpointwiseasymp}
    R_m(\bar{m})\xrightarrow{d} \mathrm{N}(0,\nu_P(\bar{m}))
\end{equation}
as $n\to \infty$, where $\nu_P(\bar{m}) \coloneq \ex_P[(m_P(W)-\bar{m}(W))^2)(Y-g_P(W))^2]$.

The question \eqref{eq:JPquestion} is now whether $R_m(\widehat{m})\xrightarrow{P} 0$ can be established from a consistency condition of $\widehat{m}$, e.g., $\nu_P(\widehat{m}) \xrightarrow{P}0$.
Unfortunately, this is not generally true, but it can be concluded under additional assumptions on $\widehat{m}$. \\ \par

\noindent \textbf{Donsker class conditions:} Let $\mathcal{F}_m \subset \mathcal{L}^2(P_W)$ denote a function space such that $\widehat{m}\in \mathcal{F}_m$ almost surely. 
A sufficient condition for the implication 
\begin{equation}\label{eq:Donskerimplication}
    \nu_P(\widehat{m}) \xrightarrow{P} 0 \quad \implies \quad 
        R_m(\widehat{m})\xrightarrow{P} 0
\end{equation}
to hold is that the collection of random variables
\begin{align*}
    \{(m_P(W)-\bar{m}(W))(Y-g_P(W)) \colon \bar{m} \in \mathcal{F}_m\}
\end{align*}
is a Donsker class. See, for example, Lemma~19.24 in \citet{van2000asymptotic}. 
Informally, this Donsker class condition means that the weak convergence \eqref{eq:Rmpointwiseasymp} holds uniformly over $\bar{m}\in \mathcal{F}_m$.
There are numerous ways to establish the Donsker condition in itself, which 
typically involve technical bounds on the `size' of the function class $\mathcal{F}_m$. 
However, these conditions are unsuitable for high-dimensional settings \citep{chernozhukov2018} or when the nuisance estimates belong to a complex function class. \\ \par

\noindent \textbf{Algorithmic stability:}
More recently, \citet{chen2022debiased} have shown that cross-fitting can be avoided when the algorithm for producing the nuisance estimates is \textit{stable}. 
Informally, this means that the learned function does change significantly when the training data is perturbed, and there are several different definitions that formalize this idea.

For the cross-term $R_m$, Proposition S15 in \citet{lundborg2022projected} describes algorithmic stability conditions that guarantee $R_m(\widehat{m})\xrightarrow{d} 0$ without the need for cross-fitting. 
It is possible to verify that certain algorithms are stable, see, for example, \citet{hardt2016train,bousquet2002stability}, and recent work has shown that \textit{bagging} can increase the stability of any algorithm \citep{soloff2024bagging,soloff2024stability}.
These works, however, are not all working under the same notion of stability, 
and it is therefore unclear to what extent bagging can be used as a replacement for cross-fitting. \\ \par

\noindent Donsker class conditions and algorithmic stability have not been investigated in depth in this thesis. However, the analysis of cross-terms plays a central role in the asymptotic theory, and these terms are controlled by cross-fitting or the aforementioned alternative conditions. 
Since the answer to question \eqref{eq:JPquestion} appears to be yes in both \LCM{} and \ACM{}, these alternative conditions must somehow be violated. A deeper understanding of why these conditions are violated could provide valuable insights into our methods.

%% file: paper_lcm/main_paper.tex
\begin{abstract}
        Conditional local independence is an asymmetric independence relation among continuous time stochastic processes. It describes whether the evolution of one process is directly influenced by another process given the histories of additional processes, and it is important for the description and learning of causal
        relations among processes. We develop a model-free framework for 
        testing the hypothesis that a counting process is conditionally locally independent 
        of another process. To this end, we introduce a new functional parameter called the Local Covariance Measure (LCM), which quantifies deviations from the hypothesis. 
        Following the principles of double machine learning, we propose an estimator 
        of the LCM and a test of the hypothesis using nonparametric estimators and sample splitting or cross-fitting.
        We call this test the (cross-fitted) Local Covariance Test ((X)-LCT), and we show that its level and power can be controlled uniformly, provided that the nonparametric estimators are consistent with modest rates. We illustrate the 
        theory by an example based on a marginalized Cox model with
        time-dependent covariates, and we show in simulations that when double 
        machine learning is used in combination with cross-fitting, then 
        the test works well without restrictive parametric assumptions.
\end{abstract}

\section{Introduction}
Notions of how one variable influences a target variable are central to 
both predictive and causal modeling. Depending on the objective, the relevant 
notion of influence can be variable importance in a predictive model of the target, 
but it can also be the causal effect of the variable on the target. 
In either case, we can investigate influence conditionally on a third 
variable -- to quantify the added predictive value, the direct 
causal effect or the causal effect adjusted for a confounder. 
Our interests are in an asymmetric notion of direct influence among stochastic 
processes, which is not adequately captured by classical (symmetric)
notions of conditional dependence. The objective of this paper is 
therefore to quantify this notion of asymmetric influence 
and specifically to develop a new nonparametric test of the hypothesis
that one stochastic process does not directly influence another.  

The hypothesis we consider is that of \emph{local independence} -- a concept 
introduced by \cite{schweder1970} for Markov processes as a continuous time
formalization of the phenomenon that the past of one stochastic process does not 
directly influence the evolution of another stochastic process. Generalizations 
to other continuous time processes were given by \cite{aalen1987} and studied by 
\cite{commenges2009}, who systematically used the term \emph{conditional} local 
independence for the general concept. We will in this paper follow that convention 
whenever we want to emphasize the conditional nature of the local independence. 
We note that (conditional) local independence is a continuous time version of 
the discrete time concept of Granger non-causality \citep{Granger:1969, aalen1987}. 

To illustrate the concept of conditional local independence we will in 
this introduction consider an example involving three processes: $X$, $Z$ and $N$ -- 
see Figure \ref{fig:lig0}. 
The process $N$ is the indicator of death, $N_t = \one(T \leq t)$, for an 
individual with survival time $T$, and $X_t$ denotes the total pension 
savings of the individual at time $t$. The process $Z$ is a covariate 
process, e.g., health variables or employment status, that may directly affect 
both the pension savings and the survival time. This is indicated in Figure \ref{fig:lig0}
by edges pointing from $Z$ to $X$ and $N$. Edges pointing from 
$N$ to $X$ and $Z$ indicate that a death event directly affects both $X$ and $Z$ (which 
take the values $X_T$ and $Z_T$, respectively, after time $T$, see Section \ref{sec:introex}).

To define conditional local independence let 
$\mathcal{F}^{N,Z}_t = \sigma(N_s, Z_s; s \leq t)$ denote the filtration generated by 
the $N$- and $Z$-processes. The $\sigma$-algebra $\mathcal{F}_t^{N,Z}$ represents 
the information contained in the $N$- and the $Z$- processes before
time $t$. Informally, the process $N_t$ is conditionally locally independent of
the process $X_t$ given $\mathcal{F}_t^{N,Z}$ if $(X_s)_{s \leq t}$ does not add 
predictable information to $\mathcal{F}_{t-}^{N,Z}$ about the infinitesimal evolution of $N_t$.
For this particular example this means that the conditional hazard function 
of $T$ does not depend on $(X_s)_{s \leq t}$ given $\mathcal{F}_t^{N,Z}$.
In Figure \ref{fig:lig0} the hypothesis of interest, that $N_t$ is conditionally 
locally independent of $X_t$ given $\mathcal{F}^{N,Z}_t$, is represented by 
the lack of an edge from $X$ to $N$.

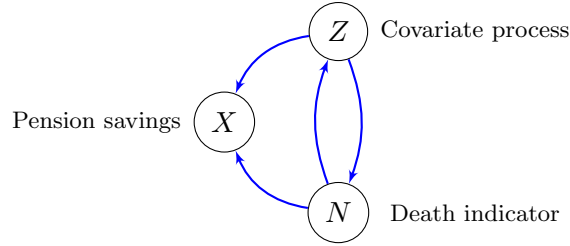
\begin{figure}
\centering
\begin{tikzpicture}[scale=0.6]
\tikzset{vertex/.style = {shape=circle,draw,minimum size=2em}}
\tikzset{vertexHid/.style = {shape=rectangle,draw,minimum size=2em}}
\tikzset{edge/.style = {->,> = latex', thick}}
\node[vertex] (Z) at  (1.5,2) {$Z$};
\node at (4.5,2) {\footnotesize{Covariate process}};
\node[vertex] (X) at  (-1,0) {$X$};
\node at (-3.8,0) {\footnotesize{Pension savings}};
\node[vertex] (T) at  (1.5,-2) {$N$};
\node at (4.5,-2) {\footnotesize{Death indicator}};

\draw[edge, bend left=20, color = blue] (Z) to (T);
\draw[edge, bend left, color = blue] (T) to (X);
\draw[edge, bend left=20, color = blue] (T) to (Z);
\draw[edge, bend right, color = blue] (Z) to (X);
\end{tikzpicture} 

\caption{Local independence graph illustrating a dependence 
structure among the three processes $X$, $Z$ and $N$. Here 
$N$ is the indicator of death for an 
individual, $X$ is their cumulative pension savings and $Z$ is a covariate process.
All nodes in this graph have implicit self-loops.
There is no edge from $X$ to $N$, which indicates that death is not directly influenced by 
pension savings. This can be formalized as $N$ being conditionally locally independent 
of $X$, which is the hypothesis we aim to test.
\label{fig:lig0}}
\end{figure}

A systematic investigation of algebraic properties of 
conditional local independence was initiated by 
\cite{didelez2006, Didelez2008, didelez2015}. She also introduced local independence 
graphs, such as the directed graph in Figure \ref{fig:lig0}, 
to graphically represent all conditional local independencies 
among several processes, and she studied the semantics of these graphs. 
This work was extended further by \cite{Mogensen:2020} to 
graphical representations of partially observed systems.
While we will not formally discuss local independence graphs, the problem 
of learning such graphs from data was an important motivation for us to 
develop a nonparametric test of conditional local
independence. A
constraint based learning algorithm of local independence graphs was given by
\cite{Mogensen:2018} in terms of a conditional local independence oracle,
but a practical algorithm requires that the oracle is replaced by 
conditional local independence tests.  

Another important motivation for considering conditional local independence 
arises from causal models.  With a
structural assumption about the stochastic process specification, a conditional
local independence has a causal interpretation \citep{aalen1987, aalen2012,
commenges2009}, and if the causal stochastic system is completely observed, a
test of conditional local independence is a test of no \emph{direct} causal
effect.
See also \citet{roysland2022graphical}, who use local independence graphs to formulate criteria for identification of causal effects in continuous-time survival models.
If the causal stochastic system is only partially observed, a
conditional local dependency need not correspond to a direct causal effect due
to unobserved confounding, but the projected local independence graph, as
introduced by \cite{Mogensen:2020}, retains a causal interpretation, and its
Markov equivalence class can be learned by conditional local independence
testing. In addition, within the framework of structural nested models, testing
the hypothesis of no \emph{total} causal effect can also be cast as a test of
conditional local independence \citep[Thm. 9.2]{lok2008}. 

To appreciate what conditional local independence means -- and, in particular, 
what it does not mean -- it is useful to compare with ordinary conditional independence. 
In our example, $N_t$ is conditionally locally independent of $X_t$ given 
$\mathcal{F}^{N,Z}_t$, but this implies neither that
\mbox{$N \ind X \mid Z$} (as processes), nor that 
$N_t \ind X_t \mid \mathcal{F}_t^Z$. In fact, 
these conditional independencies cannot hold in this example where 
$X_t = X_T$ for $t \geq T$ -- except in special cases such as $T$ being 
a deterministic function of $Z$. Theorem 2 in \cite{Didelez2008}
gives a sufficient condition for $N_t \ind X_t \mid \mathcal{F}_t^Z$
to hold in terms of the local independence graph, but this condition 
is also not fulfilled by the graph in Figure \ref{fig:lig0} due to the edge from 
$N$ to $X$. Moreover, conditional local independence is in general also different
from the baseline conditional independence $T \ind X_0 \mid Z_0$ unless 
both $X$ and $Z$ are time-independent, see Section \ref{sec:copula}. 
In Section~\ref{sec:survmodels} in the supplementary material \citep{christgau2023supplement}, we elaborate further upon the connection to semiparametric survival models. 
\cite{Didelez2008} argues that $N_t$ being conditionally locally independent of $X_t$ given 
$\mathcal{F}^{N,Z}_t$ heuristically means that $N_t \ind 
\mathcal{F}^X_{t-} \mid \mathcal{F}_{t-}^{N,Z}$, but this is technically problematic 
in continuous time. If $T$ has a continuous distribution, then for any fixed $t$, 
$N_t = N_{t-}$ almost surely, whence $N_t$ is 
almost surely $\mathcal{F}_{t-}^{N,Z}$-measurable and conditionally independent of 
anything given $\mathcal{F}_{t-}^{N,Z}$. This heuristic can thus not be used
to formally define conditional local independence in continuous time. 
See instead the formal Definition 2 by \cite{Didelez2008} or our Definition 
\ref{dfn:cli}.

Several examples from health sciences given by \cite{Didelez2008} demonstrate  
the usefulness of conditional local independence for multivariate event systems, and more recent 
attention to event systems in the machine 
learning community \citep{Zhou:2013, Xu:2016, Achab:2017, Bacry:2018, Cai:2022} 
testifies to the relevance of conditional local independence. This line of research 
relies primarily on the linear Hawkes process model, which is 
effectively used to infer local independence graphs  
-- sometimes even interpreted causally. 
The Hawkes model is attractive because conditional local independencies can be 
inferred from corresponding kernel functions being zero -- and statistical tests can readily be based on parametric or nonparametric estimation of kernels. 

A less attractive property of the Hawkes model is that it is not closed under marginalization.
As with any model based statistical test, the validity of the test is jeopardized by
model misspecification, hence even within a subsystem of a linear Hawkes process, 
a test of conditional local independence based on a Hawkes model may be invalid. 

The challenges resulting from model misspecification and marginalization 
is investigated further in Sections \ref{sec:introex} and \ref{sec:simulations}
based on an extension of our introductory 
example and Cox's survival model. Both the Hawkes model and the Cox model 
illustrate that conditional local independence might be expressed 
and tested within a (semi-)parametric model, but model misspecification
makes us question the validity of any such model based test. Thus there is a need for a nonparametric test of 
the hypothesis of conditional local independence. Moreover, since we cannot 
translate the hypothesis into an equivalent hypothesis about 
classical conditional independence, we cannot directly use existing 
nonparametric tests, such as the GCM \citep{shah2020hardness} or 
the GHCM \citep{lundborg2021conditional}, of conditional independence. 

We propose a new nonparametric test when the target process $N$ is a counting process 
and $X$ is a real valued process, and where the hypothesis is that $N$ is 
conditionally locally independent of $X$ given a filtration $\mathcal{F}_t$. 
In the context of the introductory example, $\mathcal{F}_t = \mathcal{F}_t^{N,Z}$. We 
consider a counting process target primarily because the theory of conditional 
local independence is most complete in this case, but generalizations are possible -- we 
refer to the discussion in Section~\ref{sec:dis}. Within our framework we base our
test on an infinite dimensional parameter, which we call the 
Local Covariance Measure (LCM). It is a function of time, 
which is constantly equal to zero under the hypothesis. 
Our main result is that the LCM can be estimated by using the ideas of double
machine learning \citep{chernozhukov2018} in such a way that the estimator 
converges uniformly at a $\sqrt{n}$-rate to a mean zero Gaussian martingale under the hypothesis of conditional local independence. We use the LCM to develop the (cross-fitted)
Local Covariance Test ((X)-LCT), for which we derive uniform level and power 
results.

\subsection{Organization of the paper}
In Section \ref{sec:setup} we introduce the general framework for formulating the hypothesis of
conditional local independence. This includes the introduction in 
Section \ref{sec:setup-hyp} of an abstract residual
process, which is used to define the LCM as a functional target parameter indexed by time. The LCM equals the zero-function under the hypothesis of conditional 
local independence, and to test this hypothesis we 
introduce an estimator of the LCM in Section \ref{sec:statistic}. The estimator is a stochastic process,
and we describe how sample splitting is to be used for its computation via the
estimation of two unknown components. 

In Section \ref{sec:interpretations} we give interpretations of the LCM and its estimator. 
We show that the LCM estimator is a Neyman
orthogonalized score statistic in Section \ref{sec:neyman}, and in Section \ref{sec:copula} we relate LCM to the partial copula when $X$ is time-independent.

In Section \ref{sec:asymptotics} we state the main results of the paper. We
establish in Section \ref{sec:LCMasymptotics} that the LCM estimator generally approximates the LCM with an error of order $n^{-1/2}$. Under the hypothesis of conditional local independence, we show that the (scaled) LCM estimator converges weakly to a mean zero Gaussian martingale. The estimator requires a model of the
target process $N$ as well as the process $X$ conditionally on
$\mathcal{F}_t$ to achieve the orthogonalization at the core of double machine
learning. The model of $X$ is in this paper expressed indirectly in terms of the 
residual process, and we show that if we can learn the residual process at 
rate $g(n)$ and the model of $N$ at rate $h(n)$ such that $g(n), h(n) \to 0$ and $\sqrt{n} g(n) h(n) \to 0$ for $n \to \infty$ then we achieve a $\sqrt{n}$-rate convergence 
of the LCM estimator. We also show that the variance function of the Gaussian 
martingale can be estimated consistently, and we give a general result
on the asymptotic distribution of univariate test statistics based on the 
LCM estimator. All asymptotic results are presented in the framework of 
uniform stochastic convergence.

Section \ref{sec:lct} gives explicit examples of univariate test statistics,
including the Local Covariance Test based on the normalized supremum of 
the LCM estimator. Its asymptotic distribution is derived and we present 
results on uniform asymptotic level and power. In 
Section \ref{sec:cf} we present the generalization from the sample split estimator to the cross-fit estimator. Though this estimator and the corresponding cross-fit Local Covariance Test (X-LCT) are a bit more involved to compute and analyze, X-LCT 
is more powerful and thus our recommended test for practical usage. 

The survival example from the introduction is used and elaborated upon throughout the paper. We introduce a Cox model in terms of the 
time-varying covariate processes, and we report in Section \ref{sec:simulations} the results from a simulation study based on this model.

The paper is concluded by a discussion in Section \ref{sec:dis}.

The supplementary material \citep{christgau2023supplement}, henceforth referred to as the supplement, 
consists of Sections~\ref{sec:proofs} through \ref{sec:extrafigs} and contains:
proofs of results in this paper \eqref{sec:proofs}; 
definitions and results on uniform asymptotics \eqref{app:UniformAsymptotics};  
a uniform version of Rebolledo's martingale CLT \eqref{sec:fclt}; 
an overview of achievable rate results for estimation of nuisance parameters that enter into the LCM estimator \eqref{sec:estimationofnuisance}; 
a comparison with semiparametric survival models \eqref{sec:survmodels}; 
details on Neyman orthogonality \eqref{sec:orthodetails};
and additional results from the simulation study \eqref{sec:extrafigs}.

\section{The Local Covariance Measure} \label{sec:setup}
In this section we present the general framework of the paper,
we define conditional local independence and we introduce the Local Covariance 
Measure as a means to quantify deviations from conditional local independence. 
In Section~\ref{sec:statistic} we outline how the Local Covariance Measure 
can be estimated using double machine learning and sample splitting. 
We illustrate the central concepts and methods by an 
example based on Cox's survival model with time-varying covariates. 

We consider a counting process $N = (N_t)$ and another 
real value process $X = (X_t)$, 
both defined on the probability space $(\Omega, \mathbb{F}, \mathbb{P})$. All 
processes are assumed to be defined on a common compact time interval. 
We assume, without loss of generality, that the time interval is $[0,1]$. We will 
assume that $N$ is adapted w.r.t. a right continuous and complete filtration 
$\mathcal{F}_t$, and we denote by $\mathcal{G}_t$ the right continuous 
and complete filtration generated by $\mathcal{F}_t$ and $X_t$. We  
assume throughout that $X$ is \lc{} (that is, has sample paths that are continuous from the left
and with limits from the right), which will ensure bounded sample paths and 
that the process is $\mathcal{G}_t$-predictable.

In the survival example of the introduction, $N_t = \one(T \leq t)$ is the 
indicator of whether death has happened by time $t$, and there can only be one 
event per individual observed. Furthermore, $\mathcal{F}_t = \mathcal{F}_t^{N, Z}$
and $\mathcal{G}_t = \mathcal{F}_t^{N, X, Z}$. 
Our general setup works for any counting process, thus it allows for recurrent 
events and censoring, and the filtration $\mathcal{F}_t$ can 
contain the histories of any number of processes in addition to the history of 
$N$ itself.

\subsection{The hypothesis of conditional local independence} \label{sec:setup-hyp}

The counting process $N$ is assumed to have an $\mathcal{F}_t$-intensity 
$\lambda_t$, that is, $\lambda_t$ is $\mathcal{F}_t$-predictable and with 
$$\Lambda_t = \int_0^t \lambda_s \mathrm{d}s$$
being the compensator of $N$,  
\begin{equation} \label{eq:basicmg}
M_t = N_t - \Lambda_t
\end{equation}
is a local $\mathcal{F}_t$-martingale. Within this framework we can define the 
hypothesis of conditional local independence precisely.

\begin{definition}[Conditional local independence] \label{dfn:cli}
We say that $N_t$ is conditionally 
locally independent of $X_t$ given $\mathcal{F}_t$ if 
the local $\mathcal{F}_t$-martingale $M_t$ defined by \eqref{eq:basicmg} is also a 
local $\mathcal{G}_t$-martingale.
\end{definition}

For simplicity, we may also refer to this hypothesis as \emph{local independence}
and write 
\begin{equation} \label{eq:H0}
H_0: M_t = N_t - \Lambda_t \text{ is a local } \mathcal{G}_t \text{-martingale}.
\end{equation}
As argued in the introduction, the hypothesis of local
independence is the hypothesis that observing $X$ on $[0,t]$ does not add any
information to $\mathcal{F}_{t-}$ about whether an $N$-event will happen in an
infinitesimal  time interval $[t, t + \mathrm{d}t)$. Definition \ref{dfn:cli} 
captures this interpretation by requiring that the $\mathcal{F}_t$-compensator, 
$\Lambda$, of $N$ is also the $\mathcal{G}_t$-compensator. Thus, $\lambda$ is also the $\mathcal{G}_t$-intensity under $H_0$. 

If $N$ has $\mathcal{G}_t$-intensity $\blambda$, the innovation theorem, Theorem II.T14 
in \cite{Bremaud:1981}, gives that the predictable projection $\lambda_t = \ex(\blambda_t \mid \mathcal{F}_{t-})$ is the (predictable) $\mathcal{F}_t$-intensity. 
Local independence follows if $\blambda$ is $\mathcal{F}_t$-predictable.  
Intensities are, however, only unique almost surely, and we can have local independence 
even if $\blambda$ is not \emph{a priori} $\mathcal{F}_t$-predictable but have 
an $\mathcal{F}_t$-predictable version. When 
$N$ has $\mathcal{G}_t$-intensity $\blambda$, 
$H_0$ is thus equivalent to $\blambda$ having an $\mathcal{F}_t$-predictable version.
We find Definition \ref{dfn:cli} preferable because it directly 
gives an operational criterion for determining whether $N$ has an 
$\mathcal{F}_t$-predictable version of a $\mathcal{G}_t$-intensity.

    \begin{remark}[Censoring]\label{rem:censoring}
        Suppose that the data is censored such that $(N_t,X_t,\mathcal{F}_t) 
        = (N_{t\wedge C}^*,X_{t\wedge C}^*,\mathcal{F}_{t\wedge C}^*)$,
        where $(N^*,X^*,\mathcal{F}^*)$ are uncensored data and where $C$ is the censoring time. 
        The hypothesis regarding the uncensored data,
        \begin{center}
            $H_0^*$:  $N_t^*$ is locally independent of $X_t^*$ given $\mathcal{F}_t^*$,
        \end{center} 
         might then be the hypothesis of interest.
        If $\one(C \geq t)$ happens to be $\mathcal{F}_t^*$-predictable, it is straightforward to show that $H_0^*$ implies $H_0$, and consequently a test of $H_0$ is also a test of $H_0^*$. However, $\mathcal{F}_t^*$ may not \emph{a priori} contain information about the censoring process. Suppose instead that the common condition of \textit{independent censoring} \citep{AndersenBorganGillKeiding:1993} holds, which is equivalent to $N_t^*$ being locally independent of $\mathcal{C}_t \coloneqq \one(C\leq t)$ given $\mathcal{G}_t^*$ \citep{roysland2022graphical}. Then $H_0^*$ implies that $N_t$ is locally independent of $X_t$ given $\mathcal{F}_t \vee \mathcal{F}_t^{\mathcal{C}}$. Thus, in order to test $H_0^*$, we 
        replace $\mathcal{F}_t$ by the enlarged filtration $\mathcal{F}_t \vee \mathcal{F}_t^{\mathcal{C}}$ and proceed \textit{mutatis mutandis} with testing $H_0$
        using the observed data.
    \end{remark}

Since $X$ is assumed \lc, and thus especially $\mathcal{G}_t$-predictable, the stochastic
integral 
\begin{equation} \label{eq:stochint}
\int_0^t X_{s} \mathrm{d}{M}_s,
\end{equation}
is under $H_0$ a local $\mathcal{G}_t$-martingale. A test could be based on 
detecting whether \eqref{eq:stochint} is, indeed, a local martingale.
We will take a slightly different approach where 
we replace the integrand $X$ by a residual process as defined below.
We do so for two reasons. First, to achieve a $\sqrt{n}$-rate via double 
machine learning we need the integrand to fulfill \eqref{eq:residualorthogonal} below. 
Second, other choices of integrands than $X$ could potentially lead to 
more powerful tests. 
\begin{definition}[Residual Process] \label{dfn:rescond}
A residual process $G = (G_t)_{t\in[0,1]}$ of $X_t$ given $\mathcal{F}_t$ is a \lc{} 
stochastic process that is $\mathcal{G}_t$-adapted and satisfies 
\begin{align}\label{eq:residualorthogonal}
    \ex(G_t \mid \mathcal{F}_{t-}) = 0
\end{align}
for $t\in [0,1]$.
\end{definition}
The geometric interpretation is that the residual process evolves such that $G_t$ is orthogonal to $L_2(\mathcal{F}_{t-})$ within $L_2(\mathcal{G}_{t-})$ at each time $t$. One 
obvious residual process is the \emph{additive residual process} given by
\begin{align*}
    G_t =  X_t - \Pi_t = X_t - \ex(X_{t} \mid \mathcal{F}_{t-}),
\end{align*}
where $\Pi_t = \ex(X_t \mid \mathcal{F}_{t-})$ denotes the 
predictable projection of the \lc{} process $X_{t}$, 
 see Theorem VI.19.6 in \citet{rogers2000}. The additive residual projects $X_t$ onto the orthogonal complement of $L_2(\mathcal{F}_{t-})$, but this may not necessarily remove all $\mathcal{F}_{t}$-predictable information from $X_t$. An alternative 
choice that does so under sufficient regularity conditions is the \emph{quantile residual process} given by
$$
    G_t = F_t(X_t) - \frac{1}{2},
$$
where $F_t$ is the conditional distribution function given by $F_t(x) = \mathbb{P}(X_t \leq x \mid \mathcal{F}_{t-})$.
The quantile residual process satisfies \eqref{eq:residualorthogonal} provided that $(t,x) \mapsto F_t(x)$ is continuous. In Section~\ref{sec:neyman} we discuss additional transformations of  
$X$ that can also be applied before any residualization procedure.

We will formulate the general results in terms of an abstract residual process,
but we focus on the additive residual process in the examples. Any non-degenerate residual process will contain a predictive model of (aspects of) $X_t$ given $\mathcal{F}_{t-}$ in order to satisfy \eqref{eq:residualorthogonal}. We use $\widehat{G}_t$ to
denote the residual obtained by plugging in an estimate of that predictive
model. For the additive residual process, the predictive model is $\Pi_t$ and $\widehat{G}_t =
X_t - \widehat{\Pi}_t$. For the quantile residual process, the predictive model is $F_t$ and $\widehat{G}_t = \widehat{F}_t(X_t)- \frac12$.

We can now define our functional target parameter of interest, which we call 
the \emph{Local Covariance Measure}.
\begin{definition}[Local Covariance Measure] \label{dfn:lcm}
With $G_t$ a residual process, define for $t \in [0,1]$
\begin{equation} \label{eq:gamma}
    \gamma_t = \ex\left( I_t \right),
    \qquad
    \text{where}
    \quad
    I_t = \int_0^t G_s \mathrm{d} M_s,
\end{equation}
whenever the expectation is well defined. We call the function $t \mapsto \gamma_t$ 
the Local Covariance Measure (LCM).
\end{definition}
The following propositions illuminate how $\gamma$ relates to the null hypothesis of $N_t$ being conditionally locally independent of $X_t$ given $\cF_t$.

\begin{prop} \label{prop:cli-mg} 
    Under $H_0$, the process $I = (I_t)$ is a local $\mathcal{G}_t$-martingale with $I_0 = 0$. If $I$ is a martingale, then $\gamma_t = 0$ for $t \in [0,1]$.
\end{prop}

To interpret $\gamma$ in the alternative, we assume that $N$ has $\cG_t$-intensity $\blambda$. 
\begin{prop} \label{prop:gammaalternative}
    If 
    $
        \int_0^1 \ex(|G_s| (\blambda_s + \lambda_s)) \mathrm{d}s
        <\infty
    $, then for every $t\in[0,1]$,
    \begin{align*}
        \gamma_t = \int_0^t \cov(G_s, \blambda_s - \lambda_s) \mathrm{d}s.
    \end{align*}
    In particular, $\gamma$ is the zero-function if and only if $\cov(G_s, \blambda_s - \lambda_s) = 0$ for almost all $s\in [0,1]$.
\end{prop}

We note that under $H_0$, the condition $\int_0^1 \ex(|G_s| \lambda_s) \mathrm{d}s < \infty$ 
is sufficient to ensure that $I$ is a martingale and $\gamma_t = 0$ for 
all $t \in [0,1]$. By Proposition \ref{prop:gammaalternative}, the LCM quantifies 
deviations from $H_0$ in terms of the covariance between the residual process and the difference of the $\cF_{t}$- and $\cG_t$-intensities. To this end, note that if $X$ happens to be $\mathcal{F}_t$-adapted, then 
$\mathcal{G}_t = \mathcal{F}_t$ and $N$ is trivially locally independent of
$X$. 
The hypothesis of local independence is only of interest when
$\mathcal{G}_t$ is a strictly larger filtration than $\mathcal{F}_t$, that is,
when $X$ provides information not already in $\mathcal{F}_t$. 

For the additive residual process, where $G_t = X_t - \Pi_t$, 
\begin{align*}
    \gamma_t & 
    = \ex\left(\int_0^t G_s \mathrm{d} M_s \right) 
    = \ex\left(\int_0^t X_s \mathrm{d} M_s\right) - 
    \ex\left(\int_0^t \Pi_s \mathrm{d} M_s \right)
\end{align*}
provided that the expectations are well defined. 
Since the predictable projection $\Pi_t$ has a \lc{} version and is  
$\mathcal{F}_t$-predictable, and since $M_t$ is a local $\mathcal{F}_t$-martingale,
$\int_0^t \Pi_s \mathrm{d} M_s$
is  a local $\mathcal{F}_t$-martingale. If it is a martingale, it is a 
mean zero martingale, and
\begin{equation} \label{eq:gammarep}
\gamma_t = \ex\left(\int_0^t X_s \mathrm{d} M_s\right) 
=
\ex \left( \sum_{\tau \leq t: \Delta N_\tau = 1} X_\tau 
- 
\int_0^t X_s \lambda_s \mathrm{d} s \right).
\end{equation} 
The computation above shows that the additive residual process defines the same 
functional target parameter $\gamma_t$ as the stochastic integral \eqref{eq:stochint} would. It is, however, the representation of $\gamma_t$ as the
expectation of the residualized stochastic integral that will allow us
to achieve a $\sqrt{n}$-rate of convergence of the 
estimator of $\gamma_t$ in cases where the estimator of $\lambda_t$ converges at a slower rate. 

\subsection{A Cox model with a partially observed covariate process} \label{sec:introex}
To further illustrate the hypothesis of conditional local independence and 
the Local Covariance Measure we consider an example based on Cox's
survival model with time dependent covariates. This is an extension of 
the example from the introduction with $T$ being 
the time to death of an individual, and with $X$ and $Z$ being time-varying processes. 
There is, moreover, one additional time-varying process $Y$ in the full model. 

An interpretation of the processes is as follows:
\begin{align*}
    X & = \text{Pension savings} \\
    Y & = \text{Blood pressure} \\
    Z & = \text{BMI} 
\end{align*}
Periods of overweight or obesity may influence blood pressure in the long term,
and due to, e.g., job market discrimination, high BMI could influence pension
savings negatively. Death risk is influenced directly by BMI and blood pressure
but not the size of your pension savings. Figure \ref{fig:lig} illustrates 
two possible dependence structures among the three processes and the 
death time as local independence graphs, and we will use these two graphs 
to discuss the concept of conditional local independence of pension 
savings on time to death.

\begin{figure}
\centering
\begin{tabular}{cc}
\begin{tikzpicture}[scale=0.6]
\tikzset{vertex/.style = {shape=circle,draw,minimum size=2em}}
\tikzset{vertexHid/.style = {shape=rectangle,draw,minimum size=2em}}
\tikzset{edge/.style = {->,> = latex', thick}}
\node[vertex] (Z) at  (1.5,2) {$Z$};
\node at (1.5,3) {\footnotesize{BMI}};
\node[vertex] (Y) at  (4,0) {$Y$};
\node[text width=1.1cm]  at (5.8,0) {\footnotesize{Blood} \vskip -1mm 
\footnotesize{pressure}};
\node[vertex] (X) at  (-1,0) {$X$};
\node[text width=1cm] at (-1,-1.3) {\footnotesize{Pension} \vskip -1mm 
\footnotesize{savings}};
\node[vertex] (T) at  (1.5,-2) {$N$};
\node at (1.4,-3) {\footnotesize{Death indicator}};

\draw[edge, bend left=20, color = blue] (Z) to (T);
\draw[edge, bend right=20, color = blue] (T) to (X);
\draw[edge, bend left=20, color = blue] (T) to (Z);
\draw[edge, bend left=20, color = blue] (T) to (Y);
\draw[edge, bend left, color = blue] (Z) to (Y);
\draw[edge, bend right, color = blue] (Z) to (X);
\draw[edge, bend left, color = blue] (Y) to (T);
\end{tikzpicture} & 
\begin{tikzpicture}[scale=0.6]
    \tikzset{vertex/.style = {shape=circle,draw,minimum size=2em}}
    \tikzset{vertexHid/.style = {shape=rectangle,draw,minimum size=2em}}
    \tikzset{edge/.style = {->,> = latex', thick}}
    \node[vertex] (Z) at  (1.5,2) {$Z$};
    \node at (1.5,3) {\footnotesize{BMI}};
    \node[vertex] (Y) at  (4,0) {$Y$};
    \node[text width=1.1cm]  at (5.8,0) {\footnotesize{Blood} \vskip -1mm 
    \footnotesize{pressure}};
    \node[vertex] (X) at  (-1,0) {$X$};
    \node[text width=1cm] at (-1,-1.3) {\footnotesize{Pension} \vskip -1mm 
    \footnotesize{savings}};
    \node[vertex] (T) at  (1.5,-2) {$N$};
    \node at (1.4,-3) {\footnotesize{Death indicator}};
    
    \draw[edge, bend left=20, color = blue] (Z) to (T);
\draw[edge, bend right=20, color = blue] (T) to (X);
\draw[edge, bend left=20, color = blue] (T) to (Z);
\draw[edge, bend left=20, color = blue] (T) to (Y);
    \draw[edge, bend left, color = blue] (Z) to (Y);
    \draw[edge, bend right, color = blue] (Z) to (X);
    \draw[edge, bend left, color = blue] (Y) to (T);
    \draw[edge, bend left=15, color = blue] (X) to (Y);
    \end{tikzpicture}
\end{tabular}

\caption{Local independence graphs illustrating how the three processes $X$,
$Y$, and $Z$ could affect each other and time of death in the Cox example.
There is no direct influence of $X$ (pension savings) on time of death in either
of the two graphs, but in the left graph the death indicator is furthermore
conditionally locally independent of $X$ given the history of $Z$ and $N$. 
In the right graph, $Z$ and $N$ does
not block all paths from $X$ to $N$, thus conditioning on the history of 
$Z$ and $N$ only would not render $N$ conditionally locally independent of $X$.
\label{fig:lig}}
\end{figure}
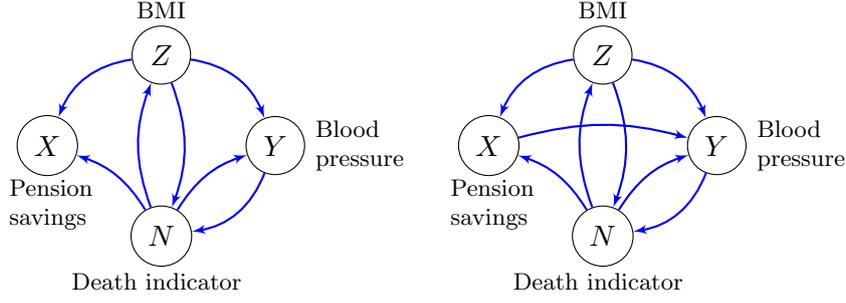

We assume that $T \in [0,1]$ and that $X$, $Y$ and $Z$ have continuous sample 
paths. Recall also that $N_t = \one(T \leq t)$ is the death indicator process. 
To maintain some form of realism, all processes are stopped at time of 
death, that is, $X_t = X_T$, $Y_t = Y_T$ and $Z_t = Z_T$ for $t \geq T$. This 
feedback from the death event to the other processes is reflected in Figure \ref{fig:lig} by the edges pointing out of $N$. Recall also that 
$$\mathcal{F}_t^{N,Z} = \sigma(N_s, Z_s; s \leq t)$$  
is the filtration generated by the $N$- and $Z$-processes. We use a similar notation 
for other processes and combinations of processes. For example, $\mathcal{F}_t^{N, X, Y, Z}$ 
is the filtration generated by $N$ and all three $X$-, $Y$-, and $Z$-processes. 
With $\lambda_t^{\text{full}}$ denoting the $\mathcal{F}_t^{N, X, Y, Z}$-intensity 
of time of death based on the history of all processes, we assume in this 
example a Cox model given by 
\begin{equation} \label{eq:coxex}
\lambda_t^{\text{full}} = \one(T\geq t)\lambda_t^0 e^{Y_t + \beta Z_t}
\end{equation}
with $\lambda_t^0$ a deterministic baseline intensity. It is not important that 
$\lambda_t^{\text{full}}$ is a Cox model for our general theory, 
but it allows for certain theoretical computations in this example.

The fact that $\lambda_t^{\text{full}}$ does not depend upon $X_t$ implies 
that $\lambda_t^{\text{full}}$ is also the $\mathcal{F}_t^{N, Y, Z}$-intensity,
and according to Definition \ref{dfn:cli}, $N_t$ is conditionally locally independent of $X_t$ given 
$\mathcal{F}_t^{N, Y, Z}$. This is in agreement with the local independence 
graphs in Figure~\ref{fig:lig} where there is no edge in either of them 
from $X$ to $N$.

We will take an interest in the case where $Y$ is unobserved and test the hypothesis:
\begin{center}
    $H_0:$ $N_t$ is conditionally locally independent of $X_t$ 
    given $\mathcal{F}_t^{N, Z}$.
\end{center}
That is, with $Y$ unobserved 
we want test if the intensity of time to death given the history of 
$N$, $X$ and $Z$ depends on $X$. To simplify notation let
$\mathcal{F}_t = \mathcal{F}_t^{N, Z}$ and 
$\mathcal{G}_t = \mathcal{F}_t^{N, X, Z}$ -- in accordance with the general notation.
The $\mathcal{G}_t$-intensity is by the innovation theorem given as
\begin{equation} \label{eq:lamb-marg}
\blambda_t = \ex(\lambda_t^{\text{full}} \mid \mathcal{G}_{t-})
= \one(T\geq t)\lambda_t^0 e^{\beta Z_t} \ex(e^{Y_t} \mid \mathcal{G}_{t-}),
\end{equation}
while the $\mathcal{F}_t$-intensity is 
\begin{equation} \label{eq:lamb}
\lambda_t = \ex(\lambda_t^{\text{full}} \mid \mathcal{F}_{t-}) = 
\one(T\geq t)\lambda_t^0 e^{\beta Z_t} \ex(e^{Y_t} \mid \mathcal{F}_{t-}),
\end{equation}
and $H_0$ is equivalent to $\lambda_t = \blambda_t$ almost surely.
Comparing \eqref{eq:lamb-marg}
and \eqref{eq:lamb} we see that $H_0$ holds in this example if 
$\ex(e^{Y_t} \mid \mathcal{G}_{t-}) = \ex(e^{Y_t} \mid \mathcal{F}_{t-})$,
and a sufficient condition for this is 
\begin{equation} \label{eq:exas}
\mathcal{F}_t^X \ind  \mathcal{F}_t^Y 
\mid  \mathcal{F}_t.
\end{equation}
The condition \eqref{eq:exas} is in concordance with the left graph in Figure \ref{fig:lig}, see Theorem 2 in \cite{Didelez2008}, but not the right, and it implies $H_0$. We will in 
Section \ref{subsec:ex} elaborate on  condition \eqref{eq:exas}  and give explicit 
examples. 

We recall that $H_0$ can  be reformulated as $\blambda_t$ not depending on $X$,
and we could investigate the hypothesis via a marginal Cox model
\begin{equation} \label{eq:coxmarg}
\blambda_t^{\text{cox}} = \one(T\geq t) \blambda_{t}^0 e^{\alpha_1 X_t + \alpha_2 Z_t}
\end{equation}
and test if $\alpha_1 = 0$. The Cox model is, however, not closed under 
marginalization and the semi-parametric model \eqref{eq:coxmarg} is quite likely 
misspecified. Consequently, a test of
$\alpha_1 = 0$ is not equivalent to a test of $H_0$. 

Our proposed nonparametric test of $H_0$ does not rely on a specific (semi-)parametric model of $\blambda_t$. To test $H_0$ we consider the LCM using the additive residual process. Then \eqref{eq:gammarep} implies that 
\begin{align*}
   \gamma_t & = \ex\left( X_{T}N_t - \int_0^{t} X_{s} \lambda_s \mathrm{d}s \right),
\end{align*}
By Proposition \ref{prop:cli-mg}, $\gamma_t = 0$ for $t \in [0,1]$ 
under $H_0$, whence conditional local independence implies $\gamma_t = 0$, and 
we test $H_0$ by estimating $\gamma_t$ and testing if it is constantly equal to $0$.

Before introducing a general estimator of the LCM in Section \ref{sec:statistic} 
we outline how to estimate the end point parameter $\gamma_1$ in this example. Due 
to $T \leq 1$ and the appearance of the indicator $\one(T\geq t)$ in \eqref{eq:lamb},
\begin{align*}
   \gamma_1 & = \ex\left( X_{T} - \int_0^{T} X_{s} \lambda_s \mathrm{d}s \right).
\end{align*}
With i.i.d. observations $(T_1, X_1, Z_1), \ldots, (T_n, X_n, Z_n)$ and 
(nonparametric) estimates, $\widehat{\lambda}_{j,t}$, based on 
$(T_1, Z_1), \ldots, (T_n, Z_n)$, we could compute the plug-in estimate 
$$\widehat{\gamma}_{1, \text{plug-in}}^{(n)} = \frac{1}{n} \sum_{j=1}^n \left( X_{j,T_j} -
\int_0^{T_j} X_{j,s} \widehat{\lambda}_{j,s} \mathrm{d}s \right).$$ 
However,  we cannot expect the plug-in estimator to 
have a $\sqrt{n}$-rate unless $\widehat{\lambda}$ has $\sqrt{n}$-rate, which 
effectively requires parametric model assumptions on the intensity. 
\begin{figure}[t]
    \centering
    \includegraphics[width=\linewidth]{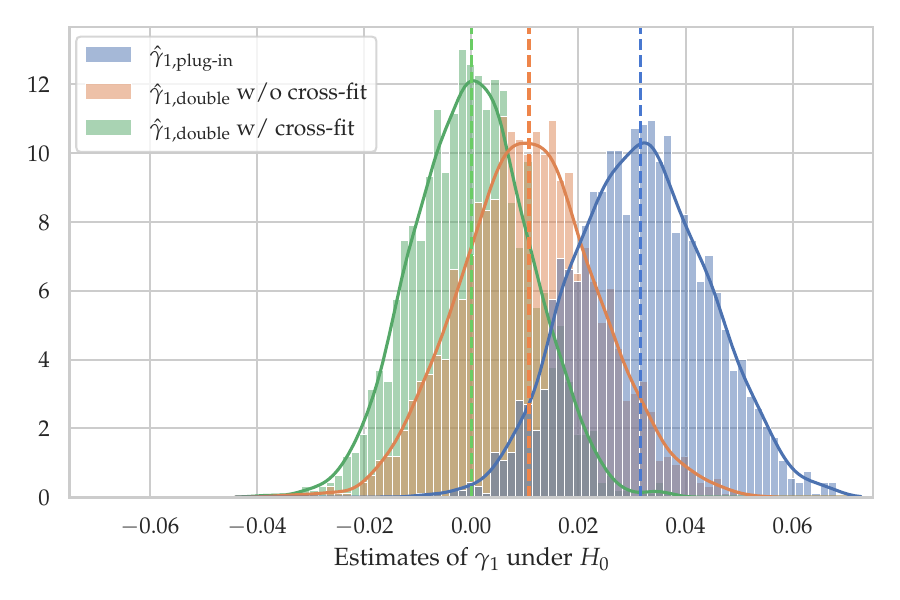}
    \caption{\small Histograms of the distributions of three different 
    estimators of $\gamma_1$.
    Each histogram contains 1000 estimates fitted to samples of size $n=500$. 
    The samples were sampled from a model that satisfies the hypothesis of 
    conditional local independence and hence the ground truth is $\gamma_1=0$. See Section \ref{sec:SamplingScheme} for further details of the data generating process.
    }
    \label{fig:endpoint_example}
\end{figure}
Using the definition of $\gamma_1$ in terms of the additive residual process 
$G_t = X_t - \Pi_t$, we also have that 
\begin{equation} \label{eq:gamrep} 
\gamma_1
= \ex\left( X_{T} - \Pi_{T} - \int_0^T (X_{s} - \Pi_s) \lambda_s \mathrm{d}s
\right).
\end{equation} 
A double machine learning estimator based on the ideas by
\cite{chernozhukov2018} is therefore obtained by plugging in two nonparametric estimators: 
$$\widehat{\gamma}_{1,\text{double}}^{(n)} = \frac{1}{n} \sum_{j=1}^n 
\left( X_{j,T_j} - \widehat{\Pi}_{j,T_j} - \int_0^{T_j} (X_{j,s} - \widehat{\Pi}_{j,s}) \widehat{\lambda}_{j,s} \mathrm{d}s \right).$$
To achieve a small bias and a $\sqrt{n}$-rate of convergence, we use 
sample splitting. The nonparametric estimates 
$\widehat{\Pi}_j$ and $\widehat{\lambda}_{j}$ are based on one part 
of the sample only, and are thus independent of the other part of 
the sample used for testing, see Section \ref{sec:statistic}. To 
obtain a fully efficient estimator, multiple sample splits can be combined,
e.g., via cross-fitting, see Section \ref{sec:cf}.

Figure \ref{fig:endpoint_example} shows the distributions of 
$\widehat{\gamma}_{1,\text{plug-in}}^{(500)}$ and 
$\widehat{\gamma}_{1,\text{double}}^{(500)}$ for the Cox example  
with $\gamma_1 = 0$, see Section \ref{sec:SamplingScheme} for details on the full 
model specification. 
The latter estimator was computed using cross-fitting but also 
without using any form of sample splitting.
The figure illustrates the bias of $\widehat{\gamma}_{1,\text{plug-in}}^{(500)}$, 
which is somewhat diminished by double machine learning without sample splitting and 
mostly eliminated by double machine learning in combination with 
cross-fitting.

\subsection{Estimating the Local Covariance Measure} \label{sec:statistic}

To estimate the LCM we assume that we have observed $n$ i.i.d. replications of the
processes, $(N_1, X_1, \mathcal{F}_1), \ldots, (N_n, X_n, \mathcal{F}_n)$, where
observing $\mathcal{F}_{j} = (\mathcal{F}_{j,t})$ signifies that anything
adapted to the $j$-th filtration is computable from observations. The 
process $N_j$ is adapted to $\mathcal{F}_j$, while $X_j$ is not, and 
$\mathcal{G}_j$ denotes the smallest right continuous and complete filtration 
generated by $X_j$ and $\mathcal{F}_j$.
 
For each $n$, we consider a sample split corresponding to a partition $J_n\cup
J_n^c = \{1,\ldots,n\}$ of the indices into two disjoint sets. We let
$\widehat{\lambda}^{(n)}$ and $\widehat{G}^{(n)}$ be estimates of the intensity and the
residualization map, respectively, fitted on data indexed by $J_n^c$ only. By an
estimate, $\widehat{\lambda}^{(n)}$, of $\lambda$ we mean a (stochastic) function
that can be evaluated on the basis of $\mathcal{F}_{j,t}$ for $j \in J_n$, and
its value, denoted by $\widehat{\lambda}_{j, t}^{(n)}$, is interpreted as a
prediction of $\lambda_{j,t}$. The stochasticity in $\widehat{\lambda}^{(n)}$ arises from
its dependence on data indexed by $J_n^c$, from which its functional form is
completely determined. Similarly, $\widehat{G}^{(n)}$ is a function that can be
evaluated on the basis of $\mathcal{G}_{j,t}$ for $j \in J_n$ to give a prediction $\widehat{G}_{j,t}^{(n)}$ of
$G_{j,t}$. In Section \ref{subsec:ex} we illustrate through the Cox example how
$\widehat{\lambda}^{(n)}$ and $\widehat{G}^{(n)}$ are to be computed in practice when we
use sample splitting. In Section~\ref{sec:estimationofnuisance} in the supplement we give
more examples of such estimation procedures and discuss their statistical properties 
in greater detail. 

To ease notation, we will throughout assume that $(N, X, \mathcal{F})$ 
denotes one additional process and filtration -- independent of and with the 
same distribution as the observed processes. Then the estimated 
intensity $\widehat{\lambda}^{(n)}$ and estimated residual process 
$\widehat{G}^{(n)}$ can be evaluated on $(N, X, \mathcal{F})$, and thus we may write 
$\widehat{\lambda}_t^{(n)}$ and $\widehat{G}_t^{(n)}$ to denote template copies of $\widehat{\lambda}^{(n)}_{j,t}$ and $\widehat{G}^{(n)}_{j,t}$ for $j \in J_n$. 

In terms of the estimates $\widehat{\lambda}^{(n)}$ and $\widehat{G}^{(n)}$  
we estimate LCM by the stochastic process $\widehat{\gamma}^{(n)}$ given by
\begin{equation} \label{eq:LCM}
\widehat{\gamma}_t^{(n)} = \frac{1}{|J_n|} \sum_{j \in J_n} \int_0^t 
\widehat{G}_{j,s}^{(n)}
\mathrm{d} \widehat{M}^{(n)}_{j, s},
\end{equation}
where $\widehat{M}_{j, t}^{(n)} = N_{j,t} - 
\int_0^t \widehat{\lambda}_{j, s}^{(n)} \mathrm{d}s$. 
We can regard $\widehat{\gamma}_t^{(n)}$ as a double machine learning estimator
of $\gamma_t$, with the observations indexed by $J_n^c$ used to learn models of
$\lambda$ and $G$, and with observations indexed by $J_n$ used to estimate
$\gamma_t$ based on these models. In Section \ref{sec:cf} we define the more
efficient estimator that uses cross-fitting, but it is instructive to study the
simpler estimator based on sample splitting first. 

In practical applications, we do not directly observe the filtration $\cF_j$, but rather samples from the stochastic processes generating the filtration. In accordance with the introductory Cox example, consider $\cF_j$ and $\cG_j$ given by $\cF_{j,t} = 
\sigma(Z_{j,s}, N_{j,s}; s\leq t)$ and $\cG_{j,t} = \sigma(X_{j,s}, Z_{j,s}, N_{j,s}; s\leq t)$ for a third stochastic process $Z_j$, with $Z_j$ possibly being multivariate. 
Within this setup, a general procedure for numerically computing the LCM is described in Algorithm \ref{alg:lcm}. Here, historical regression refers to any method which regresses the outcome at a given time on the history of the regressors up to that time. For example, historical linear regression is discussed in Section \ref{sec:simulations} and various alternative methods are discussed in Section~\ref{sec:estimationofnuisance} in the supplement. The choice of sample split will be discussed further in Section \ref{sec:cf} in the context of cross-fitting.

\begin{algorithm} \caption{Sample split estimator of LCM} \label{alg:lcm}
  \textbf{input}: processes $(N_j,X_j,Z_j)_{j=1,\ldots,n}$, partition $J_n \cup J_n^c$ of indices \;
  \textbf{options}: historical regression methods for estimation of $\lambda$ and $G$ given $N$ and $Z$, 
  
  discrete time grid $0 = t_0 <\cdots< t_k \leq 1$\; %
  \Begin{
    historically regress $(X_j)_{j\in J_n^c}$ on $(N_j,Z_j)_{j\in J_n^c}$ 
    to obtain a fitted model $\widehat{G}^{(n)}$ \; 
    historically regress $(N_j)_{j\in J_n^c}$ on $(N_j,Z_j)_{j\in J_n^c}$ 
    to obtain a fitted model $\widehat{\lambda}^{(n)}$ \; 
    compute out of sample residuals $\widehat{G}_{j,t_i}^{(n)}$ and 
    $\widehat{M}_{j,t_i}^{(n)}$ for $j \in J_n$ and $i = 0,\ldots,k$ \;
    for each $i=1,\ldots, k$, compute
    $$
    \widetilde{\gamma}_{t_i}^{(n)} = \frac{1}{|J_n|} \sum_{j \in J_n} 
    \sum_{1\leq l \leq i}
    \widehat{G}_{j,t_l}^{(n)}
    (\widehat{M}^{(n)}_{j, t_l} - \widehat{M}^{(n)}_{j, t_{l-1}})
    $$
  }
  \textbf{output}: Local Covariance Measure $\widetilde{\gamma}^{(n)}$ numerically approximated on grid\;
\end{algorithm}

As in Section \ref{sec:introex} we could suggest estimating the entire function 
$t \mapsto \gamma_t$ by a simple plug-in estimator of $\lambda$ 
using the representation \eqref{eq:gammarep}.
Figure \ref{fig:timevaryingexample} illustrates the distribution 
of estimators of the entire time dependent LCM for this plug-in 
estimator together with the double machine learning estimator with and 
without using cross-fitting. The figure also shows the distribution of 
the endpoint being the same distribution shown in Figure 
\ref{fig:endpoint_example}. The simulation is under $H_0$, 
and we see that only the double machine learning estimator with 
cross-fitting results in estimated sample paths centered around $0$.

\begin{figure}
    \includegraphics[width=\linewidth]{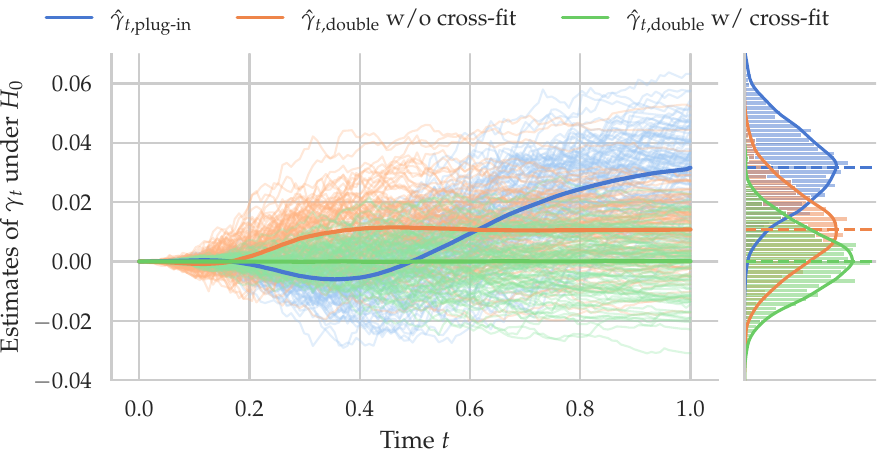}
    \caption{
        A time dependent extension of Figure \ref{fig:endpoint_example}
        showing the distribution of the sample paths  
        $t \mapsto \widehat{\gamma}_{t, \mathrm{plug-in}}^{(500)}$ and 
        $t \mapsto \widehat{\gamma}_{t, \mathrm{double}}^{(500)}$, 
        the latter with and without using cross-fitting. 
        The data were simulated under $H_0$ where 
        $t\mapsto \gamma_t$ is the zero function.
        See Section \ref{sec:SamplingScheme} for further details of the data generating process.
    }
    \label{fig:timevaryingexample}
\end{figure}
%




\section{Interpretations of the LCM estimator} \label{sec:interpretations}

In this section we provide some additional perspectives on and 
interpretations of the LCM. First 
we show that the LCM estimator can be seen as a Neyman orthogonalization
of the score statistic for a particular one-parameter family. The abstract formulation of the residual process $(G_t)$ permits that we transform $X$ into another $\mathcal{G}_t$-predictable processes. Using this perspective, we may optimize the choice of the process $X$ in terms of power. 

Next we show that when $X$ is independent of time, the 
test statistic reduces in a survival context to certain covariance measures between $X$-residuals and Cox-Snell-residuals, which we can link to existing test statistics for ordinary conditional independence.

\subsection{Neyman orthogonalization of a score statistic} \label{sec:neyman}
Consider the one-parameter family of $\mathcal{G}_t$-intensities
$$\blambda^{\beta}_t = e^{\beta X_t} \lambda_t$$
for $\beta \in \mathbb{R}$. Within this one-parameter family, the hypothesis 
of conditional local independence is equivalent to $H_0 : \beta = 0$.
The normalized log-likelihood with $n$ i.i.d. observations
in the interval $[0,t]$ is  
\begin{align*} 
    \ell_t(\beta) & = \frac{1}{n} \sum_{j=1}^n \left( 
        \int_0^t \log(\blambda_{j,s}^{\beta}) \mathrm{d}N_{j,s} - 
        \int_0^t \blambda_{j,s}^{\beta} \mathrm{d}s\right) \\
    & = \frac{1}{n} \sum_{j=1}^n \left( 
        \int_0^t \beta X_{j,s} + \log(\lambda_{j,s}) \mathrm{d}N_{j,s} - 
        \int_0^t e^{\beta X_{j,s}} \lambda_{j,s} \mathrm{d}s\right).
\end{align*}
Straightforward computations show that 
\begin{align*}
    \partial_{\beta} \ell_t(0) = \frac{1}{n}\sum_{j=1}^n \int_0^t X_{j,s} \mathrm{d} M_{j,s} 
    \quad \text{and} \quad
    - \partial_{\beta}^2 \ell_t(0) = \frac{1}{n} \sum_{j=1}^n \int_0^t X_{j,s}^2 \lambda_{j,s} \mathrm{d}s. 
\end{align*}
If $\lambda$ were known, the score statistic $\partial_{\beta} \ell_t(0)$ 
satisfies $\ex(\partial_{\beta} \ell_t(0)) = \gamma_t$. Moreover, 
under $H_0: \beta = 0$ we have that
$- \partial_{\beta}^2 \ell_t(0) = \langle \partial_{\beta} \ell_t(0) \rangle $ 
is a consistent estimate of the asymptotic variance of the 
mean zero martingale $\partial_{\beta} \ell_t(0)$. The hypothesis of 
local independence -- with $\lambda$ known -- could thus be tested 
using the score test statistic 
$- \partial_{\beta} \ell_t(0)^2 / \partial_{\beta}^2 \ell_t(0)$. 

The nuisance parameter $\lambda$ is, however, unknown and we want to 
avoid restrictive parametric assumptions about $\lambda$. Replacing 
$X_{j,t}$ by the residual process $G_{j,t}$ in the score statistic 
$\partial_{\beta} \ell_t(0)$ gives a \emph{Neyman orthogonalized} 
score 
$$
    \frac{1}{n}\sum_{j=1}^n \int_0^t G_{j,s} \mathrm{d} M_{j,s}.
$$
This score is linear in $\lambda$, which is used in supplementary Section~\ref{sec:orthodetails} to show that it satisfies the Neyman orthogonality condition under $H_0$, cf. Definition 2.1 in \cite{chernozhukov2018}. In this section, it is also shown that the act of replacing $X_t$ with $G_t = X_t - \Pi_t$ can, in fact, be viewed as \emph{concentrating out} the intensity of the score statistic in the sense of \citet{newey1994asymptotic}.
While Neyman orthogonality is never invoked explicitly, it is implicitly a central part of the asymptotic results for the LCM estimator (in particular Lemma~\ref{lem:R1} in the supplement).

The perspective on the LCM from a Neyman orthogonalized score 
statistic suggests that a test based on the LCM has most power against alternatives 
in the one-parameter family $\blambda^{\beta}$. If it happens 
that the most important alternatives are of the form  
$$\blambda^{\beta}_t = e^{\beta \bar{X}_t} \lambda_t$$
for some $\mathcal{G}_t$-predictable process $\bar{X}_t$ different from $X_t$, 
then we should replace $X_t$ by $\bar{X}_t$ in our test statistic, 
that is, in the residualization procedure. Examples of processes 
$\bar{X}_t$ are: 
\begin{itemize}
    \item transformations, $\bar{X}_t = f(X_t)$ for a function $f$
    \item time-shifts, $\bar{X}_t = X_{t-s}$ for $s > 0$
    \item linear filters, $\bar{X}_t = \int_0^t \kappa(t - s) X_s \mathrm{d} s$ for a kernel $\kappa$
    \item non-linear filters, $\bar{X}_t = \phi\left(\int_0^t \kappa(t - s) f(X_s) \mathrm{d} s\right)$
    for a kernel $\kappa$ and functions $f$ and~$\phi$.
\end{itemize}
Any finite number of such processes could, of course, also be combined 
into a vector process, and we could, indeed, generalize the LCM estimator \eqref{eq:LCM} to a vector process. 
The generalization is straightforward.

\subsection{Survival time with time-independent covariates} \label{sec:copula}
A different perspective on the test statistic is obtained if $X$ is 
constant over time, if $N_t = \one(T \leq t)$ is the counting process of a survival time $T$,
and if $\mathcal{F}_t = \sigma(N_s, Z; s \leq t)$ where $Z$ is a vector of additional 
baseline variables. Then the $\mathcal{F}_t$-intensity is 
$$\lambda_t = \one(T \geq t) h(t, Z),$$
where $h(t, Z)$ is the hazard function for $T$ given the baseline $Z$. In this 
special case, the hypothesis of conditional local independence is equivalent to the ordinary 
conditional independence 
\begin{equation} \label{eq:ordCI}
    X \ind T \mid Z.
\end{equation} 
We also find that 
$$
    \gamma_t = \ex(X (\one(T \leq t) - \Lambda_{t \wedge T})),
$$
and in particular $\gamma_1 = \ex(X (1 - \Lambda_T))$ as $T \in [0,1]$ by assumption. Since $\Lambda_T$ is exponentially distributed with mean $1$, we may write
$$
    \gamma_1 = - \mathrm{cov}(X, \Lambda_T).
$$
    Testing if $\gamma_1 \neq 0$ in this particular setup is 
    effectively a test of the conditional independence \eqref{eq:ordCI}. 
    When \eqref{eq:ordCI} is true, it further holds that $\Pi_t = \ex( X \mid \mathcal{F}_{t-})
    = \ex(X \mid Z) = \Pi_0$ is independent of $t$, and if we incorporate this 
    into our model of $\Pi$, the LCM estimator of $\gamma_1$ equals
    \begin{equation} \label{eq:LCM-baselineXadd}
     \widehat{\gamma}_1^{(n)} = \frac{1}{|J_n|} \sum_{j \in J_n} 
    (X_j - \widehat{\Pi}_{j,0})(1 - \widehat{\Lambda}_{T_j}).
    \end{equation}
    This is a (non-normalized) generalized covariance measure (GCM), see \citep{shah2020hardness},
    which is simply the (negative) empirical covariance between the additive residuals 
    $X_j - \widehat{\Pi}_{j,0}$ and the Cox-Snell residuals $\widehat{\Lambda}_{T_j}$. 
    
    Alternatively, consider the quantile residual process $G_{j,t} = F_t(X_j) - \frac{1}{2}$ 
    where $F_t(x) = \mathbb{P}(X \leq x \mid \mathcal{F}_{t-})$. If \eqref{eq:ordCI} is true, 
    it holds again that $F_t(x) = F_0(x) = \mathbb{P}(X \leq x \mid Z)$ 
    is independent of $t$, and our LCM estimator becomes
    \begin{equation} \label{eq:LCM-baselineX}
    \widehat{\gamma}_1^{(n)} = \frac{1}{|J_n|} \sum_{j \in J_n} 
    \widehat{G}_{j,0} \left(1 - \widehat{\Lambda}_{T_j}\right).
    \end{equation}
    This is likewise an empirical covariance, but now between the 
    generalized residuals and the Cox-Snell residuals. This is closely related to the 
    partial copula between $X$ and $T$ given $Z$, which can be estimated as 
    $$\frac{1}{|J_n|} \sum_{j \in J_n} 
    \widehat{G}_{j,0} \left(\frac{1}{2} - \exp(-\widehat{\Lambda}_{T_j}) \right).$$
    See \cite{Petersen:2021} for further details on the partial 
    copula and how this statistic can be used to test conditional 
    independence.
    
    Under a combined rate condition on estimation of $G$ and $\Lambda$, 
    the endpoint statistics above are known to be asymptotically 
    Gaussian with mean zero when the hypothesis of  
    conditional independence in \eqref{eq:ordCI} holds. Within this survival setting, 
    the endpoint statistics \eqref{eq:LCM-baselineXadd} can, furthermore, be seen as 
    a score test derivable from a semiparametric efficient score function. 
    Section~\ref{sec:survmodels} in the supplement gives the details for two specific 
    semiparametric survival models.

    Whenever $\widehat{G}_{j,t} = \widehat{G}_{j,0}$ is independent of time, e.g., if 
    we incorporate \eqref{eq:ordCI} into the residual model, the $t$-indexed LCM estimator 
    is 
     $$\widehat{\gamma}_t^{(n)} = \frac{1}{|J_n|} \sum_{j \in J_n} 
     \widehat{G}_{j,0} \left(\one(T_j \leq t) -  \widehat{\Lambda}_{T_j \wedge t}\right),$$
    which can be seen as a $t$-indexed extension of \eqref{eq:LCM-baselineX}. 
    For a general, time-dependent residual process, the full $t$-indexed 
    LCM estimator is 
    $$\widehat{\gamma}_t^{(n)} = \frac{1}{|J_n|} \sum_{j \in J_n} 
    \one(T_j \leq t) \widehat{G}_{j,T_j}  - \int_0^{t \wedge T_j} \widehat{G}_{j,s} \widehat{\lambda}_s \mathrm{d}s.$$
    The general results of this paper show that the $t$-indexed LCM estimator 
    is asymptotically distributed as a mean zero Gaussian martingale under $H_0$.
    This appears to be a novel result even when $X$ is constant over time. 
    However, the main contributions of this paper is to the case where $X$
    and $Z$ are stochastic processes varying with time -- where the 
    hypothesis of conditional local independence is also distinct from \eqref{eq:ordCI}.

    

\section{General asymptotic results} \label{sec:asymptotics}

In this section we derive uniform asymptotic results regarding the general LCM estimator as a stochastic process. In Section \ref{sec:lct} we discuss how to construct tests of $H_0$ based on the asymptotic results. 

We assume that $N$ has a $\cG_t$-intensity $\blambda_t$, we let $\bLambda_t = \int_0^t \blambda_s\mathrm{d}s$ denote the $\mathcal{G}_t$-compensator of $N$ and let $\mathbf{M}_t = N_t - \bLambda_t$ be the compensated local $\mathcal{G}_t$-martingale. We also recall that $\widehat{\gamma}^{(n)}$
denotes the LCM estimator based on sample splitting as defined in Section 
\ref{sec:statistic}. Within this framework we consider the decomposition
\begin{align}\label{eq:decomposition}
    \sqrt{|J_n|}\widehat{\gamma}^{(n)}
        = U^{(n)} + R_1^{(n)} + R_2^{(n)} + R_3^{(n)}
            + D_1^{(n)} + D_2^{(n)},
\end{align}
where the processes $U^{(n)}, R_1^{(n)}, R_2^{(n)}, R_3^{(n)}, D_1^{(n)}$, and $D_2^{(n)}$ are given by
    \begin{align}
    U^{(n)}_t & = \frac{1}{\sqrt{|J_n|}} \sum_{j \in J_n} 
    \int_0^t G_{j, s} \mathrm{d} \bM_{j, s}, \label{eq:oracleterm}
    \\
    R^{(n)}_{1, t} & = \frac{1}{\sqrt{|J_n|}} \sum_{j \in J_n}  \int_0^t
    G_{j, s}
    \left( \lambda_{j, s} - \widehat{\lambda}_{j, s}^{(n)} \right) 
    \mathrm{d}s,
    \\
    R^{(n)}_{2, t} & = \frac{1}{\sqrt{|J_n|}} \sum_{j \in J_n}  \int_0^t 
    \left( \widehat{G}_{j, s}^{(n)} - G_{j, s} \right) \mathrm{d} \bM_{j, s},
    \\
    R^{(n)}_{3, t} & = \frac{1}{\sqrt{|J_n|}} \sum_{j \in J_n}  \int_0^t 
    \left( \widehat{G}_{j, s}^{(n)} - G_{j, s} \right) 
    \left( \lambda_{j, s} - \widehat{\lambda}_{j, s}^{(n)} \right) 
    \mathrm{d}s,
    \\
    D^{(n)}_{1, t} & = \frac{1}{\sqrt{|J_n|}} \sum_{j \in J_n}  \int_0^t 
    G_{j, s}
    ( \blambda_{j, s} - \lambda_{j, s} ) 
    \mathrm{d}s,
    \\
    D^{(n)}_{2, t} & = \frac{1}{\sqrt{|J_n|}} \sum_{j \in J_n}  \int_0^t 
    (\widehat{G}_{j, s}^{(n)} - G_{j, s})
    (\blambda_{j, s} - \lambda_{j, s}) 
    \mathrm{d}s.
    \end{align}
Note that the processes $D_1$ and $D_2$ are (almost surely) the zero-process under $H_0$, since the null is equivalent to $\lambda_t$ being a version of $\blambda_t$. We proceed to show that:
\begin{itemize}
    \item the processes $U^{(n)}$ and $D_1^{(n)}-\sqrt{|J_n|}\gamma$ each converge in distribution,
    
    \item and the processes $R_1,R_2,R_3,D_2$ converge to the zero-process. 
\end{itemize}
For the analysis of each of $R_1$, $R_2$, and $D_2$, sample splitting is used to render the summands conditionally independent.

These asymptotic properties imply that $\sqrt{|J_n|}(\widehat{\gamma}^{(n)} - \gamma)$ is stochastically bounded in general, so the LCM estimator will asymptotically detect if the LCM is non-zero. Moreover, it will follow that $U^{(n)}$ drives the asymptotic limit of the LCM estimator under $H_0$. Based on these general asymptotic results we derive in Section \ref{sec:lct} asymptotic error control for tests based on the LCM estimator.

\subsection{Asymptotics of the LCM estimator}\label{sec:LCMasymptotics}
Our asymptotic results are formulated in terms of uniform stochastic convergence, which has also been discussed extensively in the recent literature on hypothesis testing \citep{shah2020hardness,lundborg2021conditional,lundborg2022projected,scheidegger2022weighted,neykov2021minimax}. Uniform convergence allows us to establish uniform asymptotic level of our proposed test, as well as power under local alternatives. We have collected key definitions and results related to uniform convergence in Section~\ref{app:UniformAsymptotics} in the supplement.

To state uniform assumptions and asymptotic results we 
need to indicate a range of possible sampling distributions for which the 
assumptions apply and the results hold. For this purpose, we extend our setup and allow all data to be parametrized by a fixed parameter set $\Theta$. The set $\Theta$ is not 
\emph{a priori} assumed to have any structure, and $\theta \in \Theta$ simply indicates
that $N^{\theta}$, $X^{\theta}$, $\lambda^{\theta}$, 
$G^{\theta}$ etc. have $\theta$-dependent distributions. We generally denote evaluation of processes or derived quantities for a specific $\theta$-value by a superscript, with the LCM, $\gamma^{\theta}$, in particular, depending on $\theta$. The LCM estimator is likewise written as $\widehat{\gamma}^{(n),\theta}=(\widehat{\gamma}^{(n),\theta}_t)$ for $\theta \in \Theta$ to denote its dependence on the sampling distribution. The superscript notation is, however, heavy and unnecessary in many cases and we will suppress the dependency on $\theta\in \Theta$ whenever it is not needed. Any result that does not explicitly involve $\Theta$ should be understood as a pointwise result for each $\theta \in \Theta$. 

The parametrization allows us to express convergence in distribution and probability uniformly over $\Theta$, which are denoted by $\convUD$ and $\convUP$, respectively. These concepts are defined rigorously in Definition~\ref{dfn:UniformConvergence} in the supplement. We note that uniform convergence reduces to classical (pointwise) convergence if $\Theta$ is a singleton, which corresponds to fixing the sampling distribution. We also introduce the parameter subset
\begin{equation}
    \Theta_0 \coloneqq \{\theta \in \Theta \mid H_0 \text{ is valid}\},
\end{equation}
consisting of all parameter values for which the hypothesis of conditional local independence holds. Correspondingly, we will use $\convUDnull$ and $\convUPnull$ to denote stochastic convergences uniformly over $\Theta_0$.

We are now ready to formulate the underlying assumptions on the data required for our asymptotic results. These assumptions may appear strong, but we argue in the discussion in Section 7 that they are not unreasonable from a practical viewpoint.

\begin{asm}\label{asm:UniformBounds} 
    There exist constants $C,C'>0$, such that for any parameter value $\theta\in \Theta:$
    \begin{itemize}
        \item[i)] The $\cG_t^\theta$-intensity $\blambda_t^\theta$ of $N^\theta$ is \lc{} with $\sup_{0\leq t \leq 1} \blambda_t^\theta \leq C$ almost surely.
        
        \item[ii)] The residual process $G^\theta$ is \lc{} with $\sup_{0\leq t \leq 1} |G_t^\theta| \leq C'$ almost surely.
    \end{itemize}
\end{asm}

The estimator, $\widehat{\lambda}^{(n)}_t$, of $\lambda_t$ and the estimator,
$\widehat{G}^{(n)}_t$, of the residual process are assumed to satisfy the same 
bounds as $\lambda_t$ and $G_t$. We note that Assumption \ref{asm:UniformBounds} i) implies that $\bM_t$ is a true $\cG_t$-martingale, and by the innovation theorem, $\lambda_t = \ex[\blambda_t \mid \cF_{t-}]$. As a consequence, the $\cF_t$-intensity $\lambda_t$ inherits the boundedness from the $\cG_t$-intensity $\blambda_t$, and $M_t$ is an $\cF_t$-martingale.
More generally, we have the following proposition ensuring that stochastic integrals 
are true martingales, e.g., that $I_t$ is a martingale under $H_0$.

\begin{prop} \label{lem:truemartingales}
    Under Assumption \ref{asm:UniformBounds} it holds that each of the processes
    \begin{equation*}
        \Big(\int_0^t f(G_s)\mathrm{d}\bM_s \Big)_{t\in[0,1]}
        \quad \text{ and } \quad
        \Big(\int_0^t f(\widehat{G}_s^{(n)})\mathrm{d}\bM_s \Big)_{t\in[0,1]}
    \end{equation*}
    are mean zero, square integrable $\mathcal{G}_t$-martingales
    for any $f\in C(\mathbb{R})$.  
\end{prop}

To express the asymptotic distribution of $U^{(n)}$ we need its \emph{variance function}. 

\begin{definition} We define the variance function 
$\mathcal{V}\colon [0,1] \to [0,\infty]$ as
    \begin{align} \label{eq:variance}
    \mathcal{V}(t) = \ex \left(
    \int_0^t G_s^2 \mathrm{d} N_s
    \right)
    \end{align}
for $t\in[0,1]$.
\end{definition}

As everything else, the variance function, $\mathcal{V} = \mathcal{V}^\theta$, is 
also indexed by the parameter $\theta$, which we, for notational simplicity, 
suppress unless explicitly needed.  

By taking $f(x) = x^2$ in Proposition \ref{lem:truemartingales}, 
Assumption \ref{asm:UniformBounds} implies that for each $t\in [0,1]$,
$$
    \mathcal{V}(t) 
    = \ex \left(
        \int_0^t G_s^2 \blambda_s \mathrm{d} s
        \right) < \infty.
$$ 
Moreover, $\mathcal{V}(t)$ is the variance of  $\int_0^t G_s\mathrm{d}\bM_s$, which under $H_0$ is the same as the variance of $I_t = \int_0^t G_s\mathrm{d}M_s$.

With the assumptions above we can prove the following proposition 
about the uniform distributional limit of the process 
$U^{(n)}$ in the \emph{Skorokhod space} $D[0,1]$, the space of \cl{} functions from $[0, 1]$ to $\mathbb{R}$ endowed with the Skorokhod topology. A corresponding pointwise result is an application of Rebolledo's classical
martingale CLT. Our generalization to uniform convergence is based on 
a uniform extension of Rebolledo's theorem, see Theorem~\ref{thm:URebo}
in Section~\ref{sec:fclt} in the supplement. 
\begin{prop}\label{prop:UasymptoticGaussian}
    Under Assumption \ref{asm:UniformBounds} it holds that
    \begin{align*}
        U^{(n),\theta} \convUD U^\theta
    \end{align*}
    in $D[0,1]$ as $n\to \infty$, where for each $\theta\in \Theta$, $U^\theta$ is a mean zero continuous Gaussian martingale on $[0,1]$ with variance function $\mathcal{V}^\theta$.
\end{prop}

To control the remainder terms in \eqref{eq:decomposition} we will bound 
the estimation errors in terms of the 2-norm, $\vertiii{\cdot}_{2}$, on  $L_2([0,1] \times \Omega)$, i.e.,
    \begin{align*}
    \vertiii{W}_{2}^2 = \ex \left( \int_0^1 W_s^2 \mathrm{d}s
    \right)
    \end{align*}
for any process $W\in L_2([0,1] \times \Omega)$.
We will make the following consistency assumptions on 
$\widehat{\lambda}^{(n)}$ and $\widehat{G}^{(n)}$. 
\begin{asm}\label{asm:UniformRates}
Assume that $|J_n| \to \infty$ when $n \to \infty$ and let
    \begin{align*}
        g^\theta(n) = \vertiii{G^\theta - \widehat{G}^{(n),\theta}}_{2}
        \qquad
        \text{and}
        \qquad
        h^\theta(n) = \vertiii{\lambda^\theta - \widehat{\lambda}^{(n),\theta}}_{2}.
    \end{align*}
Then each of the sequences $g^\theta(n)$, $h^\theta(n)$, and $\sqrt{|J_n|}g^\theta(n)h^\theta(n)$ converge to zero uniformly over $\Theta$ as $n\to\infty$, i.e.,
\begin{align*}
    \lim_{n\to \infty}\sup_{\theta\in\Theta} \max\{ g^\theta(n), h^\theta(n), 
        \sqrt{|J_n|}g^\theta(n)h^\theta(n) \} = 0.
\end{align*}
\end{asm}

With this assumption we can establish that the remainder terms also converge uniformly
to the zero-process.
\begin{prop}\label{prop:remainderterms}
    Under Assumptions \ref{asm:UniformBounds} and \ref{asm:UniformRates}, it holds that 
    \begin{equation*}
        \sup_{t\in [0,1]} |R_{i,t}^{(n),\theta}| \convUP 0
    \end{equation*}
    as $n\to \infty$ for $i=1,2,3$. 
\end{prop}
%
%

To control the asymptotic behavior of the LCM estimator in the alternative we 
need to control the two terms $D_1^{(n)}$ and $D_2^{(n)}$.

\begin{prop}\label{prop:Dasymptotics}
    Let Assumptions \ref{asm:UniformBounds} and \ref{asm:UniformRates} hold true. 
    \begin{enumerate}
    \item[i)] The stochastic process
    $D_1^{(n),\theta}-\sqrt{|J_n|} \gamma^{\theta} $
    converges in distribution in $(C[0,1],\|\cdot\|_\infty)$ uniformly over $\Theta$ as $n\to \infty$.
    \item[ii)] If $G_t^\theta = X_t^\theta - \Pi_t^\theta$ is the additive residual process, then $D_2^{(n),\theta}\convUP 0$ in $D[0,1]$ as $n\to \infty$.
    \end{enumerate} 
\end{prop}
We note that $D_2^{(n)}$ might not vanish without an assumption like $G_t$ 
being the additive residual process, and it is not clear if $D_2^{(n)}$ will even converge in general. We will not pursue an analysis of the asymptotic behavior of $D_2^{(n)}$ in the general case. We note, however, that if we can estimate $G$ with 
a parametric rate, that is, $\sqrt{|J_n|} g(n) = O(1)$, then 
it follows from the Cauchy-Schwarz 
inequality that $D_2^{(n)}$ is stochastically bounded, and $D_1^{(n)}$ 
still dominates in the alternative where $\gamma \neq 0$. 

We can combine all of the propositions into a single theorem regarding the asymptotics of the LCM estimator, which we consider as our main result.

\begin{thm}\label{thm:main}
    Let Assumptions \ref{asm:UniformBounds} and \ref{asm:UniformRates} hold true.
    \begin{enumerate}
        \item[i)] It holds that
            \begin{align*}
                \sqrt{|J_n|}\widehat \gamma^{(n),\theta} \convUDnull U^\theta
            \end{align*}
        in $D[0,1]$ as $n\to \infty$, where for each $\theta\in \Theta_0$, $U^\theta$ is a mean zero continuous Gaussian martingale on $[0,1]$ with variance function $\mathcal{V}^\theta$.

        \item[ii)] For the additive residual process it holds that
        for every $\varepsilon>0$ there exists $K>0$ such that
        \begin{align}\label{eq:stochboundedness}
            \limsup_{n\to \infty}\sup_{\theta\in \Theta} \mathbb{P}\pa{
            \sqrt{|J_n|} \cdot \|\widehat{\gamma}^{(n),\theta} - \gamma^\theta\|_\infty > K} < \varepsilon.
        \end{align}
    \end{enumerate}
\end{thm}
Thus we have established the weak asymptotic limit of 
$\sqrt{|J_n|} \widehat{\gamma}^{(n)}$ under $H_0$. However, the variance function $\mathcal{V}$ of the limiting Gaussian 
martingale is unknown and must be estimated from data. We 
propose to use the empirical version of \eqref{eq:variance},
    \begin{align}\label{eq:empiricalvariance}
    \widehat{\mathcal{V}}_n(t) = 
    \frac{1}{|J_n|} \sum_{j \in J_n} 
    \int_0^t \left(\widehat{G}_{j,s}^{(n)}\right)^2 \mathrm{d} N_{j,s} 
    = \frac{1}{|J_n|} \sum_{j \in J_n} 
    \sum_{\tau \leq t: \Delta N_{j,s} = 1} \left(\widehat{G}_{j,\tau}^{(n)}\right)^2,
    \end{align}
for which we have the following consistency result.

\begin{prop}\label{prop:varianceconsistent}
    Under Assumptions \ref{asm:UniformBounds} and \ref{asm:UniformRates} it holds that 
    $$
        \sup_{t\in[0,1]}|\widehat{\mathcal{V}}_n^\theta(t) - \mathcal{V}^\theta(t)| \convUP 0,
    $$
    as $n \to \infty$.
\end{prop}
We emphasize that $\mathcal{V}$ is only the asymptotic variance function of the 
LCM estimator under $H_0$. It is always the asymptotic variance function of 
$U^{(n)}$, but in the alternative the asymptotic distribution of $\widehat{\gamma}^{(n)}$
also involves the asymptotic distribution of $D_1^{(n)}$ and is thus more complicated.

Tests of conditional local independence can now be 
constructed in terms of univariate functionals of 
$\widehat{\gamma}^{(n)}$ and $\widehat{\mathcal{V}}_n$ that quantify the magnitude of the LCM. The asymptotics of such test statistics under $H_0$ are described in the following corollary, which is essentially an application of the continuous mapping theorem.
\begin{cor} \label{cor:statisticsdistribution}
    Let $\mathcal{J}\colon D[0,1]\times D[0,1] \to \mathbb{R}$ be a functional that is continuous on the closed subset $C[0,1]\times \overline{\{\mathcal{V}^\theta \colon \theta \in \Theta_0\}}$
    with respect the uniform topology, i.e., the topology generated by the norm $\|(f_1,f_2)\|= \max\{\|f_1\|_\infty,\|f_2\|_\infty\}$ for $f_1,f_2\in D[0,1]$. Define the test statistic 
    $$
        \widehat{D}_n^\theta = \mathcal{J}\left( 
            \sqrt{|J_n|}\widehat{\gamma}^{(n),\theta}, \; \widehat{\mathcal{V}}_n^\theta\right).
    $$ 
    Under Assumptions \ref{asm:UniformBounds} and \ref{asm:UniformRates},
    it holds that
        \begin{align}
            \widehat{D}_n^\theta \convUDnull \mathcal{J}(U^\theta,\mathcal{V}^\theta),
            \qquad n \to \infty,
        \end{align}
    where $U^\theta$ is a mean zero continuous Gaussian martingale with variance function $\mathcal{V}^\theta$.
\end{cor}

\section{The Local Covariance Test} \label{sec:lct}
In this section we introduce a practically applicable 
test based on the LCM estimator. Using the asymptotic distribution of the
LCM estimator we show that the asymptotic distribution of our proposed test 
is independent of the sampling distribution under $H_0$ and has an explicit representation. 
We show, in addition, uniform asymptotic level of the test, and we give a uniform 
power result for the additive residual process. Finally, we modify the test to be 
based on a cross-fitted estimator of the LCM instead of using sample splitting, 
and we show uniform level of that test.

To construct a test statistic based on the LCM estimator it is beneficial that
its distributional limit does not depend on the variance function. As a simple example,
consider the endpoint test statistic:
    \begin{align} \label{eq:endpointstatistic}
        \big(\widehat{\mathcal{V}}_n(1)\big)^{-\frac{1}{2}}\sqrt{|J_n|}\widehat{\gamma}^{(n)}_1,
    \end{align}
which under $H_0$ converges in distribution to 
$\mathcal{V}(1)^{-\frac{1}{2}}U_1$ by Corollary \ref{cor:statisticsdistribution}. The distribution of the latter is the standard normal distribution, and in particular it does not depend on $\mathcal{V}$. 

Any test statistic constructed from $\widehat{\gamma}^{(n)}$ should  
capture deviations of $\gamma_t$ away from $0$. The test statistic in
\eqref{eq:endpointstatistic} does, however, only consider the endpoint of the
process, and since $\gamma$ is not necessarily monotone, $\gamma_t$ may 
deviate more from $0$ for other $t \in [0,1]$. 
Thus in order to increase power against such alternatives we consider the test statistic
    \begin{align} \label{eq:supremumstatistic}
        \widehat{T}_n = 
        \frac{\sqrt{|J_n|}}{\sqrt{\widehat{\mathcal{V}}_n(1)}}  \sup_{0 \leq t \leq 1} 
        \big|
            \widehat{\gamma}^{(n)}_{t}
        \big|.
    \end{align}
We refer to $\widehat{T}_n$ as the \emph{Local Covariance Test statistic} (LCT statistic). We proceed to show that the LCT statistic can be turned into a test of $H_0$ with asymptotic level $\alpha$, and which has asymptotic power against any alternative with a non-zero LCM. This is the best we can hope for of any test based on the LCM estimator. 

We note that it might be possible to establish similar results for other norms of the LCM, for example, a statistic based on a weighted $L_2$-norm.
However, since other norms of the distributional limit $U$ will generally have a distribution with a complicated dependency on $\mathcal{V}$, we believe that the LCT statistic is the simplest to construct.

To establish uniform asymptotic level via Corollary
\ref{cor:statisticsdistribution} for tests based on test statistics such as \eqref{eq:supremumstatistic} we need to assume that the asymptotic variances in $t = 1$ 
are uniformly bounded away from zero.

\begin{asm}  \label{asm:UniformVariance}
    There exists $\delta_1 > 0$ such that for all $\theta \in \Theta$ it holds 
    that $\mathcal{V}^{\theta}(1) \geq \delta_1$.
\end{asm}

\subsection{Type I and type II error control}
We proceed to show that under $H_0$, the LCT statistic is distributed as $\sup_{0\le t \le 1} |B_t|$, where $(B_t)$ is a standard Brownian motion. From this point onwards, we let $S$ denote a random variable with such a distribution and note that its CDF can be written as:
    \begin{align} \label{eq:supdistribution}
        F_S(x) = \mathbb{P}(S \leq x ) = \frac{4}{\pi} \sum_{k=0}^\infty \frac{(-1)^k}{2k+1}
        \exp\left(- \frac{\pi^2(2k+1)^2}{8x^2}\right),
        \qquad x>0.
    \end{align}
See, for example, Section 12.2 in \citet{SchillingPartzsch:2012} where the formula is derived from Lévy's triple law. 

The $p$-value for a test of $H_0$ equals $1 - F_S(\widehat{T}_n)$, and since the series in \eqref{eq:supdistribution}
converges at an exponential rate, the $p$-value can be computed with high numerical precision by truncating the series. Given a significance level $\alpha \in (0,1)$, we also let $z_{1-\alpha}$ denote the $(1-\alpha)$-quantile of $F_S$, which exists and is unique since the right-hand side of \eqref{eq:supdistribution} is strictly increasing and continuous. The \emph{Local Covariance Test} (LCT) with significance level $\alpha$ is then defined by
\begin{align}\label{eq:LCT}
    \Psi_n 
    = \Psi_n^\alpha
    = \one(F_S(\widehat{T}_n)> 1-{\alpha}) 
    = \one(\widehat{T}_n > z_{1-{\alpha}}).
\end{align}

From Theorem \ref{thm:main} we can now deduce the asymptotic properties of the LCT under 
the hypothesis of conditional local independence. 
Recall that $\convUDnull$ denotes uniform convergence in distribution under $H_0$.
\begin{thm}\label{thm:LCTlevel}
    Let Assumptions \ref{asm:UniformBounds}, \ref{asm:UniformRates} and \ref{asm:UniformVariance} hold true. Then it holds that
    $$
        \widehat T_n^\theta \convUDnull S
    $$
    as $n\to \infty$. As a consequence, for any $\alpha \in (0,1)$,
    \begin{align*}
        \limsup_{n\to \infty} \sup_{\theta \in \Theta_0} \mathbb{P}(\Psi_n^{\alpha, \theta} =1) \leq \alpha.
    \end{align*}
    In other words, the Local Covariance Test defined in \eqref{eq:LCT} has uniform asymptotic level~$\alpha$.
\end{thm}
In general, we cannot expect that the test has power against alternatives to $H_0$ for which the LCM is the zero-function. This is analogous to other types of conditional independence tests based on conditional 
covariances, e.g., GCM \citep{shah2020hardness}. However, we do have the following result that establishes power against local alternatives with $\|\gamma\|_\infty$ decaying at an order of at most $|J_n|^{-1/2}$.
\begin{thm}\label{thm:rootNpower}
    Let Assumptions \ref{asm:UniformBounds} and \ref{asm:UniformRates} hold true. 
    Using the additive residual process it holds that for any $0<\alpha<\beta<1$ there exists $c>0$ such that
    \begin{align*}
        \liminf_{n\to \infty}\inf_{\theta \in \mathcal{A}_{c,n}} \mathbb{P}(\Psi_n^{\alpha,\theta} = 1) \geq \beta,
    \end{align*}
    where $\mathcal{A}_{c,n} = \{\theta \in \Theta \mid \|\gamma^\theta\|_\infty \geq c |J_n|^{-1/2}\}$.
\end{thm}

\subsection{Extension to cross-fitting} \label{sec:cf} 
In Section \ref{sec:asymptotics} we considered sample splitting with observations indexed by
$J_n^c$ used to estimate the two models and with observations indexed by $J_n$
used to estimate $\gamma$. Following \citet{chernozhukov2018}, we can improve
efficiency by cross-fitting, i.e., by flipping the roles of $J_n$ and $J_n^c$ to
obtain a second equivalent estimator of $\gamma$. Heuristically, the two
estimators are approximately independent, and thus their average should be a more
efficient estimator. This procedure generalizes directly to a partition $J_n^1
\cup \cdots \cup J_n^K = \{1,\ldots,n\}$ of the indices into $K$ disjoint folds.
The partition is assumed to have a uniform asymptotic density, meaning that
$|J_n^k|/n \to \frac{1}{K}$ as $n \to \infty$ for each $k$. 

We estimate $G$ and $\lambda$ using $(J_n^k)^c=\{1,\ldots, n\}\setminus
J_n^k$ and subsequently estimate $\gamma$ using $J_n^k$. Then the $K$-fold \emph{Cross-fitted LCM estimator}, abbreviated as X-LCM, is defined as the average LCM estimator over the $K$ folds, i.e.,
    \begin{align} \label{eq:KfoldI}
        \check{\gamma}_t^{K,(n)} 
        = \frac{1}{K} \sum_{k=1}^K \frac{1}{|J_n^k|}\sum_{j \in J_n^k} 
        \int_0^t \widehat{G}_{j,s}^{k,(n)}\mathrm{d} \widehat{M}^{k,(n)}_{j, s},
    \end{align}
where for each $j \in J_n^k$, the processes $\widehat{G}_{j}^{k,(n)}$ and $\widehat{M}^{k,(n)}_{j}$ are the model predictions of $G_j$ and $M_j$, respectively, based on training data indexed by $(J_n^k)^c$. 
We also define a $K$-fold version of the variance estimator:
    \begin{align} \label{eq:Kfoldsig}
        \check{\mathcal{V}}_{n}^K(t)
        = 
        \frac{1}{K} \sum_{k=1}^K \frac{1}{|J_n^k|} \sum_{j \in J_n^k} 
        \int_0^t \left(
            \widehat{G}_{j,s}^{k,(n)}\right)^2 \mathrm{d} N_{j,s}.
    \end{align}
Now, similarly to the LCT statistic, the cross-fitted estimator can be used to construct a test statistic,
    \begin{align} \label{eq:DMLtest}
        \check{T}_{n}^K 
        = \sqrt{\frac{n}{\check{\mathcal{V}}_{K,n}(1)}} \sup_{0\leq t \leq 1}\left| \check{\gamma}_t^{K,(n)} \right|,
    \end{align}
from which we define the following test of conditional local independence.
\begin{definition} \label{dfn:test} Let $\alpha \in (0,1)$ and let $\check{T}_{n}^K$ be
    the test statistic from \eqref{eq:DMLtest}. The $K$-fold
    \emph{Cross-fitted Local Covariance Test} (X-LCT) with
    significance level $\alpha$ is defined by
        \begin{align*}
            \check{\Psi}_n^K 
            = \one(F_S(\check{T}_{n}^K) > 1-\alpha)
            = \one(\check{T}_{n}^K > z_{1-\alpha}),
        \end{align*}
    where $z_{1-\alpha}$ is the $(1-\alpha)$-quantile of the distribution
    function $F_S$ given in \eqref{eq:supdistribution}.
\end{definition}

We provide a summary of the computation of the X-LCT in Algorithm \ref{alg:X-lct}. The asymptotic analysis of $\widehat{\gamma}^{(n)}$ generalizes to $\check{\gamma}^{K,(n)}$, but we will refrain from restating all results for the $K$-fold cross-fitted estimator. For simplicity, we focus on the fact that the X-LCT has asymptotic level $\alpha$.
\begin{thm} \label{thm:LCTXlevel}
    Suppose that Assumption \ref{asm:UniformRates} is satisfied for every sample split $J_n^k \cup (J_n^k)^c, k=1,\ldots,K$. Under Assumptions \ref{asm:UniformBounds} and \ref{asm:UniformVariance}, the X-LCT statistic satisfies
    $$
        \check{T}_{n}^{K,\theta} \convUDnull S
    $$
    for $n\to \infty$. In particular, the X-LCT has uniform asymptotic level $\alpha$.
\end{thm}
Note that cross-fitting recovers full efficiency in the sense that the scaling
factor is $\sqrt{n}$ rather than $\sqrt{|J_n|}$, which leads to a more powerful test.
Moreover, the asymptotic
distribution of $\check{T}_{n}^K$ does not depend on the number of folds $K$, and any difference between various choices of $K$ can thus be attributed to finite sample errors. Larger values of $K$ will allocate more data to estimation of $G$ and $\lambda$, which intuitively should be the harder estimation problem. Following Remark 3.1 in \citet{chernozhukov2018}, we believe that a default choice of $K=4$ or $K=5$ should be reasonable in practice.


\begin{algorithm} \caption{$K$-fold cross-fitted local covariance test (X-LCT)} \label{alg:X-lct}
  \textbf{input}: processes $(N_j,X_j,Z_j)_{j=1,\ldots,n}$, partition $J_n^1 \cup \cdots \cup J_n^K$ of indices into $K$ folds\;
  \textbf{options}: historical regression methods for estimation of $\lambda$ and $G$ given $(N,Z)$, 
  
  discrete time grid $\mathbb{T}\subset [0,1]$, significance level $\alpha\in (0,1)$\; %
  \Begin{
    \For{$k = 1,\ldots, K$}{
        apply Algorithm \ref{alg:lcm} on sample split $J_n^k \cup (J_n^k)^c$ to compute $\widetilde{\gamma}^{k,(n)}$ on the grid $\mathbb{T}$\;
        use Equation \eqref{eq:empiricalvariance} on sample split $J_n^k \cup (J_n^k)^c$ to compute $\widetilde{\mathcal{V}}_{k,n}(1)$ \;
    }
    compute 
    $\check{\gamma}^{K,(n)} 
        = \frac{1}{K} \sum_{k=1}^K \widetilde{\gamma}^{k,(n)}$
    on grid $\mathbb{T}$ \;
    compute  
        $\check{\mathcal{V}}_{K,n}(1) = \frac{1}{K} \sum_{k=1}^K \widetilde{\mathcal{V}}_{k,n}(1)$\; 
    compute the X-LCT statistic
    $\check{T}_n^K = \sqrt{n} \cdot \max_{t\in \mathbb{T}}|\check{\gamma}_t^{K,(n)}| / \sqrt{\check{\mathcal{V}}_{K,n}(1)}$ \;
    compute $p$-value $\check{p}=1-F_S(\check{T}_n^K)$ by truncating the series in Equation \eqref{eq:supdistribution}.
  }
  \textbf{output}: 
  the X-LCT $\check{\Psi}_n^K = \one(\check{p} < \alpha)$, and optionally the $p$-value $\check{p}$\;
\end{algorithm}

\section{Simulation study} \label{sec:simulations} In this section we present
the results from a simulation study based on the Cox example introduced in
Section \ref{sec:introex}. We elaborate in Section \ref{subsec:ex} on the 
full model specification used for the simulation study -- which will also 
illuminate how $\Pi$ and $\lambda$ can be modeled and estimated. 
The results from the simulation study  
focus on the distribution of the X-LCT statistic $\check{T}_{n}^K$ and validate the asymptotic level and power of the X-LCT $\check{\Psi}_n^K$. The latter is also compared to a hazard ratio test based on the marginal Cox model \eqref{eq:coxmarg}. The simulations were implemented in Python and
the code is
available\footnote{\url{https://github.com/AlexanderChristgau/nonparametric-cli-test}}.

\subsection{Cox model continued} \label{subsec:ex}
Consider the same setup as in Section \ref{sec:introex}. To fully specify the 
model we need to specify the distribution of the processes $X$, $Y$ and $Z$.
We suppose that $X$ and $Y$ can be written in terms of $Z$ as
\begin{align}\label{eq:historicallinear}
X_t = \int_0^t Z_s \rho_X(s, t) \mathrm{d}s + V_t,
\qquad \text{and} \qquad
Y_t = \int_0^t Z_s \rho_Y(s, t) \mathrm{d}s + W_t,
\end{align}
where $\rho_X$ and $\rho_Y$ are two functions defined on the triangle $\{(s,t)
\in [0,1]^2 \mid s \leq t \}$, and where $V = (V_t)_{0\leq t \leq 1}$ and $W =
(W_t)_{0\leq t \leq 1}$ are two noise processes with mean zero. The processes
$Z$, $V$ and $W$ are assumed independent, which implies \eqref{eq:exas} and thus 
that $N$ is conditionally locally independent of $X$ given $\mathcal{F}_t = \mathcal{F}_t^{N,Z}$.

The specific dependency of $X$ and $Y$ on $Z$ is known as the \emph{historical functional linear
model} in functional data analysis \citep{Malfait:2003}.
Within this model, 
\begin{equation} \label{eq:piModel}
\Pi_t = \ex(X_t \mid \mathcal{F}_{t-}) = \int_0^t Z_s \rho_X(s, t) \mathrm{d} s,
\end{equation}
and on $(T \geq t)$
\begin{align*}
    \ex(e^{Y_t} \mid \mathcal{F}_t) 
    = e^{\int_{0}^t Z_s \rho_Y(s, t) \mathrm{d} s} 
    \ex(e^{W_t} \mid T \geq t)
    = e^{\tilde{\beta}_0(t) 
        + \int_{0}^t Z_s \rho_Y(s, t) \mathrm{d} s},
\end{align*}
where $\tilde{\beta}_0(t) = \log (\ex(e^{W_t} \mid T \geq t))$. Since
$$
    \lambda_t 
    = \one(T\geq t)\lambda_t^0 
        e^{\beta Z_t} \ex(e^{Y_t} \mid \mathcal{F}_t),
$$
it follows that on $(T \geq t)$,
\begin{align} 
\log(\lambda_t) 
    = \beta_0(t) + \beta Z_t 
        + \int_{0}^t Z_s \rho_Y(s, t) \mathrm{d} s,
\label{eq:lambdaModel}
\end{align}
where the two baseline terms depending only on time have been merged into
$\beta_0$. 

The computations above suggest how the estimators $\widehat \lambda^{(n)}$ and $\widehat \Pi^{(n)}$ 
could be constructed. That is, $\widehat \lambda^{(n)}$ could be based on estimates of
$\beta$, $\beta_0$ and $\rho_Y$ from the
observations $(T_j,Z_j)_{j\in J_n^c}$, and $\widehat \Pi^{(n)}$ could be based on estimates of  
$\rho_X$ from $(X_j,Z_j)_{j\in J_n^c}$. We would then have 
$$\widehat{\Pi}^{(n)}_{j,t} =  \int_0^t Z_{j,s} \widehat{\rho}_X^{(n)}(s, t) \mathrm{d} s$$
for $j \in J_n$ where $\widehat{\rho}_{X}^{(n)}$ denotes the estimate of $\rho_X$, 
and similarly for $\widehat \lambda^{(n)}_{j,t}$.
Particular choices of estimators $\widehat{\rho}_{X}^{(n)}$ and $\widehat{\rho}_{Y}^{(n)}$ and their 
theoretical properties are reviewed in Section~\ref{sec:estimationofnuisance} in the supplement. Our conclusion from this 
review is that for the historical functional linear model, sufficient rate results 
should be possible but have not yet been established rigorously.

\subsection{Sampling scheme}\label{sec:SamplingScheme}
The actual time-discretized simulations and computations were implemented using an
equidistant grid $\mathbb{T} = (t_i)_{i=1}^{q}$ with $q=128$ time points $0=t_1
< \cdots < t_q = 1$. Inspired by \citet{harezlak2007penalized}, we generated the
processes as follows: let $\xi\in \mathbb{R}^3$ and $V,W,\VVV \in
\mathbb{R}^\mathbb{T}$ be independent random variables such that $\xi
\sim \mathcal{N}(0,\operatorname{I}_3)$ and such that $V,W$, and $\VVV$ are identically
distributed with 
    $
        V_{t_1},V_{t_2} - V_{t_1}, \ldots, V_{t_q}-V_{t_{q-1}} 
        \stackrel{\mathrm{i.i.d.}}{\sim} \mathcal{N}(0,1 / q).
    $
Then the process $Z$ is determined by
    $
        Z_t = \xi_1 + \xi_2 t + \sin(2\pi \xi_3 t) + \VVV_t
    $
for $t\in \mathbb{T}$. The processes $X$ and $Y$ were then given by the
historical linear model \eqref{eq:historicallinear} with kernels $\rho_X$ and
$\rho_Y$ being one of the following four kernels:
\begin{align*}
    \begin{array}{l l}
        \text{zero: } (s,t)\mapsto 0, 
        & \text{constant: }(s,t)\mapsto 1, \\
        \text{Gaussian: } (s,t)\mapsto e^{-2(t-s)^2}, \qquad
        &\text{sine: } (s,t)\mapsto \sin(4t-20s).
    \end{array}
\end{align*}
To compute $X$ and $Y$, we evaluated the kernels on $\{(s,t)\in \mathbb{T}^2 \mid
s\leq t\}$ and approximated the integrals by Riemann sums. The full intensity for $N_t=\one(T\leq t)$ was specified with a Weibull baseline of the form
$
    \lambda_t^{\text{full}} = \one(T\geq t) \beta_1 t^2 
        \exp\left(\beta_2 Z_t + Y_t\right),
$
for $\beta_1>0$ and a choice of $\beta_2 \in \{-1,1\}$. 
To sample $T$ we applied the inverse hazard method, which utilizes that $\Lambda_T^{\text{full}}$ is standard exponentially distributed. That is, we sampled $E \sim \mathrm{Exp}(1)$ and numerically computed $T = \max\{ t\in \mathbb{T} \mid \Lambda_t^{\text{full}} < E\}$ as a discretized approximation.
For any given parameter setting, the baseline coefficient $\beta_1$ was chosen sufficiently large to ensure that 
$\Lambda_t^{\text{full}} \geq E$
would occur before time $t=1$ in more that $\frac{q-1}{q}\cdot n$ samples. 

The simulation setting used to sample the data for Figures \ref{fig:endpoint_example} and \ref{fig:timevaryingexample} was
$\beta_2=-1$ and $\rho_X=\rho_Y =\text{constant}$.

With this setup, Assumption \ref{asm:UniformBounds} is satisfied if $V$, $W$ 
and $\VVV$ were bounded. Since we use the Gaussian distribution, they are 
technically not bounded, but they could be made bounded by introducing a 
lower and upper cap. Due to the light tails of the Gaussian distribution
such caps would have no noticeable effect on the simulation results, and 
the results we report are generated without a cap. 

The implementation details for the X-LCT and the hazard ratio test are given in Section~\ref{fig:H0pvals_full} in the supplement.

\begin{figure}
    \includegraphics[width=\linewidth]{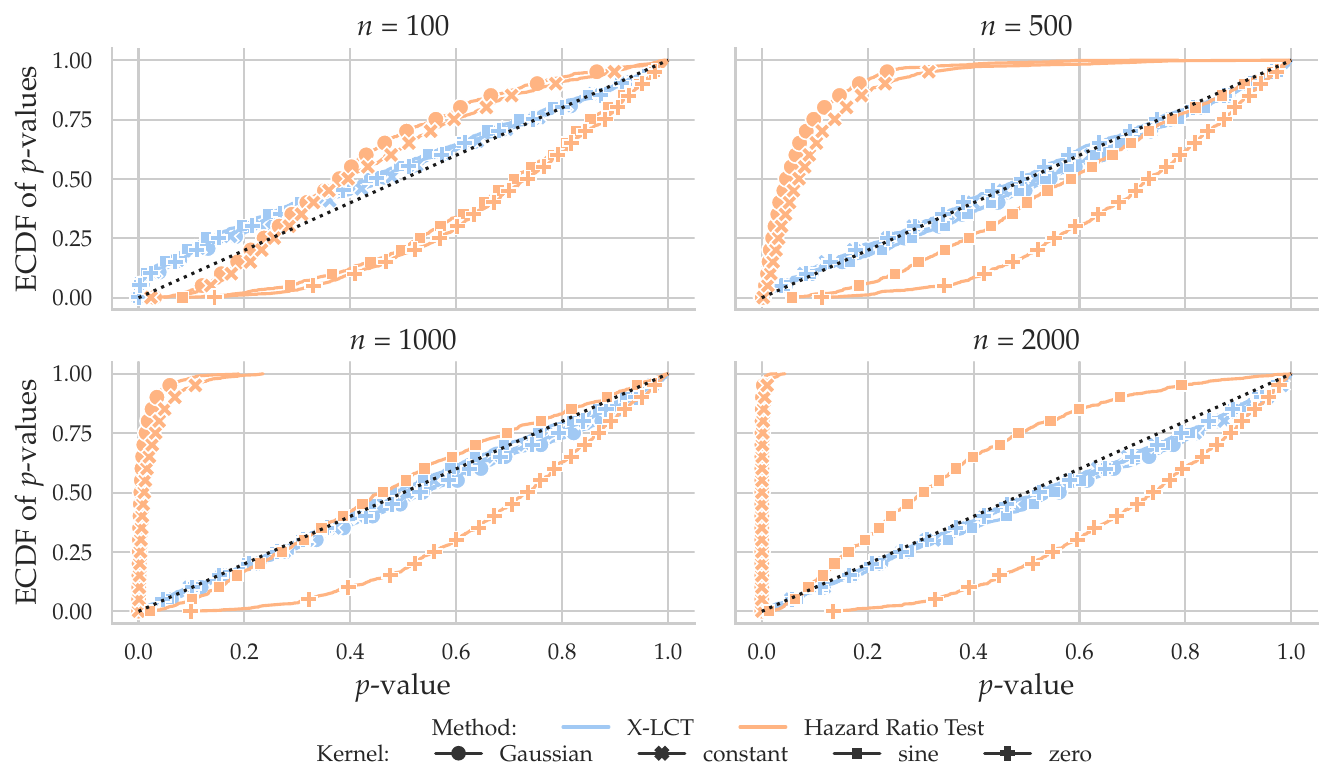}
        \caption{Empirical cumulative distribution functions of simulated
        $p$-values for the cross-fitted local covariance test and the hazard ratio test.
        The simulated data
        satisfies the hypothesis of conditional local independence, so the
        $p$-values are supposed to be uniformly distributed, and the CDF 
        should fall on the diagonal dotted line.}\label{fig:H0pvals2}
\end{figure}

\subsection{Distributions of $p$-values under $H_0$}
We examine the distributional approximation $\check{T}_{n}^K \stackrel{\mathrm{as.}}{\sim} S$, cf. Theorem \ref{thm:LCTXlevel}, by comparing the $p$-values
$1-F_S(\check{T}_{n}^K)$ to a uniform distribution. Figure \ref{fig:H0pvals2}
shows the empirical distribution functions of the $p$-values computed from data
simulated according to the scheme described in the previous section. 
The results are aggregated over the two choices of $\beta_2 \in \{-1,1\}$ since
these two settings were found to be similar. For more detailed results from
the experiment, see Figure~\ref{fig:H0pvals_full} in Section~\ref{sec:extrafigs} in the supplement, 
which also includes the $p$-values corresponding to the endpoint test statistic.

For the hazard ratio test, Figure \ref{fig:H0pvals2} shows that the $p$-values
are sub-uniform for the zero-kernel. In this case, the marginal Cox model is
correct, and the non-uniformity of the $p$-values can be explained by the
$L_2$-penalization. For the constant and Gaussian kernels the hazard ratio test
fails completely, whereas for the sine kernel, the mediated effect of $Z$ on $T$
through $Y$ is more subtle, and the model misspecification only becomes apparent
for $n=2000$. Overall, these results are consistent with the reasoning in the
Section \ref{sec:introex}: a test based on the misspecified Cox model will wrongly 
reject the hypothesis of conditional local independence. 

For the proposed X-LCT, Figure \ref{fig:H0pvals2} shows that the associated $p$-values are slightly anti-conservative for $n=100$. This is to be expected, and can be
explained by the finite sample errors leading to more extreme values of $\check{T}_{n}^K$ than the approximation by $S$. As $n$ increases, these errors
become smaller -- and for $n=2000$ the $p$-values actually seem to be
sub-uniform. The sub-uniformity may be explained by the time discretization,
since the maximum of the process is taken over $\mathbb{T}$ rather than $[0,1]$.
Figure~\ref{fig:supremumapprox} in Section~\ref{sec:extrafigs} in the supplement illustrates the asymptotic effect of
the time discretization which supports this claim. Another support of this claim
is that the endpoint test does not appear to give sub-uniform $p$-values for
large $n$, see Figure~\ref{fig:H0pvals_full}. We finally note that the 
distributions of the $p$-values for our proposed test is largely unaffected 
by the kernel used to generate the data. 

\begin{figure}
    \includegraphics[width = \linewidth]{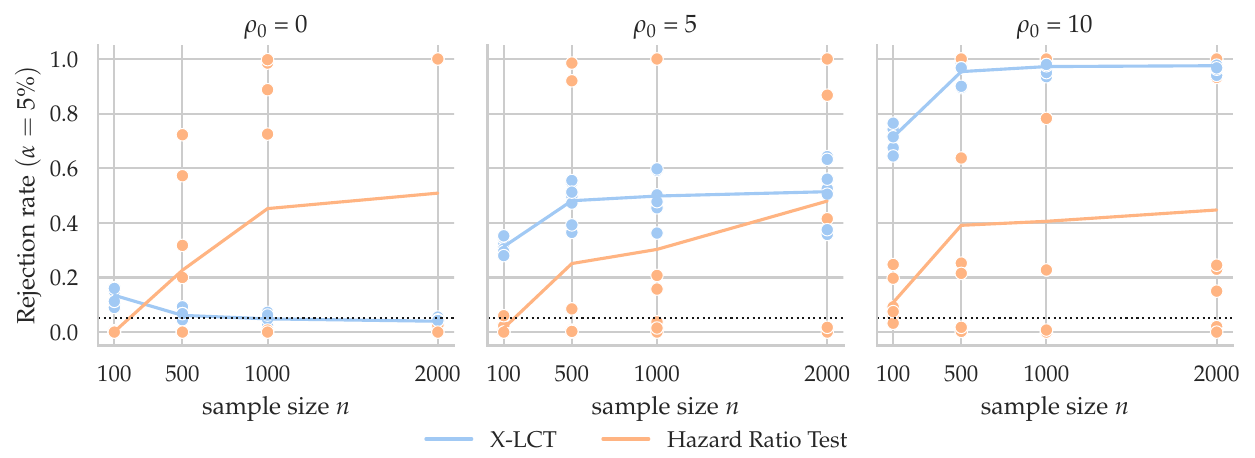}
    \caption{For each $\rho_0 \in \{0,5,10\}$, the lines show the average 
    rejection rates of our proposed test X-LCT (blue) and the hazard ratio test (orange)
    as functions of sample size, with each average taken over $8$ different 
    settings. For each setting, the rejection rate is computed from
    400 simulated datasets at a $5\%$ significance level 
    and the rejection rate is displayed with a dot.}
    \label{fig:alternatives}
\end{figure}

\subsection{Power against local alternatives} \label{subsec:localalternatives}
To investigate the power of the X-LCT we construct local alternatives to $H_0$ in
accordance with the right graph in Figure \ref{fig:lig} by replacing $Y_t$ by
the process $Y_t + \frac{\rho_0}{\sqrt{n}} X_t$. That is, for $\rho_0 \neq 0$,
blood pressure is then directly affected by pension savings, and $N_t$ is no
longer conditionally locally independent of $X_t$ given $\mathcal{F}_t$. In
terms of the full intensity, these local alternatives are equivalent to 
\begin{equation} \label{eq:alternative}
    \lambda_t^{\text{full}} = \one(T\geq t) \beta_1 t^2 
        \exp\left(\beta_2 Z_t + Y_t + \frac{\rho_0}{\sqrt{n}} X_t\right).
\end{equation}
We simulated data for the dependency parameter $\rho_0 \in \{0,5,10\}$. Note
that $\rho_0 = 0$ corresponds to our previous sampling scheme with conditional
local independence. For each of the $96 = 4\times 2 \times 4 \times 3$ choices
of kernel, $\beta_2$, $n$ and $\rho_0$ we ran the tests $400$ times and computed the
$p$-values. For simplicity, we report the rejection rate at an $\alpha = 5\%$
significance level and the results are shown in Figure \ref{fig:alternatives}.

In the leftmost panel, the data was generated under $H_0$ and the plot shows what
we noted previously, namely that the X-LCT holds level for large $n$,
whereas the hazard ratio test does not. 
For the local alternatives, $\rho_0=5$
and $\rho_0 = 10$, we note that the power of the hazard ratio test is quite
sensitive to the simulation settings. For some settings it has no power, while
for others it has some power. 

In contrast, the proposed X-LCT has power against all of the 
local alternatives. The power increases with $n$ initially but 
stabilizes from around $n=1000$. This is similar to the 
behavior observed under the null hypothesis and is not surprising.
We expect that the sample size needs to be sufficiently large for the 
nonparametric estimators to work sufficiently well, and we expect 
the sufficient sample size to be mostly 
unaffected by the value of $\rho_0$. For fixed $n$, we also note that the power 
of $\check{\Psi}_n^K$ is fairly robust with respect to the choice of $\beta_2$ 
and the choice of kernel. Overall, we find that the X-LCT
is applicable in these settings with historical effects: 
it has consistent power against the $\sqrt{n}$ alternatives while controlling type I error for $n$ reasonably large.

In Section~\ref{sec:extrafigs} in the supplement, we provide additional numerical results for time-varying alternatives, and we compare the X-LCT with its endpoint counterpart.

\section{Discussion} \label{sec:dis} 
The LCM was introduced as a functional parameter that quantifies deviations from the hypothesis $H_0$ of conditional local independence. 
We showed how the parameter may be expressed in several ways, but that it is the representation in terms of the residual process that allows us to estimate the LCM with a $\sqrt{n}$-rate under $H_0$ without parametric model assumptions. The residual process was introduced as an abstract model of $X_t$ for each $t$ given the history up to time $t$, and we showed that such a residualization could be viewed as a form of orthogonalization. Similar ideas have been used recently for classical conditional independence testing, such as GCM \citep{shah2020hardness}, tests based on the partial 
copula \citep{Petersen:2021}, and GHCM \citep{lundborg2021conditional}. It is, however,  not possible to use any of these to test $H_0$, which 
cannot be expressed as a classical conditional independence. Our test based on the LCM 
is the first nonparametric test of conditional local independence with substantial 
theoretical support, and we propose to test $H_0$ in practice by using X-LCT based on the cross-fitted estimator of LCM.

Contrary to the tests of conditional independence mentioned above, we need
sample splitting -- even under $H_0$ -- to achieve our asymptotic results. 
We do not believe that this can be avoided. The standard argument to avoid this uses 
classical conditional independence in a crucial way, which does not translate 
into our framework -- basically because we condition on information that
changes with time. Our simulation study also indicates 
that sample splitting 
or cross-fitting is needed in practice for the LCM estimator to be unbiased 
under $H_0$.

While our cross-fitted estimator of the LCM, the X-LCM, share 
some of the general patterns of other double machine learning procedures -- including 
the overall decomposition \eqref{eq:decomposition} -- our analysis and results required 
a range of generalizations of known results and some novel ideas. 
The asymptotic distribution of the leading term, $U^{(n)}$, is also a well 
known consequence of Rebolledo's CLT, see, e.g., Section V.4 in \citep{AndersenBorganGillKeiding:1993} for related results in the context of survival analysis. However, we generalized this result to uniform convergence in the 
Skorokhod space $D[0,1]$, and we introduced new techniques for handling 
the remainder terms. These novel techniques are made necessary by the 
decomposition \eqref{eq:decomposition} being a decomposition of 
stochastic processes indexed by time. 
We outline below the three most important technical contributions we made.


First, to obtain uniform control of level and power, all asymptotic results in Section \ref{sec:asymptotics} are formulated in terms of uniform stochastic convergence. Since this notion of convergence had not previously been considered on general metric spaces, and especially not on the Skorokhod space, we had to develop the necessary theory. This development could be of 
independent interest, and we have collected the general definitions and main results on uniform stochastic convergence in metric spaces in Section~\ref{app:UniformAsymptotics} in the supplement. 
This framework also allowed us to show a uniform version of Rebolledo's martingale CLT in Section~\ref{sec:fclt} in the supplement.

Second, to establish distributional convergence under $H_0$, we need to control the remainder terms $R_{i,t}^{(n)}$ uniformly over $t$. The third term, $R_{3}^{(n)}$, is simple to bound, and by exploiting Doob's submartingale inequality, the second term, $R_{2}^{(n)}$, can also be bounded. The most difficult first term, $R_{1}^{(n)}$, was controlled using  
stochastic equicontinuity via an exponential tail bound and the use of
the chaining lemma. The necessary general uniform stochastic equicontinuity and chaining arguments are collected in Section~\ref{app:UniformChaining} in the supplement.

Third, to achieve rate results in the alternative, the processes $D_1^{(n)}$ and $D_2^{(n)}$ must be controlled. The process $D_1^{(n)}$ does, like $U^{(n)}$, not involve any estimation, and its distributional convergence follows from a general CLT argument for continuous stochastic processes. The term $D_{2}^{(n)}$ is more difficult to handle, as it may not have mean zero if $G_t$ is not the additive residual process. However,  $X_t$ cancels out in $\widehat{G}_t^{(n)} - G_t$ for the additive residual process, which makes the difference $\mathcal{F}_t$-predictable, and $D_{2}^{(n)}$ can then be bounded similarly to $R_1^{(n)}$. For a general residual process, it seems possible for $D_{2}^{(n)}$ to have a bias of order $\sqrt{|J_n|}g(n)$.

Our main result, Theorem \ref{thm:main}, is stated under two assumptions. 
The second, Assumption~\ref{asm:UniformRates}, is a straightforward 
generalization to our setup of similar assumptions in the double machine 
learning literature on rates of convergence for the two estimators used. 
Both estimation errors are measured using a $2$-norm, and it is plausible 
that we can relax one norm to a weaker form of convergence if we simultaneously 
strengthen the other norm. The first assumption, Assumption \ref{asm:UniformBounds},
requires uniform bounds on both $\lambda$ and $G$. This is a strong assumption but 
perhaps not particularly problematic from a practical viewpoint. Indeed, 
$G$ is a process we can choose, and we can thus make it bounded if necessary. 
And though many theoretically interesting counting process models have unbounded 
intensities, a large cap on the intensity will make no difference in 
practice. We believe, nevertheless, that it is possible to relax Assumption \ref{asm:UniformBounds} to a weaker form of control on the 
magnitudes of $\lambda$ and $G$ as functions of time, e.g., moment bounds
uniform in $\theta$. However, such a generalization will come at the  
expense of considerably more technical proofs, and we did not pursue this 
line of research.

A major practical question is whether we can estimate $\lambda$ and $G$ 
with sufficient rates, e.g. $n^{- \frac{1}{4} + \epsilon}$. 
In Section~\ref{sec:estimationofnuisance} in the supplement we give an overview of some known and some conjectured rate results for specific forms of $\lambda$ and $\Pi$. 
Beyond parametric models we conclude that the existing rate results are scarce, and we regard it is as an independent research project to establish rates for general historical regression methods.

Another question is whether we can replace the counting process $N$ by a more general 
semimartingale. \cite{commenges2009} define conditional local independence for a 
class of special semimartingales, and \cite{Mogensen:2018} and \cite{Mogensen:2022} 
show global Markov properties for local independence graphs of certain Itô processes,
which are, in particular, special semimartingales.
Thus conditional local independence is well defined beyond counting 
processes, and we believe that most definitions and results of this paper would  
generalize beyond $N$ being a counting process. Besides some additional 
technical challenges, the major practical obstacle with such a generalization 
is that we cannot realistically assume to have completely observed sample paths 
of  Itô processes, say. The discrete time nature of the observations should 
then be included in the analysis, and this is beyond the scope of the present
paper. 

Irrespectively of the remaining open problems, the simulation 
study demonstrated some important properties of our proposed test, the X-LCT. 
First, it was fairly simple to implement for the specific example 
considered using some standard estimation techniques that were not 
tailored to the specific model class. Second, it had good level 
and power properties and clearly outperformed the test based on the 
misspecified marginal Cox model. Third, both Neyman orthogonalization 
as well as cross-fitting were pivotal for achieving
the good properties of the test.

\paragraph*{Acknowledgments}
The authors would like to thank the Associate Editor and the anonymous reviewers for constructive comments that lead to substantial improvements of the paper.
The work was supported by a grant (NNF20OC0062897) from the Novo Nordisk Foundation.

\paragraph*{Supplement}
The supplementary article \citep{christgau2023supplement} contains proofs of results from the main text, auxiliary results, additional discussions and figures. 

%% file: paper_lcm/appendix.tex
\clearpage
\section*{Supplement to `Nonparametric conditional local independence testing'}
The supplementary material is organized as follows: 
In Section~\ref{sec:proofs},
we give the proofs of the results of the main text.
In Section~\ref{app:UniformAsymptotics}, we formulate a general uniform asymptotic theory for metric spaces, whereafter we specialize the theory to the Skorokhod space $D[0,1]$ and chaining of stochastic processes.
In Section~\ref{sec:fclt}, we state Rebolledo's martingale central limit theorem, and then we generalize the result to a uniform version that is used in the proofs.
In Section~\ref{sec:estimationofnuisance}, we discuss estimation of the intensity $\lambda$ and the residual process $G$ in practice. In particular, we compare known rate results with the rates required in Assumption \ref{asm:UniformRates}.
In Section~\ref{sec:survmodels}, we compare the LCM estimator with existing work in semiparametric survival models.
In Section~\ref{sec:orthodetails}, we provide mathematical details regarding Neyman orthogonality.
Finally, Section~\ref{sec:extrafigs} contains  additional details, numerical results, and figures related to the simulation study of Section~\ref{sec:simulations}.

\supplementarysection{Proofs of results in the main text} \label{sec:proofs}

\subsection{Proof of Proposition \ref{prop:cli-mg}}

The process $G_t$ is \lc{} and $\mathcal{G}_t$-predictable by assumption, and the process
$I = (I_t)$ is a stochastic integral of $G_t$ w.r.t. a local
$\mathcal{G}_t$-martingale under the hypothesis $H_0$. It is thus 
also a local $\mathcal{G}_t$-martingale under $H_0$. By definition, 
$I_0 = 0$, and if $I$ is a martingale, $\gamma_t = \ex(I_t) =
\ex(I_0) = 0.$ \hfill \qedsymbol{} 

\subsection{Proof of proposition \ref{prop:gammaalternative}}
Suppose that $H$ is non-negative, \lc{} and $\mathcal{G}_t$-predictable, then since $\int_0^t H_s \mathrm{d}\bM_s$ is a local $\cG_t$-martingale it follows by monotone convergence along a localizing sequence that 
\begin{equation} \label{eq:Gid}
\ex\left(\int_0^t H_{s} \mathrm{d}N_{s}\right) = \ex\left(\int_0^t H_{s}\blambda_s\mathrm{d}s\right) = 
\int_0^t \ex(H_{s}\blambda_s) \mathrm{d}s
\end{equation}
for all $t\in[0,1]$. We can apply the identity 
above with $H$ the positive and negative part of $G$, respectively, 
and the integrability assumption ensures that \eqref{eq:Gid} also holds with $H = G$. It follows that
\begin{align*}
    \gamma_t = \ex (I_t)
    = \ex\left(\int_0^t G_{s}(\blambda_s 
        - \lambda_s)\mathrm{d}s\right) 
    = \int_0^t \ex\pa{G_{s}(\blambda_s 
        - \lambda_s)}\mathrm{d}s.
\end{align*}
The latter expectation is indeed a covariance since $\ex(G_s) = \ex(\ex(G_s\mid\cF_{s-})) = 0$. \hfill \qedsymbol{}

\subsection{Proof of Lemma \ref{lem:truemartingales}}
Before proving Lemma \ref{lem:truemartingales}, we first state general martingale criteria in the context of counting processes.
\begin{lem}\label{lem:countingmartingales}
    Let $(H_t)$ be a locally bounded $\cG_t$-predictable process, 
    let $N$ be a counting process with a $\cG_t$-intensity $\blambda_t$, and let $\bM_t = N_t - \int_0^t \blambda_s \mathrm{d}s$.
    
    If $\int_0^1 \blambda_s \mathrm{d}s$ 
    (or equivalently $N_1$)
    is integrable, then $\bM_t$ and $\bM_t^2-\int_0^t \blambda_s \mathrm{d}s$ are each $\cG_t$-martingales. 
    If, in addition, $\int_0^1 H_s^2 \blambda_s \mathrm{d}s$ is integrable, then
    $\int_0^t H_s \mathrm{d} \bM_s$
    is a mean zero square integrable martingale.
\end{lem}
\begin{proof}
    The first part is Lemma 2.3.2 and Theorem 2.5.3 in \citet{fleming2011counting}. 
    For the second part, assume that $\int_0^1 \blambda_s \mathrm{d}s$
    and $\int_0^1 H_s^2 \blambda_s \mathrm{d}s$ are both integrable.
    In this case, the $\cG_t$-predictable quadratic variation of $\bM$ is $\langle \bM \rangle(t) = \int_0^t\blambda_s\mathrm{d}s$ by the first part.
    Then it remains to note that $(H_t)$ is a locally bounded $\cG_t$-predictable process, so the conditions of Theorem 2.4.4 in \citet{fleming2011counting} are satisfied if $\int_0^1 H_s^2 \blambda_s \mathrm{d}s$ is integrable. This establishes the second part.
\end{proof}
We now return to the proof of Lemma \ref{lem:truemartingales}. 
Let $f\in C(\mathbb{R})$, and we shall prove that $\int_0^t f(G_s)\mathrm{d}\bM_s$ is a mean zero, square integrable $\mathcal{G}_t$-martingale. The proof for the integral with $f(\widehat{G}_s^{(n)})$ is identical. 

Continuity of $f$ implies that $C_f \coloneqq \sup_{x\in [-C',C']}|f(x)| <\infty$ and that $(f(G_t))$ is a $\mathcal{G}_t$-predictable process. By Assumption \ref{asm:UniformBounds}, the process $(f(G_t))$ is almost surely bounded by $C_f$ and therefore
\begin{align*}
    \ex\left(\int_0^1 f(G_s)^2 \blambda_s \mathrm{d}s \right) 
    \leq C_f^2 C < \infty.
\end{align*}
Thus we can apply Lemma \ref{lem:countingmartingales} to conclude that $\int_0^t f(G_s)\mathrm{d}\bM_s$ is a mean zero, square integrable $\mathcal{G}_t$-martingale.
\hfill \qedsymbol{}

\subsection{Proof of Proposition \ref{prop:UasymptoticGaussian}}

As noted elsewhere, the explicit parametrization of all objects by $\theta$ is 
notationally heavy, and there will thus be an implicit parameter value $\theta \in \Theta$
in most of the subsequent constructions and arguments.

To simplify notation we write
    \begin{align*}
    U^{(n)}_t = \sum_{j \in J_n} \int_0^t H_{j, s}^{(n)} \mathrm{d} \bM_{j, s},
    \quad \text{where} \quad 
    H_{j, s}^{(n)} = \frac{G_{j, s}}{\sqrt{|J_n|}}.
    \end{align*}
We will use a uniform extension of Rebolledo's martingale central limit theorem on the sequence $(U^{(n)})_{n \geq 1}$ to show the result. See Section~\ref{sec:fclt} for a discussion of Rebolledo's CLT and Theorem \ref{thm:URebo} for its uniform extension.

Define $\tilde{\mathcal{G}}_{t}^n$ be the smallest right continuous and complete filtration generated by the filtrations $\{\mathcal{G}_{j, t} \mid j \in J_n\}$. We can apply Lemma~\ref{lem:truemartingales} to each of the terms of $U^{(n)}$ to conclude that the $j$-th term is a square integrable, mean zero $\mathcal{G}_{j,t}^n$-martingale.
By independence of the observations for each $j$, we can enlarge the filtration for each term and conclude that they are also square integrable, mean zero $\mathcal{\tilde{G}}_t^n$-martingales. Thus $U^{(n)}$ is also a square integrable, mean zero $\mathcal{\tilde{G}}_t^n$-martingale.

To apply Theorem \ref{thm:URebo} first establish that the conditions in Equation \eqref{eq:RebolledoConditions} are fulfilled. By Proposition \ref{prop:Rebolledospecialcase}, we have that
    \begin{align*}
    \big\langle 
    U^{(n)}
    \big\rangle (t) 
    & = 
    \sum_{j \in J_n} \int_0^t \left( H_{j, s}^{(n)} \right)^2 \blambda_{j, s}
    \mathrm{d} s
    = 
    \frac{1}{|J_n|} \sum_{j \in J_n}
        \int_0^t G_{j, s}^2 \blambda_{j, s} \mathrm{d}s.
    \end{align*}
By directly applying the bounds from Assumption \ref{asm:UniformBounds}, we see that the square mean of $\int_0^t G_s^2 \blambda_s \mathrm{d}s$ is bounded by $C^2(C')^4$. 
Thus, for fixed $t\in[0,1]$, the uniform law of large numbers \citep[Lemma 19]{shah2020hardness} gives that
    \begin{align*}
        \big\langle 
        U^{(n)}
        \big\rangle (t)
        = \frac{1}{|J_n|} \sum_{j \in J_n}
        \int_0^t G_{j, s}^2 \blambda_{j, s} \mathrm{d} s 
        \convUP
        \ex \left( \int_0^t G_{s}^2 \blambda_{s}\mathrm{d} s\right) 
        = \mathcal{V}(t)
    \end{align*}
for $n \to \infty$, since the integrals are i.i.d. with the 
same distribution as $\int_0^t G_{s}^2 \blambda_{s}\mathrm{d} s$. 
This establishes the first part of the condition in Equation \eqref{eq:RebolledoConditions}.
For the second part, we also have from Proposition \ref{prop:Rebolledospecialcase} that
    \begin{align} \label{eq:rarefactionterms}
    \big\langle 
    U_{\epsilon}^{(n)}
    \big\rangle (t) 
    & =
    \sum_{j \in J_n} \int_0^t 
    \left( H_{j, s}^{(n)} \right)^2 \one \left( |H_{j, s}^{(n)}| \geq \epsilon \right)
    \mathrm{d} \bLambda_{j, s}
        \nonumber \\
    & = 
    \frac{1}{|J_n|} \sum_{j \in J_n}
    \int_0^t
    G_{j, s}^2 
    \one \left(
    \left|
    G_{j, s}
    \right| \geq \epsilon \sqrt{|J_n|}
    \right)
    \blambda_{j, s} \mathrm{d} s
    \end{align}
for each $t \in [0,1]$ and $\epsilon > 0$. 
From Assumption \ref{asm:UniformBounds}, we note that 
for $n$ sufficiently large such that $|J_n| > (C')^2/\epsilon^2$, it holds that
\(
    \mathbb{P}
    \left(
        \left|
        G_{j, s}
        \right| \geq \epsilon \sqrt{|J_n|}
    \right) = 0
\)
for all $j\in J_n$. As a consequence, the terms in \eqref{eq:rarefactionterms} are almost surely zero for $n$ sufficiently large uniformly over $\Theta$. It follows that
\(
    \big\langle 
    U_{\epsilon}^{(n)}
    \big\rangle (t) \convUP 0
\), which establishes the second part of \eqref{eq:RebolledoConditions}.

We finally note that the collection of variance functions, $( \mathcal{V}^\theta)_{\theta \in \Theta}$, is uniformly equicontinuous and bounded above under Assumption \ref{asm:UniformBounds}. This is established in Lemma \ref{lem:regularityofgammasigma} below. We have thus verified all the conditions of Theorem~\ref{thm:URebo}, so we conclude that
    \begin{align*}
    U^{(n),\theta} \convUD U^\theta
    \end{align*}
in $D[0,1]$ as $n\to \infty$, where $U^\theta$ is a mean zero continuous Gaussian martingale with variance function $\mathcal{V}^\theta$. 
\hfill \qedsymbol{} \\

Note that the convergence of \eqref{eq:rarefactionterms} is established directly from the uniform bounds in Assumption \ref{asm:UniformBounds}. However, the convergence could also be established under a milder conditions with alternative arguments. For example, under the weaker assumption of uniformly bounded variance functions, dominated convergence can be used to establish $L_1$-convergence. 

In the proof above we invoked the following lemma, which we will also use in several proofs in the sequel.

\begin{lem}\label{lem:regularityofgammasigma}
Under Assumption \ref{asm:UniformBounds}, the collections $(\gamma^\theta)_{\theta\in \Theta}$ and $(\mathcal{V}^\theta)_{\theta \in \Theta}$ are each uniformly Lipschitz and in particular uniformly equicontinuous. Moreover, it holds almost surely that
\begin{align*}
        \sup_{t\in [0,1]} |\gamma_t| \leq 2CC' 
        \qquad \text{and} \qquad
        \mathcal{V}(1) = \ex \left(\int_0^1 G_s^2 \blambda_s \mathrm{d}s\right) \leq C(C')^2.
    \end{align*}
\end{lem}
\begin{proof}
For any $0\leq s < t \leq 1$, a direct application of Assumption~\ref{asm:UniformBounds} and Proposition~\ref{lem:truemartingales} yields
\begin{align*}
    |\gamma_t-\gamma_s| 
        \leq \ex \left\lvert
                \int_s^t G_u\mathrm{d}N_u 
                    - \int_s^t G_u\lambda_u \mathrm{d}u
            \right \rvert 
        \leq \ex \left(\int_s^t |G_u|(\blambda_u  + \lambda_u)\mathrm{d}u\right)
        \leq 2CC'(t-s),
\end{align*}
and similarly,
\begin{align*}
    \mathcal{V}(t) - \mathcal{V}(s) 
        = \ex \left(\int_s^t G_u^2 \blambda_u \mathrm{d}u\right)
        \leq C(C')^2(t-s).
\end{align*}
This establishes the first part. The bounds follow from inserting $(s,t)=(0,t)$ in the first inequality and $(s,t)=(0,1)$ in the second inequality. 
\end{proof}

\subsection{Proof of Proposition \ref{prop:remainderterms}}

We will divide the proof into three lemmas for each of the remainder terms $R_1^{(n)}, R_2^{(n)}$ and $R_3^{(n)}$, where we establish convergence to the zero-process uniformly over $t$ and $\theta$. However, note that the notion of uniform convergence differs for the process index, $t\in [0,1]$, and the parameter, $\theta \in \Theta$, as we need to show that

\begin{align*}
        \forall i \in \{1,2,3\}\forall \epsilon>0: \quad 
        \lim_{n\to \infty}
        \sup_{\theta \in \Theta} \mathbb{P}\Big(
            \sup_{t\in[0,1]}|R_{i,t}^\theta|>\epsilon
        \Big) = 0.
\end{align*}

For a general discussion of the relation between weak convergence and convergence in probability uniformly as a stochastic process, see \cite{Newey:1991}. For a general discussion of uniform stochastic convergence over a distribution parameter, see Section~\ref{app:UniformAsymptotics} and the references contained therein. 
In Section~\ref{app:UniformChaining}, we discuss the combination of both convergences.

As in the proof of Proposition \ref{prop:UasymptoticGaussian}, 
$\tilde{\mathcal{G}}_t^n$ denotes the smallest right continuous and complete filtration 
generated by the filtrations $\{\mathcal{G}_{j, t} \mid j \in J_n\}$. 
Analogously, we let $\mathcal{\tilde{G}}_t^{n,c}$ be the smallest 
right continuous and complete filtration generated by the filtrations 
$\{ \mathcal{G}_{j, t} \mid j \in J_n^c \}$.
We start by considering $R_3^{(n)}$, since this is the easiest case. 

\begin{lem} \label{lem:R3}
Under Assumption~\ref{asm:UniformRates} it holds that $\sup_{t\in [0,1]} |R_{3,t}^{(n)}| \convUP 0$. 
\end{lem}

\begin{proof}
We will show the result by showing that
    \begin{align*}
    \sup_{\theta \in \Theta}
    \ex \left(
        \sup_{0 \leq t \leq 1} |R_{3, t}^{(n),\theta} |
    \right) \to 0
    \end{align*}
as $n \to \infty$. Using that the random variables 
    \begin{align*}
            \sup_{0 \leq t \leq 1} \left|
            G_{j, t} - \widehat{G}_{j, t}^{(n)}
            \right|
            \cdot 
            \sup_{0 \leq t \leq 1} \left|
            \lambda_{j, t} - \widehat{\lambda}_{j, t}^{(n)}
            \right|
    \end{align*}
for $j \in J_n$ are identically distributed for each fixed $n \geq 2$, we have that
    \begin{align*}
        & \ex \left(
        \sup_{0 \leq t \leq 1} | R_{3, t}^{(n)} |
        \right) 
        \\
        & = 
        \ex \left( \sup_{0 \leq t \leq 1} \left|
        \frac{1}{\sqrt{|J_n|}} 
        \sum_{j \in J_n}
        \int_0^t 
        \left( G_{j, s} - \widehat{G}_{j, s}^{(n)} \right) 
        \left( \lambda_{j, s} - \widehat{\lambda}_{j, s}^{(n)} \right) 
        \mathrm{d}s \right|
        \right)
        \\
        &   \leq \frac{1}{\sqrt{|J_n|}} 
            \sum_{j \in J_n}
            \ex \left( \sup_{0 \leq t \leq 1} 
            \int_0^t \left|
            G_{j, s} - \widehat{G}_{j, s}^{(n)} \right| \cdot
            \left| \lambda_{j, s} - \widehat{\lambda}_{j, s}^{(n)} \right|
            \mathrm{d}s
            \right) \\
        &   = \sqrt{|J_n|} \ex \left(  
            \int_0^1 \left|
            G_{s} - \widehat{G}_{s}^{(n)} \right| \cdot
            \left| \lambda_{s} - \widehat{\lambda}_{s}^{(n)} \right|
            \mathrm{d}s
            \right) \\
        &   \leq \sqrt{|J_n|} \ex \left( 
            \sqrt{ \int_0^1 \left( 
                G_{s} - \widehat{G}_{s}^{(n)} \right)^2 \mathrm{d}s}
            \sqrt{ \int_0^1 \left( 
                \lambda_{s} - \widehat{\lambda}_{s}^{(n)} \right)^2 \mathrm{d}s}
            \right) \\
        &   \leq \sqrt{|J_n|}
            \sqrt{\ex \left( 
                 \int_0^1 \left( 
                    G_{s} - \widehat{G}_{s}^{(n)} \right)^2 \mathrm{d}s\right)}
            \sqrt{\ex \left( \int_0^1 \left( 
                    \lambda_{s} - \widehat{\lambda}_{s}^{(n)} \right)^2 
                    \mathrm{d}s\right)} \\
        &   = \sqrt{|J_n|} g(n)h(n).
    \end{align*}
By Assumption \ref{asm:UniformRates}, $\sqrt{|J_n|} g(n)h(n)\to 0$ uniformly over $\Theta$ as $n \to \infty$, so the result follows.
\end{proof}


Next we proceed to the remainder process $R_2^{(n)}$.

\begin{lem} \label{lemma:R2}
Under Assumptions~\ref{asm:UniformBounds} and \ref{asm:UniformRates}, it holds 
that 
\(
    \sup_{t\in [0,1]} |R_{2,t}^{(n)}| \convUP 0
\).
\end{lem}

\begin{proof}
We first write
    \begin{align*}
        R^{(n)}_{2, t} 
        &= 
            \frac{1}{\sqrt{|J_n|}} \sum_{j \in J_n}  
            \int_0^t \left( G_{j, s} - \widehat{G}_{j, s}^{(n)} \right) 
            \mathrm{d} \bM_{j, s}, 
    \end{align*}
and note that $R^{(n)}_{2, t}$ is a square integrable, mean zero $\tilde{\mathcal{G}}_t^n$-martingale conditionally on $\tilde{\mathcal{G}}_1^{n,c}$.
This follows by applying Lemma~\ref{lem:truemartingales} to each of the terms, which are i.i.d. conditionally on $\tilde{\mathcal{G}}_1^{n,c}$. 
We conclude that the squared process
$(R_{2,t}^{(n)} )^2$ is a $\tilde{\mathcal{G}}_t^n$-submartingale conditionally on $\tilde{\mathcal{G}}_1^{n,c}$. By Doob's
submartingale inequality we have that 
    \begin{align*}
    \mathbb{P} \left(
        \sup_{0 \leq t \leq 1} |  R_{2, t}^{(n)} | \geq \epsilon
    \right)
    & =
    \mathbb{P} \left( 
        \sup_{0 \leq t \leq 1} \left( R_{2, t}^{(n)} \right)^2
        \geq \epsilon^2
    \right) 
    \\
    & = 
    \ex \left( 
    \mathbb{P} \left(
        \sup_{0 \leq t \leq 1} \left(  R_{2, t}^{(n)} \right)^2 \geq \epsilon^2
        \mid
        \tilde{\mathcal{G}}_1^{n,c}
    \right)
    \right)
    \\
    & \leq 
    \frac{\ex \left( \var \left( R_{2, 1}^{(n)} \mid
    \tilde{\mathcal{G}}_1^{n,c} \right) \right)}{\epsilon^2} 
    \end{align*}
for $\epsilon > 0$. The collection of random variables
    \begin{align*}
        \left(
            \int_0^1 \left( G_{j, s} - \widehat{G}_{j, s}^{(n)} \right) 
        \mathrm{d} \bM_{j, s}
        \right)_{j \in J_n}
    \end{align*}
are i.i.d. conditionally on $\tilde{\mathcal{G}}_1^{n,c}$. Therefore,
    \begin{align*}
        \var \left( R_{2, 1}^{(n)} \mid \tilde{\mathcal{G}}_1^{n,c} \right)
        & = 
        \frac{1}{|J_n|} \sum_{j \in J_n} \var \left( 
            \int_0^1 \left( G_{j, s} - \widehat{G}_{j, s}^{(n)} \right) 
            \mathrm{d} \bM_{j, s} 
            \mid \tilde{\mathcal{G}}_1^{n,c}
        \right) 
        \\
        & = 
        \ex \left( \int_0^1
        \left( G_{s} - \widehat{G}_{s}^{(n)} \right)^2 
        \mathrm{d} \langle \bM \rangle_s 
        \mid \tilde{\mathcal{G}}_1^{n,c} \right)
        \\
        & =
        \ex \left(
        \int_0^1 \left( G_{s} - \widehat{G}_{s}^{(n)} \right)^2 \blambda_{s}
        \mathrm{d}s 
        \mid \tilde{\mathcal{G}}_1^{n,c}
        \right) \\
        & \leq C \cdot \ex \left(
                \int_0^1 \left( G_{s} - \widehat{G}_{s}^{(n)} \right)^2
                \mathrm{d}s 
                \mid \tilde{\mathcal{G}}_1^{n,c}
                \right)
    \end{align*}
where we have used that $\blambda_t$ is bounded by Assumption~\ref{asm:UniformBounds} (i). Thus 
\begin{align*}
    \ex \left(
        \var \left( R_{2, 1}^{(n)} \mid \tilde{\mathcal{G}}_1^{n,c} \right)
    \right) 
    \leq 
    C \cdot \ex \left(
    \int_0^1 \left( G_{s} - \widehat{G}_{s}^{(n)} \right)^2 \mathrm{d}s
    \right) 
    = C \cdot g(n)^2,
\end{align*}
and we conclude that 
    \begin{align*}
        \mathbb{P} \left(
        \sup_{0 \leq t \leq 1} | R_{2, t}^{(n)} | \geq \epsilon
    \right)
    \leq 
    \frac{C\cdot g(n)^2}{\epsilon^2} \to 0,
    \end{align*}
as $n \to \infty$ uniformly over $\Theta$ by Assumption~\ref{asm:UniformRates}.
\end{proof}

Before proving that $R_1^{(n)}$ converges weakly to the zero-process, we will need two auxiliary lemmas. The first is a conditional version of Hoeffding's lemma, which lets us conclude conditional sub-Gaussianity. Recall that a mean zero random variable $A$ is sub-Gaussian with variance factor $\nu>0$ if 
    \begin{align*}
        \log \ex(e^{xA}) \leq \frac{x^2 \nu}{2}
    \end{align*}
for all $x \in \mathbb{R}$. See, for example, \citet{boucheron2013concentration}, Lemma 2.2, for the classical unconditional version.

\begin{lem}[conditional Hoeffding's lemma]\label{lem:conditionalHoeffding}
    Let $Y$ be a random variable taking values on a bounded interval $[a,b]$, satisfying $\ex[Y|\mathcal{G}]=0$ for a $\sigma$-algebra $\mathcal{G}$. 
    
    Then $\log \ex(e^{xY}\mid\mathcal{G}) \leq (b-a)^2 x^2 /8$ almost surely for all $x\in \mathbb{R}$.
\end{lem}
\begin{proof}
    Fix $x\in \mathbb{R}$. By convexity of the exponential function we have
    \begin{align*}
        e^{xy} \leq \frac{b-y}{b-a} e^{xa} + \frac{y-a}{b-a} e^{xb},
        \qquad y\in[a,b].
    \end{align*}
    Inserting $Y$ in place of $y$ and taking the conditional expectation yields
    \begin{align*}
        \ex[e^{xY}\mid\mathcal{G}]
        \leq \frac{b}{b-a} e^{xa} - \frac{a}{b-a} e^{xb}
        = e^{L(x(b-a))}
    \end{align*}
    almost surely,
    where $L(h) = \frac{ha}{b-a} + \log(1+\frac{a-e^ha}{b-a})$.
    Standard calculations show that $L(0)=L'(0)=0$, and the AM-GM inequality implies
    \begin{align*}
        L''(h) = - \frac{abe^h}{(b-ae^h)^2} \leq \frac{1}{4}.
    \end{align*}
    Thus, a second order Taylor expansion yields that $L(h)\leq \frac{1}{8}h^2$, and it follows that
    $\log \ex[e^{xY}\mid\mathcal{G}] \leq 
    \frac{(b-a)^2}{8}x^2$ as desired.  
\end{proof}

For the next lemma define for $s, t \in [0,1]$ with $s < t$ 
    \begin{align*}
        W^{s, t} = \frac{1}{t - s} \int_s^t G_u (\lambda_u - \widehat{\lambda}^{(n)}_u)
        \mathrm{d}u. 
    \end{align*}

\begin{lem} \label{lemma:subgaussian}
Let Assumption~\ref{asm:UniformBounds} hold true. Then, for any $0\leq s < t \leq 1$, it holds that $\ex(W^{s,t}\mid\tilde{\mathcal{G}}_1^{n,c}) = 0$ and 
that $W^{s,t}$ is sub-Gaussian conditionally on $\tilde{\mathcal{G}}_1^{n,c}$ with variance factor $\nu = (2CC')^2$, that is,
    \begin{align*}
        \log \ex (e^{xW^{s, t}} \mid \tilde{\mathcal{G}}_1^{n,c}) 
        \leq 2 (xCC')^2
    \end{align*}
for all $s < t$ and $x \in \mathbb{R}$. 
\end{lem}

\begin{proof}
For fixed $u\in [0,1]$, note that
\begin{align} \label{eq:R1meanzero}
    \ex\left( G_u
    \left( 
    \lambda_u - \widehat{\lambda}^{(n)}_u \right) 
    \mid
    \tilde{\mathcal{G}}_1^{n,c}
    \right) 
    & = 
    \ex\left( 
    \ex\left( 
    G_u
    \left( 
    \lambda_u - \widehat{\lambda}^{(n)}_u \right) 
    \mid 
    \mathcal{F}_{s-} \vee \tilde{\mathcal{G}}_1^{n,c}
    \right) 
    \mid
    \tilde{\mathcal{G}}_1^{n,c}
    \right) \nonumber \\
    & = 
    \ex\left(
    \ex\left( G_u \mid \mathcal{F}_{s-} \right)
     \left( 
        \lambda_u - \widehat{\lambda}^{(n)}_u \right) 
    \mid 
    \tilde{\mathcal{G}}_1^{n,c}
    \right) = 0,
    \end{align}
where we have used that $\lambda_{t} - \widehat{\lambda}^{(n)}_{t}$ is 
$\mathcal{F}_{t}$-predictable conditionally on $\tilde{\mathcal{G}}_1^{n,c}$, 
that $G_{t}$ is independent of $\tilde{\mathcal{G}}_1^{n,c}$ since it is $\mathcal{G}_{t}$-predictable, and that 
$\ex\left( G_{s} \mid \mathcal{F}_{s-} \right) = 0$ per definition. 
By applying the conditional Fubini theorem \citep[Theorem 27.17]{schilling2017measures}, we conclude that $\ex(W^{s,t}\mid\tilde{\mathcal{G}}_1^{n,c}) = 0$.

We can now use the conditional version of Hoeffding's lemma formulated in Lemma \ref{lem:conditionalHoeffding}. 
Indeed, we have that for all $s<t$
    \begin{align*}
        |W^{s, t}| & \leq \frac{1}{t - s} \int_s^t | G_u |
        | (\lambda_u - \widehat{\lambda}^{(n)}_u) |
        \mathrm{d}u
        \\
        & \leq \sup_{0 \leq u \leq 1}|G_u|
        \sup_{0 \leq u \leq 1}|(\lambda_u - \widehat{\lambda}^{(n)}_u)|
        \leq 2 C C' 
    \end{align*}
by Assumption~\ref{asm:UniformBounds}. Hence, for all $s < t$, Lemma \ref{lem:conditionalHoeffding} lets us conclude that 

\(
        \log \ex (e^{xW^{s, t}} \mid \tilde{\mathcal{G}}_1^{n,c}) 
        \leq 2 (xCC')^2
\), 
$x \in \mathbb{R}$.
\end{proof}

Then we have the following regarding $R_1^{(n)}$.

\begin{lem} \label{lem:R1}
Under Assumptions~\ref{asm:UniformRates} and \ref{asm:UniformBounds} it holds that $\sup_{t\in [0,1]} |R_{1,t}^{(n)}| \convUP 0$. 
\end{lem}

\begin{proof}
The proof consists of two parts. First we show that for each $t \in [0,1]$
it holds that
    \begin{align*}
    R_{1,t}^{(n)} \convUP 0
    \end{align*}
for $n \to \infty$. Then we show 
\emph{stochastic equicontinuity} of the process 
$R_{1}^{(n)}$ uniformly over $\Theta$, and by Lemma \ref{lem:uniformequcont} it follows
that 
    \begin{align*}
    \sup_{t \in [0,1]} |R_{1,t}^{(n)}| \convUP 0.
    \end{align*}
This is a direct generalization of Theorem 2.1 in \cite{Newey:1991}.
The collection of random variables 
    \begin{align*}
    \left(
        G_{j,s}
        \left( \lambda_{j, s} - \widehat{\lambda}^{(n)}_{j, s} \right)
    \right)_{j \in J_n}
    \end{align*}
are i.i.d. conditionally on $\tilde{\mathcal{G}}_1^{n,c}$. Therefore, an application of the conditional Fubini theorem yields
    \begin{align*}
    \ex( R_{1, t} \mid \tilde{\mathcal{G}}_1^{n,c}) 
    & = 
    \frac{1}{\sqrt{|J_n|}} \sum_{j \in J_n} \int_0^t 
    \ex\left( 
    G_{j,s}
    \left( \lambda_{j, s} - \widehat{\lambda}^{(n)}_{j, s} \right) 
    \mid 
    \tilde{\mathcal{G}}_1^{n,c}
    \right)
    \mathrm{d}s
    = 0
    \end{align*}
where the last equality follows from the computation in \eqref{eq:R1meanzero}.
Whence $\ex(R^{(n)}_{1, t}) = 0$, and 
$\var(R^{(n)}_{1, t}) = \ex(\var(R^{(n)}_{1, t} \mid \tilde{\mathcal{G}}_1^{n,c}))$, so
    \begin{align*}
    \var(R_{1, t}^{(n)}) 
    & =
    \ex \left( 
    \frac{1}{|J_n|} \sum_{j \in J_n} \var \left( 
        \int_0^t 
        G_{j,s}
        \left( \lambda_{j, s} - \widehat{\lambda}^{(n)}_{j, s} \right) \mathrm{d}s 
        \mid 
        \tilde{\mathcal{G}}_1^{n,c}
    \right)
    \right)  
    \\
    & =
    \ex \left(
    \ex \left(
    \left(
        \int_0^t 
        G_{s}
        \left( \lambda_{s} - \widehat{\lambda}^{(n)}_{s} \right) \mathrm{d}s 
    \right)^2
    \mid 
    \tilde{\mathcal{G}}_1^{n,c}
    \right)
    \right)
    \\
    & = 
    \ex \left(
    \left(
        \int_0^t 
        G_{s}
        \left( \lambda_{s} - \widehat{\lambda}^{(n)}_{s} \right) \mathrm{d}s 
    \right)^2
    \right)
    \\
    & \leq (C')^2 \ex \left( \int_0^t 
        \left( \lambda_{s} - \widehat{\lambda}^{(n)}_{s} \right)^2 
            \mathrm{d}s \right)\\
    & \leq (C')^2 h(n)^2
    \end{align*}
where we have used Assumption~\ref{asm:UniformBounds} (ii). Hence by Chebychev's inequality, it holds for all $\epsilon>0$ that
    \begin{align*}
    \mathbb{P}(|R^{(n)}_{1, t}| > \epsilon)
    \leq 
    \frac{\var(R^{(n)}_{1, t})}{\epsilon^2} 
    \leq \frac{(C')^2 h(n)^2}{\epsilon^2}  \longrightarrow 0
    \end{align*}
as $n \to \infty$ uniformly over $\Theta$ by Assumption~\ref{asm:UniformRates}. This completes the first part of the proof. 
For the second part, we use a chaining argument based on the exponential inequality in 
Lemma~\ref{lemma:subgaussian}. We let 
    \begin{align*}
    W^{s, t}_j = \frac{1}{t - s} \int_s^t G_{j, u}
    \left(
    \lambda_{j, u} - \widehat{\lambda}_{j, u}^{(n)}
    \right) \mathrm{d}u
    \end{align*}
and 
    \begin{align*}
        A
        =  \frac{1}{\sqrt{|J_n|}} 
            \sum_{j \in J_n} W_j^{s, t}
        = 
        \frac{1}{t - s} (R_{1, t}^{(n)} - R_{1, s}^{(n)}).
    \end{align*}
Using that $(W_j^{s, t})_{j \in J_n}$ are i.i.d. conditionally on $\tilde{\mathcal{G}}_1^{n,c}$ 
we have by Lemma~\ref{lemma:subgaussian} that
$\ex(A)=0$ and that
    \begin{align*}
    \log \ex \left( e^{x A} \right)
    & = 
    \log \ex \left( \ex \left( e^{x A} \mid \tilde{\mathcal{G}}_1^{n,c} \right) \right) 
    \\
    & = \log  \ex \left( \prod_{j \in J_n}
    \ex\left( e^{\frac{x}{\sqrt{|J_n|}} W_j^{s, t}} 
    \mid 
    \tilde{\mathcal{G}}_1^{n,c}
    \right)
    \right) 
    \\
    & \leq \log  \ex \left( e^{\frac{x^2 \nu}{2}}
    \right) 
    \\
    & = \frac{x^2 \nu}{2}.
    \end{align*}
Hence $A$ is also sub-Gaussian with variance factor $\nu$. This implies that
    \begin{align*}
        \mathbb{P}(|A| > \eta) \leq 2 e^{- \frac{\eta^2 \nu}{2}}
    \end{align*}
for all $\eta > 0$. Rephrased in terms of $R_1^{(n)}$ this bound reads
    \begin{align*}
    \mathbb{P} \left(
    |R_{1, t}^{(n)} - R_{1, s}^{(n)}| > \eta (t - s)
    \right) \leq 2 e^{-\frac{\eta^2 \nu}{2}}
    \end{align*}
for all $\eta > 0$ and $s < t$. It now follows from the chaining lemma, \cite{Pollard:1984} Lemma 
VII.9, that $R_1^{(n)}$ is stochastic equicontinuous. 
Since the variance factor $\nu=(2CC')^2$ does not depend on $\theta \in \Theta$, we have stochastic equicontinuity uniformly over $\Theta$ by Corollary \ref{cor:uniformchaining}. 
This completes the second part of the proof and we are done.
\end{proof}

Note that the second part of the proof above establishes stochastic 
equicontinuity by a bound on the probability that the increments 
of the process are large. This is a well known technique, see, e.g.,
Example 2.2.12 in \cite{Vaart:1996}, from which the same conclusion 
will follow if 
$$\ex(|R_{1, t}^{(n)} - R_{1, s}^{(n)}|^p) \leq K |t - s|^{1 + r}$$
for $K,p,r > 0$.

Proposition~\ref{prop:remainderterms} now follows from combining the Lemmas \ref{lem:R1}, \ref{lemma:R2}, and \ref{lem:R3}.
\hfill \qed{}

\subsection{Proof of Proposition \ref{prop:Dasymptotics}}
We separate the discussion of $D_1^{(n)}$ and $D_2^{(n)}$ into the Lemmas \ref{lem:D1asymptotics} and \ref{lem:D2asymptotics}, respectively, which together amount to Proposition \ref{prop:Dasymptotics}.

\begin{lem}\label{lem:D1asymptotics}
    Suppose that Assumptions \ref{asm:UniformBounds} and \ref{asm:UniformRates} hold. Then the stochastic process
    $\overline{D}^{(n)} \coloneqq D_1^{(n)}-\sqrt{|J_n|} \cdot \gamma$
    converges in distribution in $C[0,1]$ uniformly over $\Theta$.
\end{lem}
\begin{proof}
Let $\overline{D}^{(n)} \coloneqq D_1^{(n)}-\sqrt{|J_n|} \cdot \gamma$ and note that 
$$
    \overline{D}^{(n)} = |J_n|^{-\frac{1}{2}} \sum_{j\in J_n} W_j,
$$
where $W_j$ is given by $W_{j,t} \coloneqq \int_0^t G_{j,s} (\blambda_{j,s} - \lambda_{j,s})\mathrm{d}s - \gamma_t$ for each $j\in J_n$. By assumption, the variables $\{W_j\colon j\in J_n\}$ are i.i.d. with the same distribution as the process $W$ given by $W_{t} \coloneqq \int_0^t G_{s} (\blambda_{s} - \lambda_{s})\mathrm{d}s - \gamma_t$.
For each $\theta\in \Theta$, let $\Gamma^\theta$ be a Gaussian process with mean zero and covariance function $(s,t)\mapsto \cov(W_s^\theta,W_t^\theta)$, which is well-defined by computations shown below.

We will show that $\overline{D}^{(n),\theta} \convUD \Gamma^\theta$ in $C[0,1]$ by applying Lemma \ref{lem:ProkhorovsPrinciple}, which is an example of Prokhorov's method of "tightness + identification of limit".
We first prove that for any given $k\in \mathbb{N}$ and $0\leq t_1<t_2<\cdots<t_k\leq 1$,
\begin{align*}
    \mathbf{D}^{(n)} \coloneqq
    (\overline{D}_{t_1}^{(n)},\overline{D}_{t_2}^{(n)},\ldots, \overline{D}_{t_k}^{(n)})
    \convUD
    (\Gamma_{t_1}^\theta,\Gamma_{t_2}^\theta,\ldots, \Gamma_{t_k}^\theta).
\end{align*}
To this end we will apply the uniform CLT of \citet[Proposition 19]{lundborg2021conditional}
to the sequence of random vectors $\mathbf{D}^{(n)}\in \mathbb{R}^k$, i.e., the sequence of normalized sums of i.i.d. copies of $\mathbf{W} \coloneqq (W_{t_1},\ldots , W_{t_k})$.
The process $(W_t)$ is mean zero and hence $\mathbf{W}$ is also mean zero. For any $t\in [0,1]$ we observe that
\begin{align*}
    \var(W_{t}) = \var(W_{t}+\gamma_t)
    \leq \ex\left[\pa{\int_0^t |G_s|\cdot |\blambda_s - \lambda_s|\mathrm{d}s}^2
        \right]
    \leq 2C^2(C')^2.
\end{align*}
Therefore the trace of $\var(\mathbf{W})$ is uniformly bounded, which is implies the trace condition in Proposition 19 of \citet{lundborg2021conditional}. 
From Hölder's inequality and Minkowski's inequality, we note that for any $\mathbf{a},\mathbf{b}\in \mathbb{R}^k$
$$
    \|\mathbf{a} + \mathbf{b}\|_2^3 
        \leq k^{3/2}\|\mathbf{a} + \mathbf{b}\|_3^3 
        \leq k^{3/2} (\|\mathbf{a}\|_3 + \|\mathbf{b}\|_3)^3 
        \leq 8 k^{3/2} (\|\mathbf{a}\|_3^3 + \|\mathbf{b}\|_3^3).
$$
Combining the above with Assumption \ref{asm:UniformBounds} and Lemma \ref{lem:regularityofgammasigma} yields that 
$$
    \ex[\|\mathbf{W}\|_2^3] 
    \leq 
    C_k \ex\left[\pa{\int_0^1 |G_s|\cdot|\blambda_s - \lambda_s|\mathrm{d}s}^3\right]
    + C_k \sup_{t\in[0,1]} |\gamma_t|^3
    \leq 16 C_k C^3(C')^3,
$$
where $C_k=8k^{5/2}$. Hence Proposition 19 of \citet{lundborg2021conditional} lets us conclude that $\mathbf{D}^{(n)} \convUD \mathcal{N}(0,\var(\mathbf{W}))$. By definition of $\Gamma^\theta$, this is equivalent to 
$\mathbf{D}^{(n)} \convUD (\Gamma_{t_1}^\theta,\Gamma_{t_2}^\theta,\ldots, \Gamma_{t_k}^\theta)$.
 
We now argue that $(\overline{D}^{(n)})$ and $(\Gamma^\theta)$ are stochastically equicontinuous uniformly over $\Theta$. From the definition of $\Gamma^\theta$ and by Assumption \ref{asm:UniformBounds}, it follows that 
\begin{align} \label{eq:squareincrement}
    \ex[(\Gamma_t^\theta-\Gamma_s^\theta)^2] 
    = \ex[(W_t-W_s)^2] 
    \leq (2CC' (t-s))^2.
\end{align}
Hence $\frac{1}{t-s}(\Gamma_t^\theta-\Gamma_s^\theta)$ is Gaussian with a variance bounded over $\Theta$ and $0\leq s<t\leq 1$. In particular, it is sub-Gaussian with a uniform variance factor over $\Theta$ and $0\leq s<t\leq 1$.
Since $W$ is uniformly bounded over $\Theta$, an application of Hoeffding's Lemma yields that $A_j^{t,s}\coloneqq \frac{1}{t-s}(W_{j,t}-W_{j,s})$ is also sub-Gaussian with a variance factor $\nu$ that is uniform over $\Theta$, $0\leq s<t\leq 1$, and $j\in J_n$. Letting $A_\bullet^{s,t} =\frac{1}{t-s}(\overline{D}_t^{(n)}-\overline{D}_s^{(n)})$, we have
\begin{align*}
    \ex e^{x A_\bullet^{s,t}} 
    =   \prod_{j\in J_n}  
        \ex[e^{x|J_n|^{-1/2}A_j^{s,t}}]
    \leq \prod_{j\in J_n} e^{\frac{x^2\nu}{2|J_n|}} 
    = e^{x^2\nu/2}.
\end{align*}
Hence $A_\bullet^{s,t}$ is also sub-Gaussian with a variance factor uniformly over $\Theta$ and $0\leq s<t\leq 1$.

From the uniform chaining lemma, Corollary \ref{cor:uniformchaining}, we now conclude that both $(\Gamma^\theta)$ and $(\overline{D}^{(n)})$ are stochastically equicontinuous uniformly over $\Theta$. By Proposition~\ref{prop:equicontinuityistightness}, this means that the collection $(\overline{D}^{(n),\theta})$ is sequentially tight
and that $(\Gamma^\theta)$, which is constant in $n$, is uniformly tight. 

Now we have shown convergence of the finite-dimensional marginals and appropriate tightness conditions, so Lemma \ref{lem:ProkhorovsPrinciple} lets us conclude that $\overline{D}^{(n)} \convUD \Gamma^\theta$ weakly in $C[0,1]$.
\end{proof}

Before moving on to the term $D_2^{(n)}$, we first note that Lemma \ref{lem:D1asymptotics} implies that stochastic boundedness, as we will use this result in the proof of Theorem \ref{thm:main}.

\begin{lem}\label{lem:D1bounded}
    Suppose that Assumptions \ref{asm:UniformBounds} and \ref{asm:UniformRates} hold.
    Then $\overline{D}^{(n)} \coloneqq D_1^{(n)}-\sqrt{|J_n|} \cdot \gamma$ is stochastically bounded uniformly over $\Theta$, i.e., 
    for every $\varepsilon>0$ there exists $K>0$ such that
    \begin{align*}
        \limsup_{n\to \infty}\sup_{\theta\in \Theta} \mathbb{P}\pa{
        \|\overline{D}^{(n),\theta}\|_\infty > K} < \varepsilon.
    \end{align*}
\end{lem}
\begin{proof}
    We have established in the proof of Lemma \ref{lem:D1asymptotics}, under the same conditions, that $\overline{D}^{(n),\theta} \convUD \Gamma^\theta$ weakly in $C[0,1]$.
    By the uniform continuous mapping theorem formulated in Proposition \ref{prop:Uctsmapping}, it follows that $\|\overline{D}^{(n),\theta}\|_\infty \convUD \|\Gamma^\theta\|_\infty$.
    From \citet{bengs2019uniform} Theorem 4.1 we then obtain that
    \begin{align*}
        \limsup_{n\to \infty} \sup_{\theta \in \Theta}
            \mathbb{P}( \|\overline{D}^{(n),\theta}\|_\infty > K) 
        \leq  \sup_{\theta \in \Theta}\mathbb{P}(\|\Gamma^\theta\|_\infty > K) 
        \leq \frac{\ex\|\Gamma^\theta\|_\infty}{K}.
    \end{align*}
    Hence it suffices to argue that $\ex \|\Gamma^\theta\|_\infty$ is uniformly bounded over $\Theta$.
    To this end, we note that Equation \eqref{eq:squareincrement} 
    shows that square means of the increments of $\Gamma^\theta$
    are smaller that those of a standard Brownian motion scaled by $2CC'$. Then the Sudakov–Fernique comparison inequality \citep[Theorem 2.2.3]{adler2007random} allows us to leverage this relationship to the expected uniform norms, i.e., $\ex\|\Gamma^\theta\|_\infty \leq 2CC' \ex (\sup_{t\in[0,1]}|B_t|)$. It can be verified that $\ex (\sup_{t\in[0,1]}|B_t|)$ is finite, and in fact, equal to $\sqrt{\pi/2}$ as shown in \citet{saz2019expectedsup}. 
\end{proof}

\begin{lem}\label{lem:D2asymptotics}
    Suppose that Assumptions \ref{asm:UniformBounds} and \ref{asm:UniformRates}
    hold, and that $G_t = X_t - \Pi_t$ is the additive residual process. 
    Then $D_2^{(n)}\convUP 0$ in $D[0,1]$ as $n\to \infty$.
\end{lem}
\begin{proof}
Note first that the terms in $D_2^{(n)}$ are i.i.d. conditionally on $\tilde{\cG}_1^{n,c}$, with the same distribution as the process $\xi$ given by
\begin{align*}
    \xi_t = \frac{1}{\sqrt{|J_n|}}\int_0^t (\widehat{G}_s^{(n)} - G_s)(\blambda_s - \lambda_s)\mathrm{d}s.
\end{align*}
Since $\blambda_t$ is independent of $\tilde{\cG}_1^{n,c}$, we have from the innovation theorem that
$$
    \ex(\blambda_t \mid \mathcal{F}_{t-}\vee \tilde{\cG}_1^{n,c})
    = \ex(\blambda_t \mid \mathcal{F}_{t-}) = \lambda_t.
$$
For the additive residual process we also note that
$G_t-\widehat{G}_t^{(n)} = \widehat{\Pi}_t^{(n)} - \Pi_t$ is $\mathcal{F}_{t}$-predictable conditionally on $\tilde{\cG}_1^{n,c}$. It now follows that
\begin{align*}
    \sqrt{|J_n|} \cdot \ex[\xi_t\mid\tilde{\cG}_1^{n,c}]
    &= 
        \int_0^t \ex[(\widehat{G}_s^{(n)} - G_s)(\blambda_s - \lambda_s)\mid\tilde{\cG}^{n,c}]\mathrm{d}s \\
    &= 
        \int_0^t \ex[(\widehat{G}_s^{(n)} - G_s)(\ex[\blambda_s\mid\mathcal{F}_{s-}\vee \tilde{\cG}^{n,c}]- \lambda_s)\mid\tilde{\cG}^{n,c}]\mathrm{d}s
    = 0.
\end{align*}
We can therefore conclude that $D_2^{(n)}$ is mean zero conditionally on $\tilde{\cG}^{n,c}$. 
Using that the terms of $D_2^{(n)}$ are i.i.d. conditionally on $\tilde{\cG}^{n,c}$ once more, we now obtain that
\begin{align*}
    \var(D_{2,t}^{(n)}\mid\tilde{\cG}^{n,c}) 
    = |J_n| \cdot \var( \xi_t \mid\tilde{\cG}^{n,c} )
    &= \ex\left[\pa{
        \int_0^t (\widehat{G}_s^{(n)} - G_s)(\blambda_s - \lambda_s)\mathrm{d}s
        }^2 \Big|\tilde{\cG}^{n,c} \right] \\
    &\leq 4C^2 \cdot \ex\left(
        \int_0^1 (\widehat{G}_s^{(n)} - G_s)^2 \mathrm{d}s
        \Big|\tilde{\cG}^{n,c} 
        \right).
\end{align*}
Taking expectation of the above we have 
$\var(D_{2,t}^{(n)}) = \ex(\var(D_{2,t}^{(n)}\mid\tilde{\cG}^{n,c}))
\leq 4 C^2 g(n)^2$.
By Chebyshev's inequality we get for all $\epsilon>0$
\begin{align*}
    \mathbb{P}\pa{|D_{2,t}^{(n)}|>\epsilon} 
        \leq \frac{4C^2 g(n)^2}{\epsilon^2},
\end{align*}
and by Assumption \ref{asm:UniformRates} we conclude that $D_{2,t}^{(n)} \convUP 0$ for each $t\in [0,1]$. 

We now apply the same chaining argument used in the proofs of Lemma \ref{lem:R1} and Lemma \ref{lem:D1asymptotics}. From Assumption \ref{asm:UniformBounds}, we have for $0\leq s<t\leq 1$ that $|\xi_t - \xi_s| \leq 4 \sqrt{|J_n|}CC' (t-s)$. Hence the conditional Hoeffding's lemma (Lemma \ref{lem:conditionalHoeffding}) yields that
$$
    A_j^{s,t} 
        = \frac{1}{t-s}
        \int_s^t (\widehat{G}_{j,s}^{(n)} - G_{j,s})(\blambda_{j,s} - \lambda_{j,s})\mathrm{d}s
$$
is sub-Gaussian conditionally on $\tilde{\cG}_1^{n,c}$ with a variance factor $\nu$ that is uniform over $\Theta$ and $s<t$ (cf. the proof of Lemma \ref{lemma:subgaussian}). Letting $A_\bullet^{s,t} =\frac{1}{t-s}(D_{2,t}^{(n)}-D_{2,s}^{(n)})$, we have for any $x\in \mathbb{R}$
\begin{align*}
    \ex \left(e^{x A_\bullet^{s,t}} \right)
    &= \ex\left(\ex\left[e^{x A_\bullet^{s,t}}
        \mid\tilde{\cG}_1^{n,c}\right]\right) \\
    &= \ex\bigg(\prod_{j\in J_n}  
            \ex\left[e^{x|J_n|^{-1/2}A_j^{s,t}}\mid\tilde{\cG}_1^{n,c}
        \right]\bigg)
    \leq \prod_{j\in J_n} e^{\frac{x^2\nu}{2|J_n|}} 
    = e^{x^2\nu/2},
\end{align*}
so $A_\bullet^{s,t}$ is also sub-Gaussian uniformly over $s<t$ and 
$\Theta$. In terms of $D_2^{(n)}$, this means that we can apply the uniform chaining lemma, Corollary \ref{cor:uniformchaining}, and conclude that it is stochastically equicontinuous uniformly over $\Theta$. 

Since $D_{2,t}^{(n)} \convUP 0$ for each $t\in [0,1]$ and $(D_2^{(n)})$ is stochastically equicontinuous uniformly over $\Theta$, Lemma \ref{lem:uniformequcont} now lets us conclude that $\sup_{t\in [0,1]} |D_{2,t}^{(n)}| \convUP 0$ and we are done.
\end{proof}

\subsection{Proof of Theorem \ref{thm:main}}
Before proving Theorem \ref{thm:main}, we first prove that the collection of Gaussian martingales from Proposition \ref{prop:UasymptoticGaussian} is tight in $C[0,1]$ (see Definition \ref{dfn:tight}). 
\begin{lem}\label{lem:Ulimitistight}
    Let $(U^\theta)_{\theta \in \Theta}$ be the collection of Gaussian martingales from Proposition~\ref{prop:UasymptoticGaussian}, i.e., $U^\theta$ is a mean zero continuous Gaussian martingale with variance function $\mathcal{V}^\theta$. 
    Under Assumption \ref{asm:UniformBounds}, $(U^\theta)_{\theta \in \Theta}$ is uniformly tight in $C[0,1]$.
\end{lem}
\begin{proof}
We will use Theorem 7.3 in \citet{billingsley2013convergence}, which characterizes tightness of measures in $C[0,1]$. The first condition of the theorem is trivially satisfied for $(U^\theta)_{\theta \in \Theta}$ since $\mathbb{P}(U_0^\theta = 0)=1$ for all $\theta \in \Theta$. 

By Proposition \ref{prop:BMrepresentation}, $U^\theta$ has a distributional representation as a time-transformed Brownian motion such that $(U_t^\theta)_{t\in[0,1]} \overset{\mathcal{D}}{=} (B_{\mathcal{V}^\theta(t)})_{t\in [0,1]}$, where $B$ is a Brownian motion.
Recall that Brownian motion is $\alpha$-Hölder continuous for $\alpha \in (0,\frac12)$, which means 
that 
$$K(\alpha) = \sup_{s \neq t} \frac{|B_t - B_s|}{|t - s|^\alpha} < \infty.$$
Note also that the collection of variance functions is uniformly Lipschitz by Lemma \ref{lem:regularityofgammasigma} with uniform Lipschitz constant $C_0$, say. It follows that for every $\epsilon>0$,
\begin{align*}
    \lim_{\delta\to 0^+} \sup_{\theta \in \Theta}
    \mathbb{P}\Big(\sup_{|t-s|<\delta}|U_t^\theta 
        - U_s^\theta|>\epsilon \Big)
    &= \lim_{\delta\to 0^+} \sup_{\theta \in \Theta}\mathbb{P}\Big(\sup_{|t-s|<\delta}|B_{\mathcal{V}^\theta(t)} 
        - B_{\mathcal{V}^\theta(s)}|>\epsilon \Big)\\
    &\leq 
    \lim_{\delta\to 0^+} \sup_{\theta \in \Theta}\mathbb{P}\Big(K(\alpha) \sup_{|t-s|<\delta}|\mathcal{V}^\theta(t) 
        - \mathcal{V}^\theta(s)|^\alpha >\epsilon \Big) \\
    &= \lim_{\delta\to 0^+} \mathbb{P}\Big(K(\alpha) C_0^\alpha \delta^\alpha >\epsilon \Big) = 0.
\end{align*}
This establishes the second condition of Theorem 7.3 in \citet{billingsley2013convergence}, and we thus conclude that $(U^\theta)_{\theta \in \Theta}$ is uniformly tight in $C[0,1]$. 
\end{proof}

We now return to the proof of Theorem \ref{thm:main}.

For part i), we first note that under $H_0$ we can take $\blambda_t = \lambda_t$,
which implies that both $D_1^{(n)}$ and $D_2^{(n)}$ equal the zero-process. 

Combining Propositions \ref{prop:UasymptoticGaussian} and \ref{prop:remainderterms} with the uniform version of Slutsky's theorem formulated in Lemma \ref{lem:GeneralSlutsky}, we conclude that
\begin{align*}
    \sqrt{|J_n|}\widehat{\gamma}^{(n)}
        = 
        \underbrace{U^{(n)}}_{\convUDnull U^\theta} 
        + \underbrace{R_1^{(n)}+R_2^{(n)}+R_3^{(n)}}_{
            \convUP 0}
        + \underbrace{D_1^{(n)}+D_2^{(n)}}_{
            =0 \, \text{under } H_0}
        \convUPnull U^\theta,
\end{align*}
in $D[0,1]$ as $n\to \infty$, where $U^\theta$ is the Gaussian martingale from Proposition \ref{prop:UasymptoticGaussian}.

For part ii) we can, in addition to Propositions \ref{prop:UasymptoticGaussian} and \ref{prop:remainderterms}, apply Proposition \ref{prop:Dasymptotics} and Lemma~\ref{lem:D1bounded}. Using the triangle inequality on the decomposition \eqref{eq:decomposition} yields that
\begin{align*}
    \sqrt{|J_n|}\cdot \|\widehat{\gamma}^{(n)} - \gamma\|_\infty
        \leq
        &\|U^{(n)}\|_\infty + \|D_1^{(n)} - \sqrt{|J_n|}\gamma \|_\infty \\
            &+ \|R_1^{(n)}\|_\infty 
            + \|R_2^{(n)}\|_\infty 
            + \|R_3^{(n)}\|_\infty 
            + \|D_2^{(n)}\|_\infty.
\end{align*}
All the terms in the second line converge in probability to zero uniformly over $\Theta$. Combined with the convergences established in Proposition \ref{prop:UasymptoticGaussian} and Lemma \ref{lem:D1asymptotics}, we obtain that
\begin{align}\label{eq:LCMboundednes}
    &\limsup_{n\to \infty}\sup_{\theta\in \Theta} \mathbb{P}\pa{
    \sqrt{|J_n|} \cdot \|\widehat{\gamma}^{(n),\theta} - \gamma^\theta\|_\infty > K} \nonumber \\
    & \qquad \leq 
    \sup_{\theta\in \Theta} 
        \mathbb{P}\pa{ \| U^\theta \|_\infty > K/6}
    + \sup_{\theta\in \Theta} 
        \mathbb{P}\pa{ \| \Gamma^\theta \|_\infty > K/6},
\end{align}
where $\Gamma^\theta$ is the limiting Gaussian process from (the proof of) Lemma \ref{lem:D1asymptotics}. 
The last term in \eqref{eq:LCMboundednes} can be made arbitrarily small for $K$ sufficiently large by Lemma \ref{lem:D1bounded}. 
Lemma~\ref{lem:Ulimitistight} states that the family $(U^\theta)_{\theta \in \Theta}$ is tight in $C[0,1]$, and hence the family $(\|U^\theta\|_\infty)_{\theta \in \Theta}$ is tight in $\mathbb{R}_{\geq 0}$. This implies that the first term in \eqref{eq:LCMboundednes} can also be made arbitrarily small for $K$ sufficiently large.
This establishes \eqref{eq:stochboundedness} and we are done.
\hfill \qed{}

\subsection{Proof of Proposition \ref{prop:varianceconsistent}}
Consider the decomposition of the variance function estimator given by
    \begin{align*}
        \widehat{\mathcal{V}}_n(t) & = A^{(n)}_t + B^{(n)}_t + 2 C^{(n)}_t
    \end{align*}
where 
    \begin{align*}
        A^{(n)}_t & =
        \frac{1}{|J_n|} \sum_{j \in J_n} 
        \int_0^t G_{j,s}^2 \mathrm{d} N_{j,s},
        \\
        B^{(n)}_t & = 
        \frac{1}{|J_n|} \sum_{j \in J_n} 
        \int_0^t \left( G_{j,s} - \widehat{G}_{j,s}^{(n)} \right)^2 \mathrm{d} N_{j,s}, 
        \\
        C^{(n)}_t & = 
        \frac{1}{|J_n|} \sum_{j \in J_n}
        \int_0^t G_{j,s} \left(G_{j,s} - \widehat{G}_{j,s}^{(n)} \right) 
        \mathrm{d} N_{j,s}.
    \end{align*}
We first consider the asymptotic limit of $A^{(n)}$, which is the empirical mean of $|J_n|$ i.i.d. samples of the process $\int_0^t G_s^2\mathrm{d}N_s$.
Under Assumption \ref{asm:UniformBounds}, we can apply the first part of Lemma \ref{lem:countingmartingales} which states $\bM_t^2 - \bLambda_t$ is a martingale. We use this fact to note that
\begin{align*}
    \ex(N_1^2) 
        = \ex((\bM_1 + \bLambda_1)^2)
        \leq 2\left(\ex(\bM_1^2) + \ex(\bLambda_1^2)\right)
        = 4\, \ex\left(\Big(\int_0^1 \blambda_s \mathrm{d}s\Big)^2\right)
        \leq 4 C^2.
\end{align*}
Now, another use of Assumption \ref{asm:UniformBounds} shows that $\int_0^t G_s^2\mathrm{d}N_s$ has a second moment bounded by $4(CC')^2$. 
Thus we can apply the uniform law of large numbers \cite[Lemma 19]{shah2020hardness} to conclude for each $t\in [0,1]$,
\begin{align*}
    A_t^{(n)} 
        = \frac{1}{|J_n|} \sum_{j \in J_n} 
            \int_0^t G_{j,s}^2 \mathrm{d} N_{j,s}
        \convUP 
        \ex\left(\int_0^t G_{s}^2 \mathrm{d} N_s \right) = \mathcal{V}(t).
\end{align*}
Note also that $A^{(n)}$ and $\mathcal{V}$ are non-decreasing and that the collection $(\mathcal{V}^\theta)_{\theta\in\Theta}$ is uniformly equicontinuous by Lemma \ref{lem:regularityofgammasigma}. These are exactly the conditions for Lemma \ref{lem:convergenceofincreasing}, so we can automatically conclude that $\sup_{t\in [0,1]}|A_t^{(n)} - \mathcal{V}(t)|\convUP 0$.

Next we show that the remainder terms  
$B^{(n)}$ and $C^{(n)}$ converge uniformly 
to zero in expectation. Similarly to the proof of Lemma~\ref{lemma:R2}, we have under Assumptions \ref{asm:UniformBounds} and \ref{asm:UniformRates},
    \begin{align*}
        \ex \left( \sup_{0 \leq t \leq 1} B_t^{(n)} \right) 
        = \ex (B_1^{(n)}) 
        & = \ex \left(
            \frac{1}{|J_n|} \sum_{j \in J_n} 
            \int_0^1 \left( G_{j,s} - \widehat{G}_{j,s}^{(n)} \right)^2 \blambda_{j, s} \mathrm{d}s
        \right) 
        \\
        & = \ex \left(
        \ex \left( 
        \frac{1}{|J_n|} \sum_{j \in J_n} 
            \int_0^1 \left( G_{j,s} - \widehat{G}_{j,s}^{(n)} \right)^2 
            \blambda_{j, s} \mathrm{d}s
        \mid 
        \tilde{\mathcal{G}}_1^c
        \right)
        \right)
        \\
        & = 
        \ex \left( 
            \int_0^1 \left( G_{s} - \widehat{G}_{s}^{(n)} \right)^2 \blambda_{s} \mathrm{d}s
        \right)
        \\
        & \leq C \cdot g(n)^2 \longrightarrow 0
    \end{align*}
as $n \to \infty$ uniformly over $\Theta$. 
Lastly, we see that
    \begin{align*}
        \ex \left| \sup_{0 \leq t \leq 1} C_t^{(n)} \right| 
        & \leq 
        \ex \left( \sup_{0 \leq t \leq 1} |C_t^{(n)}| \right)
        \\
        & \leq 
        \ex \left( 
            \frac{1}{|J_n|} \sum_{j \in J_n}
            \sup_{0 \leq t \leq 1}
            \int_0^t |G_{j,s}| |G_{j,s} - \widehat{G}_{j,s}^{(n)}| \blambda_{j, s}
            \mathrm{d} s
        \right)
        \\
        & =
        \ex \left(
            \frac{1}{|J_n|} \sum_{j \in J_n}
            \int_0^1 |G_{j,s}| |G_{j,s} - \widehat{G}_{j,s}^{(n)}| \blambda_{j, s}
            \mathrm{d} s
        \right)
        \\
        & = 
        \ex \left( 
        \ex \left( 
            \frac{1}{|J_n|} \sum_{j \in J_n}
            \int_0^1 |G_{j,s}| |G_{j,s} - \widehat{G}_{j,s}^{(n)}| \blambda_{j, s}
            \mathrm{d} s
        \mid
        \tilde{\mathcal{G}}_1^c
        \right) 
        \right) 
        \\
        & = 
        \ex \left( 
            \int_0^1 |G_{s}||G_{s} - \widehat{G}_{s}^{(n)}| \blambda_{s}
            \mathrm{d} s
        \right)
        \\
        & \leq CC' \ex \left( 
                \int_0^1 |G_{s} - \widehat{G}_{s}^{(n)}|
                \mathrm{d} s
            \right) \\
        & \leq CC' \cdot g(n) \longrightarrow 0
    \end{align*}
as $n \to \infty$ uniformly over $\Theta$ by Assumption~\ref{asm:UniformRates}. 
Combining the convergences established for $A^{(n)}$, $B^{(n)}$, and $C^{(n)}$, we get by a generalized Slutsky (Lemma \ref{lem:SkorokhodSlutsky}) that 
$$
    \sup_{t\in[0,1]}|\widehat{\mathcal{V}}_n(t) - \mathcal{V}(t)| \convUP 0.
$$
\hfill \qed{}

\subsection{Proof of Corollary \ref{cor:statisticsdistribution}}
Under Assumptions~\ref{asm:UniformBounds} and \ref{asm:UniformRates} we know by Theorem~\ref{thm:main} and Proposition~\ref{prop:varianceconsistent} that
    \begin{align}\label{eq:LCMandVarconvergent}
    \sqrt{|J_n|} \widehat{\gamma}^{(n),\theta} \convUDnull U^\theta
    \qquad \text{and} \qquad
    \widehat{\mathcal{V}}_n^\theta \convUPnull \mathcal{V}^\theta
    \end{align}
in $D[0,1]$ as $n \to \infty$. If we were to show pointwise convergence of the test statistic, this would now be a straightforward consequence of the continuous mapping theorem. 
However, to show uniform convergence, we will need an additional tightness argument. 

Let $(\theta_n)_{n\in\mathbb{N}}\subset \Theta_0$ be an arbitrary sequence. Proposition \ref{prop:seqUni} then states that it suffices to show that there exists a subsequence $(\theta_{k(n)})_{n\in \mathbb{N}} \subseteq (\theta_{n})_{n\in \mathbb{N}}$, with $k\colon \mathbb{N} \to \mathbb{N}$ strictly increasing, such that
\begin{equation}\label{eq:ProofOfTeststatistic}
    \lim_{n\to \infty} d_{BL}\big(\widehat{D}_{k(n)}^{\theta_{k(n)}}, 
    \mathcal{J}(U^{\theta_{k(n)}},\mathcal{V}^{\theta_{k(n)}})\big)
    = 0.
\end{equation}
Here $d_{BL}$ denotes the bounded Lipschitz metric defined in Section \ref{app:UniformAsymptotics}. By Lemma~\ref{lem:Ulimitistight}, the collection $(U^\theta)_{\theta\in \Theta}$ is tight in $C[0,1]$ under Assumption \ref{asm:UniformBounds}. Therefore, Prokhorov's theorem \citep[Theorem 23.2]{kallenberg2021foundations} asserts that there exists a subsequence $(\theta_{a(n)}) \subset (\theta_n)$, and a $C[0,1]$-valued random variable $\tilde{U}$ such that $U^{\theta_{a(n)}} \xrightarrow{\mathcal{D}} \tilde U$ in $C[0,1]$.

Likewise, Lemma \ref{lem:regularityofgammasigma} states that the collection $(V^\theta)_{\theta\in \Theta}$ is uniformly bounded and uniformly equicontinuous under Assumption \ref{asm:UniformBounds}. Thus the Arzelà-Ascoli theorem yields that there exists a further subsequence $(\theta_{b(n)})\subset (\theta_{a(n)})$ and a function $\tilde{\mathcal{V}}\in C[0,1]$ such that $\|\mathcal{V}^{\theta_{b(n)}} - \tilde{\mathcal{V}}\|_\infty \to 0$. 

Combining the convergences of $U^{\theta_{b(n)}}$ and $\mathcal{V}^{\theta_{b(n)}}$ with those in Equation \eqref{eq:LCMandVarconvergent}, it follows from the triangle inequality of the metric $d_{BL}$ that also
    \begin{align*}
    \sqrt{|J_{b(n)}|} \widehat{\gamma}^{(b(n)),{\theta_{b(n)}}} \xrightarrow{\mathcal{D}} \tilde{U}
    \qquad \text{and} \qquad
    \widehat{\mathcal{V}}_{b(n)}^{\theta_{b(n)}}\xrightarrow{P} \tilde{\mathcal{V}},
    \end{align*}
in $D[0,1]$ as $n\to \infty$. Now we may use that convergence in Skorokhod topology is equivalent to convergence in uniform topology whenever the limit variable continuous, see e.g. \citet[Theorem 23.9]{kallenberg2021foundations}. Hence the convergences above also hold in $(D[0,1], \|\cdot\|_\infty)$.

Since $\mathcal{V}$ is deterministic, this implies the joint convergences
    \begin{align*}
        (U^{\theta_{b(n)}},\mathcal{V}^{\theta_{b(n)}})
        \xrightarrow{\mathcal{D}}
        (\tilde{U},\tilde{\mathcal{V}})
            \qquad \text{and} \qquad
        \Big(\sqrt{|J_{b(n)}|} \widehat{\gamma}^{(b(n)),{\theta_{b(n)}}}, \widehat{\mathcal{V}}_{b(n)}^{\theta_{b(n)}}\Big) 
        \xrightarrow{\mathcal{D}} 
        (\tilde{U},\tilde{\mathcal{V}})
    \end{align*}
in the product space $D[0,1]\times D[0,1]$ endowed with the uniform topology. 
Since $(\tilde{U},\tilde{\mathcal{V}}) \in C[0,1]\times \overline{\{\mathcal{V}^\theta \colon \theta \in \Theta_0\}}$ takes values in the continuity set of $\mathcal{J}$ by assumption, the classical continuous mapping theorem lets us conclude that
\begin{align*}
    \mathcal{J}(U^{\theta_{b(n)}},\mathcal{V}^{\theta_{b(n)}})
    \xrightarrow{\mathcal{D}}
    \mathcal{J}(\tilde{U},\tilde{\mathcal{V}})
        \quad \text{and} \quad
    \widehat{D}_{b(n)}^{\theta_{b(n)}}=\mathcal{J}\Big(\sqrt{|J_{b(n)}|} \widehat{\gamma}^{(b(n)),{\theta_{b(n)}}}, \widehat{\mathcal{V}}_{b(n)}^{\theta_{b(n)}}\Big) 
    \xrightarrow{\mathcal{D}} 
    \mathcal{J}(\tilde{U},\tilde{\mathcal{V}})
\end{align*}
as $n \to \infty$. Now another application of the triangle inequality with $\mathcal{J}(\tilde{U},\tilde{\mathcal{V}})$ as intermediate value shows that \eqref{eq:ProofOfTeststatistic} holds with $k(n)=b(n)$, so we are done.
\hfill \qedsymbol{}


\subsection{Proof of Theorem \ref{thm:LCTlevel}}
We will apply Corollary \ref{cor:statisticsdistribution} with the functional $\mathcal{J}$ given by
$$
    \mathcal{J}(f_1,f_2) = 
    \one(f_2 \neq 0)\frac{\|f_1\|_\infty }{\sqrt{|f_2(1)|}}, \qquad f_1,f_2 \in D[0,1].
$$
Under Assumption \ref{asm:UniformVariance}, it suffices to check continuity of $\mathcal{J}$ on the set $\Upsilon$ given by
$$ 
    \Upsilon \coloneqq 
    C[0,1] \times \{f \in C[0,1] \mid \delta_1 \leq |f(1)| \} 
    \supset 
    C[0,1]\times \overline{\{\mathcal{V}^\theta \colon \theta \in \Theta_0\}} .
$$
To see that $\mathcal{J}$ is continuous on $\Upsilon$ in the uniform topology, we note that it can be written as a composition of the continuous maps
\begin{align*}
    \Upsilon \longrightarrow [0,\infty)\times [\delta_1, \infty), &
        \qquad (f_1,f_2) \mapsto (\|f_1\|_\infty,|f_2(1)|), \\
    [0,\infty)\times [\delta_1, \infty) \longrightarrow \mathbb{R}, &
        \qquad (x_1,x_2)\mapsto \frac{x_1}{\sqrt{x_2}}.
\end{align*}
Thus it follows from Corollary \ref{cor:statisticsdistribution} that
\begin{align*}
    \widehat T_n = \frac{\sqrt{|J_n|}\sup_{t\in [0,1]} |\widehat \gamma_t^{(n)}|
                }{\sqrt{\widehat{\mathcal{V}}_n(1)}}
     = \mathcal{J}\left( 
        \sqrt{|J_n|}\widehat{\gamma}^{(n)}, \; \widehat{\mathcal{V}}_n\right)
    \convUDnull  
    \mathcal{J}(U, \mathcal{V})
    = \frac{\|U\|_\infty}{\mathcal{V}(1)}.
\end{align*}
With $(B_u)$ a Brownian motion it follows by Proposition \ref{prop:BMrepresentation} 
that 
\begin{align}\label{eq:UnormalizedIsS}
    \dfrac{\|U\|_\infty}{\sqrt{\mathcal{V}(1)}}
    \overset{\mathcal{D}}{=}
        \dfrac{\sup_{0\leq t\leq 1}|B_{\mathcal{V}(t)}|}{\sqrt{\mathcal{V}(1)}}
    =   \dfrac{\sup_{0\leq u\leq \mathcal{V}(1)}|B_{u}|}{\sqrt{\mathcal{V}(1)}}
    \stackrel{\mathcal{D}}{=} 
        \sup_{0\leq t\leq 1}|B_{t}| \stackrel{\mathcal{D}}{=} S,
\end{align}
where we have used that $\mathcal{V}$ is continuous and that Brownian motion is scale invariant. This establishes the first part of the theorem.

For the second part, we first note that the distribution of $S$ is absolutely continuous with respect to Lebesgue measure, which follows from Equation \eqref{eq:supdistribution}.
Then we can use Theorem~4.1 of \citet{bengs2019uniform} to conclude that
$$
    \limsup_{n\to \infty} \sup_{\theta \in \Theta} 
    |\mathbb{P}(\widehat T_n \leq z_{1-\alpha}) - (1-\alpha)| = 0.
$$
It follows from the triangle inequality that
\begin{align*}
    \limsup_{n\to \infty}\sup_{\theta \in \Theta} \mathbb{P}(\Psi_n^\alpha = 1)
    = \limsup_{n\to \infty} \sup_{\theta \in \Theta} \mathbb{P}(\widehat T_n > z_{1-\alpha})
    \leq \alpha.
\end{align*}
\hfill \qedsymbol{}

\subsection{Proof of Theorem \ref{thm:rootNpower}}
Let $0<\alpha<\beta<1$ be given.
The second part of Theorem~\ref{thm:main} permits us to choose $K>0$ sufficiently large such that
\begin{align}
    \limsup_{n\to \infty}\sup_{\theta\in \Theta} \mathbb{P}\pa{
    (\sqrt{|J_n|}\|\widehat{\gamma}^{(n),\theta} - \gamma^\theta\|_\infty )> K} < 1 - \beta.
\end{align}
We then choose $c > K + z_{1-\alpha}\sqrt{1+C(C')^2}$ such that for all $\theta \in \mathcal{A}_{c,n}$, it holds that
\begin{align*}
    \sqrt{|J_n|}\|\gamma^\theta\|_\infty-z_{1-\alpha}\sqrt{1+\mathcal{V}^\theta(1)} 
    \geq c - z_{1-\alpha} \sqrt{1+C(C')^2} > K,
\end{align*}
where we have used Lemma \ref{lem:regularityofgammasigma} in the first inequality. The (reverse) triangle inequality now yields that for any $\theta \in \mathcal{A}_{c,n}$
\begin{align*}
    (\Psi_n^\theta = 0)
    = (\widehat{T}_n^\theta \leq z_{1-\alpha}) 
    &= \pa{\|\widehat\gamma^{(n),\theta}\|_\infty 
        \leq \sqrt{\widehat{\mathcal{V}}_n^\theta(1)}\frac{z_{1-\alpha}}{\sqrt{|J_n|}}} \\
    &\subseteq \pa{
        \|\gamma^\theta\|_\infty
        -
        \left\|
            \widehat\gamma^{(n),\theta}
                - \gamma^\theta
        \right\|_\infty
        \leq \sqrt{\widehat{\mathcal{V}}_n^\theta(1)}\frac{z_{1-\alpha}}{\sqrt{|J_n|}}} \\
    &\subseteq E_1^{(n),\theta} \cup E_2^{(n),\theta},    
\end{align*}
where
\begin{align*}
    E_1^{(n),\theta} &= \pa{\sqrt{|J_n|} \left\| \widehat\gamma^{(n),\theta} 
            - \gamma^\theta\right\|_\infty > K},\\
    E_2^{(n),\theta} &= \pa{\widehat{\mathcal{V}}_n^\theta(1) 
        > 1+\mathcal{V}^\theta(1)}
        \subseteq
            \pa{|\widehat{\mathcal{V}}_n^\theta(1)-\mathcal{V}^\theta(1)| > 1}.
\end{align*}
From Proposition \ref{prop:varianceconsistent} we know that $\limsup_{n\to \infty} \sup_{\theta \in \Theta} \mathbb{P}(E_2^{(n),\theta}) = 0$, so from the choice of $K$ we conclude that
\begin{align*}
    \limsup_{n\to \infty}\sup_{\theta \in \Theta} \mathbb{P}(\Psi_n=0)
    \leq \limsup_{n\to \infty} \sup_{\theta \in \Theta} \mathbb{P}(E_1^{(n),\theta})
    < 1- \beta.
\end{align*}
The desired statement follows from substituting $\mathbb{P}(\Psi_n=0) = 1 - \mathbb{P}(\Psi_n=1)$ into the above equation and simplifying. \hfill \qedsymbol

\subsection{Proof of Theorem \ref{thm:LCTXlevel}}
Assume that $H_0$ holds and note that Assumptions \ref{asm:UniformBounds} and \ref{asm:UniformRates} are satisfied for every sample split $J_n^k \cup (J_n^k)^c$, $k=1,\ldots,K$. 

We consider the decomposition in Equation \eqref{eq:decomposition} for each sample split $J_n^k \cup (J_n^k)^c$, and denote the corresponding processes by
$U^{k,(n)}$, $R_1^{k,(n)}$, $R_2^{k,(n)}$, $R_3^{k,(n)}$, $D_1^{k,(n)}$, and $D_2^{k,(n)}$. 
For each fold $k\in\{1,\ldots,K\}$, we can then apply the results in Section \ref{sec:asymptotics} for a single data split:
\begin{itemize}
    \item By Proposition~\ref{prop:UasymptoticGaussian}, we have that $U^{k,(n)} \convUD U$ in $D[0,1]$, where $U$ is a mean zero continuous Gaussian martingale with variance function $\mathcal{V}$. 

    \item By Proposition~\ref{prop:remainderterms}, 
    $R_\ell^{k,(n)} \convUP 0$ in $(D[0,1],\|\cdot\|_\infty)$ as $n \to \infty$.
    
    \item Under $H_0$, the processes $D_1^{k,(n)}$ and $D_2^{k,(n)}$ are equal to the zero process almost surely.
\end{itemize}
Recall that the folds are assumed to have uniform asymptotic density, which is equivalent to $\frac{\sqrt{n}}{\sqrt{K|J_n^k|}} \to 1$ as $n\to \infty$.
Thus we may also conclude that for each fixed $k$ and $\ell$,
\begin{align*}
    \frac{\sqrt{n}}{\sqrt{K|J_n^k|}} U^{k,(n)}\convUD U 
    \quad \text{and} \quad 
    \frac{\sqrt{n}}{K\sqrt{|J_n^k|}} R_\ell^{k,(n)} \convUP 0,
\end{align*}
where the convergences hold in the Skorokhod and uniform topology, respectively. Now the key observation is that 
\begin{align*}
    U^{1,(n)} 
    \ind \cdots \ind 
    U^{K,(n)}.
\end{align*}
To see this, note that $U^{k,(n)}$ is constructed from $(G_j,M_j)_{j\in J^k}$ only, and by the i.i.d. assumption of the data, the collections $(G_j,M_j)_{j\in J^1},\ldots,(G_j,M_j)_{j\in J_n^K}$ are jointly independent. 
We can therefore apply Lemma \ref{lem:sumofindependent} iteratively to the sequences 
$$
    \frac{\sqrt{n}}{\sqrt{K|J_n^1|}} U^{1,(n)}, \ldots, \frac{\sqrt{n}}{\sqrt{K|J_n^K|}} U^{K,(n)}
$$ to conclude that their sum is uniformly convergent to the sum of $K$ independent copies of $U$. Using the convolution property of the Gaussian distribution, it therefore follows that 
\begin{align*}
    \check{U}^{K,(n)}
    \coloneqq
    \frac{1}{\sqrt{K}}\sum_{k=1}^K 
        \frac{\sqrt{n}}{\sqrt{K|J_n^k|}} U^{k,(n)}
        \convUD U
\end{align*}
in $D[0,1]$ as $n\to \infty$. By the uniform Slutsky theorem formulated in Lemma \ref{lem:SkorokhodSlutsky}, we can therefore conclude that
    \begin{align*}
        \sqrt{n}\check{\gamma}^{K,(n)}      
        = 
        \check{U}^{K,(n)}
        +   \sum_{k=1}^K 
            \frac{\sqrt{n}}{K\sqrt{|J_n^k|}} \left( 
            R_1^{k,(n)} + R_2^{k,(n)} + R_3^{k,(n)} + 
            D_1^{k,(n)} + D_2^{k,(n)} 
        \right)
        \convUDnull U
    \end{align*}
in $D[0,1]$ as $n\to \infty$. Note that we use $\theta\in \Theta_0$ to ensure that $D_1^{k,(n)} + D_2^{k,(n)}$ is equal to the zero process almost surely. 
Since the limit $(U^\theta)_{\theta \in \Theta_0}$ is tight in $C[0,1]$ by Lemma~\ref{lem:Ulimitistight}, 
Proposition~\ref{prop:convergencetoCont2} lets us conclude that $\sqrt{n}\|\check{\gamma}^{K,(n)}\|_\infty \convUDnull \|U\|_\infty$.

Consider now the cross-fitted variance estimator at its endpoint
\begin{equation*}
    \check{\mathcal{V}}_{K,n}(1) = \frac{1}{K} \sum_{k=1}^K \frac{1}{|J_n^k|}\sum_{j \in J_n^k} \int_0^1 \left(\widehat{G}_{j,s}^{k,(n)}\right)^2 \mathrm{d} N_{j,s}.
\end{equation*}
From Proposition~\ref{prop:varianceconsistent}, we see that $\check{\mathcal{V}}_{K,n}(1)$ is an average of $K$ variables converging uniformly in probability to $\mathcal{V}(1)$ in the uniform topology. Hence $\check{\mathcal{V}}_{K,n}(1)$ also converges uniformly in probability to $\mathcal{V}(1)$ in the uniform topology. We can then apply Theorem 6.3 of \citet{bengs2019uniform}, which is a uniform version of Slutsky's theorem, to conclude that
\begin{align*}
    \check{T}_{n}^K = \frac{\sqrt{n} \|\check{\gamma}^{K,(n)}\|_\infty}{\check{\mathcal{V}}_{K,n}(1)} 
    \convUDnull
    \frac{\|U\|_\infty}{\mathcal{V}(1)} \stackrel{\mathcal{D}}{=} S,
\end{align*}
as $n\to \infty$, where last equality in distribution was established in \eqref{eq:UnormalizedIsS}. 

Following the second part of the proof of Theorem \ref{thm:LCTlevel}, we conclude in the X-LCT has uniform asymptotic level. \hfill \qedsymbol{}

\supplementarysection{Uniform stochastic convergence}\label{app:UniformAsymptotics}
In this section, we discuss weak convergence of random variables with values in a metric space uniformly over a parameter set $\Theta$. The uniformity over the parameter set can be used, for example, to establish uniform asymptotic level as well as power under local alternatives.

The content of this section extends the works of \citet{bengs2019uniform} and \citet{kasy2019uniformity}, and we especially build upon Appendix B of \citet{lundborg2021conditional}, in which uniform stochastic convergence is considered in \emph{separable} Banach spaces and Hilbert spaces. 
The space space $(D[0,1], \|\cdot\|_\infty)$ of \cl{} functions endowed with the uniform norm is a Banach space, but it is unfortunately not separable. Therefore we extend the notion of uniform stochastic convergence to random variables in metric spaces, with the condition that the limit is supported on a separable set. This allows to consider uniform weak convergence in two important special cases: 
i) convergence in $(D[0,1], \|\cdot\|_\infty)$ towards variables in $(C[0,1], \|\cdot\|_\infty)$, and ii) convergence in $D[0,1]$ endowed with the Skorokhod metric. 

The Skorokhod space $D[0,1]$ is, if not specified otherwise, equipped with the complete Skorokhod metric $d^\circ$, which makes it a Polish space, i.e., a complete and separable metric space. See for example Section 12 in \citet{billingsley2013convergence} for a discussion of the Skorokhod space and in particular Equation (12.16) for a definition of $d^\circ$.

\subsection{Uniform stochastic convergence in metric spaces}
Throughout this section we consider a background probability space $(\Omega, \mathbb{F}, \mathbb{P})$ and let $(\mathbb{D},d_{\mathbb{D}})$ denote a generic metric space. 
We define $BL_1(\mathbb{D})$ as the set of real-valued functions on $\mathbb{D}$ with Lipschitz norm bounded by $1$, that is, functions $f\colon \mathbb{D} \to \mathbb{R}$ with 
$\|f\|_\infty \leq 1$ and $|f(x)-f(y)| \leq d_{\mathbb{D}}(x,y)$ for every $x,y\in \mathbb{D}$.
Let $\mathcal{M}_1(\mathbb{D})$ denote the set of Borel probability measures on $\mathbb{D}$. We then define the \emph{bounded Lipschitz metric} on $\mathcal{M}_1(\mathbb{D})$ by 
\begin{align*}
    d_{BL}(\mu,\nu) \coloneqq 
    \sup_{f\in BL_1(\mathbb{D})}
    \Big|\int f \mathrm{d}\mu -\int f \mathrm{d}\nu\Big|,
    \qquad \mu,\nu \in \mathcal{M}_1(\mathbb{D}).
\end{align*}
For any pair $(X,Y)$ of $\mathbb{D}$-valued random variables we use the shorthand notation
\begin{align*}
    d_{BL} (X,Y)
    = d_{BL} (X(\mathbb{P}),Y(\mathbb{P}))
    = \sup_{f\in BL_1(\mathbb{D})} \lvert \mathbb{E}(f(X)-f(Y))\rvert.
\end{align*}
If the underlying metric space is ambiguous for $d_{BL}$, we will specify that it is the bounded Lipschitz metric on $\mathcal{M}_1(\mathbb{D})$ by writing $d_{BL(\mathbb{D})}$. 
Our interest in the bounded Lipschitz metric is due to its characterization of weak convergence.

\begin{prop}\label{prop:boundedLipconv}
    Let $X, X_1, X_2, \ldots$ be a sequence of $\mathbb{D}$-valued random variables. Assume that there exists a separable subset $\mathbb{D}_0 \subseteq \mathbb{D}$
    such that $\mathbb{P}(X \in \mathbb{D}_0)=1$. Then the following are equivalent:
    \begin{itemize}
        \item The sequence $(X_n)_{n\geq 1}$ converges in distribution to $X$, i.e., for all $f\in C_b(\mathbb{D})$ it holds that
        $\mathbb{E}[f(X_n)] \to \mathbb{E}[f(X)]$ as $n\to \infty$.
        \item It holds that $d_{BL}(X_n,X) \to 0$ as $n\to \infty$.
    \end{itemize}
\end{prop}
\begin{proof}
    See Theorem 1.12.2, Addendum 1.12.3, and the following discussion in \citet{Vaart:1996}.
\end{proof}
To discuss uniform stochastic convergence, we will for the remaining part of this section let $\Theta$ be fixed set, which is used as a (possible) parameter set for every random variable. We say that a collection $(X^{\theta})_{\theta \in \Theta}$ of $\mathbb{D}$-valued random variables is \emph{separable} if there exists a separable subset $\mathbb{D}_0 \subseteq \mathbb{D}$ such that $\mathbb{P}(X^\theta \in \mathbb{D}_0)=1$ for all $\theta \in \Theta$. If $\mathbb{D}$ is a separable metric space, then any collection of $\mathbb{D}$-valued random variables is automatically separable.

Now Lemma~\ref{prop:boundedLipconv} justifies the following generalization of weak convergence uniformly over $\Theta$:
\begin{definition}\label{dfn:UniformConvergence}
    Let $(X_{n}^\theta)_{n\in \mathbb{N}, \theta \in \Theta}$ and $(X^{\theta})_{\theta \in \Theta}$ be collections of $\mathbb{D}$-valued random variables and assume that $(X^{\theta})_{\theta \in \Theta}$ is separable. We say that:
    \begin{enumerate}[label=(\roman*)]
        \item \emph{$X_{n}^\theta$ converges uniformly in distribution over $\Theta$ to $X^\theta$ in $\mathbb{D}$}, and write $X_{n}^\theta \convUD X^\theta$, if
        \begin{align*}
            \lim_{n\to \infty}\sup_{\theta \in \Theta}
            d_{BL(\mathbb{D})}(X_{n}^\theta,X^\theta)
            = 0.
        \end{align*}
        
        \item \emph{$X_{n}^\theta$ converges uniformly in probability over 
        $\Theta$ to $X^\theta$ in $\mathbb{D}$}, and write $X_{n}^\theta \convUP X^\theta$, if 
        \begin{align*}
            \lim_{n\to \infty}\sup_{\theta \in \Theta}\mathbb{P}(d_{\mathbb{D}}(X_{n}^\theta, X^\theta)>\epsilon)
            = 0
        \end{align*}
        for every $\epsilon > 0$.
    \end{enumerate}
\end{definition}
If for some $\mu\in \mathcal{M}_1(\mathbb{D})$, it holds that $X_n^\theta \convUD X^\theta$ with $X^\theta(\mathbb{P}) = \mu$ for all $\theta \in \Theta$, we also write $X_n^\theta \convUD \mu$. Similarly, we may replace the limit random variable $X^\theta$ by a point $x\in \mathbb{D}$ by interpreting $x$ as the constant map $(\theta,\omega)\mapsto x$ for $\theta \in \Theta$ and $\omega \in \Omega$.

Note that if the parameter set $\Theta = \{\theta_0\}$ is a singleton, then each type of uniform convergence reduces to the corresponding classical definition of convergence in distribution or probability. 
If $\mathbb{D}$ is a separable Banach space, we note that Definition \ref{dfn:UniformConvergence} coincides with Definition 3 in \citet{lundborg2021conditional}.
\begin{prop}\label{prop:seqUni}
    Let $(X_n^\theta)_{n\in \mathbb{N}, \theta \in \Theta}$ and $(X^\theta)_{\theta \in \Theta}$ be collections of \(\mathbb{D}\)-valued random variables and assume $(X^\theta)_{\theta \in \Theta}$ is separable. Then the following are equivalent:
    \begin{itemize}
        \item[a)] $X_n^\theta \convUD X^\theta$ as $n\to \infty$.

        \item[b)] For any sequence $(\theta_n)_{n\in\mathbb{N}}\subseteq \Theta$ it holds that $d_{BL}(X_{n}^{\theta_{n}},X^{\theta_{n}}) \to 0$ as $n\to 0$.
        
        \item[c)] For any sequence $(\theta_n)_{n\in\mathbb{N}}\subseteq \Theta$ there exists a subsequence $(\theta_{k(n)})_{n\in\mathbb{N}}$, with $k\colon \mathbb{N} \to \mathbb{N}$ strictly increasing,
        such that
        \begin{align*}
            \lim_{k\to \infty}d_{BL}(X_{k(n)}^{\theta_{k(n)}},X^{\theta_{k(n)}})
            =0.
        \end{align*}
    \end{itemize}
    Moreover, $X_n^\theta \convUP X^\theta$ if and only if for any sequence $(\theta_n)_{n\in\mathbb{N}}\subseteq \Theta$ and any $\epsilon>0$ it holds that
        \begin{align*}
            \lim_{n\to \infty}\mathbb{P}(d_{\mathbb{D}}(X_n^{\theta_n}, X^{\theta_n})>\epsilon) = 0.
        \end{align*}
\end{prop}
\begin{proof}
    This is essentially Lemma 1 in \citet{kasy2019uniformity} for $\mathbb{D}$-valued random variables, except that we have added the equivalent condition c). The proof for the characterization of uniform convergence in probability is identical to the one given by \citet{kasy2019uniformity}, so we focus on the equivalence between a), b), and c).
    To this end, we to prove that $a) \implies b) \implies c) \implies a)$.
    
    The fact that a) implies b) follows directly from applying the bound 
    \[
       d_{BL}(X_n^{\theta_n},X^{\theta_{n}})
            \leq \sup_{\theta \in \Theta} d_{BL}(X_n^\theta,X^\theta)
    \]
    and taking the limit as $n\to \infty$. We also see that b) implies c) since any sequence is a subsequence of itself. 
    
    We show that c) implies a) by contraposition.
    Assume the negation of a), that is, there exists an $\epsilon>0$ and a sequence $(\theta_n)_{n\in \mathbb{N}}\subseteq \Theta$ such that 
    \begin{align*}
        d_{BL}(X_{n}^{\theta_{n}},X^{\theta_{n}}) > \epsilon
    \end{align*}
    for all $n\in \mathbb{N}$. Then, for all subsequences $(\theta_{k(n)})$ of $(\theta_n)$, it holds that $d_{BL}(X_{k(n)}^{\theta_{k(n)}},X^{\theta_{k(n)}})$ does not converge to zero. This implies the negation of c).
\end{proof}

Proposition \ref{prop:seqUni} will allow us to extend many results for classical stochastic convergence to uniform stochastic convergence.
\begin{cor}\label{cor:ConvUDtoconstant}
    Let $(X_n^\theta)_{n\in \mathbb{N},\theta \in \Theta}$ be a collection of \(\mathbb{D}\)-valued random variables and let $x\in \mathbb{D}$. Then $X_n^\theta \convUD x$ if and only if $X_n^\theta \convUP x$.
\end{cor}
\begin{proof}
    For any sequence $(\theta_n)_{n\in\mathbb{N}}\subseteq \Theta$, recall that $X_n^{\theta_n} \xrightarrow{\mathcal{D}} x$ if and only if $X_n^{\theta_n} \xrightarrow{P} x$, see e.g. Lemma 5.1 in \citet{kallenberg2021foundations}.  Hence the statement follows directly from Proposition \ref{prop:seqUni} (combined with 
    Proposition \ref{prop:boundedLipconv}).
\end{proof}

Our goal is to prove uniform versions of Slutsky's theorem for $D[0,1]$, Rebolledo's central limit theorem, and the chaining lemma for stochastic processes. To prove Slutsky's lemma for $D[0,1]$, we first prove a general result for metric spaces.
\begin{lem}\label{lem:GeneralSlutsky}
    Let $(X^\theta, X_n^\theta, Y_n^\theta)_{n\in \mathbb{N},\theta \in \Theta}$ be a collection of $\mathbb{D}$-valued random variables and assume that $(X^\theta)_{\theta \in \Theta}$ is separable. 
    If $X_n^\theta \convUD X^\theta$ and $d_{\mathbb{D}}(X_n^\theta,Y_n^\theta) \convUP 0$,
    then it also holds that $Y_n^\theta \convUD X^\theta$.
\end{lem}
\begin{proof}
    By the triangle inequality of the bounded Lipschitz metric, we observe that
    \begin{align*}
        \sup_{\theta \in \Theta}d_{BL}(Y_n^\theta, X^\theta) 
            \leq \sup_{\theta \in \Theta} d_{BL}(Y_n^\theta, X_n^\theta) 
                + \sup_{\theta \in \Theta} d_{BL}(X_n^\theta, X^\theta).
    \end{align*}
    The last term converges to zero by the assumption of $X_n^\theta \convUD X^\theta$. 
    For the other term, let $\epsilon>0$ and use the partition
    $$
        (d_{\mathbb{D}}(X_n^\theta,Y_n^\theta)\leq \epsilon) \cup (d_{\mathbb{D}}(X_n^\theta,Y_n^\theta)>\epsilon)
    $$
    to obtain that
    \begin{align*}
        d_{BL}(X_n^\theta, Y_n^\theta) 
            &= \sup_{f\in BL_1(\mathbb{D})} |\mathbb{E}[f(X_n^\theta)-f(Y_n^\theta)]|\\
            &\leq \epsilon 
                + \sup_{f\in BL_1(\mathbb{D})} 
             \mathbb{E}[|f(X_n^\theta)-f(Y_n^\theta)|\one(d_{\mathbb{D}}(X_n^\theta,Y_n^\theta)>\epsilon)]\\
            &\leq \epsilon 
            + \mathbb{P}(d_{\mathbb{D}}(X_n^\theta,Y_n^\theta)>\epsilon).
    \end{align*}
    Taking the supremum over $\Theta$ and the limit superior for $n\to \infty$ finishes the proof.
\end{proof}
The following formulation of the continuous mapping theorem is analogous to Theorem~1 in \citet{kasy2019uniformity}. The proof is almost identical, but we repeat it here for completeness.

\begin{prop}\label{prop:Uctsmapping}
    Let $(\mathbb{D}_1,d_1)$ and $(\mathbb{D}_2,d_2)$ be metric spaces, and let $\Phi \colon \mathbb{D}_1 \longrightarrow \mathbb{D}_2$ be a Lipschitz continuous map.
    Let $(X_n^\theta)_{n\in \mathbb{N}, \theta \in \Theta}$ and $(X^\theta)_{\theta \in \Theta}$ be collections of $\mathbb{D}_1$-valued random variables, and assume $(X^\theta)_{\theta \in \Theta}$ is separable. 
    
    If $X_n^\theta\convUD X^\theta$ in $\mathbb{D}_1$, then $\Phi(X_n^\theta)\convUD \Phi(X^\theta)$ in $\mathbb{D}_2$.
\end{prop}
\begin{proof}
    Note first that if $X^\theta$ is in a separable subset $\mathbb{D}_0 \subseteq \mathbb{D}_1$, then the variables $\Phi(X_\theta)$ for $\theta \in \Theta$ 
    are all in the separable subset $\Phi(\mathbb{D}_0)\subseteq \mathbb{D}_2$. 
    Hence it is well-defined to consider uniform convergence in distribution towards $(\Phi(X^\theta))_{\theta \in \Theta}$. Let $f\in BL_1(\mathbb{D}_2)$ and let $K$ be the Lipschitz constant of $\Phi$. Consider the map 
    \[
        g\colon \mathbb{D}_1\longrightarrow \mathbb{R}, \qquad g(x) = \min(1,K^{-1}) f(\Phi(x)).
    \]
    Then $\|g\|_\infty\leq \|f\|_\infty \leq 1$ and for all $x,y\in \mathbb{D}_1$,
    \begin{align*}
        |g(x)-g(y)| &\leq \min(1,K^{-1}) d_2(\Phi(x),\Phi(y)) \\
        &\leq \min(1,K^{-1})K d_1(x,y)\leq d_1(x,y)
    \end{align*}
    Hence $g\in BL_1(\mathbb{D}_1)$. It follows that
    \begin{align*}
        d_{BL_1(\mathbb{D}_2)}(\Phi(X_n^\theta),\Phi(X^\theta)) 
        &= \sup_{f\in BL_1(\mathbb{D}_2)} |\mathbb{E}[f(\Phi(X_n^\theta))-f(\Phi(X^\theta))]|\\
        &\leq \frac{1}{\min(1,K^{-1})} \sup_{g\in BL_1(\mathbb{D}_1)} |\mathbb{E}[g(X_n^\theta)-g(X^\theta)]|\\
        &\leq \max(1,K) \cdot d_{BL_1(\mathbb{D}_1)}(X_n^\theta,X^\theta) 
    \end{align*}
    Taking the supremum over $\Theta$ and the limit superior as $n\to \infty$ finish the proof.
\end{proof}

We will also need the following two notions of tightness.

\begin{definition} \label{dfn:tight}
Let $(\mu^\theta)_{\theta \in \Theta}$ be a family of probability measures on $\mathbb{D}$, and let $(X^\theta)_{\theta \in \Theta}$ and $(X_n^\theta)_{n \in \mathbb{N}, \theta \in \Theta}$ be collections of $\mathbb{D}$-valued random variables.
\begin{itemize}
    \item[i)] We say that $(\mu^\theta)_{\theta \in \Theta}$ is \emph{tight} if for any $\varepsilon > 0$, there exists a compact set $K\subseteq \mathbb{D}$ such that
    $\sup_{\theta \in \Theta} \mu^\theta(K^c)  < \varepsilon$.
    We say that $(X^\theta)_{\theta \in \Theta}$ is \emph{uniformly tight} if the collection of distributions $(X^\theta(\mathbb{P}))_{\theta \in \Theta}$ is tight.
    
    \item[ii)] The sequence $((X_n^\theta)_{\theta \in \Theta})_{n\in \mathbb{N}}$ of collections is said to be \emph{sequentially tight} if for any sequence $(\theta_n)_{n \in \mathbb{N}} \subset \Theta$, the sequence of distributions $(X_{n}^{\theta_n}(\mathbb{P}))_{n \in \mathbb{N}}$ is tight.
\end{itemize}
\end{definition}
Definition \ref{dfn:tight} i) is a classical concept, whereas sequential tightness was introduced by \citet{lundborg2021conditional} and relaxes uniform tightness for sequences of variables parametrized over an infinite set.

The importance of tightness is mainly due to Prokhorov's theorem \citep[Theorem 23.2]{kallenberg2021foundations}, which states that if $\mathbb{D}$ is a Polish space\footnote{The `only if' part does not require separability nor completeness.}, then $(\mu^\theta)_{\theta \in \Theta}$ is tight if and only if all sequences in $(\mu^\theta)_{\theta \in \Theta}$ have a weakly convergent subsequence.

The continuous mapping theorem in Proposition \ref{prop:Uctsmapping} is more restrictive than the classical theorem as it requires Lipschitz continuity. However, we also have an alternative version of uniform continuous mapping when the limit variable is tight.
\begin{prop}\label{prop:UctsmappingT}
    Let $(\mathbb{D}_1,d_1)$ and $(\mathbb{D}_2,d_2)$ be Polish spaces, and let $(X_n^\theta)_{n\in \mathbb{N}, \theta \in \Theta}$ and $(X^\theta)_{\theta \in \Theta}$ be collections of $\mathbb{D}_1$-valued random variables. Assume $(X^\theta)_{\theta \in \Theta}$ is uniformly tight, and let $\Phi \colon \mathbb{D}_1 \longrightarrow \mathbb{D}_2$ be a map that is continuous on the support of $(X^\theta)_{\theta \in \Theta}$.
    
    If $X_n^\theta\convUD X^\theta$ in $\mathbb{D}_1$, then $\Phi(X_n^\theta)\convUD \Phi(X^\theta)$ in $\mathbb{D}_2$.
\end{prop}
\begin{proof}
    Same as the proof of Proposition 10 in \citet{lundborg2021conditional}, but with norms of differences replaced by metric distances.
\end{proof}

\subsection{Uniform stochastic convergence in Skorokhod space}
In this section we consider the special case where $(\mathbb{D},d_\mathbb{D})$ is the Skorokhod space $(D[0,1],d^\circ)$.  We can also equip $D[0,1]$ with the uniform norm, $\|x\|_\infty = \sup_{t\in [0,1]}|x_t|$, and it known that weak convergence based on either $\|\cdot\|_\infty$ or $d^\circ$ are equivalent when the limit is continuous. We now extend this result to stochastic convergence uniformly over $\Theta$. 

\begin{prop}[Skorokhod equivalence]
\label{prop:convergencetoCont2}
    Let $(X_n^\theta)_{n\in \mathbb{N},\theta \in \Theta}$ be a collection of $D[0,1]$-valued random variables and let $(X^\theta)_{\theta \in \Theta}$ be a uniformly tight collection of $C[0,1]$-valued random variables. Then $X_n^\theta \convUD X^\theta$ in $(D[0,1], d^\circ)$ if and only if $X_n^\theta \convUD X^\theta$ in $(D[0,1], \|\cdot\|_\infty)$. In the affirmative, $\|X_n^\theta\|_\infty \convUD \|X\|_\infty$.
\end{prop}
\begin{proof}
    To avoid ambiguity in the topology on $D[0,1]$, we will throughout this proof use $\mathbb{D}^\circ$ to denote the metric space $(D[0,1],d^\circ)$ and we use $\mathbb{D}_\infty$ to denote the Banach space $(D[0,1], \|\cdot\|_\infty)$. Note also that $C[0,1]$ is separable within $\mathbb{D}_\infty$, so $(X^\theta)_{\theta\in\Theta}$ is separable, and hence the convergence $X_n^\theta \convUD X^\theta$ is well-defined in the non-separable space $\mathbb{D}_\infty$. 
    
    The `if' part is clear since $d^\circ(x,y)\leq \|x-y\|_\infty$ for all $x,y\in D[0,1]$. 
    
    For the `only if' part, assume that $X_n^\theta \convUD X^\theta$ in $\mathbb{D}^\circ$ and let $(\theta_n)\subseteq \Theta$ be an arbitrary sequence. Since $(X^{\theta_n}(\mathbb{P}))$ is tight, Prokhorov's Theorem asserts that there exists a subsequence 
    $(\theta_{k(n)})$ and a probability distribution $\mu$ on $C[0,1]$ such that $X^{\theta_{k(n)}}(\mathbb{P}) \xrightarrow{wk} \mu$ in $\mathbb{D}^\circ$. By the triangle inequality
    \begin{align*}
        d_{BL(\mathbb{D}^\circ)}(X_{k(n)}^{\theta_{k(n)}},\mu) 
            \leq 
            d_{BL(\mathbb{D}^\circ)}(X_{k(n)}^{\theta_{k(n)}},X^{\theta_{k(n)}}) 
            + d_{BL(\mathbb{D}^\circ)}(X^{\theta_{k(n)}},\mu) \to 0, \quad n \to 0.
    \end{align*}
    This shows that also $X_{k(n)}^{\theta_{k(n)}}(\mathbb{P})\xrightarrow{wk} \mu$ in $\mathbb{D}^\circ$. Now we can use that weak convergence in the Skorokhod topology and the uniform topology are equivalent when the limit is continuous \citep[Theorem 23.9 (iii)]{kallenberg2021foundations}.
    We therefore conclude that the convergences $X^{\theta_{k(n)}}(\mathbb{P}) \xrightarrow{wk} \mu$ and $X_{k(n)}^{\theta_{k(n)}}(\mathbb{P})\xrightarrow{wk} \mu$ also hold in $\mathbb{D}_\infty$. 
    But then another use of the triangle inequality shows that 
    $$
        d_{BL(\mathbb{D}_\infty)}(X_{k(n)}^{\theta_{k(n)}},X^{\theta_{k(n)}})  
        \leq 
        d_{BL(\mathbb{D}_\infty)}(X_{k(n)}^{\theta_{k(n)}},\mu)
        + d_{BL(\mathbb{D}_\infty)}(\mu,X^{\theta_{k(n)}}) \to 0.
    $$
    Since $(\theta_{k(n)})$ is a subsequence of the arbitrarily chosen sequence $(\theta_n)$, we conclude that $X_n^\theta \convUD X^\theta$ in $\mathbb{D}_\infty$ by Proposition \ref{prop:seqUni}.
    
    Finally, as the uniform norm is Lipschitz continuous as a map from $\mathbb{D}_\infty$ to $\mathbb{R}$, the continuous mapping theorem formulated in Proposition \ref{prop:Uctsmapping}
    yields that
    \begin{align*}
        X_n^\theta \convUD X^\theta 
            \: \text{  in  } \: \mathbb{D}_\infty
        \qquad \implies \qquad 
        \|X_n^\theta\|_\infty \convUD \|X\|_\infty.
    \end{align*}
    This establishes the last part of the lemma.
\end{proof}

Using $\|\mu\|_\infty$ to denote the pushforward measure for any $\mu \in \mathcal{M}_1(D([0,1]))$ we restate the result above for a fixed limit distribution.

\begin{cor}\label{cor:convergencetoCont}
    Let $(X_n^\theta)_{n\in \mathbb{N},\theta \in \Theta}$ be a collection of $D[0,1]$-valued random variables and let $\mu$ be a probability measure on $C[0,1]$.
    Then $X_n^\theta \convUD \mu$ in $(D[0,1], d^\circ)$ if and only if $X_n^\theta \convUD \mu$ in $(D[0,1], \|\cdot\|_\infty)$. In the affirmative, $\|X_n^\theta\|_\infty \convUD \|\mu\|_\infty$.
\end{cor}
\begin{proof}
    Since $\mu$ is a probability measure on the Polish space $C[0,1]$, it is, in particular, tight \citep[Theorem 1.3]{billingsley2013convergence}. Hence the statement is a special case of Proposition~\ref{prop:convergencetoCont2}. 
\end{proof}

Now we are ready to prove a uniform version of Slutsky's theorem in the Skorokhod space.

\begin{lem}[Uniform Slutsky in Skorokhod space]\label{lem:SkorokhodSlutsky}
    Let $(X^\theta, X_n^\theta, Y_n^\theta)_{n\in \mathbb{N},\theta \in \Theta}$ be a collection of $D[0,1]$-valued random variables such that $Y_n^\theta \convUP 0$ and $X_n^\theta\convUD X^\theta$ in $D[0,1]$. 
    Then it holds that $X_n^\theta+Y_n^\theta \convUD X^\theta$.
\end{lem}
\begin{proof}
    Since $Y_n^\theta \convUP 0$, Corollary \ref{cor:convergencetoCont} implies that $\|Y_n^\theta\|_\infty \convUD 0$, and Corollary \ref{cor:ConvUDtoconstant} implies that $\|Y_n^\theta\|_\infty \convUP 0$. 
    Using the trivial estimate $d^\circ(x+y,x) \leq \|(x+y)-x\|_\infty = \|y\|_\infty$ for $x,y\in D[0,1]$, it follows that 
    $d^\circ(X_n^\theta+Y_n^\theta,X_n^\theta) \convUP 0$.
    Combining the latter with $X_n^\theta\convUD X^\theta$, the desired conclusion now follows from Lemma \ref{lem:GeneralSlutsky}.
\end{proof}

We also have a related result for sums of independent sequences.
\begin{lem}\label{lem:sumofindependent}
    Let $(X_n^\theta, Y_n^\theta)_{n\in\mathbb{N},\theta \in \Theta}$ be a collection of $D[0,1]$-valued random variables and let $(X^\theta)_{\theta \in \Theta}$ and $(Y^\theta)_{\theta \in \Theta}$ be uniformly tight collections of $C[0,1]$-valued random variables.
    Assume that $X_n^\theta \convUD X^\theta$ and $Y_n^\theta \convUD Y^\theta$ in $D[0,1]$, and that for each $\theta \in \Theta$ and $n\in \mathbb{N}$, it holds that 
    $X_n^\theta \ind Y_n^\theta$.
    Let $Z^\theta$ have distribution 
    $X^\theta(\mathbb{P})* Y^\theta(\mathbb{P})$, that is, the same distribution as the sum of two independent copies of each of $X^\theta$ and $Y^\theta$. 
    
    Then it also holds that 
    $X_n^\theta + Y_n^\theta \convUD Z^\theta $
    in $(D[0,1],\|\cdot\|_\infty)$.
\end{lem}
\begin{proof}
    We may assume without loss of generality that $X^\theta \ind Y^\theta$ and that $Z^\theta = X^\theta + Y^\theta$. Let $(\theta_n)\subseteq\Theta$ be an arbitrary sequence. By tightness of $(X^\theta)_{\theta \in \Theta}$ and $(Y^\theta)_{\theta \in \Theta}$, we can apply Prokhorov's theorem twice to obtain probability measures $\mu$ and $\nu$ on $C[0,1]$, and a subsequence $(\theta_{k(n)})$, such that 
    $X^{\theta_{k(n)}}(\mathbb{P}) \xrightarrow{wk} \mu$
    and 
    $Y^{\theta_{k(n)}}(\mathbb{P}) \xrightarrow{wk} \nu$.
    Hence the product measures converge,
    $$
    X^{\theta_{k(n)}}(\mathbb{P}) \otimes Y^{\theta_{k(n)}}(\mathbb{P}) \xrightarrow{wk} \mu \otimes \nu,
    $$
    in $C[0,1]$ as $n\to \infty$, see, for example, Theorem 2.8 (ii) in \citet{billingsley2013convergence}. 

    Since $X_n^\theta \convUD X^\theta$ in $D[0,1]$ by assumption and $(X^\theta)$ is uniformly tight in $C[0,1]$, Proposition \ref{prop:convergencetoCont2} implies that the convergence also holds in $(D[0,1],\|\cdot\|_\infty)$.
    The triangle inequality now yields
    \begin{align*}
        d_{BL}(X_{k(n)}^{\theta_{k(n)}}, \mu) 
            \leq d_{BL}(X_{k(n)}^{\theta_{k(n)}}, X^{\theta_{k(n)}}) 
            + d_{BL}( X^{\theta_{k(n)}}, \mu) \to 0,
    \end{align*}
    so also $X_{k(n)}^{\theta_{k(n)}}(\mathbb{P}) \xrightarrow{wk} \mu$ in $(D[0,1],\|\cdot\|_\infty)$. An analogous computation shows that $Y_{k(n)}^{\theta_{k(n)}}(\mathbb{P}) \xrightarrow{wk} \nu$, and hence also
    \begin{align*}
        X_{k(n)}^{\theta_{k(n)}}(\mathbb{P}) \otimes Y_{k(n)}^{\theta_{k(n)}}(\mathbb{P}) 
        \xrightarrow{wk} \mu \otimes \nu
    \end{align*}
    in the product space $D[0,1]\times D[0,1]$ endowed with the uniform product topology. From the independence statements 
    $X^\theta \ind Y^\theta$ and 
    $X_n^\theta \ind Y_n^\theta$, we have thus shown that
    \begin{equation*}
        (X^{\theta_{k(n)}},Y^{\theta_{k(n)}}) \xrightarrow{\mathcal{D}} \mu \otimes \nu
        \quad \text{and} \quad
       (X_{k(n)}^{\theta_{k(n)}},Y_{k(n)}^{\theta_{k(n)}})
       \xrightarrow{\mathcal{D}} \mu \otimes \nu
    \end{equation*}
    in the uniform product topology. Since addition $+\colon D[0,1]\times D[0,1] \to D[0,1]$ is continuous with respect to this topology, we conclude by the classical continuous mapping theorem that
    \begin{align*}
        Z^{\theta_{k(n)}} = X^{\theta_{k(n)}}+Y^{\theta_{k(n)}} 
            \xrightarrow{\mathcal{D}}
            \mu * \nu
        \quad \text{and} \quad
        X_{k(n)}^{\theta_{k(n)}}+Y_{k(n)}^{\theta_{k(n)}}
            \xrightarrow{\mathcal{D}} \mu * \nu.
    \end{align*}
    It now follows that
    \begin{align*}
        d_{BL}&(X_{k(n)}^{\theta_{k(n)}}+Y_{k(n)}^{\theta_{k(n)}},Z^{\theta_{k(n)}}) \\
        &\leq 
        d_{BL}(X_{k(n)}^{\theta_{k(n)}}+Y_{k(n)}^{\theta_{k(n)}}, \mu * \nu)
        +
        d_{BL}(\mu * \nu, Z^{\theta_{k(n)}}) \to 0.
    \end{align*}
    Since $(\theta_{k(n)})$ is a subsequence of the arbitrarily chosen sequence $(\theta_n)$, we conclude that $X_n^\theta + Y_n^\theta \convUD Z^\theta $ in $(D[0,1],\|\cdot\|_\infty)$ by Proposition \ref{prop:seqUni}.
\end{proof}

We also need the following lemma, which is a generalization of the classical result: pointwise convergence of a sequence of monotone functions towards a continuous limit is in fact uniform over compact intervals. 

\begin{lem} \label{lem:convergenceofincreasing}
    Let $(X_n^\theta)_{n\in \mathbb{N},\theta \in \Theta}$ be a collection of $D[0,1]$-valued random variables with non-decreasing sample paths. 
    Let $(f^\theta)_{\theta \in \Theta}\subset C[0,1]$ be a uniformly equicontinuous collection of non-decreasing functions. If $X_n^\theta(t)\convUP f^\theta(t)$ for each $t\in [0,1]$, then it also holds that 
    \begin{align*}
        \sup_{t\in [0,1]}|X_n^\theta(t)-f^\theta(t)| \convUP 0.
    \end{align*}
\end{lem}
\begin{proof}
    Let $\epsilon>0$. By uniform equicontinuity we can find 
    $0=t_1<\cdots < t_k = 1$ such that $f^\theta(t_i)-f^\theta(t_{i-1})<\epsilon/2$ for all $\theta$ and $i$. Using that $X_n^\theta$ and $f^\theta$ are non-decreasing, we observe that for $t_{i-1}\leq t \leq t_{i}$:
    \begin{align*}
        X_n^\theta(t)-f^\theta(t) &\leq X_n^\theta(t_i)-f^\theta(t_i) + \epsilon/2, \\
        X_n^\theta(t)-f^\theta(t) &\geq X_n^\theta(t_{i-1})-f^\theta(t_{i-1}) - \epsilon/2.
    \end{align*}
    Combining the inequalities over the entire grid we have 
    \[
            \sup_{t\in [0,1]}|X_n^\theta(t)-f^\theta(t)| \leq \max_{i=0,\ldots,k} |X_n^\theta(t_i)-f^\theta(t_i)| + \epsilon/2.
    \]
    By assumption, $X_n^\theta(t)\convUP f^\theta(t)$ for each $t$, and in particular
    \[
        \max_{i=0,\ldots,k} |X_n^\theta(t_i)-f^\theta(t_i)| \convUP 0
    \]
    as $n\to \infty$. We therefore conclude that
    \begin{align*}
        \sup_{\theta \in \Theta}
        \mathbb{P}
        \Big(\sup_{t\in [0,1]}|X_n^\theta(t)-f^\theta(t)| > \epsilon\Big) 
        \leq 
        \sup_{\theta \in \Theta}
        \mathbb{P}
        \Big(\max_{i=0,\ldots,k} |X_n^\theta(t_i)-f^\theta(t_i)| > \epsilon/2\Big) 
                \longrightarrow 0
    \end{align*}
    as $n\to \infty$.
\end{proof}

The last auxiliary result of this section is an example of Prokhorov's method of ``tightness + identification of limit''.

\begin{lem}\label{lem:ProkhorovsPrinciple}
    Let $(\mathbb{D},d_\mathbb{D})$ be either $(C[0,1],\|\cdot\|_\infty)$ or $(D[0,1],d^\circ)$, and let $(X^\theta, X_n^\theta)_{n\in \mathbb{N},\theta \in \Theta}$ be a collection of $\mathbb{D}$-valued random variables with $(X^\theta)_{\theta \in \Theta}$ separable. Suppose that
    \begin{itemize}
        \item The finite dimensional marginals converge uniformly:        for any $0\leq t_1< \cdots < t_k\leq 1$
            \begin{align*}
                \pi_{t_1,\ldots, t_k}(X_n^\theta) \convUD 
                \pi_{t_1,\ldots, t_k}(X^\theta), 
                    \qquad n\to \infty,
            \end{align*}
            where $\pi_{t_1,\ldots, t_k} \colon \mathbb{D} \to \mathbb{R}^k$ is the projection given by $\pi_{t_1,\ldots, t_k}(x) = (x(t_1), \ldots, x(t_k))$.
        \item $(X_n^\theta)_{n\in \mathbb{N},\theta \in \Theta}$ is sequentially tight.
        \item $(X^\theta)_{n\in \mathbb{N},\theta \in \Theta}$ is uniformly tight.
    \end{itemize}
    Then $X_n^\theta \convUD X^\theta$ as $n\to \infty$.
\end{lem}
\begin{proof}
    The statement is analogous to Proposition 18 in \citet{lundborg2021conditional}, the difference being that the functionals $\langle \cdot, h\rangle$ in 
    \cite{lundborg2021conditional} have been replaced by the functionals $\pi_{t_1,\ldots,t_k}$.
    
    The proof of \citet{lundborg2021conditional} also works in our case, given that the finite dimensional marginals form a separating class for the both the Borel algebra on $C[0,1]$ and the Borel algebra on $D[0,1]$. This is established in \citet{billingsley2013convergence}, Example 1.3 and Theorem 12.5 (iii).
\end{proof}

\subsection{Chaining in time uniformly over a parameter} \label{app:UniformChaining}
We extend the basic chaining arguments to hold uniformly over $\Theta$. Our arguments closely follow those of \citet[Chapter VII.2.]{Pollard:1984} and \citet{Newey:1991}. The results are formulated for processes indexed over a general metric space $T$, but we will only apply the results in the case $T=[0,1]$. We have the following extension of stochastic equicontinuity to the uniform setting.
\begin{definition}\label{dfn:UniformChaining}
    A collection of sequences
    $$
        \big(Z^{(n),\theta}\big)_{n\in \mathbb{N},\theta \in \Theta}
        =
        \Big(Z_t^{(n),\theta}\Big)_{t\in T, n\in \mathbb{N},\theta \in \Theta}
    $$
    of stochastic processes indexed over a metric space $(T,d)$ is called \emph{stochastically equicontinuous uniformly over $\Theta$} if for all $\epsilon, \eta > 0$ there exists $\delta >0$ such that
    \begin{align*}
        \limsup_{n\to \infty} \sup_{\theta \in \Theta}\mathbb{P}\Big(\sup_{s,t\in T \colon d(s,t)\leq \delta} \big|Z_s^{(n),\theta}-Z_t^{(n),\theta}\big| 
                > \epsilon\Big)
            < \eta.
    \end{align*}
\end{definition}
In Section 2.8.2 of \citet{Vaart:1996}, the same definition is given in the context of empirical processes.
Recall that we write, e.g., $Z^{(n)}$ as a shorthand for $Z^{(n),\theta}$ and let the dependency on $\theta$ be implicit for notational ease. We also write $\sup_{d(s,t)\leq \delta}$ as a shorthand for $\sup_{s,t\in T \colon d(s,t)\leq \delta}$.
Definition \ref{dfn:UniformChaining} is a direct extension of pointwise stochastic equicontinuity.
Accordingly, Theorem 2.1 from \citet{Newey:1991} generalizes as follows:
\begin{lem}\label{lem:uniformequcont}
    Let $(Z_t^{(n)})_{t\in T, n\in \mathbb{N}}$ be a sequence of stochastic processes indexed by a compact metric space $T$. Assume that $(Z_t^{(n)})$ is stochastically equicontinuous uniformly over $\Theta$ and that for each $t\in T$ it holds that $Z_t^{(n)} \convUP 0$. Then $\sup_{t\in T} |Z_t^{(n)}|\convUP 0$ as $n\to \infty$.
\end{lem}
\begin{proof}
    Let $\varepsilon,\eta>0$ be given, and let $\delta>0$ be the corresponding distance obtained from the uniform stochastic equicontinuity of $(Z^{(n)})$. By compactness of $T$ there exists a finite set $T^*\subseteq T$ such that $T = \bigcup_{t\in T^*}B(t,\delta)$. By the triangle inequality we get that
    \begin{align*}
        \sup_{t\in T} |Z_t^{(n)}|
        = \sup_{t\in T^*}\sup_{s\in B(t,\delta)} |Z_s^{(n)}|
        \leq 
            \sup_{t\in T^*}|Z_t^{(n)}| +
            \sup_{t\in T^*}\sup_{s\in B(t,\delta)} |Z_s^{(n)}-Z_t^{(n)}|.
    \end{align*}
    Since $T^*$ is finite, it follows that $\sup_{t\in T^*} |Z_t^{(n)}|\convUP 0$, which combined with the inequality implies that
    \begin{align*}
        &\limsup_{n \to \infty} \sup_{\theta \in \Theta} \mathbb{P}
            \big(\sup_{t\in T} |Z_t^{(n)}| > 2\varepsilon\big) \\
        &\leq 0 + 
        \limsup_{n \to \infty} \sup_{\theta \in \Theta}
        \mathbb{P}\Big(
            \sup_{t\in T^*}\sup_{s\in B(t,\delta)} |Z_t^{(n)}-Z_t^{(n)}|>\varepsilon
        \Big)
        \leq \eta.
    \end{align*}
    As $\varepsilon,\eta>0$ were chosen arbitrarily, we conclude that $\sup_{t\in T} |Z_t^{(n)}| \convUP 0$.
\end{proof}
To establish uniform stochastic equicontinuity we extend the chaining lemma to a uniform setting. To formulate the theorem we first need some classical definitions related to chaining.
\begin{definition}
    Let $T$ be a compact metric space. A subset $T^* \subseteq T$ is called a $\delta$-net if $\bigcup_{t\in T^*} B(t,\delta) = T$. 
    The covering number
    $$
        N(\delta) =N(\delta,T) \coloneqq \min\{\,|T^*| \colon
            T^*\subseteq T, T^* \text{ is a } \delta\text{-net}\,\}
    $$
    is the smallest possible cardinality of a $\delta$-net, which is finite by compactness. The associated covering integral is
    \begin{align*}
        J(\delta) = \int_0^\delta \pa{2 \log(N(\epsilon)/\epsilon)}^{\frac{1}{2}}\mathrm{d}\epsilon,
        \qquad 0\leq \delta \leq 1.
    \end{align*}
\end{definition}

\begin{lem}
    Let $(T,d)$ be a metric space with finite covering integral $J(\cdot)$ and let $(Z_t^\theta)_{t\in T,\theta \in \Theta}$ be a collection of stochastic processes indexed by $T$ with continuous sample paths. Assume there is a uniform constant $\varsigma>0$ such that, for all $s,t\in T$ and $\eta>0$,
    \begin{align*}
        \sup_{\theta \in \Theta} \mathbb{P}\pa{|Z_s^\theta - Z_t^\theta| > \eta \cdot d(s,t)}
            \leq 2 e^{- \frac{\eta^2}{2\varsigma^2}}.
    \end{align*}
    Then, for all $0<\epsilon<1$,
    \begin{align*}
        \sup_{\theta \in \Theta} \mathbb{P}
        \Big(\sup_{d(s,t)\leq \epsilon}|Z_s^\theta - Z_t^\theta| > 26 \varsigma J(\epsilon)\Big)
        \leq 2 \epsilon.
    \end{align*}
\end{lem}
\begin{proof}
    The lemma is a direct consequence of classical chaining lemma \citep[page 144]{Pollard:1984}. For each $\theta \in \Theta$, the conditions of the chaining lemma are met for $(Z_t^\theta)_{t\in T}$ with sub-exponential factor $\varsigma$. This implies, in particular, that for any $\theta \in \Theta$ and $0<\epsilon<1$,
    \begin{align*}
        \mathbb{P}
        \Big(\sup_{d(s,t)\leq \epsilon}|Z_s^\theta - Z_t^\theta| > 26 \varsigma J(\epsilon)\Big)
        \leq 2 \epsilon,
    \end{align*}
    which is equivalent to the conclusion of the lemma.
\end{proof}

This immediately implies the following corollary.
\begin{cor}\label{cor:uniformchaining}
    Let $(T,d)$ be a metric space with finite covering integral $J(\cdot)$ and let $(Z^{(n),\theta})$ be a sequence of stochastic processes on $T$ with continuous sample paths. Assume there exists a constant $\varsigma>0$ such that, for all $s,t\in T$ and $\eta>0$ and $n\in \mathbb{N}$,
    \begin{align*}
        \sup_{\theta \in \Theta} \mathbb{P} \pa{|Z_s^{(n),\theta} - Z_t^{(n),\theta}| > \eta \cdot d(s,t)}
            \leq 2 e^{- \frac{\eta^2}{2\varsigma^2}}.
    \end{align*}
    Then $(Z^{(n)})$ is stochastically equicontinuous uniformly over $\Theta$.
\end{cor}

For stochastic processes with continuous sample paths, stochastic equicontinuity turns out to be equivalent to sequential tightness (Definition \ref{dfn:tight} ii)).
\begin{prop}\label{prop:equicontinuityistightness}
    Let $(Z^{(n),\theta})_{n\in \mathbb{N},\theta\in\Theta}$ be a collection of $C[0,1]$-valued random variables such that $\mathbb{P}(Z_0^{(n),\theta}=0)=1$ all $n\in \mathbb{N}$ and $\theta \in \Theta$.
    The following are equivalent:
    \begin{enumerate}
        \item $(Z^{(n),\theta})$ is stochastically equicontinuous uniformly over $\Theta$.
        
        \item $(Z^{(n),\theta})$ is sequentially tight.
    \end{enumerate}
\end{prop}
\begin{proof}
    The equivalence is a straightforward application of Theorem 7.3 in \citet{billingsley2013convergence}. 
    Condition $(i)$ of the aforementioned theorem is satisfied for any sequence of measures from the collection $(Z^{(n),\theta}(\mathbb{P}))_{n\in \mathbb{N},\theta \in \Theta}$, since $Z_0^{(n),\theta}=0$ almost surely for all $n$ and $\theta$.
    For any sequence $(\theta_n)\subseteq \Theta$, stochastic equicontinuity uniformly over $\Theta$ implies condition $(ii)$ of Theorem 7.3 in \citet{billingsley2013convergence} for the measures $((Z^{(n),\theta_n})(\mathbb{P}))$. 
    We therefore conclude that stochastic equicontinuity uniformly over $\Theta$ implies sequential tightness. 
    
    On the contrary, assume that $(Z^{(n),\theta})$ is sequentially tight and let $\epsilon,\eta>0$ be given. For each $n$, choose $\theta_n$ such that 
    \begin{align*}
        \sup_{\theta\in\Theta} \mathbb{P} \Big(\sup_{|s-t|\leq \delta} \big|Z_s^{(n),\theta}-Z_t^{(n),\theta}\big| 
                \geq \epsilon\Big)
        \leq \mathbb{P}\Big(\sup_{|s-t|\leq \delta} \big|Z_s^{(n),\theta_n}-Z_t^{(n),\theta_n}\big| 
                \geq \epsilon\Big) + \frac{1}{n}.
    \end{align*}
    Since $((Z^{(n),\theta_n})(\mathbb{P}))$ is tight by assumption, condition $(ii)$ of Theorem 7.3 asserts that there exists $\delta, N>0$ such that
    $$ 
        \mathbb{P}\Big(\sup_{|s-t|\leq \delta} \big|Z_s^{(n),\theta_n}-Z_t^{(n),\theta_n}\big| 
                \geq \epsilon\Big) < \eta
    $$ 
    for $n\geq N$. Combining both inequalities and taking the limit superior finish the proof.
\end{proof}


\supplementarysection{The Functional Martingale CLT} \label{sec:fclt}
In this section we state Rebolledo's martingale CLT \citep{rebolledo1980central} based on its formulation in \citet{AndersenBorganGillKeiding:1993}, and then we extend the result to a uniform version without fixed variance functions.
The one-dimensional case suffices for our purpose, so for simplicity, every local martingale in the following is a real-valued stochastic process. For a local square integrable martingale $(M_t)$, we let $\langle M \rangle(t)$ denote its quadratic characteristic.
The theorem requires a condition on the jumps of the local martingales, for which we will need the following definition.
\begin{definition}
    Let $M_t$ be a local square integrable $\cF_t$-martingale. For any $\varepsilon>0$, we define $\langle M_\varepsilon \rangle(t)$ to be the quadratic characteristic of the pure jump-process given by 
    $$
        t\mapsto \sum_{0\leq s\leq t} M_s \one(|\Delta M_s|>\varepsilon).
    $$
\end{definition}
We also need a representation of Gaussian martingales, which ensures their continuity.

\begin{prop} \label{prop:BMrepresentation}
    Let $(B_t)_{t\in[0,\infty)}$ be a Brownian motion on $[0,\infty)$ with continuous sample paths. For every non-decreasing $f\in C[0,1]$, the process $(B_{f(t)})_{t\in [0,1]}$ is a continuous mean zero Gaussian martingale on $[0,1]$ with variance function $f$. 

    Consequently, if $U=(U_t)_{t\in [0,1]}$ is a mean zero Gaussian martingale with a continuous variance function $V$, then $U$ has the distributional representation 
    \begin{equation}\label{eq:BMrepresentation}
            (U_t)_{t\in [0,1]} 
            \stackrel{\mathcal{D}}{=}
            (B_{V(t)})_{t\in [0,1]}.
    \end{equation}
\end{prop}
\begin{proof}
    Let $f\in C[0,1]$ be non-decreasing. From the properties of Brownian motion, it follows directly that the time-transformed process $(B_{f(t)})_{t\in [0,1]}$ is a mean zero Gaussian process with variance function $f$. 
    Since $f$ is continuous, each sample path $t\mapsto B_{f(t)}$ is a composition of continuous functions and thus continuous itself. 
    Since $f$ is non-decreasing, the time-transformation also preserves the martingale property. This establishes the first part.
    
    For the second part, recall that the covariance function of a martingale is determined by its variance function.
    Hence the first part implies that the right-hand side in \eqref{eq:BMrepresentation} is a Gaussian process with the same mean and covariance structure as the left-hand side. Since the distribution of a Gaussian processes is uniquely determined by its mean and covariance structure, the equality in distribution follows.
\end{proof}

Proposition \ref{prop:BMrepresentation} is a simple, distributional variant of the 
Dubins-Schwarz theorem, see \cite{revuz2013continuous}, Chapter V, Theorems 1.6 and 1.7. 
The Dubins-Schwarz theorem implies that, in fact,  $U_t = B_{V(t)}$ for $t\in [0,1]$, 
where $B$ is a Brownian motion on $[0,V(1)]$. For the purpose of this work we only need 
the simpler, distributional equality \eqref{eq:BMrepresentation}.

We can now formulate Rebolledo's CLT for local martingales. 
To this end, note that Proposition \ref{prop:BMrepresentation} ensures the existence of the continuous 
Gaussian limit martingale $U$ when the variance function $V$ is continuous. 

\begin{thm}[Rebolledo's CLT]\label{thm:fclt}
    Let $(U^{(n)})_{n\in \mathbb{N}}$ be a sequence a local square integrable martingales in $D[0,1]$, possibly defined on different sample spaces and with different filtrations for each $n\in \mathbb{N}$. Let $U$ be a continuous Gaussian martingale with continuous variance function $V\colon [0,1] \to [0,\infty)$, and assume that $U^{(n)}_0 = U_0 = 0$.
    Suppose that for every $t\in [0,1]$ and $\varepsilon>0$,
    \begin{align*}
        \langle U^{(n)} \rangle(t) \xrightarrow{P} V(t)
        \qquad \text{ and } \qquad 
        \langle U_\varepsilon^{(n)} \rangle(t) \xrightarrow{P} 0,
    \end{align*}
    as $n\to \infty$. 
    Then it holds that $U^{(n)} \xrightarrow{\mathcal{D}} U$
    in $D[0,1]$ as $n\to \infty$. 
\end{thm}
\begin{proof}
    This is a special case of Theorem $\mathrm{II}.5.2$ in \citet{AndersenBorganGillKeiding:1993}.
\end{proof}

The general formulation of Rebolledo's CLT above, which allows for $n$-dependent sample spaces and filtrations, can now be leveraged to obtain a uniform version via the sequential characterization of uniform stochastic convergence.

\begin{thm}[Uniform Rebolledo CLT]\label{thm:URebo}
    For each $n \in \mathbb{N}$ and $\theta \in \Theta$:
    \begin{itemize}[leftmargin=15pt]
        \item Let $\mathcal{F}^{(n),\theta}=(\mathcal{F}_t^{(n),\theta})_{t\in [0,1]}$ be a filtration satisfying the usual conditions.
        
        \item Let $U_t^{(n),\theta}$ be a local square integrable $\mathcal{F}_t^{(n),\theta}$-martingale in $D[0,1]$ with $U_0^{(n),\theta}=0$.
        
        \item Let $V^\theta \colon [0,1] \to [0,\infty)$ be a non-decreasing function with $V^\theta(0)=0$.
    \end{itemize}
    Assume that $(V^\theta)_{\theta \in \Theta}$ is uniformly equicontinuous and that $\sup_{\theta \in \Theta} V^\theta(1) < \infty$. 
    Assume further that for every $\varepsilon>0$ and $t\in [0,1]$,
    \begin{align}\label{eq:RebolledoConditions}
        \langle U^{(n),\theta} \rangle(t) \convUP V^\theta(t)
        \qquad \text{ and } \qquad
        \langle U_\varepsilon^{(n),\theta}\rangle(t) \convUP 0,
    \end{align}
    as $n\to \infty$. Then it holds that
    \begin{align*}
        U^{(n),\theta} \convUD U^\theta, \qquad n\to \infty,
    \end{align*}
    in $D[0,1]$ uniformly over $\Theta$, where for each $\theta \in \Theta$, $U^\theta$ is a mean zero continuous Gaussian martingale on $[0,1]$ with variance function $V^\theta$.
\end{thm}
\begin{proof}
    We will use the characterization of uniform convergence as stated in Proposition~\ref{prop:seqUni} c). To this end, let $(\theta_n)\subseteq \Theta$ be an arbitrary sequence. By assumption $(V_{\theta_n})_{n\in \mathbb{N}}$ is a uniformly equicontinuous and bounded sequence of functions on a compact interval, so the Arzelà–Ascoli theorem states that there exists a subsequence $\theta_{k(n)}$, with $k\colon \mathbb{N} \to \mathbb{N}$ strictly increasing, and a function $\tilde{V}\in C[0,1]$ such that 
    \[
        \sup_{t\in [0,1]} |V^{\theta_{k(n)}}(t) - \tilde{V}(t)| 
        \longrightarrow 0, \qquad n \to \infty.
    \]
    Since each function $V^{\theta_{k(n)}}$ is non-decreasing, it follows that $\tilde{V}$ is non-decreasing. It also holds that 
    $\tilde{V}(0)=\lim_{n\to \infty}V^{\theta_{k(n)}}(0)=0$, and therefore $\tilde{V}$ is 
    the variance function of a continuous Gaussian martingale $\tilde{U}$ with $\tilde{U}_0 = 0$. 
    
    By assumption of the convergences in \eqref{eq:RebolledoConditions}, we may conclude that 
    \begin{align*}
        |\langle U^{(k(n)),\theta_{k(n)}} \rangle(t) - \tilde{V}(t)|
        \leq
        \underbrace{
            |\langle U^{(k(n)),\theta_{k(n)}} \rangle(t) - V^{\theta_{k(n)}}(t)|
        }_{\xrightarrow{P} 0}
        +
        \underbrace{
        |V^{\theta_{k(n)}}(t) - \tilde{V}(t)|
        }_{\to 0}
        \xrightarrow{P} 0
    \end{align*}
    and that
    \(
        \langle U_\epsilon^{(k(n)),\theta_{k(n)}} \rangle(t) \to 0
    \) as $n\to \infty$. Thus we have established the conditions of the classical Rebolledo CLT -- Theorem \ref{thm:fclt} -- for the sequence $U^{(k(n)),\theta_{k(n)}}$ and the Gaussian martingale $\tilde{U}$ with variance function $\tilde{V}$. We therefore conclude that
    \begin{align*}
        U^{(k(n)),\theta_{k(n)}} 
            \xrightarrow{\mathcal{D}}
        \tilde{U}
    \end{align*}
    in $D[0,1]$ as $n\to \infty$. 

    We now establish that the sequence $(U^{\theta_{k(n)}})$ also converges in distribution to $\tilde{U}$ in $C[0,1]$, and in particular also in $D[0,1]$. To this end, we use the characterization of convergence in distribution in $C[0,1]$ from Theorem 7.5 in \citet{billingsley2013convergence}, which states that we need to show that
    \begin{enumerate}
        \item For all $0\leq t_1 < \cdots < t_m \leq 1$, it holds that
        \[
            (U_{t_1}^{\theta_{k(n)}}, \ldots, U_{t_m}^{\theta_{k(n)}})
                \xrightarrow{\mathcal{D}}
            (\tilde{U}_{t_1},\ldots, \tilde{U}_{t_m}), 
            \qquad n \to \infty.
        \]
        \item For all $\epsilon>0$
        \begin{align*}
            \lim_{\delta\to 0^+} \limsup_{n \to \infty}
            \mathbb{P}\Big(\sup_{|t-s|<\delta}|U_t^{\theta_{k(n)}}
            - U_s^{\theta_{k(n)}}|>\epsilon \Big)=0.
        \end{align*}
    \end{enumerate}
    The first condition is clear since all the marginals are multivariate Gaussian, and the mean and variance of the sequence converges to the mean and variance of the limit distribution. The second condition follows from the same computation as in the proof of Lemma~\ref{lem:Ulimitistight}. By Theorem 7.5 in \citet{billingsley2013convergence} we therefore conclude that 
    \begin{align*}
        U^{\theta_{k(n)}} \xrightarrow{\mathcal{D}} \tilde{U},
        \qquad \text{for} \:\:  n \to \infty,
    \end{align*}
    in $C[0,1]$, and hence also in $D[0,1]$. 

    We can now apply the triangle inequality for the bounded Lipschitz metric to conclude that
    \begin{align*}
        d_{BL}(U^{(k(n)),\theta_{k(n)}}, U^{\theta_{k(n)}} )
        \leq 
        d_{BL}(U^{(k(n)),\theta_{k(n)}}, \tilde{U} )
        +
        d_{BL}( \tilde{U}, U^{\theta_{k(n)}} )
        \longrightarrow 0.
    \end{align*}
    Since $(\theta_n) \subseteq \Theta$ was an arbitrary sequence, we conclude that $U^{(n),\theta} \convUD U^\theta$ by Proposition~\ref{prop:seqUni}.
\end{proof}

The following proposition gives explicit expressions for the quadratic characteristics that 
appear in Rebolledo's CLT in the special case where the local martingales are given as stochastic integrals with respect to a compensated counting processes. 

\begin{prop}\label{prop:Rebolledospecialcase}
    Let $N_1,\ldots, N_n$ be counting processes and assume that for each $j=1,\ldots,n$,
    $N_j$ has an absolutely continuous $\mathcal{F}_t^{(n)}$-compensator $\Lambda_{j,t}$ such that $M_{j,t}=N_{j,t}-\Lambda_{j,t}$ is a locally square integrable $\mathcal{F}_t^{(n)}$-martingale.
    Let $H_1,\ldots,H_n$ be locally bounded $\mathcal{F}_t^{(n)}$-predictable processes, and define the process $U_t^{(n)} = \sum_{j=1}^n \int_0^t H_{j,s} \mathrm{d}M_{j,s}$. Then $U_t^{(n)}$ is a local square integrable $\mathcal{F}_t^{(n)}$-martingale, and for any $t,\epsilon>0$ it holds that
    \begin{align*}
        \langle U^{(n)} \rangle(t) 
            &= \sum_{j=1}^n \int_0^t H_{j,s}^2 \mathrm{d}\Lambda_{j,s}, \\
        \langle U_\epsilon^{(n)} \rangle(t) 
            &= \sum_{j=1}^n \int_0^t H_{j,s}^2\one(|H_{j,s}|\geq \epsilon)\mathrm{d}\Lambda_{j,s}.
    \end{align*}
\end{prop}
\begin{proof}
    See the discussion following Theorem $\mathrm{II}.5.2$ in \citet{AndersenBorganGillKeiding:1993}, in particular equations $(2.5.6)$ and $(2.5.8)$.
\end{proof}

\supplementarysection{Estimation of \texorpdfstring{$\lambda$}{intensity} and \texorpdfstring{$G$}{residual process}} \label{sec:estimationofnuisance} 
The asymptotic theory for estimation of the LCM crucially relies on $\widehat \lambda^{(n)}$ and $\widehat G^{(n)}$ being consistent,
and more importantly, having a product error decaying at an $n^{-1/2}$-rate.
Therefore, a central question when applying the test, is how to model $\lambda$ and $G$. 

In principle, we could use parametric models to learn $\widehat \lambda^{(n)}$ and $\widehat G^{(n)}$, and under such models it should be possible to achieve $n^{-1/2}$-rates. For example, if we consider a parametrization $(t,\theta) \mapsto \lambda_t(\theta)$ which is $\kappa(t)$-Lipschitz in $\theta\in \Theta \subseteq \mathbb{R}^p$ for each $t$, then
\begin{align*}
    h(n)^2 
    = E\left(\int_0^1 (\lambda_t(\theta_0)-\lambda_t(\widehat \theta^{(n)}))^2 
    \mathrm{d}t \right) 
    \leq 
    \|\kappa\|_{L_2([0,1])}^2 E\| \theta_0 - \widehat \theta^{(n)}\|_{\mathbb{R}^p}^2.
\end{align*}
Thus the rates from parametric asymptotic theory can be converted to rates for $g$ and~$h$. 

However, it is of greater interest if sufficient rates can be achieved with
nonparametric estimators. Below we give
concrete examples of nonparametric models and discuss which rates are
achievable. For simplicity, we focus on the case where $\mathcal{F}_t =
\mathcal{F}_t^{N,Z}$ and where $G_t = X_t - \Pi_t$ as in the introductory example.

\subsection{Nonparametric functional estimation of \texorpdfstring{$\Pi$}{the predictable projection}}
As seen in Section~\ref{subsec:ex}, assumptions on the form of $\Pi$ turn the
general estimation problem into a concrete problem of estimating a function. 

If the system is Markovian, it can be reasonable to assume a
\emph{functional concurrent model}. The model asserts that $\Pi_t = \mu(t,Z_t)$
for a bivariate function $\mu$, and a survey of methods for estimating $\mu$ is
given by \citet{maity2017nonparametric}. Notably, \citet{jiang2011functional}
achieve an $n^{-1/3}$-rate of $g(n)$ under certain regularity and moment
assumptions, see their Theorem 3.3. That result also holds if $Z$ is 
replaced by a linear predictor $\beta^T\mathbf{Z}$ of several covariates.

Consider again the historical linear regression model from Section \ref{subsec:ex},
and assume that the effect of $Z$ on $X$ is homogeneous over time. 
That is, $\rho_X(s,t) = \tilde \rho_X(t-s)$ for some function $\tilde \rho_X$.
This submodel is known as the
\emph{functional convolution model}, since $\Pi$ can be written as the
convolution of $Z$ and $\tilde \rho_X$. Applying the Fourier transform converts
it into a (complex) linear concurrent model, so by Plancherel's theorem one can leverage the convergence rates from the concurrent model. \citet[Theorem
16]{manrique2016functional} uses this idea to transfer the $n^{-1/4}$-rate of the
\emph{functional ridge regression estimator} \citep{manrique2018ridge} 
to the convolution model, which holds under modest moment conditions on the data. With additional distributional assumptions, we conjecture that faster rate results for the linear concurrent model can also be leveraged to the convolution model. \citet{csenturk2010functional} consider a similar model under
the assumption that
$$
    \rho_X(s,t) = \one(t-\Delta \leq s\leq t)\tilde \rho_X^1(t)\tilde \rho_X^2(t-s),
$$
for two functions $\tilde \rho_X^1$ and $\tilde \rho_X^2$ and a lag $\Delta>0$.
They establish a pointwise rate result for the response curve, but it is not 
obvious how to cast their result as a polynomial rate for $g(n)$.

For the full historical functional linear model we are not aware of any
published rate results. \citet{yuan2010reproducing,cai2012minimax} establish
rates on the prediction error for scalar-on-function regression, and
\citet{Yao:2005} establish various rates for function-on-function regression,
but in a non-historical setting. Based on the former, we give a heuristic for
which rates are achievable for $g(n)$ in this model. If $\widehat{\Pi}$ is based on
a kernel estimate $\widehat{\rho}_X^{(n)}$ of $\rho_X$, then Tonelli's theorem yields
    \begin{align*}
        g(n)^2 = \vertiii{ 
            \Pi - \widehat{\Pi}^{(n)}
        }_{2}^2
        &= \ex \left(
            \int_0^1 \left(
                \int_0^t (\rho_X(s,t)-\widehat{\rho}_X^{(n)}(s,t))Z_s\mathrm{d}s
            \right)^2\mathrm{d}t
        \right) \\
        &= \int_0^1 \ex
            \left(\left(
                \int_0^t (\rho_X(s,t)-\widehat{\rho}_X^{(n)}(s,t))Z_s\mathrm{d}s
            \right)^2\right)\mathrm{d}t.
    \end{align*}
Theorem 4 in \citet{cai2012minimax} asserts that we, under certain regularity
conditions, can estimate $\rho_X(\cdot,t)$ such that 
\begin{equation*}
    \ex \left(
    \left(
        \int_0^t (\rho_X(s,t)-\widehat{\rho}_X^{(n)}(s,t))Z_s\mathrm{d}s
    \right)^2
    \right)
\end{equation*}
decays at a $n^{-2r_t/(2r_t+1)}$-rate for a fixed $t$. Here $r_t$ is a constant
describing the eigenvalue decay of a certain operator related to the 
autocovariance of $Z$ and the regularity of $\rho_X$. As a concrete example, if
$Z$ is a Wiener process and $\rho_X(\cdot,t)\in \mathcal{W}_2^m([0,t])$ is in
the $m$-th Sobolev space for each $t>0$, then $r_t = 1 + m$ and $g(n)$ will
converge at an $n^{-(1+m)/(2m+3)}$-rate, see the discussion after Corollary 8 in
\cite{yuan2010reproducing}. Based on these arguments, we believe that the
desired $n^{-(1/4+\varepsilon)}$-rate for $g(n)$ is achievable with 
suitable regularity assumptions on $Z$ and $\rho_X$. 

\subsection{Estimation of \texorpdfstring{$\lambda$}{the intensity}}
Within the framework of the Cox model, \cite{Wells:1994} demonstrate 
that the baseline intensity can be estimated with rate $n^{-2/5}$ using 
a standard kernel smoothing technique. With the parametric $n^{-1/2}$-rate 
on the remaining parameters, this translates readily into $h(n) = O(n^{-2/5})$.

We suspect that the same rate should also be attainable in a sparse setting with high-dimensional covariates, for example by applying the smoothing approach of \citet{Wells:1994} to the baseline hazard estimators of e.g. \citet{fang2017testing} and \citet{hou2023treatment}.

\citet{hiabu2021smooth} consider the more general multiplicative intensity model with $\lambda_t = \one(T\geq t)f(t,Z_t))$, where $f$ has a multiplicative structure over its arguments. They introduce an estimator with optimal rate $h(n)=n^{-2/(5+d)}$, where $d$ is the dimension of $Z$. For $d>3$, we therefore need faster rates on $g(n)$ in order for the LCT to maintain type I error control.

Omitting the multiplicative structure on $f$, \citet{bender2020general}
propose a general framework for nonparametric estimation of Markovian
intensities. They
survey existing methods such as gradient boosted trees and neural networks and
relate them to this setting. Based on real and synthetic data, they find that
both gradient boosted trees and neural networks outperform the Cox
model in terms of predictive performance as measured by the Brier score. In
essence, the framework relies on discretizing time and approximating the
intensity with successive Poisson regressions. 
Using the same idea,
\citet{rytgaard2021estimation} argue that $h(n) = o(n^{-1/4})$
can be achieved for time-independent covariates. 

Similarly, \citet{rytgaard2022continuous} mention that $h(n) = o(n^{-1/4})$ can
be achieved for estimation of intensities in a multivariate point process with a
uniformly bounded number of events, which we place into a general modeling 
framework below.

\subsection{Estimation of \texorpdfstring{$\lambda$}{the intensity} and \texorpdfstring{$\Pi$}{the predictable projection} for counting processes}
In Sections~\ref{sec:setup} and \ref{sec:asymptotics} we considered the setup where $N$ was a counting
process adapted to a filtration $\mathcal{F}_t$, which could contain information
on baseline covariates and covariate processes that were not necessarily
counting processes. In this section we explore how our testing framework can be
applied when all stochastic processes of interest are
counting processes.

More specifically, let $(N_t^d)_{d \in [p]}$ be 
a $p$-dimensional counting process. For $a, b \in [p]$ and 
$C \subset [p] \setminus \{b\}$ with 
$a \neq b$ and $a \in C$ we are interested in testing 
the hypothesis that $N^a$ is conditionally locally independent 
of $N^b$ given the filtration, $\mathcal{F}^C_t$, generated 
by $N^C = (N^d)_{d \in C}$.

We can cast this setup in the framework of Section~\ref{sec:setup} as follows.
Naturally, we let $N = N^a$ and $\mathcal{F}_t = \mathcal{F}^C_t$. The auxiliary
process $X$ is chosen to be \lc{} and predictable with respect to the
filtration, $\mathcal{F}^b_t$, generated by $N^b$. For example, we could choose
$X_t = N^b_{t-}$. But $X_t$ could be any functional of $N^b$ such
as $X_t = f(N_{t-}^b)$ for a suitable function $f$ or a linear filter of $N^b$,
    \begin{align*}
        X_t = \int_0^{t-} \kappa(t - s) \mathrm{d} N_s^b,
    \end{align*}
where $\kappa$ is a suitable kernel function, see also Section \ref{sec:neyman}.
In principle, the process $X$ could also depend on the process $N^C$, but it is
important that the filtration, $\mathcal{G}_t$, generated by $\mathcal{F}_t$ and
$X_t$ is strictly larger than $\mathcal{F}_t$, i.e., $X_t$ should depend on
$N^b$, in order to get a non-trivial test as explained in Section~\ref{sec:setup}. 

In the framework of counting processes, we can approach the estimation of both
$\lambda$ and $\Pi$ in a unified and general way as follows: Let $(\tau_j,
z_j)_{j \geq 1}$ be the marked point process associated with the counting
process $N^C$, i.e., $(\tau_j)_{j \geq 1}$ is a sequence of almost surely
strictly increasing event times located at the jumps of $N^C$, and $(z_j)_{j
\geq 1}$ for $z_j \in C$ are the corresponding event types. 

Since both $\lambda_t$ and $\Pi_t$ are real-valued and
$\mathcal{F}^C_{t-}$-measurable for each fixed $t \geq 0$, they can be
represented as measurable functions of $\{(\tau_j, z_j) \mid \tau_j < t, z_j \in
C\}$. Hence, we can model both $\lambda$ and $\Pi$ using any sequence-to-number
model. For the intensity process, \citet{rytgaard2022continuous} propose a
sequence of HAL estimators when the total event count is uniformly bounded. As
an alternative, \citet{Xiao:2019} propose using a recurrent neural network
(LSTM). Unless there is a uniform bound on the total number of events, as
assumed by \cite{rytgaard2022continuous}, there are currently no published
results available on the rates of convergence for nonparametric estimation of
sequence-to-number functions.

\supplementarysection{Relation to semiparametric survival models} \label{sec:survmodels}
In this section, we relate the LCM to existing work on treatment effects in survival analysis. We resume to the setting of Section \ref{sec:copula}, that is, the case where $N_t=\one(T\geq t)$ is the counting process of a survival time, $X$ is a baseline treatment variable, and where $\mathcal{F}_t = \sigma(N_s, Z; s \leq t)$ for additional baseline covariates $Z$. Supposing that $X$ is also non-negative, we may consider two different models for the intensity:
\begin{align} 
    \blambda_t^\times 
        &= \one(T\geq t)\underline{\lambda}(t) \exp(\theta X + \phi(Z)),
            \label{eq:PLCM} \\
    \blambda_t^+
        &= \one(T\geq t)(\vartheta X + \varphi(t,Z)),
            \label{eq:AHM}
\end{align}
where $\theta,\vartheta\in \mathbb{R}$ are treatment parameters of interest, and where $\underline{\lambda},\phi,\varphi$ deterministic nuisance functions. The model in \eqref{eq:PLCM} is known as the partially linear Cox model (PLCM, \citet{sasieni1992information}), and the additive model in \eqref{eq:AHM} was considered by \citet{dukes2019doubly} among others.

While the parameters $\theta$ and $\vartheta$ are difficult to compare directly, the hypothesis of conditional local independence corresponds to the hypothesis of zero treatment effect within each of the models, and testing this hypothesis can be done using a score test.

\citet{sasieni1992information} shows that within the PLCM, the efficient score for $\theta$ is given by
\begin{align} \label{eq:PLCMscore}
    S^\times(\theta; \underline{\lambda},\phi)  = 
    \int_0^1  (X-\alpha^*(t)-h^*(Z))(\mathrm{d}N_t -  \one(T\geq t) \underline{\lambda}(t)e^{\theta X + \phi(Z)}\mathrm{d}t), 
\end{align}
where $(\alpha^*,h^*)$ are defined as the minimizers $\ex (X - \alpha(T) - h(Z) )^2$.
Recall that, when $X$ and $Z$ are time-independent, the null hypothesis $H_0$ of conditional local independence reduces to the conditional independence statement $X\ind  T \mid Z$. Consequently, it holds that $\alpha^*(T)=0$ and $h^*(Z) = \ex[X\mid Z] = \Pi_0$ under $H_0$. Evaluating the efficient score at $\theta=0$ under $H_0$ therefore gives
\begin{align*}
    S^\times(0; \underline{\lambda},\phi) 
    = \int_0^1  (X-\Pi_0)(\mathrm{d}N_t - \one(T\geq t) \underline{\lambda}(t)e^{\phi(Z)}\mathrm{d}t)
    = (X-\Pi_0)(1-\Lambda_T(\underline{\lambda},\phi)).
\end{align*}
We see that the empirical version of $S^\times$ is exactly the endpoint of the LCM estimator with additive residual process (cf. Equation \ref{eq:LCM-baselineX}). This means that our test, the X-LCT, can under the PLCM be interpreted as a score test based on the efficient score. 

A similar connection can be made for the additive model. \citet{dukes2019doubly} show that the efficient score for $\vartheta$ is given by
\begin{align*}
S^+(\vartheta; \varphi) = 
    \int_0^1 \left(
    X - \frac{\ex[(\vartheta X + \varphi(t,Z))^{-1}Xe^{-\vartheta Xt} \mid Z]
    }{
        \ex[(\vartheta X + \varphi(t,Z))^{-1}e^{-\vartheta Xt} \mid Z]
    }
    \right)
    \frac{(\mathrm{d}N_t-\blambda_t(\vartheta,\phi)\mathrm{d}t)
    }{
        (\vartheta X + \varphi(t,Z))
    }.
\end{align*}
Plugging in $\vartheta=0$ and simplifying under $H_0$ yields
\begin{align*}
S^+(0; \varphi) = 
    \int_0^1 
    \frac{
         (X - \ex[X \mid Z])
    }{
        \varphi(t,Z)
    }
    (\mathrm{d}N_t-\one(T\geq t)\varphi(t,Z)\mathrm{d}t)
    = \int_0^1 
    \frac{
         (X - \Pi_0)
    }{
        \lambda_t
    }
    (\mathrm{d}N_t - \lambda_t\mathrm{d}t).
\end{align*}
We recognize the empirical version of $S^+$ as the endpoint of the LCM estimator
with the time-constant $X$ replaced by the hazard weighted process $X / \lambda_t$.

Other works that consider effect estimation based on orthogonal scores include \cite{huang1999efficient,fang2017testing,niu2022reconciling,zhong2022deep} for the PLCM and \citet{hou2023treatment} for the additive model in \eqref{eq:AHM}.

We also suspect, as the derivations in Section \ref{sec:neyman} likewise suggest, that the LCM 
is still an efficient score for certain semiparametric survival models even when the covariates 
vary with time, but we are not aware of existing results on such a connection.

\supplementarysection{Details on Neyman orthogonality} \label{sec:orthodetails}
In this section, we first show by direct computation that the LCM is Neyman orthogonal with respect to both general residual processes and intensities. We then show that 
the LCM with an additive residual process can be viewed as a concentrated-out score in the sense of \citet{newey1994asymptotic}.

\subsection{General Neyman orthogonality}
The definition of Neyman orthogonality by \citet[Def. 2.1.]{chernozhukov2018} requires that we formally define function spaces for the collections of nuisance parameters. 
However, to avoid extensive technical specifications (and redundant model assumptions), we prove a simpler -- but more general -- condition, from which Neyman orthogonality can be derived within specific semiparametric models.
 
First, we generalize the integral $I_t$ from Definition \ref{dfn:lcm} to a function of pairs $(x,y)$ of \lc{} functions, given by
\begin{align*}
    I_t(x,y) = \int_0^t x_s (\mathrm{d}N_s - y_s \mathrm{d}s).
\end{align*}
With this notation, the LCM is given by $\gamma_t = \ex[I_t(G,\lambda)]$, where $G$ is a residual process and $\lambda$ is the $\mathcal{F}_t$-intensity of $N$. We assume for simplicity that $\lambda$ and $G$ are bounded such that the expectation is well-defined.

Now let $\tilde{G}_t$ be an arbitrary bounded $\mathcal{G}_t$-predictable \lc{} process, and let $\tilde{\lambda}_t$ be an arbitrary bounded $\mathcal{F}_t$-predictable \lc{} process.
We establish the following orthogonality condition: under $H_0$ it holds that
\begin{align}\label{eq:generalortho}
    \partial_r \ex[
    I_t(G + r(\tilde{G}-G),\lambda + r (\tilde{\lambda}-\lambda)]\big\vert_{r=0} = 0.
\end{align}
Indeed, observe that
\begin{align*}
    I_t(G + r (\tilde{G}-G),&\lambda + r (\tilde{\lambda}-\lambda)) \\
    &= (1-r)^2 I_t(G, \lambda) + r(1-r) (I_t(G,\tilde{\lambda}) 
        + I_t(\tilde{G},\lambda)) + r^2 I_t(\tilde{G},\tilde{\lambda}),
\end{align*}
from which it follows that
\begin{align*}
    \partial_r \ex[
    I_t(G + r(\tilde{G}-G),\lambda + r (\tilde{\lambda}-\lambda)]\big\vert_{r=0} 
    = 2\cdot\ex[I_t(G, \lambda)] + \ex[I_t(G,\tilde{\lambda})]
        + \ex[I_t(\tilde{G},\lambda)].
\end{align*}
The first term is zero under $H_0$ by Proposition \ref{prop:cli-mg}, and the third term vanishes under $H_0$ by the same argument. For the second term, we note that
\begin{align*}
    \ex[I_t(G,\tilde{\lambda})]
    = \gamma_t + \ex\left[\int_0^t G_s (\lambda_s - \tilde{\lambda}_s)\mathrm{d}s\right]
    = \gamma_t + \int_0^t \ex\left[\ex\left[G_s\mid \mathcal{F}_{s-}\right] (\lambda_s - \tilde{\lambda}_s)\right]\mathrm{d}s
    = \gamma_t,
\end{align*}
which also vanishes under $H_0$. This lets us conclude that \eqref{eq:generalortho} holds under $H_0$.

\subsection{Concentrating-out}
To derive the concentrated-out score, we first need to formalize the nuisance parameters and the collection thereof. We consider the case where $\mathcal{F}_t = \mathcal{F}_t^{N,Z}$ for a process $Z=(Z_t)$, and let $D_X$ and $D_Z$ denote the respective sample spaces of the $X$ and $Z$. We posit the following semiparametric model for the intensity:
$$
    \blambda_t(\beta,h) = \one(T\geq t)e^{\beta X_t} h(t,Z),
$$
where $h\colon [0,1]\times D_Z \to [0,\infty)$ is a function such that $t\mapsto h(t, Z)$ is an $\mathcal{F}_t^Z$-predictable \lc{} process. Denote the collection of such functions by $\mathcal{T}_1$. To make the space not dependent on the particular instantiation of the process $Z$, we could also take it to be the set:
\begin{align*}
    \Big\{
        (h_t)_{t\in [0,1]} \in 
            \int_{[0,1]}^{\oplus} L_2({D_Z\vert}_{[0,t)},\mathbb{R}) \mathrm{d}t \: \Big\vert \: 
            &t\mapsto h_t((z_s)_{s< t}) \text{ is a bounded nonnegative} \\ 
            &\qquad \text{\lc{} function for all } (z_s)_{0\leq s\leq 1} \in D_Z \Big\},
\end{align*}
where $\int^{\oplus}$ denotes the direct integral, and where ${D_Z\vert}_{[0,t)}$ is the path space $D_Z$ restricted to the domain $[0,t)$.

Recall that the likelihood at $t=1$ is given by
\begin{align*}
    \ell_1(\beta,h) = 
    \int_0^1 \log(\blambda_{s}(\beta,h)) \mathrm{d}N_s - 
        \int_0^1 \blambda_s(\beta,h) \mathrm{d}s.
\end{align*}
We show that concentrating out the nuisance parameter $h$ of $\ell_1(\beta,h)$ yields to the LCM with additive residual process under $H_0$. 

Let $P_0$ be fixed distribution satisfying $H_0$ with ground truth $\beta^0=0$ and $h^0\in\mathcal{T}_1$. Suppose that for each $\beta\in(-\epsilon,\epsilon)$ in a neighborhood of zero, there is a function $\varpi_\beta \in \mathcal{T}_1$ such that 
$$
    \varpi_\beta(t,Z) = \ex_{P_0}[e^{\beta X_t}\mid T\geq t, \mathcal{F}_{t-}^Z].
$$
The first step of concentrating-out is to maximize the expected likelihood over $h$. We claim that for each fixed $\beta$, the function $h_\beta^* \coloneqq h^0/\varpi_\beta \in \mathcal{T}_1$
maximizes the objective $u(h)\coloneqq \ex_{P_0}[\ell_1(\beta,h)]$ over $h\in \mathcal{T}_1$. 

Since the logarithm is concave and $\blambda(\beta,h)$ is linear in $h$, it follows that $u$ is a concave objective. 
Moreover, one can show that $h_\beta^*$ is the unique critical point of $u$ by equating its Gateaux derivative to zero and invoking the fundamental lemma of the calculus of variations. For simplicity, we settle with verifying (see below) that the Gateaux derivative is zero at $h_\beta^*$. From these properties we may conclude that the global maximum of $u$ is attained at $h_\beta^*$.

A straightforward differentiation shows that for any $\tilde{h}\in \mathcal{T}_1$,
\begin{align*}
    \partial_r\ell_1(\beta,h_\beta^* + r\cdot \tilde{h})\mid_{r=0}
    =
    \int_0^1 \frac{\tilde{h}(s,Z)}{h_\beta^*(s,Z)} \mathrm{d}N_s 
    - \int_0^1 \blambda_s(\beta,\tilde{h})\mathrm{d}s.
\end{align*}
Because $\frac{\tilde{h}(t,Z)}{h_\beta^*(t,Z)}$ is $\mathcal{G}_t$-predictable and $N_t - \int_0^t \blambda_s(0,h^0)\mathrm{d}s$ is a $\mathcal{G}_t$-martingale, taking expectation under $P_0$ yields that
\begin{align*}
    \partial_r u(h_\beta^* + r\cdot \tilde{h})\mid_{r=0}
    &= \ex[\partial_r\ell_1(\beta,h^* + r\cdot \tilde{h})\mid_{r=0}] \\
    &= \ex\left[
        \int_0^1  \frac{\tilde{h}(s,Z)}{h_\beta^*(s,Z)}
        \blambda_s(0,h^0)\mathrm{d}s
        - \int_0^1\blambda_s(\beta,\tilde{h})\mathrm{d}s
    \right] \\
    &= \ex\left[
        \int_0^1 \one(T\geq s)\left(
        \varpi_\beta(t,Z)
        -
        e^{\beta X_s} \right)\tilde{h}(s,Z)\mathrm{d}s
    \right] \\
    &= 
        \int_0^1 \ex\left[\one(T\geq s)\left(
        \varpi_\beta(t,Z)
        -
        \ex [e^{\beta X_s}\mid T\geq t, \mathcal{F}_{s-}^Z] \right)\tilde{h}(s,Z)
        \right]     
    \mathrm{d}s
    = 0,
\end{align*}
where we have used that $\tilde{h}(s,Z)$ is $\mathcal{F}_s$-predictable. 

Now, by the method of concentrating-out, we are led to consider nuisance functions of the form
\begin{align*}
    \eta_h &\colon (-\epsilon, \epsilon) \longrightarrow \mathcal{T}_1, \\
    \eta_h(\beta) &= \left((t,z) \mapsto \frac{h(t,z)}{\varpi_\beta(t,z)}\right), 
\end{align*}
for any $h \in \mathcal{T}_1$, which in particular includes $\eta_{h_0}(\beta) = h_\beta^*$. The resulting concentrating-out score is therefore
\begin{align*}
    &\psi(\beta, \eta_h) = \frac{\partial \ell_1(\beta,\eta_h(\beta))
        }{\partial \beta} \\
        &=
        \int_0^1
        \left(X_s - \frac{\partial_\beta \varpi_\beta(t,Z)}{\varpi_\beta(t,Z)}\right) \mathrm{d}N_s
        - \int_0^1
        \left(\frac{X_s}{\varpi_\beta(t,Z)} 
            - \frac{\partial_\beta{\varpi_\beta(t,Z)}}{\varpi_\beta(t,Z)^2}
        \right) \blambda(\beta,h)\mathrm{d}s \\
        &=
        \int_0^1
        \left(X_s - \frac{\ex_{P_0}[X_s e^{\beta X_s}\mid T\geq s, \mathcal{F}_{s-}^Z]}{\ex[e^{\beta X_s}\mid T\geq s, \mathcal{F}_{s-}^Z]}\right) \mathrm{d}N_s \\
        &\quad - \int_0^1
        \left(\frac{X_s}{\ex_{P_0}[e^{\beta X_s}\mid T\geq s, \mathcal{F}_{s-}^Z]} 
            - \frac{\ex_{P_0}[X_s e^{\beta X_s}\mid T\geq s, \mathcal{F}_{s-}^Z]}{(\ex_{P_0}[e^{\beta X_s}\mid T\geq s, \mathcal{F}_{s-}^Z])^2}
        \right) \blambda(\beta,h)\mathrm{d}s.
\end{align*}
After plugging in $\beta=0$ and simplifying we see that
\begin{align*}
    \psi(0, \eta_h) = \int_0^1 (X_s - \ex[X_s \mid \mathcal{F}_{s-}])(\mathrm{d}N_s - \blambda(0,h)\mathrm{d}s).
\end{align*}
This shows that the concentrating-out score evaluated at $\beta=0$ is exactly the score for the endpoint of the LCM estimator.

\supplementarysection{Additional details of simulation study} \label{sec:extrafigs}

This section contains additional details, numerical results, and figures related to the simulations of Section~\ref{sec:simulations}.

\begin{figure}
    \centering
    \includegraphics[width=\linewidth]{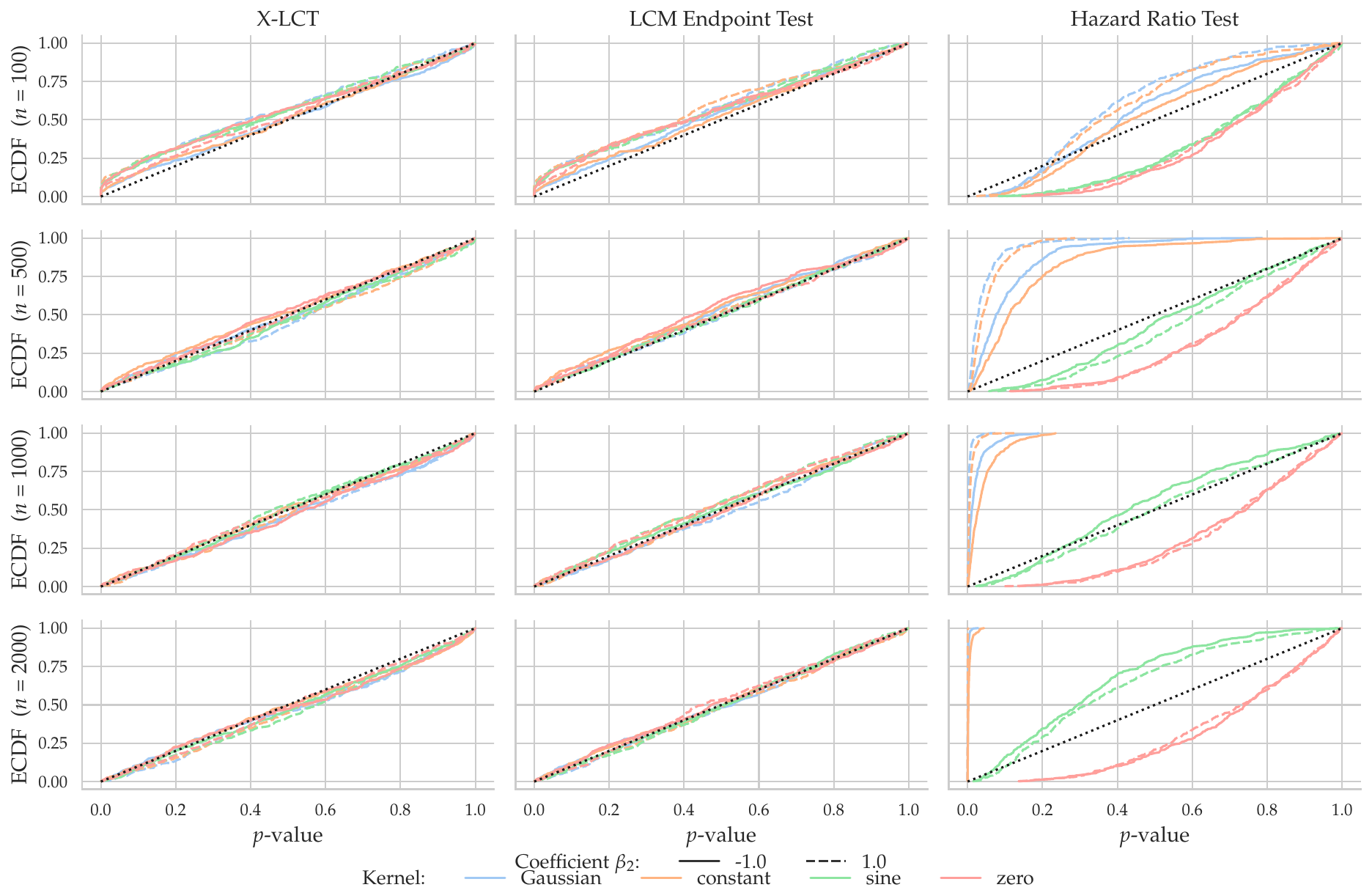}
    \caption{Empirical distribution functions of $p$-values for the three
    different conditional local independence tests considered, simulated under
    the sampling scheme described in Section~\ref{sec:simulations}. The dotted line
    shows $y=x$ corresponding to a uniform distribution.}\label{fig:H0pvals_full}
\end{figure}

\begin{figure}
    \centering
    \includegraphics[width=\linewidth]{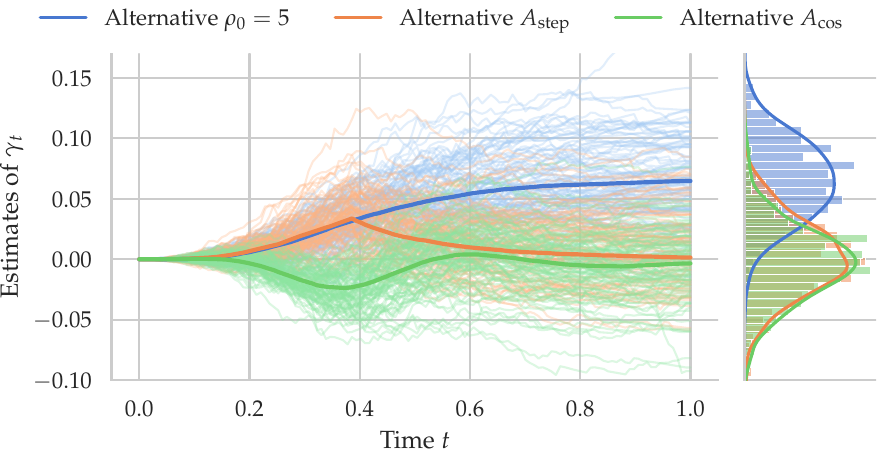}
    \caption{Sample paths of $\check{\gamma}^{K,(500)}$ fitted on data
    sampled from three different alternatives as described in Section 
    \ref{subsec:localalternatives}. 
    Here $(X,Y,Z)$ are sampled from the scheme described in Section \ref{sec:simulations},
    with both $\rho_X$ and $\rho_Y$ being the constant kernel and with $\beta = -1$.
    For each alternative, 100 paths are shown. The 
    empirical mean functions and the endpoint distributions
    are highlighted and computed based on 500 samples.}
     \label{fig:pathsalternative}
\end{figure}

\begin{figure}
    \centering
    \includegraphics[width = .7\linewidth]{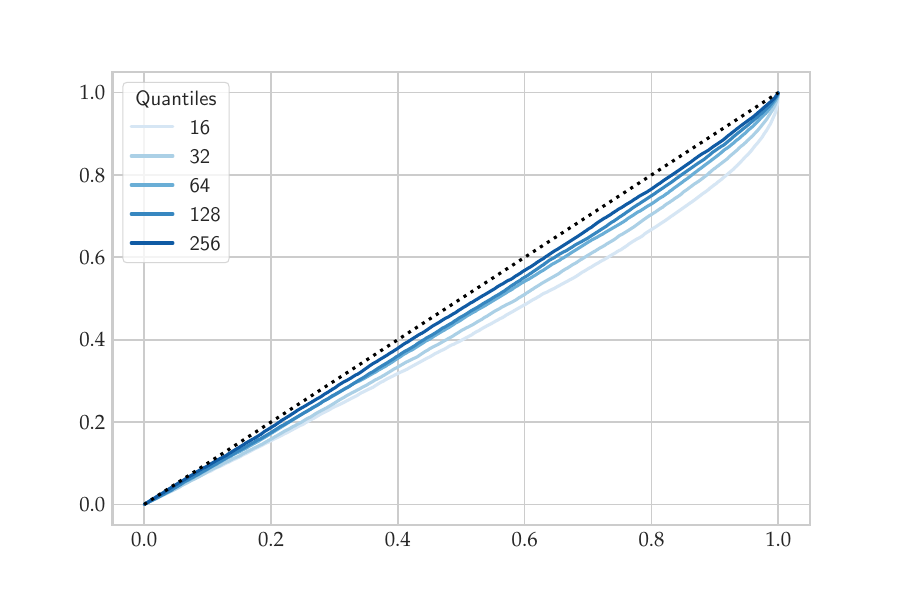}
    \caption{Empirical distribution functions of $p^{(q)} = 1- F_S(M^{(q)})$,
    where $M^{(q)} = (M_t^{(q)})_{t=1,\ldots,q}$ is a random walk with Gaussian
    increments such that $M_q^{(q)}$ has unit variance for each
    $q\in \{2^\ell \colon \ell=4,\ldots,8\}$. 
    Each empirical distribution function is based on $N=20\, 000$ samples.}
    \label{fig:supremumapprox}
\end{figure}

\subsection{Implementation of estimators and tests} \label{sec:implementation}
For our proof-of-concept implementation we used two simple off-the-shelf
estimators.

To estimate $\lambda$ we used the \texttt{BoXHED2.0} estimator from
\citet{pakbin2021boxhed2}, based on the works of \citet{wang2020boxhed} and
\citet{lee2021boosted}. In essence, the estimator is a gradient boosted forest
adapted to the setting of hazard estimation with time-dependent covariates. The
maximum depth and number of trees were tuned by 5-fold cross-validation over
the same grid as in \citet{pakbin2021boxhed2}. For computational ease, the
hyperparameters were tuned once on the entire dataset instead of tuning them on
each fold $J_n^k$. In principle, this may invalidate the asymptotic properties of $\check{\Psi}_n^K$ since it breaks the independence between $\widehat{\lambda}^{k,(n)}$ and $(T_j,X_j,Z_j)_{j\in J_n^k}$, but we believe that this dependency is negligible.

To estimate the predictable projection $\Pi_t = \ex(X_t \mid \mathcal{F}_{t-})$,
we fitted a series of linear least squares estimators by regressing $X_t$ on
$(Z_s)_{s\in \mathbb{T}:s<t}$ for each $t\in \mathbb{T}$. To stabilize the
estimation error $g(n)$, we added a small $L_2$-penalty with coefficient $0.001$
fixed across all experiments for simplicity. Since $X_t$ was sampled from a discretized historical linear model, the error $g(n)$ should in principle converge with a classical $n^{-1/2}$-rate. 
The historical linear regression estimator from the \texttt{scikit-fda} library was also considered initially, but we found that fitting this model was too computationally expensive for a simulation study with cross-fitting. 
In principle, in our time-continuous setting, we would like to use a functional estimator 
of $\Pi$ that would utilize the regularity along $s$ and $t$. Initial experiments, 
however, suggested that the simpler historical regression described above
gave similar results as using the \texttt{scikit-fda} library,
and we went with the less time consuming implementation.

Based on these estimators, the X-LCT was implemented based on Algorithm \ref{alg:X-lct}. Following the
recommendation by \citet[Remark 3.1.]{chernozhukov2018}, we computed the X-LCT with $K=5$ folds. The associated $p$-value was
computed with the series representation of $F_S$ truncated to the first $1000$ terms.

We compared our results for X-LCT with a hazard ratio test in the possibly misspecified 
marginal Cox model given by \eqref{eq:coxmarg}. This test was computed using the lifelines library
\citep{Davidson-Pilon2019}, specifically the \texttt{CoxTimeVaryingFitter}
model. The model was fitted with an $L_2$-penalty with a coefficient set
to $0.1$ (the default), and as a consequence the hazard ratio test is expected to be conservative. 

\subsection{Comparison with endpoint statistic} \label{sec:endpointsim}
We compare the X-LCT, which is based on the uniform norm of the X-LCM, with its endpoint counterpart. More precisely, we consider the test statistic
$$
    \left(\check{\mathcal{V}}_{K,n}(1)\right)^{-\frac{1}{2}}
    \sqrt{n}\check{\gamma}_1^{K,(n)},
$$
which is asymptotically standard normal under $H_0$.
With the simulation settings in Section~\ref{subsec:localalternatives}, the X-LCT turns out to be more 
or less indistinguishable from the corresponding endpoint test. 
This is because the alternatives considered have corresponding LCMs, which are most extreme towards $t=1$.
Therefore, the supremum and the endpoint behave similarly in these cases.

For this reason we consider local alternatives that result in a non-monotonic LCM. Using the same expression for the 
intensity \eqref{eq:alternative}, but with a time-varying $\rho_0$, 
we consider the alternatives
    \begin{align*}
        A_\text{step}\colon& 
            \rho_0(t) = 5 \cdot \one(t\leq 0.4) - 5 \cdot \one(t>0.4),\\
        A_\text{cos} \colon& 
            \rho_0(t) = 7 \cdot \cos(4\pi \cdot t).
    \end{align*}
The idea behind the alternative $A_\text{step}$ is that the LCM
should be increasing on $[0,0.4]$ and decreasing on $(0.4,1]$. 
Figure \ref{fig:pathsalternative} shows sample paths
of $\check{\gamma}^{K,(n)}$ for data simulated under each of the alternatives
$\rho_0=5$, $A_\text{step}$ and $A_\text{cos}$. The figure illustrates that
$t \mapsto |\check{\gamma}_t^{K,(n)}|$ is, indeed, mostly maximal towards $t=1$ 
for the alternative $\rho_0=5$, but not for the time-varying alternatives 
$A_\text{step}$ and $A_\text{cos}$.

With the same sampling scheme for $(X,Y,Z)$ as in Section \ref{sec:SamplingScheme}, we conducted
an analogous experiment with 400 runs for each setting.
Figure \ref{fig:timevaryingrho} shows the rejection rates for the 
two tests.

Under the hypothesis of conditional local independence, the left plot 
in Figure \ref{fig:timevaryingrho} shows that
the endpoint test behaves similarly to $\check{\Psi}_n^K$ as expected. Both
tests have power against the local alternatives, but for $A_\text{step}$ the
power does not seem to stabilize before $n=2000$. This is different from the
previous settings, and can be explained by a slower convergence of the intensity
estimator due to the more complex dependency on $X$. For both of the local
alternatives, we observe that $\check{\Psi}_n^K$ is more powerful than the
endpoint test, with the difference being largest for $A_\text{step}$. 
In conclusion, these results show that the supremum test dominates the 
endpoint test in certain situations.

\begin{figure}
    \includegraphics[width=\linewidth]{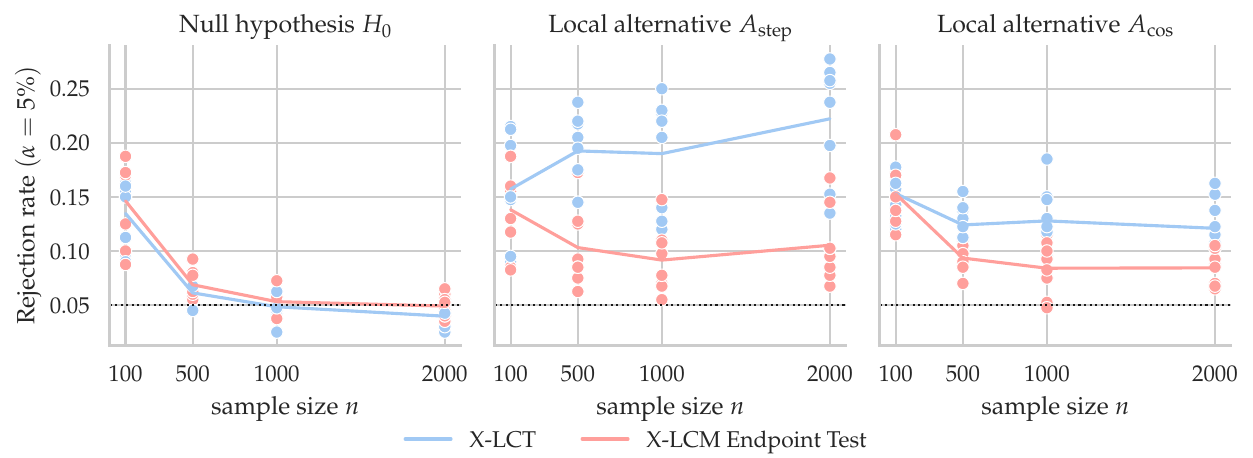}
    \caption{The plots show the average rejection rate of the 
    double machine learning tests based on the supremum statistic (blue)
    and the endpoint statistic (red).} \label{fig:timevaryingrho}
\end{figure}

%% file: paper_dope/main_paper.tex
\begin{abstract}
Covariate adjustment is a ubiquitous method used to estimate the average treatment effect (ATE) from observational data. Assuming a known graphical structure of the data generating model, recent results give graphical criteria for optimal adjustment, which enables efficient estimation of the ATE. 
However, graphical approaches are challenging for high-dimensional and complex data, and it is not straightforward to specify a meaningful graphical model of non-Euclidean data such as texts.
We propose a general framework that accommodates adjustment for any subset of information expressed by the covariates.
We generalize prior works and leverage these results to identify the optimal covariate information for efficient adjustment. This information is minimally sufficient for prediction of the outcome conditionally on treatment. 

Based on our theoretical results, we propose the Debiased Outcome-adapted Propensity Estimator (DOPE) for efficient estimation of the ATE, and we provide asymptotic results for the DOPE under general conditions. Compared to the augmented inverse propensity weighted (AIPW) estimator, the DOPE can retain its efficiency even when the covariates are highly predictive of treatment. We illustrate this with a single-index model, and with an implementation of the DOPE based on neural networks, we demonstrate its performance on simulated and real data. Our results show that the DOPE provides an efficient and robust methodology for ATE estimation in various observational settings.
\end{abstract}




\section{Introduction}
Estimating the population average treatment effect (ATE) of a treatment on an outcome variable is a fundamental statistical task. A naive approach is to contrast the mean outcome of a treated population with the mean outcome of an untreated population. Using observational data this is, however, generally a flawed approach due to confounding. 
If the underlying confounding mechanisms are captured by a set of pre-treatment covariates~$\bW$, it is possible to adjust for confounding by conditioning on $\bW$ in a certain manner. 
Given that multiple subsets of $\bW$ may be valid for this adjustment, it is natural to ask if there is an `optimal adjustment subset' that enables the most efficient estimation of the ATE.

Assuming a causal linear graphical model, \citet{henckel2022graphical} established the existence of -- and gave graphical criteria for -- an \emph{optimal adjustment set} for the OLS estimator. \citet{rotnitzky2020efficient} extended the results of \citet{henckel2022graphical}, by proving that the optimality was valid within general causal graphical models and for all regular and asymptotically linear estimators. 
Critically, this line of research assumes knowledge of the underlying graphical structure. 

To accommodate the assumption of \textit{no unmeasured confounding}, 
observational data is often collected with as many covariates as possible, which means that $\bW$ can be high-dimensional. In such cases, assumptions of a known graph are unrealistic, and graphical estimation methods are
statistically unreliable \citep{uhler2013geometry,shah2020hardness,chickering2004large}. 
Furthermore, for non-Euclidean data such as images or texts, it is not clear how to impose any graphical structure pertaining to causal relations. 
Nevertheless, we can in these cases still imagine that the information that $\bW$ represents can be separated into distinct components that affect treatment and outcome directly, as illustrated in Figure~\ref{fig:semistructuredW}.
\begin{figure}
    \centering
      \begin{tikzpicture}[node distance=1.5cm, thick, roundnode/.style={circle, draw, inner sep=1pt,minimum size=7mm},
        squarenode/.style={rectangle, draw, inner sep=1pt, minimum size=7mm}]
        \node[roundnode] (T) at (0,-0.3) {$T$};
        \node (t) at (0,-0.78) {\scriptsize Treatment};
        \node[roundnode] (Y) at (2,-0.3) {$Y$};
        \node (y) at (2,-0.78) {\scriptsize Outcome};
        \node (Z) at (1,1.85) {$\bW$};
        \node[roundnode,dashed] (Z1) at (-0.5,1.4) {$\mathbf{C}_1$};
        \node[roundnode,dashed] (Z2) at (1,1.1) {$\mathbf{C}_2$};
        \node[roundnode,dashed] (Z3) at (2.5,1.4) {$\mathbf{C}_3$};
    
        \draw (1,1.35) ellipse (2.5cm and 0.9cm);
        \path [-latex,draw,thick, blue] (T) edge [bend left = 0] node {} (Y);
        \path [-latex,draw,thick] (Z1) edge [bend left = 0] node {} (T);
        \path [-latex,draw,thick] (Z2) edge [bend left = 0] node {} (T);
        \path [-latex,draw,thick] (Z2) edge [bend left = 0] node {} (Y);
        \path [-latex,draw,thick] (Z3) edge [bend left = 0] node {} (Y);
      \end{tikzpicture}
    \caption{The covariate $\bW$ can have a complex data structure, even if the information it represents is structured and can be categorized into components that influence treatment and outcome separately.}
    \label{fig:semistructuredW}
\end{figure}
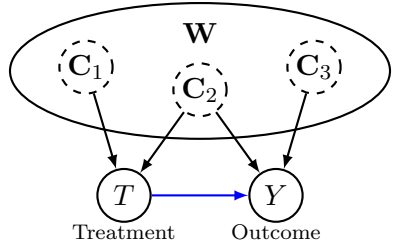

In this paper, we formalize this idea and formulate a novel and general adjustment theory with a focus on efficiency bounds for the estimation of the average treatment effect. 
Based on the adjustment theory, we propose a general estimation procedure and analyze its asymptotic behavior. 

\subsection{Setup}
Throughout we consider a discrete treatment variable $T\in \mathbb{T}$, a square-integrable outcome variable $Y\in \real$, and pre-treatment covariates $\bW \in \mathbb{W}$. For now we only require that $\mathbb{T}$ is a finite set and that $\mathbb{W}$ is a measurable space.
The joint distribution of $(T,\bW,Y)$ is denoted by $P$ and it is assumed to belong to a collection of probability measures~$\mathcal{P}$. If we need to make the joint distribution explicit, we denote expectations and probabilities with $\ex_P$ and $\mathbb{P}_P$, respectively, but usually the $P$ is omitted for ease of notation.

Our model-free target parameters of interest are of the form
\begin{align}\label{eq:adjustedmean}
    \chi_t
    \coloneq \ex[\ex[Y\given T=t,\bW]]
    = \ex\left[
        \frac{\one(T=t)Y}{\mathbb{P}(T=t \given \bW)}
    \right],
    \qquad t\in \mathbb{T}.
\end{align}
In words, these are treatment specific means of the outcome when adjusting for the covariate $\bW$. To ensure that this quantity is well-defined, we assume the following condition, commonly known as \textit{positivity}.
\begin{asm}[Positivity] \label{asm:positivity}
    It holds that $0<\mathbb{P}(T=t \given \bW)<1$ almost surely for each $t\in\mathbb{T}$.
\end{asm}

Under additional assumptions common in causal inference literature -- which informally entail that $\bW$ captures all confounding -- the target parameter $\chi_t$ has the interpretation as the \emph{interventional mean}, which is expressed by $\ex[Y\given \mathrm{do}(T=t)]$ in \textit{do-notation} or by $\ex[Y^{t}]$ using \textit{potential outcome} notation \citep{peters2017elements,van2011targeted}. Under such causal assumptions, the average treatment effect is identified, and 
 when $\mathbb{T} = \{0,1\}$ it is typically expressed as the contrast $\chi_1-\chi_0$. 
The theory in this paper is agnostic with regards to whether or not $\chi_t$ has this causal interpretation, although it is the primary motivation for considering $\chi_t$ as a target parameter.

Given $n$ i.i.d. observations of $(T,\bW,Y)$, one may proceed to estimate $\chi_t$ by estimating either of the equivalent expressions for $\chi_t$ in Equation \eqref{eq:adjustedmean}. Within parametric models, the outcome regression function 
$$
    g(t,\bw) = \ex[Y\given T=t,\bW=\bw]
$$ 
can typically be estimated with a $\sqrt{n}$-rate. In this case, the sample mean of the estimated regression function yields a $\sqrt{n}$-consistent estimator of $\chi_t$ under Donsker class conditions or sample splitting. However, many contemporary datasets indicate that parametric model-based regression methods get outperformed by nonparametric methods such as boosting and neural networks \citep{bojer2021kaggle}.

Nonparametric estimators of the regression function typically converge at rates slower than $\sqrt{n}$, and likewise for estimators of the propensity score 
$$
    m(t\mid \bw) = \mathbb{P}(T=t\given \bW=\bw).
$$
Even if both nonparametric estimators have rates slower than $\sqrt{n}$, it is in some cases possible to achieve a $\sqrt{n}$-rate of $\chi_t$ by modeling both $m$ and $g$, and then combining their estimates in a way that achieves `rate double robustness' \citep{smucler2019unifying}. That is, an estimation error of the same order as the product of the errors for $m$ and $g$.
Two prominent estimators that have this property are the Augmented Inverse Probability Weighted estimator (AIPW) and the Targeted Minimum Loss-based Estimator (TMLE) \citep{robins1995semiparametric,chernozhukov2018,van2011targeted}.

In what follows, the premise is that even with a $\sqrt{n}$-rate estimator of $\chi_t$, it might still be intractable to model $m$ -- and possibly also $g$ -- as a function of $\bW$ directly. This can happen for primarily two reasons:
(i) the sample space $\mathbb{W}$ is high-dimensional or has a complex structure, or (ii) the covariate is highly predictive of treatment, leading to unstable predictions of the inverse propensity score $\mathbb{P}(T=t\given \bW)^{-1}$.

In either of these cases, which are not exclusive, we can try to manage these difficulties by instead working with a \emph{representation} $\bZ = \phi(\bW)$, given by a measurable mapping $\phi$ from $\mathbb{W}$ into a more tractable space such as $\real^d$. In the first case above, such a representation might be a pre-trained word embedding, e.g., the celebrated BERT and its offsprings \citep{devlin2018bert}. The second case has been well-studied in the special case where $\mathbb{W}=\real^k$ and where $\mathcal{P}$ contains the distributions that are consistent with respect to a fixed DAG (or CPDAG). We develop a general theory that subsumes both cases, and we discuss how to represent the original covariates to efficiently estimate the adjusted mean $\chi_t$.

\subsection{Relations to existing literature}
Various studies have explored the adjustment for complex data structures by utilizing a (deep) representation of the covariates, as demonstrated in works such as \citet{shi2019adapting,veitch2020adapting}. In a different research direction, the \textit{Collaborative TMLE} \citep{van2010collaborative} has emerged as a robust method for estimating 
average treatment effects by collaboratively learning the outcome regression and propensity score, particularly in scenarios where covariates are highly predictive of treatment \citep{ju2019scalable}. Our overall estimation approach shares similarities with the mentioned strategies; for instance, our proof-of-concept estimator in the experimental section employs neural networks with shared layers. However, unlike the cited works, it incorporates the concept of efficiently tuning the representation specifically for predicting outcomes, rather than treatment. Related to this idea is another interesting line of research, which builds upon the \textit{outcome adapted lasso} proposed by \citet{shortreed2017outcome}. Such works include \citet{ju2020robust,benkeser2020nonparametric,greenewald2021high,balde2023reader}. These works all share the common theme of proposing estimation procedures that select covariates based on $L_1$-penalized regression onto the outcome, and then subsequently estimate the propensity score based on the selected covariates adapted to the outcome. The theory of this paper generalizes the particular estimators proposed in the previous works, and also allows for other feature selection methods than $L_1$-penalization. Moreover, our generalization of (parts of) the efficient adjustment theory from \citet{rotnitzky2020efficient} allows us to theoretically quantify the efficiency gains from these estimation methods. Finally, our asymptotic theory considers a novel regime, which, according to the simulations, seems more adequate for describing the finite sample behavior than the asymptotic results of \citet{benkeser2020nonparametric} and \citet{ju2020robust}.

Our general adjustment results in Section~\ref{sec:informationbounds} draw on the vast literature on classical adjustment and confounder selection, for example \citet{rosenbaum1983central,hahn1998role,henckel2022graphical,rotnitzky2020efficient,guo2022confounder,perkovic2018complete,peters2017elements,forre2023mathematical}. 
In particular, two of our results are direct extensions of results from \citet{rotnitzky2020efficient,henckel2022graphical}.

\subsection{Organization of the paper}
In Section~\ref{sec:adjustment} we discuss generalizations of classical adjustment concepts to abstract conditioning on information.
In Section~\ref{sec:informationbounds} we discuss information bounds in the framework of Section \ref{sec:adjustment}. 
In Section~\ref{sec:estimation} we propose a novel method, the DOPE, for efficient estimation of adjusted means, and we discuss the asymptotic behavior of the resulting estimator.
In Section~\ref{sec:experiments} we implement the DOPE and demonstrate its performance on synthetic and real data.
The paper is concluded by a discussion in Section~\ref{sec:discussion}.

\section{Generalized adjustment concepts}\label{sec:adjustment}
In this section we discuss generalizations of classical adjustment concepts. These generalizations are motivated by the premise from the introduction: it might be intractable to model the propensity score directly as a function of~$\bW$, so instead we consider adjusting for a representation $\bZ=\phi(\bW)$. This is, in theory, confined to statements about conditioning on $\bZ$, and is therefore equivalent to adjusting for any representation of the form $\widetilde{\bZ} = \psi \circ \phi (\bW)$, where $\psi$ is a bijective and bimeasureable mapping. 
The equivalence class of such representations is characterized by the $\sigma$-algebra generated by~$\bZ$, denoted by $\sigma(\bZ)$, which informally describes the information contained in $\bZ$. In view of this, we define adjustment with respect to sub-$\sigma$-algebras contained in $\sigma(\bW)$. 

\begin{remark}\label{rmk:conditional}
    Conditional expectations and probabilities are, unless otherwise indicated, defined conditionally on $\sigma$-algebras as in \citet{kolmogoroff1933grundbegriffe}, see also \citet[Ch. 8]{kallenberg2021foundations}. Equalities between conditional expectations are understood to hold almost surely. When conditioning on the random variable $T$ and a $\sigma$-algebra $\cZ$ we write `$\given T,\cZ$' as a shorthand for `$\given \sigma(T)\vee \cZ$'. Finally, we define conditioning on both the event $(T=t)$ and on a $\sigma$-algebra $\cZ\subseteq \sigma(\bW)$ by
    $
        \ex[Y \given T=t, \cZ] 
        \coloneq \frac{\ex[Y \one(T=t) \given \cZ]
            }{\mathbb{P}(T=t \given \cZ)},
    $
    which is well-defined under Assumption \ref{asm:positivity}.
\end{remark}

\begin{definition}\label{def:adjustment} A sub-$\sigma$-algebra $\cZ\subseteq \sigma(\bW)$ is 
called a \emph{description} of $\bW$. For each $t\in \mathbb{T}$ and $P\in\mathcal{P}$, 
and with $\cZ$ a description of $\bW$, we define
    \begin{align*}
        \pi_t(\cZ; P) &\coloneq \mathbb{P}_P(T=t\given \cZ), \\
        b_t(\cZ; P) &\coloneq \ex_P[Y\given T=t, \cZ], \\
        \chi_t(\cZ;P) &\coloneq \ex_P[b_t(\cZ;P)] 
        = \ex_P[\pi_t(\cZ;P)^{-1}\one(T=t)Y].
    \end{align*}
    If a description $\cZ$ of $\bW$ is given as $\cZ = \sigma(\bZ)$ for a representation $\bZ=\phi(\bW)$,
    we may write $\chi_t(\bZ;P)$ instead of $\chi_t(\sigma(\bZ);P)$ etc. 
    
    We say that  
    \begin{align*}
        \textnormal{$\cZ$ is $P$-valid if:} 
            &\qquad
            \chi_t(\cZ; P) = \chi_t(\bW; P), &\text{ for all } 
            t \in \mathbb{T}, \\
        \textnormal{$\cZ$ is $P$-OMS if:} 
            &\qquad
            b_t(\cZ; P) = b_t(\bW; P), &\text{ for all } 
            t \in \mathbb{T}, \\
        \textnormal{$\cZ$ is $P$-ODS if:} 
            &\qquad
            Y \indP \bW \given T,\cZ.
    \end{align*}
    Here OMS means \emph{Outcome Mean Sufficient} and ODS means 
    \emph{Outcome Distribution Sufficient}.
    If $\cZ$ is $P$-valid for all $P\in \mathcal{P}$, we say that it is $\mathcal{P}$-valid. We define $\mathcal{P}$-OMS and $\mathcal{P}$-ODS analogously. 
\end{definition}
A few remarks are in order.
\begin{itemize}
    \item We have implicitly extended the shorthand notation described in Remark \ref{rmk:conditional} to the quantities in Definition \ref{def:adjustment}. 

    \item The quantity $\chi_t(\cZ; P)$ is deterministic, whereas $\pi_t(\cZ;P)$ and $b_t(\cZ;P)$ are $\cZ$-measurable real valued random variables. Thus, if $\cZ$ is generated by a representation $\bZ$, then by the Doob-Dynkin lemma \citep[Lemma 1.14]{kallenberg2021foundations} these random variables can be expressed as functions of $\bZ$. 
    This fact does not play a role until the discussion of estimation in Section~\ref{sec:estimation}.

    \item There are examples of descriptions $\cZ \subseteq \sigma(\bW)$ that are not given by representations. Such descriptions might be of little practical importance, but our results do not require $\cZ$ to be given by a representation. 
    We view the $\sigma$-algebraic framework as a convenient abstraction that generalizes equivalence classes of representations. 

    \item We have the following hierarchy of the properties in Definition~\ref{def:adjustment}:
    \begin{align*}
    P\textnormal{-ODS} 
    \Longrightarrow P \textnormal{-OMS} 
    \Longrightarrow P \textnormal{-valid},
    \end{align*}
    and the relations also hold if $P$ is replaced by $\mathcal{P}$.

    \item The $P$-ODS condition relates to existing concepts in statistics. The condition can be viewed as a $\sigma$-algebraic analogue of \textit{prognostic scores} \citep{hansen2008prognostic}, and it holds that any prognostic score generates a $P$-ODS description. Moreover, if a description $\cZ = \sigma(\bZ)$ is $P$-ODS, then its generator $\bZ$ is \textit{$c$-equivalent} to $\bW$ in the sense of \citet{pearl2009causality}, see in particular the claim following his Equation (11.8). In Remark~\ref{rmk:minimalsufficiency} we discuss the relations between $\mathcal{P}$-ODS descriptions and classical statistical sufficiency.
\end{itemize}

The notion of $\mathcal{P}$-valid descriptions can be viewed as a generalization of valid adjustment sets, where subsets are replaced with sub-$\sigma$-algebras.

\begin{example}[Comparison with adjustment sets in causal DAGs]
    Suppose $\bW \in \real^k$ and let $\mathcal{D}$ be a DAG on the nodes $\mathbf{V}=(T,\bW,Y)$. Let $\mathcal{P}=\mathcal{M}(\mathcal{D})$ be the collection of continuous distributions (on $\real^{k+2}$) that are Markovian with respect to $\mathcal{D}$ and with $\ex|Y|<\infty$.

    Any subset $\bZ\subseteq \bW$ is a representation of $\bW$ given by a coordinate projection, and the corresponding $\sigma$-algebra $\sigma(\bZ)$ is a description of $\bW$. In this framework, a subset $\bZ\subseteq \bW$ is called a \textit{valid adjustment set} for $(T,Y)$ if for all $P\in\mathcal{P}$, $t\in \mathbb{T}$, and $y\in \real$
    \begin{align*}
        \ex_P\left[
            \frac{\one(T=t)\one(Y\leq y)}{\mathbb{P}_P(T=t\given \mathrm{pa}_{\mathcal{D}}(T))}
            \right]
        = \ex_P[\mathbb{P}_P(Y\leq y\given T=t, \bZ)],
    \end{align*}
    where $\mathrm{pa}_{\mathcal{D}}$ denotes the parents $T$ in $\mathcal{D}$, see for example Definition 2 in \citet{rotnitzky2020efficient}.
    It turns out that $\bZ$ is a valid adjustment set if and only if 
    \begin{align*}
        \ex_P\left[
            \frac{\one(T=t)Y}{\mathbb{P}_P(T=t\given \mathrm{pa}_{\mathcal{D}}(T))}
            \right]
        = \ex_P[\ex_P[Y\given T=t, \bZ]] 
        = \chi_t(\cZ; P)
    \end{align*}
    for all $P \in\mathcal{P}$ and $t\in \mathbb{T}$. 
    This follows\footnote{We thank Leonard Henckel for pointing this out.} from results of \citet{perkovic2018complete}, which we discuss for completeness in Proposition~\ref{prop:adjset} in the supplement.
    Thus, if we assume that $\bW$ is a valid adjustment set, then any subset $\bZ\subseteq \bW$ is a valid adjustment set if and only if the corresponding description $\sigma(\bZ)$ is $\mathcal{P}$-valid. 
\end{example}

In general, $\mathcal{P}$-valid descriptions are not necessarily generated from valid adjustment sets. This can happen if structural assumptions are imposed on the conditional mean, which is illustrated in the following example.

\begin{example}\label{exm:magnitude}
    Suppose that the outcome regression is known to be invariant under rotations of $\bW \in \real^{k}$ such that for any $P\in \mathcal{P}$,
    \begin{align*}
        \ex_P[Y \given T,\bW] = \ex_P[Y \given T,\|\bW\|\,] \eqcolon g_P(T,\|\bW\|).
    \end{align*}
    Without further graphical assumptions, we cannot deduce that any proper subset of $\bW$ is a valid adjustment set. In contrast, the magnitude $\|\bW\|$ generates a $\mathcal{P}$-OMS -- and hence also $\mathcal{P}$-valid -- description of $\bW$ by definition.
    
    Suppose that there is a distribution $\tilde{P}\in \mathcal{P}$ for which 
    $g_{\tilde{P}}(t,\cdot)$ is bijective. 
    Then $\sigma(\|\bW\|)$ is also the smallest $\mathcal{P}$-OMS description up to $\mathbb{P}_{\tilde{P}}$-negligible sets: if $\cZ$ is another $\mathcal{P}$-OMS description, we see that $b_t(\cZ; \tilde{P}) = b_t(\bW; \tilde{P}) = g_{\tilde{P}}(t,\|\bW\|)$ almost surely. 
    Hence 
    $$
        \sigma(\|\bW\|) =\sigma(g_{\tilde{P}}(t,\|\bW\|))  
        \subseteq \ol{\sigma(b_t(\cZ; \tilde{P}))} 
        \subseteq \ol{\cZ},
    $$
    where overline denotes the $\mathbb{P}_{\tilde{P}}$-completion of a $\sigma$-algebra. That is, $\ol{\cZ}$
    is the smallest $\sigma$-algebra containing $\cZ$ and all its $\mathbb{P}_{\tilde{P}}$-negligible sets.
\end{example}
    

Even in the above example, where the regression function is known to depend on a one-dimensional function of $\bW$, it is generally not possible to estimate the regression function at a $\sqrt{n}$-rate without restrictive assumptions. 
Thus, modeling of the propensity score is required in order to obtain a $\sqrt{n}$-rate estimator of the adjusted mean. If $\bW$ is highly predictive of treatment, then naively applying a doubly robust estimator (AIPW, TMLE) can be statistically unstable due to large inverse propensity weights. Alternatively, since $\sigma(\|\bW\|)$ is a $\mathcal{P}$-valid description, we could also base our estimator on the \textit{pruned propensity} $\mathbb{P}(T=t \given \|\bW\|)$. 
This approach should intuitively provide more stable weights, as we expect $\|\bW\|$ to be less predictive of treatment. We proceed to analyze the difference between the asymptotic efficiencies of the two approaches and we show that, under reasonable conditions, the latter approach is never worse asymptotically.

\section{Efficiency bounds for adjusted means}\label{sec:informationbounds}
We now discuss efficiency bounds based on the concepts introduced in Section~\ref{sec:adjustment}. To this end, let $\bZ = \phi(\bW)$ be a representation of $\bW$. 
Under a sufficiently dense model $\mathcal{P}$, the \textit{influence function} for $\chi_t(\bZ;P)$ is given by 
\begin{equation}\label{eq:IF}
    \psi_{t}(\bZ; P) 
        = b_t(\bZ; P) + \frac{\one(T=t)}{\pi_t(\bZ; P)}(Y-b_t(\bZ; P))
            -\chi_t(\bZ;P).
\end{equation}
The condition that $\mathcal{P}$ is sufficiently dense entails that no `structural assumptions' are imposed on the functional forms of $\pi_t$ and $b_t$, see also \citet{robins1994estimation,robins1995semiparametric,hahn1998role}. 
Structural assumptions do not include smoothness conditions, but do include, for example, parametric model assumptions on the outcome regression.

Using the formalism of Section \ref{sec:adjustment}, we define 
$\psi_{t}(\cZ; P)$ analogously to \eqref{eq:IF} for any description $\cZ$ of $\bW$. 
We denote the variance of the influence function by
\begin{align}\label{eq:asympvar}
    \bV_t(\cZ;P) 
    &\coloneq \var_P[\psi_{t}(\cZ; P)] \nonumber \\
    &= \ex_P\left[
        \frac{\var_P(Y\given T=t, \cZ)}{\pi_t(\cZ;P)}
    \right]
    + \var_P(b_t(\cZ;P)).
\end{align}

The importance of this variance was highlighted by \citet{hahn1998role}, who computed $\bV_t(\bW;P)$ as the semiparametric efficiency bound for regular asymptotically linear (RAL) estimators of $\chi_t(\bW;P)$. That is, for any $\mathcal{P}$-consistent RAL estimator $\widehat{\chi}_t(\bW;P)$ based on $n$ i.i.d. observations, the asymptotic variance of $\sqrt{n}(\widehat{\chi}_t(\bW;P) - \chi_t(\bW;P))$ is bounded below by $\bV_t(\bW; P)$, provided that $\mathcal{P}$ does not impose structural assumptions on $\pi_t(\bW)$ and $b_t(\bW)$. 
Moreover, the asymptotic variance of both the TMLE and AIPW estimators achieve this lower bound when the propensity score and outcome regression can be estimated with sufficiently fast rates \citep{chernozhukov2018,van2011targeted}.

Since the same result can be applied for the representation $\bZ = \phi(\bW)$, we see that estimation of $\chi_t(\bZ; P)$ has a semiparametric efficiency bound of $\bV_t(\bZ ; P)$. Now suppose that $\sigma(\bZ)$ is a $\mathcal{P}$-valid description of $\bW$, which means that $\chi_t \coloneq \chi_t(\bW;P) = \chi_t(\bZ; P)$ for all $P$. It is then natural to ask if we should estimate $\chi_t$ based on $\bW$ or the representation $\bZ$? 
We proceed to investigate this question based on which of the efficiency bounds $\bV_t(\bZ ; P)$ or $\bV_t(\bW ; P)$ is smallest. 

If $\bV_t(\bZ ; P)$ is smaller than $\bV_t(\bW ; P)$, then $\bV_t(\bW ; P)$ is not an actual efficiency bound for $\chi_t(\bW;P)$. This does not contradict the result of \citet{hahn1998role}, as we have assumed the existence of a non-trivial $\mathcal{P}$-valid description of $\bW$, which implicitly imposes structural assumptions on the functional form of $b_t(\bW)$. Nevertheless, it is sensible to compare the variances $\bV_t(\bZ ; P)$ and $\bV_t(\bW ; P)$, as these are the asymptotic variances of the AIPW when using either $\bZ$ or $\bW$, respectively, as a basis for the nuisance functions. 

We formulate our efficiency bounds in terms of the more general contrast parameter
\begin{equation}\label{eq:contrasttarget}
    \Delta = \Delta(\bW ; P) = \sum_{t\in \mathbb{T}} c_t \chi_t(\bW; P),
\end{equation}
where $\mathbf{c}\coloneq (c_t)_{t\in \mathbb{T}}$ are fixed real-valued coefficients. 
The prototypical example, when $\mathbb{T} = \{0,1\}$, is $\Delta =\chi_1 - \chi_0$, which is the average treatment effect under causal assumptions, cf. the discussion following Assumption~\ref{asm:positivity}. 
Note that the family of $\Delta$-parameters includes the adjusted mean $\chi_t$ as a special case.

To estimate $\Delta$, we consider estimators of the form
\begin{equation}\label{eq:hatDelta}
    \widehat \Delta (\cZ;P) =
        \sum_{t\in \mathbb{T}} c_t\widehat\chi_t(\cZ;P),
\end{equation}
where $\widehat\chi_t(\cZ;P)$ denotes a consistent RAL estimator of $\chi_t(\cZ;P)$.
Correspondingly, the efficiency bound for such an estimator is
\begin{equation}
    \avar(\cZ;P) \coloneq \var_P\Big(\sum_{t\in \mathbb{T}} c_t\psi_{t}(\cZ; P)\Big).
\end{equation}

It turns out that two central results by \citet{rotnitzky2020efficient}, 
specifically their Lemmas 4 and 5, can be generalized from covariate subsets to descriptions. One conceptual difference is that there is \textit{a priori} no natural generalization of \textit{precision variables} and \textit{overadjustment (instrumental) variables}. To wit, if $\cZ_1$ and $\cZ_2$ are descriptions, there is no canonical way\footnote{equivalent conditions to the existence of an independent complement are given in Proposition 4 in \citet{emery2001vershik}.} to subtract $\cZ_2$ from $\cZ_1$ in a way that maintains their join $\cZ_1 \vee \cZ_2 \coloneq \sigma(\cZ_1,\cZ_2)$. 
Apart from this technical detail, the proofs translate more or less directly. The following lemma is a direct extension of \citet[Lemma 4]{rotnitzky2020efficient}.

\begin{lem}[Deletion of overadjustment]\label{lem:overadj}
    Fix a distribution $P\in \mathcal{P}$ and let $\cZ_1\subseteq \cZ_2$ be $\sigma$-algebras such that $Y\indP \cZ_2 \given T,\cZ_1$. Then it always holds that
    $$
    \avar(\cZ_2;P) - \avar(\cZ_1;P)
    = \sum_{t\in\mathbb{T}} c_t^2 D_t(\cZ_1,\cZ_2;P) \geq 0,
    $$
    where for each $t\in \mathbb{T}$,
    \begin{align*}
        D_t&(\cZ_1,\cZ_2;P) \coloneq \bV_t(\cZ_2 ; P)-\bV_t(\cZ_1 ; P) \\
        &= \ex_P\,\Big[
            \pi_t(\cZ_1;P)
            \var(Y\given T=t,\cZ_1)
            \var\big(\pi_t(\cZ_2;P)^{-1}
                \big| \; T=t,\cZ_1 \big)\Big].
    \end{align*}
    Moreover, if $\cZ_2$ is a description of $\bW$ then $\cZ_1$ is $P$-valid if and only if $\cZ_2$ is $P$-valid.
\end{lem}
The lemma quantifies the efficiency lost in adjustment when adding information that is irrelevant for the outcome. 

We proceed to apply this lemma to the minimal information in $\bW$ that is predictive of $Y$ conditionally on $T$. To define this information, we use the regular conditional distribution function of $Y$ given $T=t,\bW=\bw$, which we denote by
\begin{align*}
    F(y\given t,\bw; P) \coloneq \mathbb{P}_P (Y\leq y \given T=t,\bW=\bw), 
        \qquad y\in \real, t\in\mathbb{T}, \bw\in \mathbb{W}.
\end{align*}
See \citet[Sec. 8]{kallenberg2021foundations} for a rigorous treatment of regular conditional distributions.
We will in the following, by convention, take
\begin{equation} \label{eq:btw}
    b_t(\bw;P) = \int y\,\mathrm{d}F(y\given t,\bw; P),
\end{equation}
so that $b_t(\bW;P)$ is given in terms of the regular conditional distribution.

\begin{definition}\label{def:OutcomeAlgebras}
    Define the $\sigma$-algebras
    \begin{align*}
        \mathcal{Q} &\coloneq \bigvee_{P\in\mathcal{P}} \mathcal{Q}_P,
            \qquad \mathcal{Q}_P \coloneq \sigma(F(y\given t,\bW; P) ; \, y\in \real, t\in \mathbb{T}), \\
        \mathcal{R} &\coloneq \bigvee_{P\in\mathcal{P}} \mathcal{R}_P,
            \qquad \mathcal{R}_P \coloneq \sigma(b_t(\bW; P) ; \, t\in \mathbb{T}).
    \end{align*}
\end{definition}

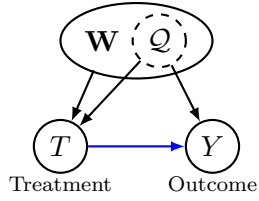
\begin{figure}
    \centering
    \begin{tikzpicture}[node distance=1.5cm, thick, roundnode/.style={circle, draw, inner sep=1pt,minimum size=7mm},
      squarenode/.style={rectangle, draw, inner sep=1pt, minimum size=7mm}]
        \node[roundnode] (T) at (0,0) {$T$};
        \node (t) at (0,-0.5) {\scriptsize Treatment};
        \node[roundnode,dashed] (W) at (1.3,1.4) {$\mathcal{Q}$};
        \node (Z) at (0.55,1.4) {$\bW$};
        \node (z) at (0.5,1.15) {};
        \node[roundnode] (Y) at (2,0) {$Y$};
        \node (y) at (2,-0.5) {\scriptsize Outcome};
        
        \draw (1,1.4) ellipse (1cm and 0.5cm);
        
        \path [-latex,draw,thick] (W) edge [bend left = 0] node {} (T);
        \path [-latex,draw,thick] (z) edge [bend left = 0] node {} (T);
        \path [-latex,draw,thick] (W) edge [bend left = 0] node {} (Y);
        \path [-latex,draw,thick, blue] (T) edge [bend left = 0] node {} (Y);
    \end{tikzpicture}
    \caption{The $\sigma$-algebra $\mathcal{Q}$ given in Definition~\ref{def:OutcomeAlgebras}
    as a description of $\bW$.}
    \label{fig:COA}
\end{figure}
Note that $\mathcal{Q}_P$, $\mathcal{Q}$, $\mathcal{R}_P$ and $\mathcal{R}$ are all descriptions 
of $\bW$, see Figure \ref{fig:COA} for a depiction of the information contained in $\mathcal{Q}$. 
Note also that $\mathcal{R}_P \subseteq \mathcal{Q}_P$ by the convention \eqref{eq:btw}.

We now state one of the main results of this section. 
\begin{thm}\label{thm:COSefficiency}
    Fix $P\in \mathcal{P}$. Then, under Assumption \ref{asm:positivity}, it holds that
    \begin{enumerate}
        \item[(i)] If $\cZ$ is a description of $\bW$ and $\mathcal{Q}_P\subseteq \ol{\cZ}$, then $\cZ$ is $P$-ODS. In particular, $\mathcal{Q}_P$ is a $P$-ODS description of $\bW$.
        \item[(ii)] If $\cZ$ is $P$-ODS description, then $\mathcal{Q}_P\subseteq \ol{\cZ}$.
        \item[(iii)] For any $P$-ODS description $\cZ$ it holds that
        \begin{equation}\label{eq:COApointwise}
            \avar(\cZ; P) - \avar(\mathcal{Q}_P; P) 
            = \sum_{t\in \mathbb{T}} 
                c_t^2 D_t(\mathcal{Q}_P,\cZ; P)
            \geq 0,
        \end{equation}
        where $D_t$ is given as in Lemma~\ref{lem:overadj}.
    \end{enumerate}
\end{thm}

Together parts (i) and (ii) state that a description of $\bW$ is $P$-ODS if and only if its $\mathbb{P}_P$-completion contains $\mathcal{Q}_P$. Part (iii) states that, under $P\in \mathcal{P}$, $\mathcal{Q}_P$ leads to 
the optimal efficiency bound among all $P$-ODS descriptions.

In the following corollary we use $\ol{\cZ}^P$ to denote the $\mathbb{P}_P$-completion of a $\sigma$-algebra $\cZ$.

\begin{cor}\label{cor:Qefficiency} Let $\cZ$ be a description of $\bW$. Then 
   $\cZ$ is $\mathcal{P}$-ODS if and only if $\mathcal{Q}_P \subseteq \ol{\cZ}^P$ for all $P \in \mathcal{P}$.  
   A sufficient condition for $\cZ$ to be $\mathcal{P}$-ODS is that 
    $\mathcal{Q} \subseteq \ol{\cZ}^P$ for all $P \in \mathcal{P}$, in which case
    \begin{equation}\label{eq:COAuniform}
        \avar(\cZ; P) - \avar(\mathcal{Q}; P) 
        = \sum_{t\in \mathbb{T}} 
            c_t^2 D_t(\mathcal{Q},\cZ; P)\geq 0,
            \qquad P\in \mathcal{P}.
    \end{equation}
    In particular, $\mathcal{Q}$ is a $\mathcal{P}$-ODS description of $\bW$, and \eqref{eq:COAuniform} holds with $\cZ = \sigma(\bW)$.
\end{cor}
\begin{remark}
    It is \textit{a priori} not obvious if $\mathcal{Q}$ is given by a representation, i.e., if  $\mathcal{Q}=\sigma(\phi(\bW))$ for some measurable mapping $\phi$.
    In Example \ref{exm:magnitude} it is, since the arguments can be reformulated to conclude that $\mathcal{Q} = \sigma(\|\bW\|)$.
\end{remark}

Instead of working with the entire conditional distribution, it suffices to work with the conditional mean when assuming, e.g., independent additive noise on the outcome.
\begin{prop}\label{prop:COMSefficiency}
    For fixed $P\in \mathcal{P}$, then $\mathcal{Q}_P = \mathcal{R}_P$ 
    if any of the following statements are true:
    \begin{itemize}
        \item $F(y\given t,\bW;P)$ is $\sigma(b_t(\bW;P))$-measurable for each $t\in \mathbb{T},y\in \real$.

        \item $Y$ is a binary outcome. 

        \item $Y$ has independent additive noise, i.e.,
            $Y = b_T(\bW) + \varepsilon_Y$ with $\varepsilon_Y \indP T,\bW$.
    \end{itemize}
    If $\mathcal{Q}_P = \mathcal{R}_P$ holds,
    then \eqref{eq:COApointwise} also holds for any $P$-OMS description.
\end{prop}

\begin{remark} When $\mathcal{Q}_P = \mathcal{R}_P$
    for all $P \in \mathcal{P}$ we have $\mathcal{Q} = \mathcal{R}$. There is thus the same information 
    in the $\sigma$-algebra $\mathcal{R}$ generated by the conditional means of the outcome as 
    there is in the $\sigma$-algebra $\mathcal{Q}$ generated by the entire conditional distribution of the outcome. 
    The three conditions in Proposition~\ref{prop:COMSefficiency} are sufficient but not necessary
    to ensure this.
\end{remark}


We also have a result analogous to Lemma 4 of \citet{rotnitzky2020efficient}:

\begin{lem}[Supplementation with precision]\label{lem:underadj}
    Fix $P\in\mathcal{P}$ and let $\cZ_1\subseteq \cZ_2$ be descriptions 
    of $\bW$ such that $T\indP \cZ_2 \given \cZ_1$. Then $\cZ_1$ is $P$-valid if and only if $\cZ_2$ is $P$-valid. Irrespectively, it always holds that
    \begin{align*}
        \avar(\cZ_1 ; P)-\avar(\cZ_2 ; P)
        = 
        \mathbf{c}^\top \var_P[\mathbf{R}(\cZ_1,\cZ_2;P)] \mathbf{c} \ge 0
    \end{align*}
    where $\mathbf{R}(\cZ_1,\cZ_2;P) \coloneq (R_t(\cZ_1,\cZ_2;P))_{t\in \mathbb{T}}$ with
    \begin{align*}
        R_t(\cZ_1,\cZ_2;P) 
            &\coloneq \pa{\frac{\one(T=t)}{\pi_t(\cZ_2;P)} - 1}
            \pa{b_t(\cZ_2;P) - b_t(\cZ_1;P)}.
    \end{align*}
    Writing $R_t = R_t(\cZ_1,\cZ_2;P)$, the components of the covariance matrix of $\mathbf{R}$ are given by
    \begin{align*}
        \var_P(R_t) 
            &= \ex_P \left[
                \pa{\frac{1}{\pi_t(\cZ_1;P)} - 1}
                    \var_P[b_t(\cZ_2;P)\given \cZ_1]\right], \\
        \cov_P(R_s,R_t) &= 
            - \ex_P[\cov_P(b_s(\cZ_2;P)b_t(\cZ_2;P)\given \cZ_1)].
    \end{align*}
\end{lem}

As a consequence, we obtain the well-known fact that the propensity score is a valid adjustment if $\bW$ is, cf. Theorems 1--3 in \citet{rosenbaum1983central}.
\begin{cor}
    Let $\Pi_P=\sigma(\pi_t(\bW;P)\colon \, t\in \mathbb{T})$. If $\cZ$ is a description of $\bW$ containing $\Pi_P$, then $\cZ$ is $P$-valid and 
    $$
        \avar(\Pi_P ; P)-\avar(\cZ ; P)
        = \mathbf{c}^\top \var_P[\mathbf{R}(\mathcal{R}_P ,\cZ;P)] \mathbf{c} 
        \ge 0.
    $$
\end{cor}
The corollary asserts that while the information contained in the propensity score is valid, it is asymptotically inefficient to adjust for in contrast to all the information of $\bW$. This is in similar spirit to Theorem 2 of \citet{hahn1998role}, which states that the efficiency bound $\avar(\bW; P)$ remains unaltered if the propensity is considered as known. However, the corollary also quantifies the difference of the asymptotic efficiencies.


Corollary \ref{cor:Qefficiency} asserts that $\mathcal{Q}$ is maximally efficient over all $\mathcal{P}$-ODS descriptions $\mathcal{Z}$ satisfying $\mathcal{Q} \subseteq \ol{\mathcal{Z}}^P$, and Proposition \ref{prop:COMSefficiency} asserts that in special cases, $\mathcal{Q}$ reduces to $\mathcal{R}$. 
Since $\mathcal{R}$ is $\mathcal{P}$-OMS, hence $\mathcal{P}$-valid, it is natural to ask if $\mathcal{R}$ is generally more efficient than $\mathcal{Q}$. The following example shows that their efficiency bounds may be incomparable uniformly over~$\mathcal{P}$. 

\begin{example}\label{ex:symmetric}
    Let $0<\delta<\frac12$ be fixed and let $\mathcal{P}$ be the collection of data generating distributions that satisfy:
    \begin{itemize}
        \item $\bW \in [\delta,1-\delta]$ with a symmetric distribution, i.e., $\bW\stackrel{\mathcal{D}}{=}1-\bW$.
        \item $T\in\{0,1\}$ with $\ex(T\given \bW) = \bW$.
        \item $Y = T + g(|\bW-\frac{1}{2}|) + v(\bW)\varepsilon_Y$, where $\varepsilon_Y\ind{} (T,\bW)$, $\ex[\varepsilon_Y^2]<\infty$, $\ex[\varepsilon_Y] = 0$, and where $g\colon [0,\frac{1}{2}-\delta]\to \real$ and $v\colon [\delta,1-\delta]\to [0,\infty)$ are continuous functions.
    \end{itemize}
    Letting $\bZ = |\bW-\frac{1}{2}|$, it is easy to verify directly from Definition \ref{def:OutcomeAlgebras} that 
    $$
        \mathcal{Q} = \sigma(\bW) \neq \sigma(\bZ) = \mathcal{R}.
    $$
    It follows that $\bZ$ is $\mathcal{P}$-OMS but not $\mathcal{P}$-ODS. However, $\bZ$ is $\mathcal{P}_1$-ODS in the homoscedastic submodel $\mathcal{P}_1= \{P\in \mathcal{P}\given v\equiv 1\}$. In fact, it generates the 
    $\sigma$-algebra $\mathcal{Q}$ within this submodel, i.e., $\sigma(\bZ) = \vee_{P\in \mathcal{P}_1}\mathcal{Q}_P$. Thus $\avar(\bZ; P) \leq \avar(\bW ; P)$ for all $P\in\mathcal{P}_1$. 

    We refer to the supplementary Section \ref{sec:symmetriccomputations} for complete calculations of the subsequent formulas.
    
    From symmetry it follows that $\pi_1(\bZ)=0.5$ and hence we conclude that $T\ind \bZ$. By Lemma \ref{lem:underadj}, it follows that $0$ (the trivial adjustment) is $\mathcal{P}$-valid, but with $\avar(\bZ; P) \leq \avar(0 ; P)$ for all $P\in\mathcal{P}$.
    
    Alternatively, direct computation yields that
    \begin{align*}
        \bV_t(0)
        &= 2\var(g(\bZ)) + 2\, \ex[v(\bW)^2]\ex[\varepsilon_Y^2] \\
        \bV_t(\bZ)
        &= \var\left( g(\bZ)\right)  + 2\,\ex[v(\bW)^2] \ex[\varepsilon_Y^2] \\
        \bV_t(\bW) 
        &= \var\left( g(\bZ)\right) + \ex\left[v(\bW)^2/\bW\right] \ex[\varepsilon_Y^2].
    \end{align*}
    With $\Delta = \chi_t$, the first two equalities confirm that $\avar(\bZ; P) \leq \avar(0 ; P)$, $P\in\mathcal{P}$, and the last two yield that indeed $\avar(\bZ; P) \leq \avar(\bW ; P)$ for $P\in\mathcal{P}_1$ by applying Jensen's inequality.
    In fact, these are strict inequalities whenever $g(\bZ)$ and $\varepsilon_Y$ are non-degenerate. 
    
    Finally, we show that it is possible for $\avar(\bZ; P) > \avar(\bW ; P)$ for $P\notin\mathcal{P}_1$. 
    Let $\tilde{P}\in \mathcal{P}$ be a data generating distribution with
    $\ex_{\tilde{P}}[\varepsilon_Y^2]>0$, $v(\bW)=\bW^2$, and with $\bW$ uniformly distributed on $[\delta,1-\delta]$. Then
    \begin{align*}
        2\,\ex_{\tilde{P}}[v(\bW)^2] 
            &= \,2\ex_{\tilde{P}}[\bW^{4}] 
            = 2\frac{(1-\delta)^{5}-\delta^{5}}{5(1-2\delta)} 
            \xrightarrow{\delta\to 0} \frac{2}{5} \\
        \ex_{\tilde{P}}\left[\frac{v(\bW)^2}{\bW}\right] 
            &= \ex_{\tilde{P}}[\bW^{3}] 
            = \frac{(1-\delta)^{4}-\delta^{4}}{4(1-2\delta)} 
            \xrightarrow{\delta\to 0} \frac{1}{4}.
    \end{align*}
    So for sufficiently small $\delta$, it holds that $\avar(\bZ; \tilde{P}) > \avar(\bW ; \tilde{P})$. The example can also be modified to work for other $\delta>0$ by taking $v$ to be a sufficiently large power of $\bW$.
\end{example}

\begin{remark}
    Following Theorem 1 of \citet{benkeser2020nonparametric}, it is stated that the asymptotic variance $\avar(\mathcal{R}_P; P)$ is generally smaller than $\avar(\bW; P)$. The example demonstrates that this requires some assumptions on the outcome distribution, such as the conditions in Proposition~\ref{prop:COMSefficiency}.
\end{remark}
\begin{remark}\label{rmk:minimalsufficiency}
    Suppose that $\mathcal{P}=\{f_\theta\cdot \mu \colon \theta \in \Theta\}$ is a parametrized family of measures with densities $\{f_\theta\}_{\theta\in\Theta}$ with respect to a $\sigma$-finite measure $\mu$. Then informally, a \emph{sufficient sub-$\sigma$-algebra} is any subset of the observed information for which the remaining information is independent of $\theta\in \Theta$, see  \citet{billingsley2017probability} for a formal definition. 
    Superficially, this concept seems similar to that of a $\mathcal{P}$-ODS description. In contrast however, the latter is a subset of the covariate information $\sigma(\bW)$ rather than the observed information $\sigma(T,\bW,Y)$, and it concerns sufficiency for the outcome distribution rather than the entire data distribution. 
    Moreover, the Rao-Blackwell theorem asserts that conditioning an estimator of $\theta$ on a sufficient sub-$\sigma$-algebra leads to an estimator that is never worse. Example~\ref{ex:symmetric} demonstrates that the situation is more delicate when considering statistical efficiency for adjustment.
\end{remark}

\section{Estimation based on outcome-adapted representations}\label{sec:estimation}
In this section we develop a general estimation method based on the insights of Section~\ref{sec:informationbounds}, and we provide an asymptotic analysis of this methodology under general conditions. Our method modifies the AIPW estimator, which we proceed to discuss in more detail. 

We have worked with the propensity score, $\pi_t$, and the outcome regression, $b_t$, as random variables. We now consider, for each $P\in \mathcal{P}$, their function counterparts obtained as regular conditional expectations:
\begin{align*} 
    m_P&\colon \mathbb{T}\times \mathbb{W} \to \real, 
        &m_P(t\given \bw) \coloneq \mathbb{P}_P(T=t\given \bW=\bw), \\
    g_P&\colon \mathbb{T}\times \mathbb{W} \to \real,
        &g_P(t, \bw)  \coloneq  \ex_P[Y\given T=t,\bW=\bw].
\end{align*}
To target the adjusted mean, $\chi_t$, the AIPW estimator utilizes the influence function -- given by \eqref{eq:IF} -- as a score equation. To be more precise, given estimates $(\widehat{m},\widehat{g})$ of the nuisance functions $(m_P,g_P)$, the AIPW estimator of $\chi_t$ is given by
\begin{equation}\label{eq:AIPW}
    \widehat{\chi}_t^{\mathrm{aipw}}(\widehat{m},\widehat{g}) \coloneq
        \mathbb{P}_n\Big[
            \widehat{g}(t,\bW) + \frac{\one(T=t)(Y-\widehat{g}(t,\bW))}{\widehat{m}(t\given\bW)}
        \Big],
\end{equation}
where $\mathbb{P}_n[\cdot]$ is defined as the empirical mean over $n$ i.i.d. observations from $P$. The natural AIPW estimator of $\Delta$ given by $\eqref{eq:hatDelta}$ is then $\widehat{\Delta}^{\mathrm{aipw}}(\widehat{m},\widehat{g}) = \sum_{t\in \mathbb{T}} c_t \widehat{\chi}_t^{\mathrm{aipw}}(\widehat{m},\widehat{g})$.

Roughly speaking, the AIPW estimator $\widehat{\chi}_t^{\mathrm{aipw}}$ converges to $\chi_t$ with a $\sqrt{n}$-rate if
\begin{equation}\label{eq:productbias}
    n\cdot \mathbb{P}_n\big[(\widehat{m}(t\given \bW)-m_P(t\given \bW))^2\big]
        \cdot \mathbb{P}_n \big[(\widehat{g}(t, \bW)-g_P(t, \bW))^2 \big]
        \longrightarrow 0.
\end{equation}
If $(\widehat{m},\widehat{g})$ are estimated using the same data as $\mathbb{P}_n$, 
then the above asymptotic result relies on Donsker class conditions on the spaces containing $\widehat{m}$ and $\widehat{g}$. Among others, \citet{chernozhukov2018} propose circumventing Donsker class conditions by sample splitting techniques such as $K$-fold cross-fitting.

\subsection{Representation adjustment and the DOPE}
Suppose that the outcome regression function factors through an intermediate representation:
\begin{align} \label{eq:factorizedmodel}
    g_P(t,\bw) = h_P(t, \phi(\theta_P,\bw)), \qquad t\in \mathbb{T}, \bw\in\mathbb{W},
\end{align}
where $\phi \colon \Theta \times \mathbb{W} \to \real^d$ is a known measurable mapping with unknown parameter $\theta\in \Theta$, and where $h_P\colon \mathbb{T}\times \real^d \to \real$ is an unknown function. 
We use $\bZ_\theta = \phi(\theta,\bW)$ to denote the corresponding representation parametrized by $\theta\in \Theta$. If a particular covariate value $\bw\in \mathbb{W}$ is clear from the context, we also use the implicit notation $\bz_{\theta}= \phi(\theta,\bw)$.

\begin{example}[Single-index model]\label{ex:singleindex}
    The \textit{(partial) single-index model} applies to $\bW \in \real^k$ 
    and assumes that \eqref{eq:factorizedmodel} holds with $\phi(\theta,\bw)=\bw^\top\theta$, where $\theta \in \Theta \subseteq \real^k$.
    In other words, it assumes that the outcome regression factors through the \textit{linear predictor} $\bZ_{\theta} = \bW^\top \theta$ such that 
    \begin{align} \label{eq:singleindex}
    \forall P\in \mathcal{P}:\qquad
    Y = h_P(T, \bZ_{\theta_P}) + \varepsilon_Y, 
         \qquad \ex_P[\varepsilon_Y \given T, \bW] = 0.
    \end{align}
    For each treatment $T=t$, the model extends the generalized linear model (GLM) by assuming that the (inverse) link function $h_P(t, \cdot)$ is unknown. 
     
    The \textit{semiparametric least squares} (SLS) estimator proposed by \citet{ichimura1993semiparametric} estimates $\theta_P$ and $h_P$ by performing a two-step regression procedure. Alternative estimation procedures are given in \citet{powell1989semiparametric,delecroix2003efficient}.
\end{example}

Given an estimator $\hat{\theta}$ of $\theta_P$, independent of $(T, \bW, Y)$, we use the following notation for various functions and estimators thereof:
\begin{subequations}
\label{eq:linknotation}
\begin{align} 
    g_{\hat{\theta}}(t,\bw)
        &\coloneq \ex_P[Y\given T=t,\bZ_{\hat{\theta}}=\bz_{\hat{\theta}},\hat{\theta}\,], \nonumber \\
    m_{\hat{\theta}}(t\given \bw) 
        &\coloneq \mathbb{P}_P(T=t \given \bZ_{\hat{\theta}}=\bz_{\hat{\theta}},\hat{\theta}\,), \nonumber \\
    \widehat{g}(t,\bw) &\colon 
        \text{an estimator of } g_P(t,\bw),
        \nonumber \\
    \widehat{g}_{\hat{\theta}}(t,\bw) 
        &\colon \text{an estimator of } g_{\hat{\theta}}(t,\bw) \text{ of the form } 
        \widehat{h}(t,\bz_{\hat{\theta}}),  \\
    \widehat{m}_{\hat{\theta}}(t\given \bw) 
        &\colon \text{an estimator of } m_{\hat{\theta}}(t\given \bw) \text{ of the form } 
        \widehat{f}(t, \bz_{\hat{\theta}}).
\end{align}
\end{subequations}
In other words, $g_{\hat{\theta}}$ and $m_{\hat{\theta}}$ are the theoretical propensity score and outcome regression, respectively, when using the \emph{estimated} representation $\bZ_{\hat{\theta}}= \phi(\hat{\theta}, \bW)$. Note that we have suppressed $P$ from the notation on the left-hand side, as we will be working under a fixed $P$ in this section. 

Sufficient conditions are well known that ensure 
\begin{align*}
    \sqrt{n}(\widehat{\chi}_t^{\mathrm{aipw}}(\widehat{m},\widehat{g}) - \chi_t) 
        \xrightarrow{d} \mathrm{N}(0,\bV_t(\bW; P)).
\end{align*}
Such conditions, e.g., Assumption 5.1 in \citet{chernozhukov2018}, primarily entail the condition in \eqref{eq:productbias}. We will leverage \eqref{eq:factorizedmodel} to derive more efficient estimators
under similar conditions.

Suppose for a moment that $\theta_P$ is known. Since $\bZ_{\theta_P}$ is $P$-OMS under \eqref{eq:singleindex}, it holds that $\chi_t(\bZ_{\theta_P};P) = \chi_t$. Then under analogous conditions for the estimators $(\widehat{m}_{\theta_P},\widehat{g}_{\theta_P})$, we therefore have
\begin{align*}
    \sqrt{n}(\widehat{\chi}_t^{\mathrm{aipw}}(\widehat{m}_{\theta_P},\widehat{g}_{\theta_P}) - \chi_t)
        \xrightarrow{d} \mathrm{N}(0,\bV_t(\bZ_{\theta_P};P)).
\end{align*}
Under the conditions of Proposition \ref{prop:COMSefficiency}, it holds that $\bV_t(\bZ_{\theta_P};P)\leq \bV_t(\bW; P)$. In other words,
$\widehat{\chi}_t^{\mathrm{aipw}}(\widehat{m}_{\theta_P},\widehat{g}_{\theta_P})$ is asymptotically at least as efficient as $\widehat{\chi}_t^{\mathrm{aipw}}(\widehat{m},\widehat{g})$.

In general, the parameter $\theta_P$ is, of course, unknown, and we therefore consider adjusting for an estimated representation $\bZ_{\hat{\theta}}$. For simplicity, we present the special case $\Delta = \chi_t$, but the results are easily extended to general contrasts $\Delta$.
Our generic estimation procedure is described in Algorithm \ref{alg:generalalg}, where 
\begin{equation*}
    \bZ_{\hat{\theta},i} \coloneq \phi(\hat{\theta},\bW_i)
\end{equation*}
denotes the estimated representation of the $i$-th observed covariate. We refer to the resulting estimator as the \textit{Debiased Outcome-adapted Propensity Estimator} (DOPE), and it is denoted by $\widehat{\chi}_t^{\mathrm{dope}}$.

\begin{algorithm} \caption{Debiased Outcome-adapted Propensity Estimator} \label{alg:generalalg}
  \textbf{input}: observations $(T_i,\bW_i,Y_i)_{i\in[n]}$, index sets $\mathcal{I}_1,\mathcal{I}_2,\mathcal{I}_3 \subseteq [n]$\;
  \textbf{options}: method for computing $\hat{\theta}$,
  regression methods for propensity score and outcome regression of the form \eqref{eq:linknotation}\;
  \Begin{
    compute estimate $\hat{\theta}$ based on data $(T_i,\bW_i,Y_i)_{i\in\mathcal{I}_1}$\;
  
    regress outcomes $(Y_i)_{i\in \mathcal{I}_2}$ onto $(T_i,\bZ_{\hat{\theta},i})_{i\in\mathcal{I}_2}$ to obtain $\widehat{g}_{\hat{\theta}}(\cdot,\cdot)$\;
    
    regress treatments $(T_i)_{i\in\mathcal{I}_2}$ onto 
    $(\bZ_{\hat{\theta},i})_{i\in \mathcal{I}_2}$ to obtain $\widehat{m}_{\hat{\theta}}(\cdot\given \cdot)$\;

    compute AIPW based on data $(T_i,\bW_i,Y_i)_{i\in\mathcal{I}_3}$ and nuisance estimates $(\widehat{m}_{\hat{\theta}}, \widehat{g}_{\hat{\theta}})$:
    \begin{align*}
    \widehat{\chi}_t^{\mathrm{dope}}
    =
        \frac{1}{|\mathcal{I}_3|}\sum_{i\in\mathcal{I}_3}\left(
            \widehat{g}_{\hat{\theta}}(t,\bW_i) + \frac{\one(T_i=t)(Y_i-\widehat{g}_{\hat{\theta}}(t,\bW_i))}{\widehat{m}_{\hat{\theta}}(t\given\bW_i)}
        \right)
    \end{align*}
  }
  \Return{\textnormal{DOPE:} $\widehat{\chi}_t^{\mathrm{dope}}$.}
\end{algorithm}

Algorithm \ref{alg:generalalg} is formulated such that $\mathcal{I}_1$, $\mathcal{I}_2$, and $\mathcal{I}_3$ can be arbitrary subsets of $[n]$, and for the asymptotic theory we assume that they are disjoint. However, in practical applications it might be reasonable to use the full sample for every estimation step, i.e., employing the algorithm with $\mathcal{I}_1=\mathcal{I}_2=\mathcal{I}_3=[n]$. In this case, we also imagine that line 4 and line 5 are run simultaneously given that $\hat{\theta}$ may be derived from an outcome regression, cf. the SLS estimator in the single-index model (Example \ref{ex:singleindex}). 
In the supplementary Section~\ref{sup:crossfitting}, 
we also describe a more advanced cross-fitting scheme for Algorithm~\ref{alg:generalalg} 

\begin{remark}\label{rmk:benkeser}
\citet{benkeser2020nonparametric} use a similar idea as the DOPE, but their propensity factors through the final outcome regression function instead of a general intermediate representation. 
That our general formulation of Algorithm~\ref{alg:generalalg} contains their collaborative one-step estimator as a special case is seen as follows.
Suppose $T\in\{0,1\}$ is binary and let
$$
    \Theta = \{\theta = (g_P(0,\cdot),g_P(1,\cdot)) \colon P \in \mathcal{P} \}
$$ 
be the collection of outcome regression functions. Define the intermediate representation as the canonical pairing $\phi(\theta,\bW) = (\theta_1(\bW),\theta_2(\bW))$. 
Then, for any outcome regression estimate $\hat{\theta}$, 
we may set $\hat{g}_{\hat{\theta}}(t,\bw) = (1-t)\hat{\theta}_1(\bw) + t\hat{\theta}_2(\bw)$, and the propensity score $\hat{m}_{\hat{\theta}}(t\mid \bw)$ factors through the outcome regressions. 
In this case and with $\mathcal{I}_1=\mathcal{I}_2=\mathcal{I}_3$, Algorithm~\ref{alg:generalalg} yields the collaborative one-step estimator of \citet[Appendix D]{benkeser2020nonparametric}. Accordingly, we refer to this special case as DOPE-BCL (Benkeser, Cai and van der Laan).
\end{remark}

\subsection{Asymptotics of the DOPE}
We proceed to discuss the asymptotics of the DOPE. For our theoretical analysis, we assume that the index sets in Algorithm \ref{alg:generalalg} are disjoint. That is, the theoretical analysis relies on sample splitting.
\begin{asm}\label{asm:samples}
    The observations $(T_i,\bW_i,Y_i)_{i\in[n]}$ used to compute $\widehat{\chi}_t^{\mathrm{dope}}$ are i.i.d. with the same distribution as $(T,\bW,Y)$, and $[n]=\mathcal{I}_1\cup\mathcal{I}_2\cup\mathcal{I}_3$ is a partition such that $|\mathcal{I}_3| \to \infty$ as $n\to \infty$.
\end{asm}
In our simulations, employing sample splitting did not seem to enhance performance, and hence we regard Assumption~\ref{asm:samples} as a theoretical convenience rather than a practical necessity in all cases. Our results can likely also be established under alternative assumptions that avoid sample splitting, in particular Donsker class conditions.

We also henceforth use the convention that each of the quantities $\pi_t,b_t,\chi_t$, and $\bV_t$ are defined conditionally on $\hat{\theta}$, e.g.,
$$
    \chi_t(\bZ_{\hat{\theta}})
    = \ex[b_t(\bZ_{\hat{\theta}})\given \hat{\theta} \,]
    = \ex[\ex[Y\given T=t,\bZ_{\hat{\theta}},\hat{\theta}\,]\given \hat{\theta}\,].
$$
The error of the DOPE estimator can then be decomposed as 
\begin{equation}\label{eq:biasvardecomp}
    \widehat{\chi}_t^{\mathrm{dope}} - \chi_t
    = (\widehat{\chi}_t^{\mathrm{dope}} - \chi_t(\bZ_{\hat{\theta}}))
        + (\chi_t(\bZ_{\hat{\theta}}) - \chi_t).
\end{equation}
The first term is the error had our target been the adjusted mean when adjusting for the estimated representation $\bZ_{\hat{\theta}}$, whereas the second term is the adjustment bias that arises from adjusting for $\bZ_{\hat{\theta}}$ rather than $\bW$ (or $\bZ_{\theta}$). 

\subsubsection{Estimation error conditionally on representation}
To describe the asymptotics of the first term of \eqref{eq:biasvardecomp}, we consider the decomposition
\begin{align}\label{eq:DOPEdecomp}
    \sqrt{|\mathcal{I}_3|}\widehat{\chi}_t^{\mathrm{dope}}
    &= U_{\hat{\theta}}^{(n)} + R_1 + R_2 + R_3,
\end{align}
where
\begin{align*}
    U_{\hat{\theta}}^{(n)} &\coloneq \frac{1}{\sqrt{|\mathcal{I}_3|}}\sum_{i\in \mathcal{I}_3} u_i(\hat \theta), 
    &u_i(\hat \theta) \coloneq g_{\hat \theta}(t,\bW_i) + \frac{\one(T_i=t)(Y_i-g_{\hat \theta}(t,\bW_i))}{m_{\hat \theta}(t\given \bW_i)},\\
    R_1 &\coloneq \frac{1}{\sqrt{|\mathcal{I}_3|}}\sum_{i\in \mathcal{I}_3} r_i^1,
    &r_i^1 \coloneq 
        \left(\widehat{g}_{\hat{\theta}}(t,\bW_i)-g_{\hat{\theta}}(t,\bW_i)\right)
        \left(1-\frac{\one(T_i=t)}{m_{\hat{\theta}}(t\given \bW_i)}\right), \\
    R_2 &\coloneq  \frac{1}{\sqrt{|\mathcal{I}_3|}}\sum_{i\in \mathcal{I}_3} r_i^2,
    &r_i^2 \coloneq  (Y_i-g_{\hat{\theta}}(t,\bW_i))\left(
            \frac{\one(T_i=t)}{\widehat{m}_{\hat{\theta}}(t\given \bW_i)}
            -\frac{\one(T_i=t)}{m_{\hat{\theta}}(t\given \bW_i)}\right), \\
    R_3 &\coloneq  \frac{1}{\sqrt{|\mathcal{I}_3|}}\sum_{i\in \mathcal{I}_3} r_i^3,
    &r_i^3 \coloneq  
        \big(\widehat{g}_{\hat{\theta}}(t,\bW_i)
            -g_{\hat{\theta}}(t,\bW_i)\big)\left(
        \frac{\one(T_i=t)}{\widehat{m}_{\hat{\theta}}(t\given \bW_i)}
            -\frac{\one(T_i=t)}{m_{\hat{\theta}}(t\given \bW_i)}\right).
\end{align*}

We show that the oracle term, $U_{\hat{\theta}}^{(n)}$, drives the asymptotic limit, and that the terms $R_1,R_2,R_3$ are remainder terms, subject to the conditions stated below. 

\begin{asm}\label{asm:AIPWconv}
    For $(T,\bW,Y)\sim P$ satisfying the representation model~\eqref{eq:factorizedmodel}, it holds that:
    \begin{enumerate}
        \item[(i)] There exists $c>0$ such that $\max\{|\widehat{m}_{\hat{\theta}}-\frac{1}{2}|,|m_{\hat{\theta}}-\frac{1}{2}|\}\leq \frac{1}{2}-c$.
        \label{asm:strictoverlap}
        
        \item[(ii)] There exists $C>0$ such that 
        $\ex[Y^2\given \bW, T] \leq C$.
        \label{asm:CVarbound}

        \item[(iii)] There exists $\delta>0$ such that $\ex\big[|Y|^{2+\delta}\big]<\infty$.
        \label{asm:varlowerbound}
        
        \item[(iv)] It holds that \(
                \mathcal{E}_{1,t}^{(n)} \coloneq \frac{1}{|\mathcal{I}_3|} \sum_{i\in \mathcal{I}_3} (\widehat{m}_{\hat{\theta}}(t\given \bW_i) - m_{\hat{\theta}}(t\given \bW_i))^2 \xrightarrow{P} 0
            \).
        \label{asm:propensityconsistency}
        
        \item[(v)] It holds that \(
                \mathcal{E}_{2,t}^{(n)} \coloneq\frac{1}{|\mathcal{I}_3|} \sum_{i\in \mathcal{I}_3} (\widehat{g}_{\hat{\theta}}(t,\bW_i) - g_{\hat{\theta}}(t,\bW_i))^2 \xrightarrow{P} 0
            \).
        \label{asm:ORconsistency}
        
        \item[(vi)] It holds that $|\mathcal{I}_3|\cdot\mathcal{E}_{1,t}^{(n)} \mathcal{E}_{2,t}^{(n)} 
                \xrightarrow{P} 0$.
        \label{asm:producterror}
    \end{enumerate}    
\end{asm}

Classical convergence results of the AIPW are proven under similar conditions, but with $(vi)$ replaced by the stronger convergence in \eqref{eq:productbias}.
We establish conditional asymptotic results under conditions on the conditional errors $\mathcal{E}_{1,t}^{(n)}$ and $\mathcal{E}_{2,t}^{(n)}$. 
To the best of our knowledge, the most similar results that we are aware of are those of \citet{benkeser2020nonparametric}, and our proof techniques are most similar to those of \citet{chernozhukov2018,lundborg2023perturbation}.

We can now state our first asymptotic result for the DOPE. 

\begin{thm} \label{thm:conditionalAIPWconv}
    Under Assumptions \ref{asm:samples} and \ref{asm:AIPWconv}, it holds that
    \begin{align*}
        \mathbb{V}_t(\bZ_{\hat{\theta}})^{-1/2}
        \cdot
        U_{\hat{\theta}}^{(n)} \xrightarrow{d} \mathrm{N}(0,1),
        \qquad \text{and} \qquad 
        R_i \xrightarrow{P} 0, \quad i=1,2,3,
    \end{align*}
    as $n\to \infty$. As a consequence,
    \begin{align*}
        \sqrt{|\mathcal{I}_3|} \cdot \bV_t(\bZ_{\hat{\theta}})^{-1/2}\left(
        \widehat{\chi}_t^{\mathrm{dope}} - \chi_t(\bZ_{\hat{\theta}})
        \right) \xrightarrow{d} \mathrm{N}(0, 1).
    \end{align*}
\end{thm}

In other words, we can expect the DOPE to have an asymptotic distribution, conditionally on $\hat{\theta}$, approximated by 
$$
    \widehat{\chi}_t^{\mathrm{dope}}(\widehat{g}_{\hat{\theta}},\widehat{m}_{\hat{\theta}}) \given \hat{\theta}
    \quad \stackrel{as.}{\sim} \quad 
    \mathrm{N}\Big(\chi_t(\bZ_{\hat{\theta}}), \, 
    \frac{1}{|\mathcal{I}_3|}\mathbb{V}_t(\bZ_{\hat{\theta}})\Big).
$$
Note that if $|\mathcal{I}_3| = \lfloor n/3 \rfloor$, say, then the asymptotic variance is $\frac{3}{n}\mathbb{V}_t(\bZ_{\hat{\theta}})$. Our simulation study indicates that the asymptotic approximation may be valid in some cases without the use of sample splitting, in which case $|\mathcal{I}_3|=n$ and the asymptotic variance is $\frac{1}{n}\mathbb{V}_t(\bZ_{\hat{\theta}})$. A direct implementation of the sample splitting procedure thus comes with an efficiency cost. In the supplementary Section~\ref{sup:crossfitting} we discuss how the cross-fitting procedure makes use of sample splitting without an efficiency cost.

Given nuisance estimates $(\widehat{g},\widehat{m})$
we consider the empirical variance estimator given by
\begin{align}\label{eq:varestimator}
    \widehat{\mathcal{V}}_t(\widehat{g},\widehat{m}) 
    = \frac{1}{|\mathcal{I}_3|} \sum_{i\in \mathcal{I}_3} \widehat{u}_i^2 
        - \Big(\frac{1}{|\mathcal{I}_3|} \sum_{i\in \mathcal{I}_3}
        \widehat{u}_{i}\Big)^2 
    = \frac{1}{|\mathcal{I}_3|} \sum_{i\in \mathcal{I}_3}
        \Big(\widehat{u}_i - \frac{1}{|\mathcal{I}_3|} \sum_{j\in \mathcal{I}_3}\widehat{u}_{j}\Big)^2, 
\end{align}
where
\begin{align*}
    \widehat{u}_i =\widehat{u}_i(\widehat{g},\widehat{m})  
        = \widehat{g}(t,\bW_i) + \frac{\one(T=t)(Y_i-\widehat{g}(t,\bW_i))}{\widehat{m}(t\given \bW_i)}.
\end{align*}

The following theorem states that the variance estimator with nuisance functions $(\widehat{g}_{\hat{\theta}},\widehat{m}_{\hat{\theta}})$ is consistent for the asymptotic variance in Theorem~\ref{thm:conditionalAIPWconv}. 
\begin{thm}\label{thm:varconsistent}
    Under Assumptions \ref{asm:samples} and \ref{asm:AIPWconv}, it holds that
    \begin{align*}
    \widehat{\mathcal{V}}_t(\widehat{g}_{\hat{\theta}},
    \widehat{m}_{\hat{\theta}})
    -\mathbb{V}_t(\bZ_{\hat{\theta}}) 
    \xrightarrow{P} 0
    \end{align*}
    as $n\to \infty$.
\end{thm}

\subsubsection{Asymptotics of representation induced error}
We now turn to the discussion of the second term in \eqref{eq:biasvardecomp}, i.e., the difference between the adjusted mean for the estimated representation $\bZ_{\hat{\theta}}$ and the adjusted mean for the full covariate $\bW$. Under sufficient regularity, the delta method \citep[Thm. 3.8]{van2000asymptotic} describes the distribution of this error:
\begin{prop}\label{prop:deltamethod}
    Assume $(T,\bW,Y)\sim P$ satisfies the model \eqref{eq:factorizedmodel} with $\theta_P \in \Theta\subseteq \real^p$. Let $u\colon \Theta \to \real$ be the function given by $u(\theta) = \chi_t(\phi(\theta,\bW);P)$ and assume that $u$ is differentiable in $\theta_P$. Suppose that $\hat{\theta}$ is an estimator of $\theta_P$ with rate $r_n$ such that
    \begin{align*}
        r_n \cdot (\hat{\theta} -\theta_P) 
        \xrightarrow{d} \mathrm{N}(0,\Sigma).
    \end{align*}
    Then
    \begin{align*}
        r_n \cdot (\chi_t(\bZ_{\hat{\theta}}) - \chi_t)
        \xrightarrow{d} \mathrm{N}(0, \nabla u(\theta_P)^\top \Sigma \nabla u(\theta_P))
    \end{align*}
    as $n \to \infty$.
\end{prop}

The delta method requires that the adjusted mean, $\chi_t = \chi_t(\phi(\theta,\bW);P)$, is differentiable with respect to $\theta\in\Theta$. The theorem below showcases that this is the case for the single-index model in Example~\ref{ex:singleindex}.
\begin{thm}\label{thm:SIregularity}
    Let $(T,\bW,Y)\sim P$ be given by the single-index model in Example~\ref{ex:singleindex} with $h_t(\cdot)\coloneq h(t,\cdot)\in C^1(\mathbb{R})$. Assume that $\bW$ has a distribution with density $p_{\bW}$ with respect to Lebesgue measure on $\real^d$ and that $p_{\bW}$ is continuous almost everywhere with bounded support. Assume also that the propensity $m(t\mid \bw) = \mathbb{P}(T=t\given \bW=\bw)$ is continuous in $\bw$.
    
    Then $u\colon \real^d \to \real$, defined by $u(\theta) = \chi_t(\bW^\top \theta; P)$, is differentiable at $\theta=\theta_P$ with
    \begin{align*}
    \nabla u (\theta_P)=
    \ex_P\Big[
     h_t'(\bW^\top \theta_P)
    \Big(
    1
    -
    \frac{\mathbb{P}(T=t\given \bW)}{\mathbb{P}(T=t\given \bW^\top \theta_P)}
    \Big)\bW\Big].
    \end{align*}
\end{thm}
The theorem is stated with some restrictive assumptions that simplify the proof, but these are likely not necessary. 
It should also be possible to extend this result to more general models of the form \eqref{eq:factorizedmodel}, but we leave such generalizations for future work. In fact, the proof technique of Theorem~\ref{thm:SIregularity} has already found application in \citet{gnecco2023boosted} in a different context.

The convergences implied in Theorem \ref{thm:conditionalAIPWconv} and Proposition \ref{prop:deltamethod}, with $r_n=\sqrt{|\mathcal{I}_3|}$, suggest that the MSE of the DOPE is of order
\begin{align}\label{eq:asymptoticMSE}
    \ex[(\widehat{\chi}_t^{\mathrm{dope}} - \chi_t)^2]
    &= \ex[(\widehat{\chi}_t^{\mathrm{dope}} - \chi_t(\bZ_{\hat{\theta}}))^2]
        + \ex[(\chi_t(\bZ_{\hat{\theta}}) - \chi_t)^2] \nonumber \\
        &\qquad + 2\,\ex[(\widehat{\chi}_t^{\mathrm{dope}} - \chi_t(\bZ_{\hat{\theta}}))(\chi_t(\bZ_{\hat{\theta}}) - \chi_t)] \nonumber \\
    &\approx \frac{1}{|\mathcal{I}_3|}\Big(\ex[\bV_t(\bZ_{\hat{\theta}})] + 
        \nabla u(\theta_P)^\top \Sigma \nabla u(\theta_P)
        \Big)
        + o\left(|\mathcal{I}_3|^{-1}\right)
\end{align}
The informal approximation `$\approx$' can be turned into an equality by establishing (or simply assuming) uniform integrability, which enables the distributional convergences to be lifted to convergences of moments.

\begin{remark}\label{rmk:confidenceinterval}
    The expression in \eqref{eq:asymptoticMSE} suggests an approximate confidence interval of the form
    \begin{equation*}
        \widehat{\chi}_t^{\mathrm{dope}} \pm \frac{z_{1-\alpha}}{\sqrt{|\mathcal{I}_3|}} (\widehat{\mathcal{V}}_t + \widehat{\nabla}^\top \widehat{\Sigma} \widehat{\nabla}),
    \end{equation*}
     where $\widehat{\Sigma}$ is a consistent estimator of the asymptotic variance of $\sqrt{|\mathcal{I}_3|}(\hat{\theta}-\theta_P)$, and where $\widehat{\nabla}$ is a consistent estimator of $\nabla_\theta \chi_t(\phi(\theta, \bW))\vert_{\theta=\theta_P}$. However, the requirement of constructing both $\widehat{\Sigma}$ and $\widehat{\nabla}$ adds further to the complexity of the methodology, so it might be preferable to rely on bootstrapping techniques in practice. In the simulation study we return to the question of inference and examine the coverage of a naive interval that does not include the term $\widehat{\nabla}^\top \widehat{\Sigma} \widehat{\nabla}$.
\end{remark}

\section{Experiments}\label{sec:experiments}
In this section, we showcase the performance of the DOPE on simulated and real data.
All code is publicly available on \href{https://github.com/AlexanderChristgau/OutcomeAdaptedAdjustment}{GitHub}\footnote{
\url{https://github.com/AlexanderChristgau/OutcomeAdaptedAdjustment}}.

\subsection{Simulation study}\label{sec:simulation}
We present results of a simulation study based on the single-index model from Example~\ref{ex:singleindex}. We demonstrate the performance of the DOPE from
Algorithms \ref{alg:generalalg} and \ref{alg:crossfitted}, and compare with various alternative estimators.

\subsubsection{Sampling scheme}
We simulated datasets consisting of $n$ i.i.d. copies of $(T,\bW,Y)$ sampled according to the following scheme:
\begin{align}\label{eq:samplescheme}
    \bW = (W^1,\ldots,W^d)
        &\sim \mathrm{Unif}(0,1)^{\otimes d} \nonumber \\
    T \given \bW
        &\sim \mathrm{Bern}( 0.01 + 0.98 \cdot \one(W^1 > 0.5)) \nonumber\\
    Y \given T,\bW, \beta_Y
        &\sim \mathrm{N}(h(T, \bW^\top \beta_Y ),1)
\end{align}
where $\beta_Y$ is sampled once for each dataset with
\begin{align*}
    \beta_Y 
        &= (1,\tilde{\beta}_Y ), \qquad \tilde{\beta}_Y \sim \mathrm{N}(0,\mathbf{I}_{d-1}),
\end{align*}
and where $n, d$, and $h$ are experimental parameters.
The settings considered were $n\in \{300,900,2700\}$, $d\in\{4,12,36\}$, and with $h$ being one of
\begin{align}\label{eq:simulationlinks}
    \begin{array}{ll}
       h_{\text{lin}}(t,z) = t+3z,   & \quad
            h_{\text{square}} (t,z) = z^{1+t}, \\
       h_{\text{cbrt}} (t,z) = (2+t)z^{1/3},  & \quad 
            h_{\text{sin}} (t,z) = (3+t)\sin(\pi z).
    \end{array}
\end{align}
For each setting, $N=300$ datasets were simulated.

Note that while $\ex[T] = 0.01 + 0.98 \cdot \mathbb{P}(W^1 > 0.5) = 0.5$, the propensity score $m(t\mid \bw) = 0.01 + 0.98 \cdot \one(w^1 > 0.5)$ takes the rather extreme values $\{0.01,0.99\}$. Even though it technically satisfies (strict) positivity, these extreme values of the propensity makes the adjustment for $\bW$ a challenging task. For each dataset, the adjusted mean $\chi_1$ (conditional on $\beta_Y$) was considered as the target parameter, and the ground truth was numerically computed as the sample mean of $10^7$ observations of $h(1,\bW^\top \beta_Y)$.

\subsubsection{Simulation estimators}\label{sec:simestimators} This section contains an overview of the estimators used in the simulation. For a complete description see Section \ref{sup:simdetails} in the supplementary material.

Two settings were considered for outcome regression (OR):
\begin{itemize}
    \item \textbf{Linear}: Ordinary Least Squares (OLS). 

    \item \textbf{Neural network}: A feedforward neural network with two hidden layers: a linear layer with one neuron, followed by a fully connected ReLU-layer with 100 neurons.
    The first layer is a linear bottleneck that enforces the single-index model, and we denote the weights by $\theta\in \real^d$. An illustration of the architecture can be found in the supplementary Section~\ref{sup:simdetails}.
    For a further discussion of leaning single-and multiple-index models with neural networks, see \citet{parkinson2023linear} and references therein.        
\end{itemize}
For propensity score estimation, logistic regression was used across all settings. A ReLU-network with one hidden layer with 100 neurons was also considered for estimation of the propensity score, but it did not enhance the performance of the resulting ATE estimators in this setting. Random forests and other methods were also initially used for both outcome regression and propensity score estimation. They were subsequently excluded because they did not seem to yield any noteworthy insights beyond what was observed for the methods above.

For each outcome regression, two implementations were explored: a stratified regression where $Y$ is regressed onto $\bW$ separately for each stratum $T=1$ and $T=2$, and a joint regression where $Y$ is regressed onto $(T,\bW)$ simultaneously. In the case of joint regression, the neural network architecture represents the regression function as a single index model, given by $Y = h(\alpha T + \theta^\top \bW) + \varepsilon_Y$. This representation differs from the description of the regression function specified in the sample scheme~\eqref{eq:samplescheme}, which allows for a more complex interaction between treatment and the index of the covariates. In this sense, the joint regression is misspecified. 

Based on these methods for nuisance estimation, we considered the following estimators of the adjusted mean: 
\begin{itemize}
    \item The regression estimator 
    \(
    \widehat{\chi}_1^{\mathrm{reg}}
        \coloneq \mathbb{P}_n[\widehat{g}(1,\bW)] 
    \).
    \item The AIPW estimator given in Equation \eqref{eq:AIPW}.
    
    \item The DOPE from Algorithm~\ref{alg:generalalg}, 
    where $\bZ_{\hat{\theta}} = \bW^\top\hat{\theta}$ and where    
    $\hat{\theta}$ contains the weights of the first layer of the neural network designed for single-index regression.
    We refer to this estimator as the DOPE-IDX.

    \item The DOPE-BCL described in Remark~\ref{rmk:benkeser}, where the propensity score is based on the final outcome regression. 
    See also \citet[App. D]{benkeser2020nonparametric}.
\end{itemize}
We considered two versions of each DOPE estimator: one without sample splitting and another using 4-fold cross-fitting\footnote{consistent with Rem. 3.1. in \citet{chernozhukov2018}, which recommends using 4 or 5 folds for cross-fitting.}. For the latter, the final empirical mean is calculated based on the hold-out fold, while the nuisance parameters are fitted using the remaining three folds. For each fold $k=1,\ldots,4$, this means employing Algorithm~\ref{alg:generalalg} with $\mathcal{I}_1=\mathcal{I}_2 = [n]\setminus J_k$ and $\mathcal{I}_3 = J_k$. The two versions showed comparable performance for larger samples. In this section we only present the results for the DOPE without sample splitting, which performed better overall in our simulations. However, a comparison with the cross-fitted version can be found in Section \ref{sup:simdetails} in the supplement.

Numerical results for the IPW estimator 
\( \widehat{\chi}_1^{\mathrm{ipw}}
    \coloneq \mathbb{P}_n[\one(T=t)Y / \widehat{m}(1\given \bW)] \) 
were also gathered. It performed poorly in our settings, and we have omitted the results for the sake of clarity in the presentation.

\subsubsection{Empirical performance of estimators}
\begin{figure}
    \centering
    \includegraphics[width=\linewidth]{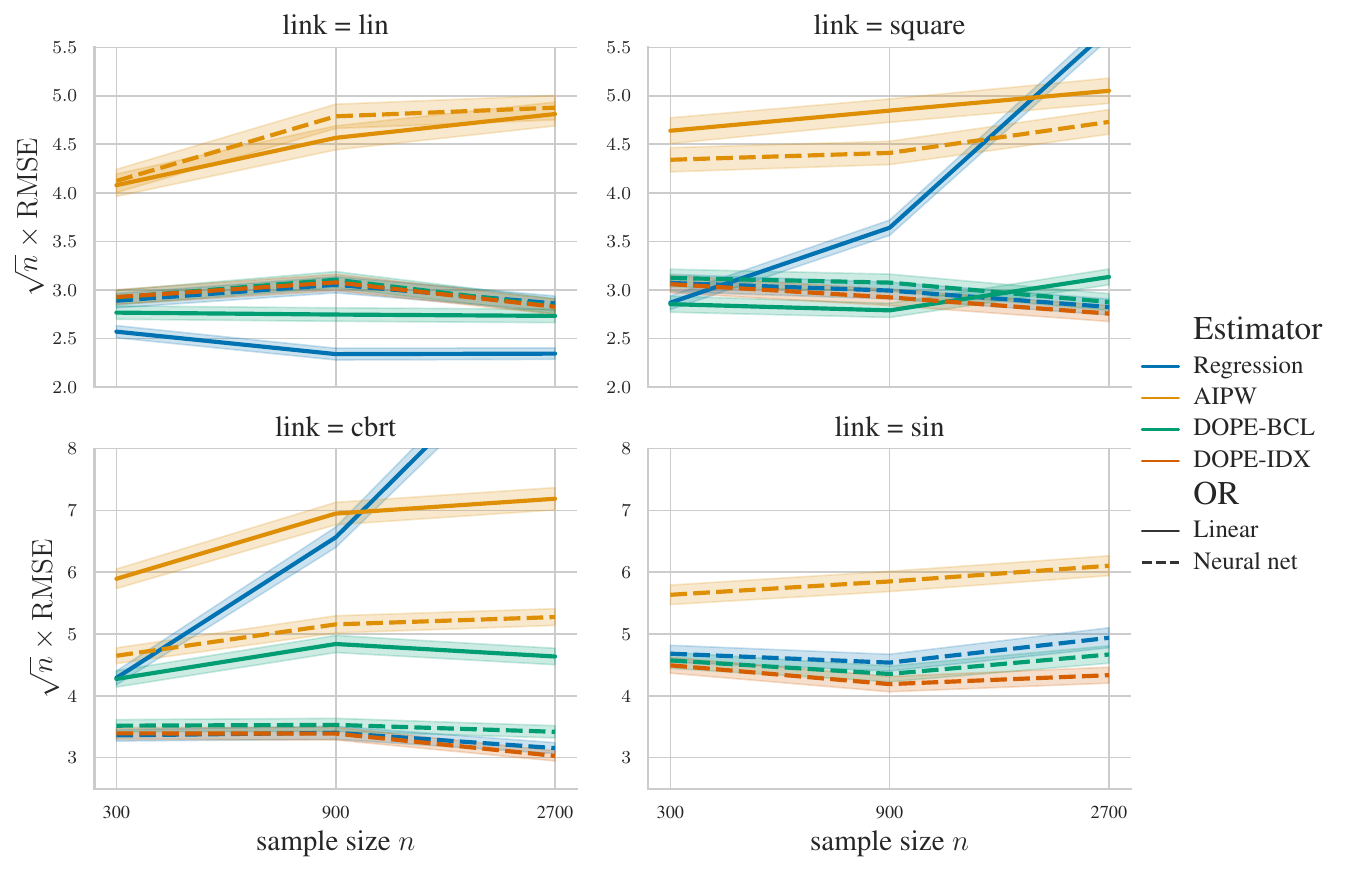}
    \caption{Root mean square errors for various estimators of $\chi_1$ plotted against sample size. Each data point is an average over 900 datasets, 300 for each choice of $d\in\{4,12,36\}$. The bands around each line correspond to $\pm 1$ standard error.    
    For this plot, the outcome regression (OR) was fitted separately for each stratum $T=1$ and $T=2$.}
    \label{fig:RMSE-nocf}
\end{figure}

The results for stratified outcome regression are shown in Figure \ref{fig:RMSE-nocf}. Each panel, corresponding to a link function in \eqref{eq:simulationlinks}, displays the RMSE for various estimators against sample size. Across different values of $d \in \{4,12,36\}$, the results remain consistent, and thus the RMSE is averaged over the $3\times 300$ datasets with varying $d$. In the upper left panel, where the link is linear, the regression estimator $\widehat{\chi}_1^{\mathrm{reg}}$ with OLS performs best as expected. The remaining estimators exhibit similar performance, except for AIPW, which consistently performs poorly across all settings. This can be attributed to extreme values of the propensity score, resulting in a large efficiency bound for AIPW. All estimators seem to maintain an approximate $\sqrt{n}$-consistent RMSE for the linear link as anticipated.

For the nonlinear links, we observe that the OLS-based regression estimators perform poorly and do not exhibit approximate $\sqrt{n}$-consistency. 
For the sine link, the RMSEs for all OLS-based estimators are large, and they are not shown for ease of visualization.
For neural network outcome regression, the AIPW still performs the worst across all nonlinear links. 
The regression estimator and the DOPE estimators share similar performance when used with the neural network, with DOPE-IDX being more accurate overall. Since the neural network architecture is tailored to the single-index model, it is not surprising that the outcome regression works well, and as a result there is less need for debiasing. On the other hand, the debiasing introduced in the DOPE does not hurt the accuracy, and in fact, improves it in this setting.

\begin{figure}
    \centering
    \includegraphics[width=\linewidth]{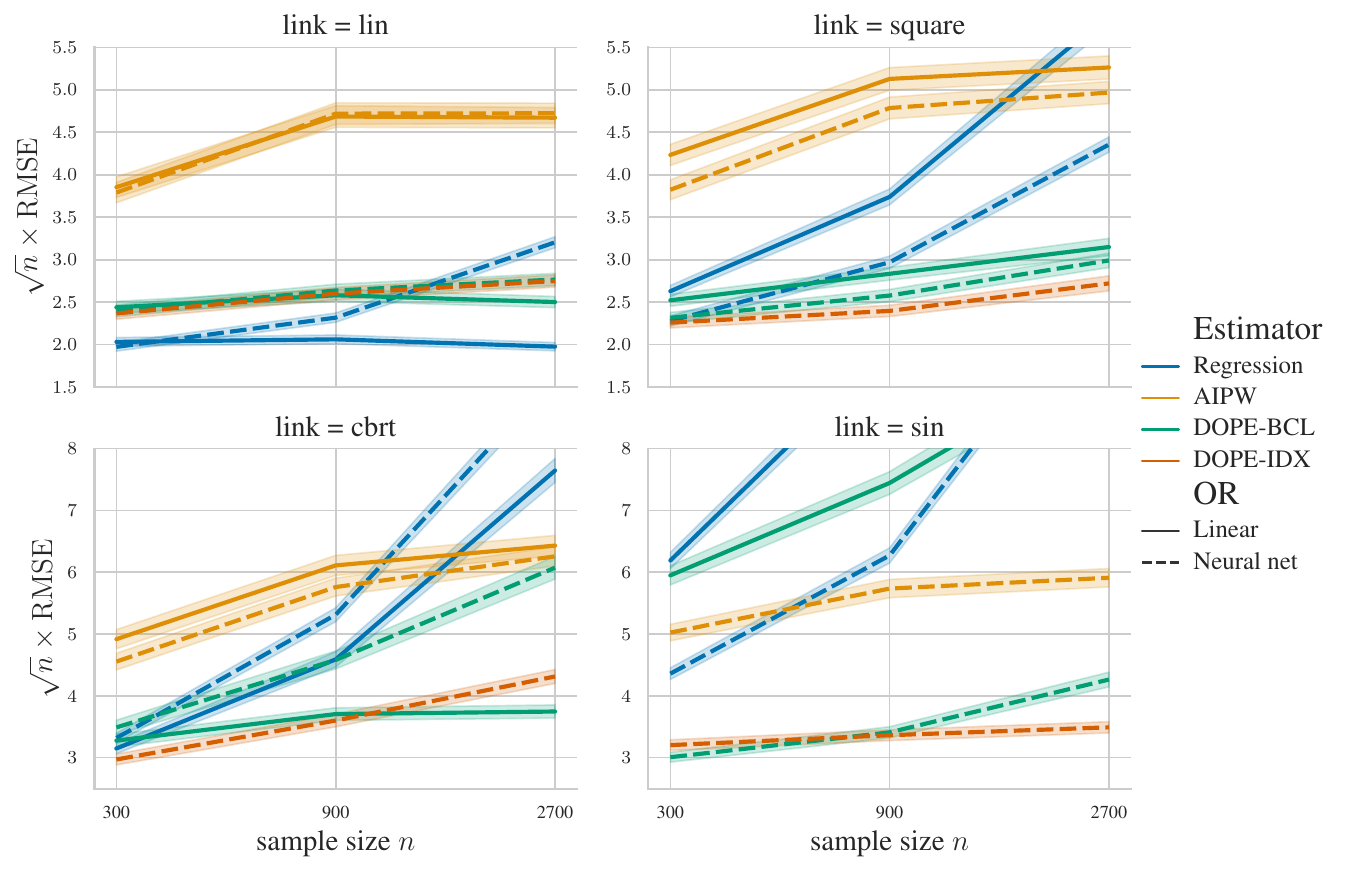}
    \caption{Root mean square errors for various estimators of $\chi_1$ plotted against sample size. Each data point is an average over 900 datasets, 300 for each choice of $d\in\{4,12,36\}$. The bands around each line correspond to $\pm 1$ standard error.    
    For this plot, the outcome regression (OR) was fitted jointly onto $(T,\bW)$.}
    \label{fig:RMSE-joint}
\end{figure}

The results for joint regression of $Y$ on $(T,\bW)$ are shown in Figure \ref{fig:RMSE-joint}. 
The results for the OLS-based estimators provide similar insights as previously discussed, so we focus on the results for the neural network based estimators. The jointly trained neural network is, in a sense, misspecified for the single-index model (except for the linear link), as discussed in Section~\ref{sec:simestimators}. 
Thus it is not surprising that the regression estimator fails to adjust effectively for larger sample sizes. What is somewhat unexpected, however, is that the precision of DOPE, especially the DOPE-IDX, does not appear to be compromised by the misspecified outcome regression. In fact, the DOPE-IDX even seems to perform better with joint outcome regression. 
We suspect that this could be attributed to the joint regression producing more robust predictions for the rare treated subjects with $W^1\leq 0.5$, for which $m(1\given \bW)=0.01$. The predictions are more robust since the joint regression can leverage some of the information from the many untreated subjects with $W^1\leq 0.5$ at the cost of introducing systematic bias, which the DOPE-IDX deals with in the debiasing step.
While this phenomenon is interesting, a thorough exploration of its exact details, both numerically and theoretically, is a task we believe is better suited for future research.

In summary, the DOPE serves as a middle ground between the regression estimator and the AIPW. It provides an additional safeguard against biased outcome regression, all while avoiding the potential numerical instability entailed by using standard inverse propensity weights.

\subsubsection{Inference}\label{sec:simulationinference}
We now consider approximate confidence intervals obtained from the empirical variance estimator $\widehat{\mathcal{V}}$ defined in \eqref{eq:varestimator}. 
Specifically, we consider intervals of the form $\widehat{\chi}_1\pm \frac{z_{0.975}}{\sqrt{n}}\widehat{\mathcal{V}}^{-1/2}$, where $z_{0.975}$ is the $0.975$ quantile of the standard normal distribution.

\begin{figure}
    \centering
    \includegraphics[width=\linewidth]{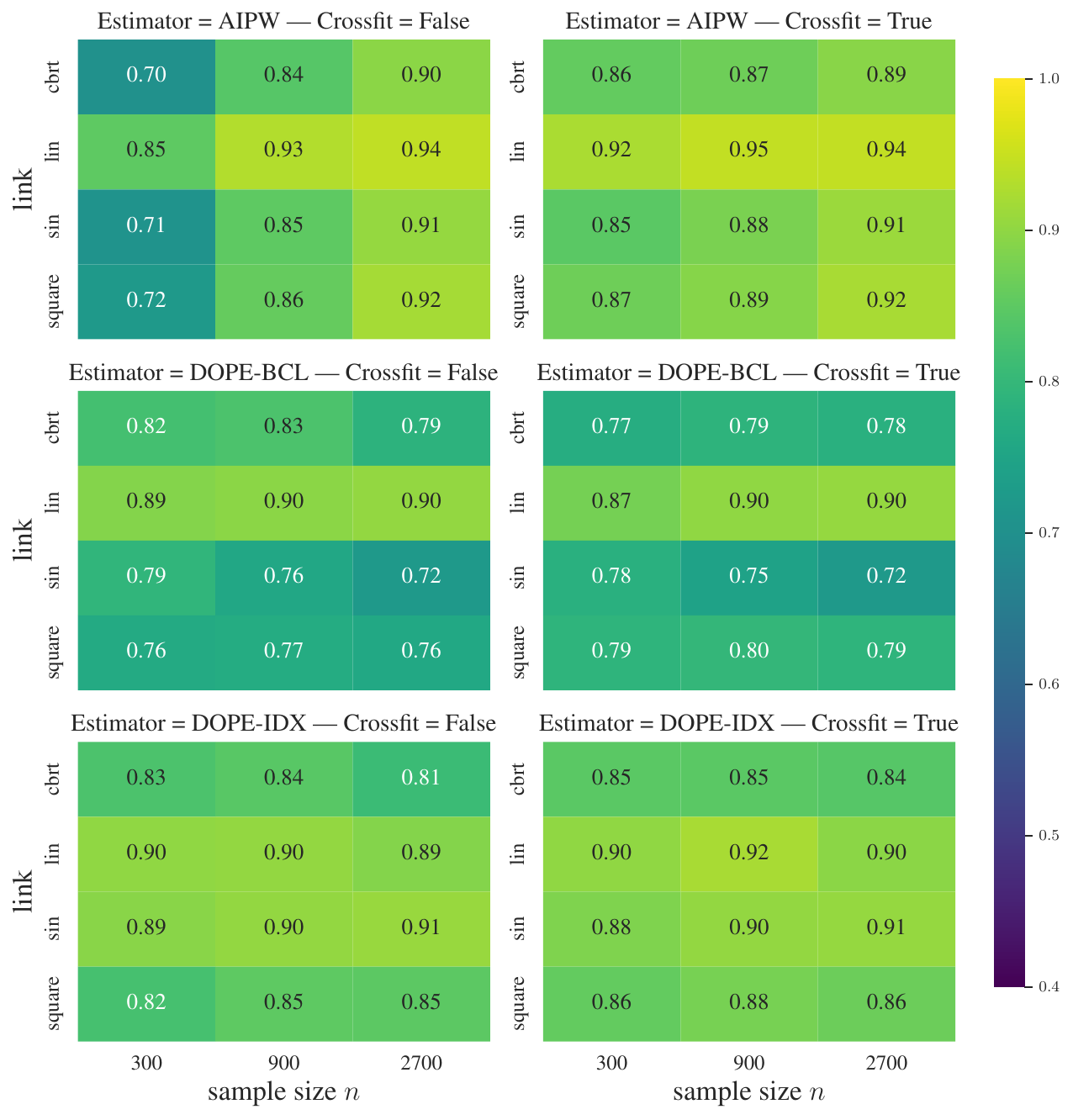}
    \caption{Coverage of asymptotic confidence intervals of the adjusted mean $\chi_1$, aggregated over $d\in\{4,12,36\}$ and with the true $\beta$ given in \eqref{eq:samplescheme}.
    } 
    \label{fig:Coverage-cf}
\end{figure}
Figure \ref{fig:Coverage-cf} shows the empirical coverage of $\chi_1$ for the AIPW, DOPE-BCL and DOPE-IDX, all based on neural network outcome regression fitted jointly. Based on the number of repetitions, we can expect the estimated coverage probabilities to be accurate up to $\pm 0.01$ after rounding. 

The left column shows the results for no sample splitting, whereas the right column shows the results for 4-fold cross-fitting. For all estimators, cross-fitting seems to improve coverage slightly for large sample sizes. The AIPW is a little anti-conservative, especially for the links $h_{\text{sin}}$ and $h_{\text{cbrt}}$, but overall maintains approximate coverage.
The slight miscoverage is to be expected given the challenging nature with the extreme propensity weights.
For the DOPE-BCL and DOPE-IDX, the coverage probabilities are worse, but with the DOPE-IDX having better coverage overall. According to our asymptotic analysis, specifically Theorem~\ref{thm:conditionalAIPWconv}, the DOPE-BCL and DOPE-IDX intervals 
will only achieve asymptotic $95\%$ coverage of the \emph{random} parameters 
$\chi_1(\widehat{h}(t, \bZ_{\hat{\theta}}))$ and $\chi_1(\bZ_{\hat{\theta}})$, respectively. Thus we do not anticipate these intervals to asymptotically achieve the nominal $95\%$ coverage of $\chi_1$. 
However, we show in Figure~\ref{fig:Coverage-with-length} in the supplementary Section~\ref{sup:simdetails} that the average widths of the confidence intervals are significantly smaller for the DOPE intervals than the AIPW interval. Hence it is possible that the DOPE intervals can be corrected -- as we have also discussed in Remark~\ref{rmk:confidenceinterval} -- to gain advantage over the AIPW interval, but we leave such explorations for future work. 

\clearpage
In summary, the lack of full coverage for the naively constructed intervals are consistent with the asymptotic analysis conditionally on a representation $\bZ_{\hat{\theta}}$. 
In view of this, our asymptotic results might be more realistic and pragmatic than the previously studied regime \citep{benkeser2020nonparametric,ju2020robust}, which assumes fast unconditional rates of the nuisance parameters 
$(\widehat{m}_{\hat{\theta}},\widehat{g}_{\hat{\theta}})$.

\subsection{Application to NHANES data}
We consider the mortality dataset collected by the
National Health and Nutrition Examination Survey I Epidemiologic Followup Study \citep{cox1997plan}, henceforth referred to as the NHANES dataset.
The dataset was initially collected as in \citet{lundberg2020local}\footnote{
See \url{https://github.com/suinleelab/treeexplainer-study}
for their GitHub repository.
}.
The dataset contains several baseline covariates, and our outcome variable is the indicator of death at the end of study. 
To deal with missing values, we considered both a complete mean imputation as in \citet{lundberg2020local}, or a trimmed dataset where covariates with more than $50\%$ of their values missing are dropped and the rest are mean imputed. 
The latter approach reduces the number of covariates from $79$ to $65$. 
The final results were similar for the two imputation methods, so we only report the results for the mean imputed dataset here. In the supplementary Section~\ref{sup:simdetails} we show the results for the trimmed dataset.

The primary aim of this section is to evaluate the different estimation methodologies on a realistic and real data example. For this purpose we consider a treatment variable based on \textit{pulse pressure} and study its effect on mortality. Pulse pressure is defined as the difference between systolic and diastolic blood pressure (BP). 
High pulse pressure is not only used as an indicator of disease but is also reported to increase the risk of cardiovascular diseases \citep{franklin1999pulse}. We investigate the added effect of high pulse pressure when adjusting for other baseline covariates, in particular systolic BP and levels of white blood cells, hemoglobin, hematocrit, and platelets. 
We do not adjust for diastolic BP, as it determines the pulse pressure when combined with systolic BP, which would therefore lead to a violation of positivity. 

\begin{figure}
    \centering
    \includegraphics[width=0.75\linewidth]{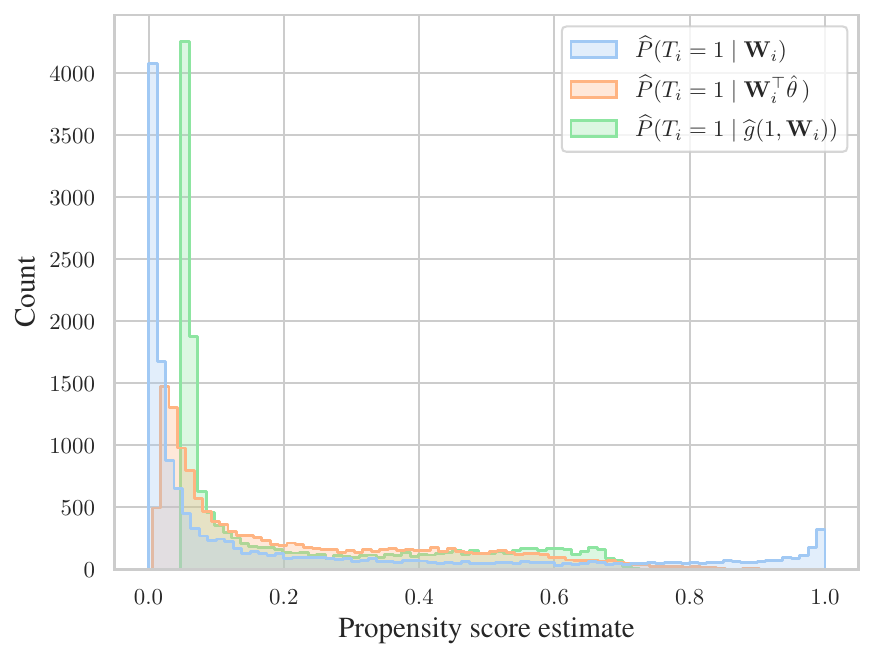}
    \caption{Distribution of estimated propensity scores based on the full covariate set $\bW$, the single-index $\bW^\top \hat{\theta}$, and the outcome predictions $\widehat{g}(1,\bW)$, respectively.}
    \label{fig:NHANESpropensities}
\end{figure}

A pulse pressure of 40 mmHg is considered normal, and as the pressure increases past 50 mmHg, the risk of cardiovascular diseases is reported to increase.
Some studies have used 60 mmHg as the threshold for high pulse pressure \citep{homan2024physiology}, and thus we consider the following treatment variable corresponding to high pulse pressure:
\begin{align*}
    T = \one\left( \textrm{pulse pressure} > 60 \textrm{ mmHg}\right).
\end{align*}

For the binary outcome regression we consider logistic regression and a variant of the neural network from 
Section~\ref{sec:simestimators} with an additional `sigmoid-activation' on the output layer. 
Based on 5-fold cross-validation, logistic regression yields a log-loss of $0.50$, whereas the neural network yields a log-loss of $0.48$ (smaller is better).

Figure~\ref{fig:NHANESpropensities} shows the distribution of the estimated propensity scores used in the AIPW, DOPE-IDX, and DOPE-BCL, based on neural network regression. As expected, we see that the full covariate set yields propensity scores that essentially violate positivity, with this effect being less pronounced for DOPE-IDX and even less so for DOPE-BCL, where the scores are comparatively less extreme.

\begin{table}
    \centering
    \begin{tabular}{lrrr}
    \toprule
    Estimator &  Estimate &  BS se &             BS CI \\
    \midrule
    Regr. (NN)          &     0.022 &  0.009 &    (0.004, 0.037) \\
    Naive contrast      &     0.388 &  0.010 &    (0.369, 0.407) \\
    Regr. (Logistic)    &     0.026 &  0.010 &    (0.012, 0.049) \\
    DOPE-BCL (Logistic) &     0.023 &  0.010 &    (0.003, 0.040) \\
    DOPE-BCL (NN)       &     0.021 &  0.010 &    (0.001, 0.039) \\
    DOPE-IDX (NN)       &     0.024 &  0.012 &    (0.002, 0.048) \\
    AIPW (NN)           &     0.023 &  0.015 &   (-0.009, 0.051) \\
    AIPW (Logistic)     &     0.026 &  0.015 &   (-0.004, 0.056) \\
    IPW (Logistic)      &    -0.040 &  0.025 &  (-0.103, -0.006) \\
    \bottomrule \\
    \end{tabular}

    \caption{Adjusted effect estimates of unhealthy pulse pressure on mortality based on NHANES dataset. Estimators are sorted in increasing order after bootstrap (BS) standard error. }
    \label{tab:pulsepressure}
\end{table}
Table~\ref{tab:pulsepressure} shows the estimated treatment effects 
$\widehat{\Delta} = \widehat{\chi}_1-\widehat{\chi}_0$ for various estimators together with $95\%$ bootstrap confidence intervals. The estimators are sorted according to bootstrap variance based on $B=1000$ bootstraps.  
Because we do not know the true effect, we cannot directly decide which estimator is best. 
The naive contrast provides a substantially larger estimate than all of the adjustment estimators, which indicates that the added effect of high pulse pressure on mortality when adjusting is different from the unadjusted effect. 
The IPW estimator, $\widehat{\chi}_1^{\mathrm{ipw}}$, is the only estimator to yield a negative estimate of the adjusted effect, and its bootstrap variance is also significantly larger than the other estimators. Thus it is plausible, that the IPW estimator fails to adjust appropriately. This is not surprising given extreme distribution of the propensity weights, as shown in blue in Figure~\ref{fig:NHANESpropensities}.

The remaining estimators yield comparable estimates of the adjusted effect, with the logistic regression based estimators having a marginally larger estimates than their neural network counterparts. The (bootstrap) standard errors are comparable for the DOPE and regression estimators, but the AIPW estimators have slightly larger standard errors. As a result, the added effect of high pulse pressure on mortality cannot be considered statistically significant for the AIPW estimators, whereas it can for the other estimators. 

In summary, the results indicate that the choice treatment effect estimator impacts the final estimate and confidence interval. While it is uncertain which estimator is superior for this particular application, the DOPE estimators seem to offer a reasonable and stable estimate. 

\section{Discussion}\label{sec:discussion}
In this paper, we address challenges posed by complex data with unknown underlying structure, an increasingly common scenario in observational studies. Specifically, we have formulated a refined and general formalism for studying efficiency of covariate adjustment. This formalism extends and builds upon the efficiency principles derived from causal graphical models, in particular results provided by \citet{rotnitzky2020efficient}. Our theoretical framework has led to the identification of the optimal covariate information for adjustment. 
From this theoretical groundwork, we introduced the DOPE for general estimation of the ATE with increased efficiency.

Several key areas within this research merit further investigation:

\textbf{Extension to other targets and causal effects}: Our focus has predominantly been on the adjusted mean $\chi_t = \ex[\ex[Y\given T=t,\bW]]$. However, the extension of our methodologies to continuous treatments, instrumental variables, or other causal effects, such as the average treatment effect among the treated, is an area that warrants in-depth exploration. We suspect that many of the fundamental ideas can be modified to prove analogous results for such target parameters. 

\textbf{Beyond neural networks}: While our examples and simulation study has applied the DOPE with neural networks, its compatibility with other regression methods such as kernel regression and gradient boosting opens new avenues for alternative adjustment estimators. Investigating such estimators can provide practical insights and broaden the applicability of our approach. 
The Outcome Highly Adapted Lasso \citep{ju2020robust} is a similar method in this direction that leverages the Highly Adaptive Lasso \citep{benkeser2016highly} for robust treatment effect estimation.

\textbf{Causal representation learning integration}: Another possible direction for future research is the integration of our efficiency analysis into the broader concept of causal representation learning \citep{scholkopf2021toward}. This line of research, which is not concerned with specific downstream tasks, could potentially benefit from some of the insights of our efficiency analysis and the DOPE framework.

\textbf{Implications of sample splitting}: Our current asymptotic analysis relies on the implementation of three distinct sample splits. Yet, our simulation results suggest that sample splitting might not be necessary. Using the notation of Algorithm~\ref{alg:generalalg}, we hypothesize that the independence assumption $\mathcal{I}_3 \cap (\mathcal{I}_1 \cup \mathcal{I}_2)= \emptyset$ could be replaced with Donsker class conditions. However, the theoretical justifications for setting $\mathcal{I}_1$ identical to $\mathcal{I}_2$ are not as evident, thus meriting additional exploration into this potential simplification.

\textbf{Valid and efficient confidence intervals}: 
The naive confidence intervals explored in the simulation study, cf. 
Section~\ref{sec:simulationinference}, do not provide an adequate coverage of the unconditional adjusted mean. As indicated in Remark~\ref{rmk:confidenceinterval}, it would be interesting to explore methods for correcting the intervals, for example by: (i) adding a bias correction, (ii) constructing a consistent estimator of the unconditional variance, or (iii) using bootstrapping. Appealing to bootstrapping techniques, however, might be computationally expensive, in particular when combined with cross-fitting. To implement (i), one approach could be to generalize
Theorem~\ref{thm:SIregularity} to establish general differentiability of the adjusted mean with respect to general representations, and then construct an estimator of the gradient. However, one must asses whether this bias correction could negate the efficiency benefits of the DOPE.

\section*{Acknowledgments}
We thank Leonard Henckel and Anton Rask Lundborg for helpful discussions. 
AMC and NRH were supported by a research grant (NNF20OC0062897) from Novo Nordisk Fonden.

%% file: paper_dope/appendix.tex
\clearpage
\section*{Supplement to `Efficient adjustment for complex covariates: Gaining efficiency with DOPE'}

The supplementary material is organized as follows: 
Section~\ref{sup:proofs} contains proofs of the results in the main text and a few auxiliary results; 
Section~\ref{sup:simdetails} contains details on the experiments in Section~\ref{sec:experiments}; 
and Section~\ref{sup:crossfitting} contains a description of the cross-fitting algorithm.

\supplementarysection{Auxiliary results and proofs}\label{sup:proofs}
The proposition below relaxes the criteria for being a valid adjustment set in a causal graphical model. It follows from the results of \citet{perkovic2018complete}, which was pointed out to the authors by Leonard Henckel, and is stated here for completeness.
\begin{prop}\label{prop:adjset}
    Let $T,\bZ$, and $Y$ be pairwise disjoint node sets in a DAG $\mathcal{D}=(\mathbf{V},\mathbf{E})$, and let $\mathcal{M}(\mathcal{D})$ denote the collection of continuous distributions that are  Markovian with respect to $\mathcal{D}$ and with $\ex|Y|<\infty$. Then $\bZ$ is an adjustment set relative to $(T,Y)$ if and only if for all $P\in \mathcal{M}(\mathcal{D})$ and all $t$
    \begin{align*}
        \ex_P[Y\given \mathrm{do}(T = t)] 
        = \ex_P\left[
            \frac{\one(T=t)Y}{\mathbb{P}_P(T=t\given \mathrm{pa}_{\mathcal{D}}(T))}
            \right]
        = \ex_P[\ex_P[Y\given T=t, \bZ]]
    \end{align*}
\end{prop}

\begin{proof}
The `only if' direction is straightforward if we use an alternative characterization of $\bZ$ being an adjustment set: for any density $p$ of a distribution $P\in \mathcal{P}$ it holds that
$$
    p(y\given \mathrm{do}(T=t)) 
    = \int p(\mathbf{y}|\bz,t) p(\bz)\mathrm{d}\bz,
$$
where the `do-operator' (and notation) is defined as in, e.g., \citet{peters2017elements}. Fubini's theorem then yields 
\begin{align*}
    \ex_P[Y\given \mathrm{do}(T=t)] 
    &= \int_{\real^l} \mathbf{y} p(\mathbf{y}|do(\mathbf{T}=\mathbf{t})) \mathrm{d}\mathbf{y}\\
    &= \int_{\real^l} \mathbf{y} \int_{\real^k} p(\mathbf{y}|\bz,t)
    p(\bz)\mathrm{d}\bz\mathrm{d}\mathbf{y} \\
    &=  \int_{\real^k}\int_{\real^l} \mathbf{y} p(\mathbf{y}|\bz,t) \mathrm{d}\mathbf{y}
    p(\bz)\mathrm{d}\bz \\
    &= \ex_P[\ex_P[Y\given T=t,\bZ]],
\end{align*}
which is equivalent to the `only if' part. On the contrary, assume that $\bZ$ is \emph{not} an adjustment set for $(T,Y)$. Then Theorem~56 and the proof of Theorem~57 in \citet{perkovic2018complete} imply the existence of a Gaussian distribution $\tilde{P}\in \mathcal{P}$ for which 
\begin{align*}
    \ex_{\tilde{P}}[Y\given \mathrm{do}(T = 1)] \neq \ex_{\tilde{P}}[\ex_{\tilde{P}}[Y\given T=1, \bZ]]
\end{align*}
This implies the other direction.
\end{proof}
The following lemma is a generalization of Lemma 27 in \citet{rotnitzky2020efficient}.
\begin{lem}\label{lem:invpropexp}
    For any $\sigma$-algebras $\cZ_1\subseteq \cZ_2 \subseteq \sigma(\bW)$
    it holds, under Assumption~\ref{asm:positivity}, that
    $$
        \ex_P\left[
            \pi_t(\cZ_2; P)^{-1} \given T=t,\cZ_1
        \right]
         = \pi_t(\cZ_1; P)^{-1}
    $$
    for all $P\in\mathcal{P}$.
\end{lem}
\begin{proof}
    Direct computation yields
    \begin{align*}
        &\ex_P\left[
            \pi_t(\cZ_2; P)^{-1} \given T=t,\cZ_1
        \right] \mathbb{P}_P(T=t\given\cZ_1) \\
        &= 
        \ex_P\left[
            \frac{\one(T=t)}{\pi_t(\cZ_2; P)} \given \cZ_1
        \right] 
        =
        \ex_P\left[
            \frac{\ex_P[\one(T=t)\given\cZ_2]}{\pi_t(\cZ_2; P)} \given \cZ_1
        \right]
        = 1.
    \end{align*}
\end{proof}

\subsection{Proof of Lemma \ref{lem:overadj}}
Since $P\in \mathcal{P}$ is fixed, we suppress it from the notation in the following computations.

From $Y\ind \cZ_2 \given \cZ_1,T$ and $\cZ_1\subseteq \cZ_2$, it follows that $\ex[Y\given\cZ_1,T] = \ex[Y\given\cZ_2,T]$. Hence, for each $t\in \mathbb{T}$,
\begin{equation}\label{eq:overadjeq1}
    b_t(\cZ_1) = \ex[Y\given\cZ_1,T=t] 
    = \ex[Y\given\cZ_2,T=t]=b_t(\cZ_2).
\end{equation} 
Therefore $\chi_t(\cZ_1) = \ex[b_t(\cZ_1)] = \ex[b_t(\cZ_2)] = \chi_t(\cZ_2)$. If $\mathcal{Z}_2$ is description of $\bW$, this identity shows that $\cZ_2$ is $\mathcal{P}$-valid if and only if $\cZ_1$ is $\mathcal{P}$-valid. 

To compare the asymptotic efficiency bounds for $\cZ_2$ with $\cZ_1$ we use the law of total variance. Note first that
\begin{align*}
    \ex[\psi_{t}(\cZ_2)\given T,Y,\cZ_1] 
    &= \ex\Big[\one(T=t)(Y-b_t(\cZ_2))
        \frac{1}{\pi_t(\cZ_2)} + b_t(\cZ_2)-\chi_t \given T,Y,\cZ_1 
             \Big] \\
    &= \one(T=t)(Y-b_t(\cZ_1))\,
        \ex\Big[\frac{1}{\pi_t(\cZ_2)}\given T=t,Y,\cZ_1 \Big]
            + b_t(\cZ_1)-\chi_t \\
    &= \one(T=t)(Y-b_t(\cZ_1))
        \ex\left[\frac{1}{\pi_t(\cZ_2)}\given T=t,\cZ_1 \right]
            + b_t(\cZ_1)-\chi_t \\
    &= \one(T=t)(Y-b_t(\cZ_1))
        \frac{1}{\pi_t(\cZ_1)} + b_t(\cZ_1)-\chi_t
        = \psi_t(\cZ_1).
\end{align*}
Second equality is due to $\cZ_1$-measurability and \eqref{eq:overadjeq1}, whereas the third equality follows from $Y\ind \cZ_2 \given \cZ_1,T$ and the last inequality is due to Lemma~\ref{lem:invpropexp}.
On the other hand,
\begin{align*}
    &\phantom{=}\ex\big[\var(\psi_{t}(\cZ_2)\given T,Y,\cZ_1)\big]\\
    &= \ex\Big[\var\Big(\one(T=t)(Y-b_t(\cZ_1))
        \frac{1}{\pi_t(\cZ_2)} + b_t(\cZ_1)-\chi_t(P)\; \big| \; T,Y,\cZ_1 
             \Big)\Big] \\
    &= \ex\Big[\one(T=t)(Y-b_t(\cZ_1))^2\var\Big(
        \frac{1}{\pi_t(\cZ_2)} \; \big| \; T=t,Y,\cZ_1 
             \Big)\Big] \\
    &= \ex\Big[\pi_t(\cZ_1)\var(Y \given T=t,\cZ_1)\var\Big(
        \frac{1}{\pi_t(\cZ_2)} \; \big| \; T=t,\cZ_1 
             \Big)\Big].
\end{align*}
Combined we have that
\begin{align*}
    \bV_t(\cZ_2)
    &= \var(\psi_{t}(\cZ_2))
    = \var(\ex[\psi_{t}(\cZ_2) \given T,Y,\cZ_1]) 
        + \ex[\var[\psi_{t}(\cZ_2) \given T,Y,\cZ_1]] \\
    &= \bV_t(\cZ_1) 
        + \ex\Big[\pi_t(\cZ_1)\var(Y \given T=t,\cZ_1)\var\Big(
        \frac{1}{\pi_t(\cZ_2)}\Big| \; T=t,Y,\cZ_1 
             \Big)\Big].
\end{align*}
To prove the last part of the lemma, we first note that $\psi_t(\cZ_2)$ and $\psi_{t'}(\cZ_2)$ are conditionally uncorrelated for $t\neq t'$. This follows from \eqref{eq:overadjeq1}, from which the conditional covariance reduces to
\begin{align*}
    \cov(\psi_t(\cZ_2), \psi_{t'}(\cZ_2) \given T,Y,\cZ_1)
    = \one(T=t)\one(T=t')
    \cov(
        *,* \given T,Y,\cZ_1
    )
    = 0.
\end{align*}
Thus, letting $\psi(\cZ_2) = (\psi_t(\cZ_2))_{t\in \mathbb{T}}$ and using the computation for $\bV_t$, we obtain that:
\begin{align*}
    &\bV_\Delta(\cZ_2) = \var(\mathbf{c}^\top \psi(\cZ_2)) \\
    &= \var(\ex[\mathbf{c}^\top \psi(\cZ_2) \given T,Y,\cZ_1]) 
        +  \ex[\var[\mathbf{c}^\top \psi(\cZ_2) \given T,Y,\cZ_1]] \\
    &= \bV_\Delta(\cZ_1) 
        + \sum_{t\in \mathbb{T}} c_t^2 \, \ex\Big[\pi_t(\cZ_1)\var(Y \given T=t,\cZ_1)\var\Big(
        \frac{1}{\pi_t(\cZ_2)}\Big| \; T=t,Y,\cZ_1 
             \Big)\Big].
\end{align*}
This concludes the proof. \hfill \qedsymbol

\subsection{Proof of Theorem \ref{thm:COSefficiency}}
(i) Assume that $\cZ\subseteq \sigma(\bW)$ is a description of $\bW$ such that $\mathcal{Q}_P\subseteq \ol{\cZ}$. Observe that almost surely,
$$
    \mathbb{P}_P(Y\leq y \given T, \bW) 
    = \sum_{t\in \mathbb{T}} \one(T=t)F(y\mid t,\bW;P)
$$ 
is $\sigma(T,\mathcal{Q}_P)$-measurable and hence also 
$\sigma(T,\ol{\cZ})$-measurable.
It follows that $\mathbb{P}_P(Y\leq y \given T, \bW)$ is also a version of $\mathbb{P}_P(Y\leq y \given T, \overline{\cZ})$. 
As a consequence, $\mathbb{P}_P(Y\leq y \given T, \bW) = \mathbb{P}_P(Y\leq y \given T, \cZ)$ almost surely. From Doob's characterization of conditional independence \citep[Theorem 8.9]{kallenberg2021foundations}, we conclude that $Y\indP \bW \mid T,\cZ$, or equivalently that $\cZ$ is $P$-ODS. The `in particular' follows from setting $\cZ=\mathcal{Q}_P$.

(ii) Assume that $\cZ \subseteq \sigma(\bW)$ is $P$-ODS. Using Doob's characterization of conditional independence again, 
$\mathbb{P}_P(Y\leq y \given T, \bW) = \mathbb{P}_P(Y\leq y \given T, \cZ)$
almost surely. Under Assumption~\ref{asm:positivity}, this entails that
$F(y\given T=t,\bW;P) = \mathbb{P}_P(Y\leq y \given T=t, \cZ)$ almost surely, and hence $F(y\given T=t,\bW;P)$ must be $\ol{\cZ}$-measurable. As the generators of $\mathcal{Q}_P$ are $\ol{\cZ}$-measurable, we conclude that $\mathcal{Q}_P \subseteq \ol{\cZ}$.

(iii) Let $\cZ$ be a $P$-ODS description. 
From (i) and (ii) it follows that $\mathcal{Q}_P\subseteq \ol{\cZ}$, and since $\cZ \subseteq \sigma(\bW)$ it 
also holds that $Y\indP \ol{\cZ} \mid T,\mathcal{Q}_P$. 
Since $\bV_\Delta(\cZ; P) = \bV_\Delta(\ol{\cZ}; P)$, Lemma~\ref{lem:overadj} gives the desired conclusion. 
\hfill \qed{}

\subsection{Proof of Corollary \ref{cor:Qefficiency}}
Theorem~\ref{thm:COSefficiency} (i,ii) implies that
$\cZ$ is $\mathcal{P}$-ODS if and only if~$\ol{\cZ}^P$ contains $\mathcal{Q}_P$ for all $P \in \mathcal{P}$.

If $\mathcal{Q} \subseteq \ol{\cZ}^P$ for all $P \in \mathcal{P}$ then $\mathcal{Q}_P \subseteq \ol{\cZ}^P$ 
for all $P \in \mathcal{P}$ by definition, and $\cZ$ is $\mathcal{P}$-ODS. 
Equation \eqref{eq:COAuniform} follows from the same argument as Theorem~\ref{thm:COSefficiency} $(iii)$. That is, 
$\mathcal{Q} \subseteq \ol{\cZ}^P$,  and since $\mathcal{Q}$ is $P$-ODS and $\cZ \subseteq \sigma(\bW)$ we have $Y \indP \ol{\cZ}^P \mid T,\mathcal{Q}$. Lemma~\ref{lem:overadj} can be applied for each $P\in \mathcal{P}$ to obtain the desired conclusion. 
\hfill \qed{}

\subsection{Proof of Proposition \ref{prop:COMSefficiency}}
Since $\mathcal{R}_P \subseteq \mathcal{Q}_P$ always holds, so it suffices to prove that $\mathcal{Q}_P \subseteq \mathcal{R}_P$. 
If $F(y\given t,\bW;P)$ is $\sigma(b_t(\bW;P))$-measurable, then $\mathcal{Q}_P \subseteq \mathcal{R}_P$ follows directly from definition. 
If $Y$ is a binary outcome, then the conditional distribution is determined by the conditional mean, i.e., $F(y\given t,\bW;P)$ is $\sigma(b_t(\bW;P))$-measurable. If $Y = b_T(\bW) + \varepsilon_Y$ with $\varepsilon_Y \ind T, \bW$, 
then $F(y\given t,\bw;P) = F_P(y - b_t(\bw))$ where $F_P$ is the distribution function for the distribution of $\varepsilon_Y$. Again, $F(y\given t,\bW;P)$ becomes $\sigma(b_t(\bW;P))$-measurable.


To prove the last part, let $\cZ$ be a $P$-OMS description. 
Since $b_t(\bW;P) = b_t(\cZ;P)$ is $\cZ$-measurable for every $t\in \mathbb{T}$, it follows that $\mathcal{R}_P\subseteq \ol{\cZ}$.
Hence the argument of Theorem~\ref{thm:COSefficiency} (iii) establishes  
Equation~\eqref{eq:COApointwise} when $\mathcal{Q}_P = \mathcal{R}_P$. 
\hfill \qed{}

\subsection{Proof of Lemma~\ref{lem:underadj}}

Since $P\in \mathcal{P}$ is fixed throughout this proof, we suppress it from notation. From $\cZ_1\subseteq \cZ_2 \subseteq \sigma(\bW)$ and $T\indP \cZ_2 \given \cZ_1$, it follows that $\pi_t(\cZ_1) = \pi_t(\cZ_1)$, and hence 
$$
    \chi_t(\cZ_1)
        =\ex[\pi_t(\cZ_1)^{-1}\one(T=t)Y]
        =\ex[\pi_t(\cZ_2)^{-1}\one(T=t)Y]
        = \chi_t(\cZ_2).
$$
This establishes the first part. For the second part, we use that
$\pi_t(\cZ_1) = \pi_t(\cZ_2)$ and $\chi_t(\cZ_1) = \chi_t(\cZ_2)$
to see that
\begin{align*}
    \psi_t(\cZ_1) - \psi_t(\cZ_2)
        &= b_t(\cZ_1)-b_t(\cZ_2) 
        + \frac{\one(T=t)}{\pi_t(\cZ_2)}(b_t(\cZ_2)-b_t(\cZ_1)) \\
        &= \Big(\frac{\one(T=t)}{\pi_t(\cZ_2)}-1\Big)(b_t(\cZ_2)-b_t(\cZ_1))
        \eqcolon R_t(\cZ_1,\cZ_2) .
\end{align*}
Since $\ex[R_t(\cZ_1,\cZ_2)  \given \cZ_2]=0$ we have that $\ex[\psi_t(\cZ_2)R_t(\cZ_1,\cZ_2) ]=0$, which means that $\psi_t(\cZ_2)$ and $R_t(\cZ_1,\cZ_2) $ are uncorrelated. Note that from $\pi_t(\cZ_1) = \pi_t(\cZ_2)$ it also follows that 
\begin{align*}
    b_t(\cZ_1)
    = \ex\Big[\frac{Y\one(T=t)}{\pi_t(\cZ_1)}\given \cZ_1\Big] 
    = \ex\Big[\ex\Big[\frac{Y\one(T=t)}{\pi_t(\cZ_2)}\given \cZ_2\Big]
        \given \cZ_1\Big]
    = \ex[b_t(\cZ_2) \given \cZ_1]
\end{align*}
This means that $\ex[(b_t(\cZ_2)-b_t(\cZ_1))^2\given \cZ_1]= \var(b_t(\cZ_2)\given \cZ_1)$. 
These observations let us compute that
\begin{align} \label{eq:varRtcomputation}
    \bV_t(\cZ_1)
    &= \var(\psi_t(\cZ_1)) \nonumber\\
    &= \var(\psi_t(\cZ_2)) 
        +\var\left(R_t(\cZ_1,\cZ_2) \right)  \nonumber\\
    &= \var(\psi_t(\cZ_2)) 
        +\ex[\ex[R_t(\cZ_1,\cZ_2)^2\given \cZ_2]]  \nonumber\\
    &= \bV_t(\cZ_2)
        + \ex\Big[
        (b_t(\cZ_2)-b_t(\cZ_1))^2
        \ex\Big[\frac{\one(T=t)^2}{\pi_t(\cZ_1)^2}
            -\frac{2\one(T=t)}{\pi_t(\cZ_1)}+1\given \cZ_2\Big]\Big] \nonumber\\
    &= \bV_t(\cZ_2)
        + \ex\Big[
        \var(b_t(\cZ_2)\given \cZ_1)
        \Big(\frac{1}{\pi_t(\cZ_1)}-1\Big)\Big].
\end{align}
This establishes the formula in the case $\Delta = \chi_t$. 

For general $\Delta$, note first that $\mathbf{R}$ and $(\psi_t(\cZ_2))_{t\in\mathbb{T}}$ are uncorrelated since $\ex[\mathbf{R}\given \cZ_2]=0$. Therefore we have
\begin{align*}
    \avar(\cZ_1) 
        &= \var\Big(\sum_{t\in\mathbb{T}}c_t\psi_t(\cZ_1)\Big) \\
        &= \var\Big(\sum_{t\in\mathbb{T}}c_t\psi_t(\cZ_2)\Big) +
        \var\Big(\sum_{t\in\mathbb{T}}c_t R_t(\cZ_1,\cZ_2)\Big) \\
        &= \avar(\cZ_2) + \mathbf{c}^\top \var\Big( \mathbf{R}\Big)\mathbf{c}.
\end{align*}
It now remains to establish the covariance expressions for the last term, since the expression of $\var\left(R_t(\cZ_1,\cZ_2) \right)$ was found in \eqref{eq:varRtcomputation}. 
To this end, note first that for any $s,t\in\mathbb{T}$ with $s\neq t$,
\begin{align*}
    \ex\left[
        \Big(\frac{\one(T=s)}{\pi_s(\cZ_2)}-1\Big)\Big(\frac{\one(T=t)}{\pi_t(\cZ_2)}-1\Big)
        \given \cZ_2
    \right]
    = -1.
\end{align*}
Thus we finally conclude that
\begin{align*}
    \cov(R_s,R_t)
    &=
        -\ex\left[
            (b_s(\cZ_2)-b_s(\cZ_1))
            (b_t(\cZ_2)-b_t(\cZ_1))
        \right] \\
    &=  -\ex\left[
            \cov(b_s(\cZ_2),b_t(\cZ_2) \given \cZ_1)
        \right],
\end{align*}
as desired.
\hfill \qed{}

\subsection{Computations for Example \ref{ex:symmetric}}\label{sec:symmetriccomputations}
    Using symmetry, we see that
    \begin{align*}
        \pi_1(\bZ) &= \ex[T \mid \bZ ] = \ex[\bW \mid \bZ] \\
        &= 
            (0.5-\bZ)\mathbb{P}(\bW\leq 0.5\given \bZ)
             + (0.5+\bZ)\mathbb{P}(\bW > 0.5\given \bZ)\
        = \frac{1}{2}.
    \end{align*}
    From $T\ind \bZ$ we observe that
    \begin{align*}
        b_t(\bW) &= t + \ex[g(\bZ)\given  T=t,\bW] = t + g(\bZ), \\
        b_t(\bZ) &= t + \ex[g(\bZ)\given  T=t,\bZ] = t + g(\bZ), \\
        b_t(0) &= t + \ex[g(\bZ)\given T=t] = t + \ex[g(\bZ)].
    \end{align*}
    Plugging these expressions into the asymptotic variance yields:
    \begin{align*}
        \bV_t(0)  
        &= \var\left( b_t(0) + \frac{\one(T=t)}{\pi_t(0)}(Y-b_t(0))\right)\\
        &= \pi_t(0)^{-2}\var\left(\one(T=t)(Y-b_t(0))\right)\\
        &= 4\ex[\one(T=t)( g(\bZ) - \ex[ g(\bZ)]+v(\bW)\varepsilon_Y)^2] \\
        &= 4\ex[\bW( g(\bZ) - \ex[ g(\bZ)])^2]
            +2\ex[v(\bW)\varepsilon_Y^2] \\
        &= 2\var(g(\bZ)) + 2 \ex[v(\bW)^2]\ex[\varepsilon_Y^2] 
    \end{align*}
    \begin{align*}
        \bV_t(\bZ) 
        &= \var\left( b_t(\bZ) + \frac{\one(T=t)}{\pi_t(\bZ)}(Y-b_t(\bZ))\right) \\
        &= \var\left( g(\bZ) + 2\one(T=t)v(\bW)\varepsilon_Y\right) \\
        &= \var\left( g(\bZ)\right)  + 2\ex[v(\bW)^2] \ex[\varepsilon_Y^2] 
    \end{align*}
    \begin{align*}
        \bV_t(\bW) 
        &= \var\left( b_t(\bW) + \frac{\one(T=t)}{\pi_t(\bW)}(Y-b_t(\bW))\right) \\
        &= \var\left( g(\bZ) + \frac{\one(T=t)v(\bW)\varepsilon_Y}{\bW}\right) \\
        &= \var\left( g(\bZ)\right)  + \ex\left[\frac{v(\bW)^2}{\bW}\right] \ex[\varepsilon_Y^2].
    \end{align*}
\hfill \qed{}

\subsection*{Proof of Theorem \ref{thm:conditionalAIPWconv}}    
We first prove that $R_i \xrightarrow{P} 0$ for $i=1,2,3$. 
The third term is taken care of by combining Assumption \ref{asm:AIPWconv} $(i,iv)$ with Cauchy-Schwarz:
\begin{align*}
    |R_3| 
    &\leq \sqrt{|\mathcal{I}_3|} 
        \sqrt{\frac{1}{|\mathcal{I}_3|} \sum_{i\in \mathcal{I}_3} 
        \big(\widehat{m}_{\hat{\theta}}(t\given \bW_i)^{-1} 
            - m_{\hat{\theta}}(t\given \bW_i)^{-1}\big)^2}
        \sqrt{\mathcal{E}_{2,t}^{(n)}} \\
    &\leq \ell_c \sqrt{|\mathcal{I}_3| 
        \mathcal{E}_{1,t}^{(n)}\mathcal{E}_{2,t}^{(n)}} \xrightarrow{P} 0,
\end{align*}
where $\ell_c>0$ is the Lipschitz constant of $x \mapsto x^{-1}$ 
on $[c,1-c]$. 

For the first term, $R_1$, note that conditionally on $\hat{\theta}$ and the observed  estimated representations
$\bZ_{\hat{\bullet}} \coloneq (\bZ_{\hat{\theta},i})_{i\in\mathcal{I}_3} =(\phi(\hat{\theta},\bW_i))_{i\in\mathcal{I}_3}$, the summands are conditionally i.i.d. due to Assumption~\ref{asm:samples}. They are also conditionally mean zero since
\begin{align*}
    \ex[r_i^1 \given \bZ_{\hat{\bullet}}, \hat{\theta}\,]
    &= 
    \left(\widehat{g}_{\hat{\theta}}(t,\bW_i)-g_{\hat{\theta}}(t,\bW_i)\right)
        \Big(1-\frac{\ex[\one(T_i=t)\given \bZ_{\hat{\bullet}}, \hat{\theta}\,]}{m_{\hat{\theta}}(t \given \bW_i)}\Big) \\
    &= 
    \left(\widehat{g}_{\hat{\theta}}(t,\bW_i)-g_{\hat{\theta}}(t,\bW_i)\right)
        \Big(1-\frac{\ex[\one(T_i=t)\given \bZ_{\hat{\theta},i}, \hat{\theta}]}{m_{\hat{\theta}}(t \given \bW_i)}\Big)
     = 0,
\end{align*}
for each $i\in \mathcal{I}_3$.
Using Assumption \ref{asm:AIPWconv}~$(i)$, we can bound the conditional variance by
\begin{align*}
    \var\left(r_i^1
        \given
        \bZ_{\hat{\bullet}}, \hat{\theta} \,
        \right) 
    &=\frac{\left(
        \widehat{g}_{\hat{\theta}}(t,\bW_i)-g_{\hat{\theta}}(t,\bW_i)\right)^2
        }{m_{\hat{\theta}}(t \given \bW)^2} 
        \var\left(
        \one(T=t)
        \given
        \bZ_{\hat{\bullet}}, \hat{\theta}
        \right) \\
    &=\frac{\left(
        \widehat{g}_{\hat{\theta}}(t,\bW_i)-g_{\hat{\theta}}(t,\bW_i)\right)^2
        }{m_{\hat{\theta}}(t \given \bW)} 
        (1-m_{\hat{\theta}}(t \given \bW)) \\
    &\leq \frac{(1-c)}{c}
    \left(\widehat{g}_{\hat{\theta}}(t,\bW_i)-g_{\hat{\theta}}(t,\bW_i)\right)^2
\end{align*}
Now it follows from Assumption~\ref{asm:AIPWconv}~$(v)$ and the conditional Chebyshev's inequality that
\begin{align*}
    \mathbb{P}\left(|R_1| >\epsilon\given \bZ_{\hat{\bullet}}, \hat{\theta}\,\right)
    &\leq \epsilon^{-2}\var(R_1\given \bZ_{\hat{\bullet}}, \hat{\theta}) 
    \leq \frac{1-c}{\epsilon^2 c} \mathcal{E}_{2,t}^{(n)}
    \xrightarrow{P} 0,
\end{align*}
from which we conclude that $R_1\xrightarrow{P}0$.

The analysis of $R_2$ is similar. Observe that the summands are conditionally i.i.d. given $\bZ_{\hat{\bullet}},\hat{\theta}$, and $T_\bullet \coloneq (T_i)_{i\in\mathcal{I}_3}$. They are also conditionally mean zero since
\begin{align*}
    \ex\left[r_i^2
        \given \bZ_{\hat{\bullet}}, T_\bullet, \hat{\theta}
        \right] 
    &= 
     \left(\ex[Y_i \given \bZ_{\hat{\bullet}}, T_\bullet, \hat{\theta}]-g_{\hat{\theta}}(t,\bW_i)\right)\left(
            \frac{\one(T_i=t)}{\widehat{m}_{\hat{\theta}}(\bW_i)}
            -\frac{\one(T_i=t)}{m_{\hat{\theta}}(\bW_i)}\right) \\
    &= 
     (\ex[Y_i \given \phi(\hat{\theta},\bW_i), T_i, \hat{\theta}]-g_{\hat{\theta}}(t,\bW_i))\left(
            \frac{\one(T_i=t)}{\widehat{m}_{\hat{\theta}}(\bW_i)}
            -\frac{\one(T_i=t)}{m_{\hat{\theta}}(\bW_i)}\right) = 0,
\end{align*}
where we use that on the event $(T_i = t)$,
$$
\ex[Y_i \given \bZ_{\hat{\theta},i}, T_i, \hat{\theta}\,] 
= \ex[Y_i \given \bZ_{\hat{\theta},i}, T_i=t, \hat{\theta}\,] 
= g_{\hat{\theta}}(t,\bW_i).
$$ 
To bound the conditional variance, note first that Assumption~\ref{asm:samples} and Assumption~\ref{asm:CVarbound} imply
\begin{align*}
    \var\big(Y_i\given\bZ_{\hat{\theta},i}, T_i, \hat{\theta}\,\big)
    \leq \ex\big[Y_i^2\given \bZ_{\hat{\theta},i}, T_i, \hat{\theta}\,\big]
    = \ex[\ex\big[Y_i^2\given \bW_i, T_i, \hat{\theta}\,\big]
        \given \bZ_{\hat{\theta},i}, T_i, \hat{\theta}\,\big]
    \leq C,
\end{align*}
which in return implies that
\begin{align*}
    \var\left(
        r_i^2
        \given \bZ_{\hat{\bullet}}, T_\bullet, \hat{\theta}
    \right) 
    &= \one(T_i=t) \left(
        \frac{1}{\widehat{m}_{\hat{\theta}}(\bW_i)}
        -\frac{1}{m_{\hat{\theta}}(\bW_i)}
        \right)^2
        \var\left(Y_i\given\bZ_{\hat{\theta},i}, T_i, \hat{\theta}\,\right) \\
    &\leq C \left(
        \frac{1}{\widehat{m}_{\hat{\theta}}(\bW_i)}
        -\frac{1}{m_{\hat{\theta}}(\bW_i)}
        \right)^2.
\end{align*}
It now follows from the conditional Chebyshev's inequality that
\begin{align*}
    \mathbb{P}(|R_2|>\epsilon \given \bZ_{\hat{\bullet}}, T_\bullet, \hat{\theta}) 
        &\leq \epsilon^{-2}\var(R_2 \given \bZ_{\hat{\bullet}}, T_\bullet, \hat{\theta}) \\
        &\leq \frac{C}{|\mathcal{I}_3|}\sum_{i\in\mathcal{I}_3}\left(
            \frac{1}{\widehat{m}_{\hat{\theta}}(\bW_i)}
            -\frac{1}{m_{\hat{\theta}}(\bW_i)}
            \right)^2
        \leq C \ell_c^2  \mathcal{E}_{1,t}^{(n)}.
\end{align*}
Assumptions~\ref{asm:AIPWconv}~$(iv)$ implies convergence to zero in probability, and from this we conclude that also $R_2 \xrightarrow{P}0$. 

We now turn to the discussion of the oracle term $U_{\hat{\theta}}^{(n)}$.
From Assumption~\ref{asm:samples} it follows that $(T_i,\bW_i,Y_i)_{i\in\mathcal{I}_3}\ind \hat{\theta}$, and hence the conditional distribution $U_{\hat{\theta}}^{(n)} \given \hat{\theta}=\theta$ is the same as $U_\theta^{(n)}$. We show that this distribution is asymptotically Gaussian uniformly over $\theta\in \Theta$. To see this, note that for each $\theta\in \Theta$, the terms of $U_\theta^{(n)}$ are i.i.d. with mean $\chi_t(\bZ_\theta)$ and variance $\mathbb{V}_t(\bZ_\theta)$.
The conditional Jensen's inequality and Assumption~\ref{asm:AIPWconv}~$(i)$ imply that
\begin{align*}
    \ex[|u_i(\theta)|^{2+\delta}]
    &\leq C_{2+\delta}(\ex[|g_\theta(t,\bW)|^{2+\delta}]
        + c^{-(1+\delta)}\ex[|Y-g_\theta(t,\bW)|^{2+\delta}\,]) \\
    &\leq C' \ex[|Y|^{2+\delta}],
\end{align*}
where $C'>0$ is a constant depending on $c$ and $\delta$.
Thus we can apply the Lindeberg-Feller CLT, as stated in \citet[Lem. 18]{shah2020hardness}, to conclude that
\begin{equation*}
    S_n(\theta) \coloneq \sqrt{|\mathcal{I}_3|}\mathbb{V}_t(\bZ_\theta)^{-1/2}(U_{\theta}^{(n)} - \chi_t(\bZ_\theta)) 
        \xrightarrow{d/\Theta} \mathrm{N}(0,1),
\end{equation*}
where the convergence holds uniformly over $\theta\in \Theta$. With $\Phi(\cdot)$ denoting the CDF of a standard normal, this implies that
\begin{align}\label{eq:uniformimpliescondtional}
    \sup_{t\in\real}|\mathbb{P}(S_n(\hat{\theta})\leq t) - \Phi(t)|
    &= \sup_{t\in\real} \left\lvert\int (\mathbb{P}(S_n(\hat{\theta})\leq t \given \hat{\theta}=\theta)
    -\Phi(t))P_{\hat{\theta}}(\mathrm{d}\theta)\right\rvert 
    \nonumber \\
    &\leq \sup_{\theta\in\Theta} \sup_{t\in\real} |\mathbb{P}(S_n(\theta)\leq t)-\Phi(t)|
    \longrightarrow 0.
\end{align}
This shows that $S_n(\hat{\theta}) \xrightarrow{d} \mathrm{N}(0,1)$ as desired.
The last part of the theorem is a simple consequence of Slutsky's theorem. 
\hfill \qedsymbol{}


\subsection*{Proof of Theorem \ref{thm:varconsistent}}
For each $i\in \mathcal{I}_3$, let $\widehat{u}_i = u_i(\hat{\theta}) + r_i^1 + r_i^2 + r_i^3$
be the decomposition from \eqref{eq:DOPEdecomp}. 
For the squared sum, it is immediate from Theorem \ref{thm:conditionalAIPWconv} that
\begin{align*}
    \left(\frac{1}{|\mathcal{I}_3|}
        \sum_{i\in \mathcal{I}_3}\widehat{u}_{i}\right)^2 
    = \left(\frac{1}{|\mathcal{I}_3|}
        \sum_{i\in \mathcal{I}_3} u_{i}(\hat{\theta})\right)^2 + o_P(1).
\end{align*}
For the sum of squares, we first expand the squares as
\begin{align}\label{eq:uhatexpansion}
    \widehat{u}_i^2 = u_i(\hat{\theta})^2 + 2(r_i^1 + r_i^2 + r_i^3) u_i(\hat{\theta}) + (r_i^1 + r_i^2 + r_i^3)^2.
\end{align}
We show that the last two terms are convergent to zero in probability. For the cross-term, we note by Cauchy-Schwarz that
\begin{align*}
    \Big|\frac{1}{|\mathcal{I}_3|}
        \sum_{i\in \mathcal{I}_3} (r_i^1 + r_i^2 + r_i^3) u_i(\hat{\theta}) \Big|
    \leq 
    \left(\frac{1}{|\mathcal{I}_3|}
        \sum_{i\in \mathcal{I}_3} (r_i^1 + r_i^2 + r_i^3)^2\right)^{1/2}
    \left(\frac{1}{|\mathcal{I}_3|}
        \sum_{i\in \mathcal{I}_3} u_i(\hat{\theta})^2\right)^{1/2}.
\end{align*}
We show later that the sum $\frac{1}{|\mathcal{I}_3|}\sum_{i\in \mathcal{I}_3} u_i^2$ is convergent in probability. Thus, to show that the last two terms of \eqref{eq:uhatexpansion} converge to zero, it suffices to show that
$\frac{1}{|\mathcal{I}_3|}\sum_{i\in \mathcal{I}_3} 
(r_i^1 + r_i^2 + r_i^3)^2 \xrightarrow{P}0$. 
To this end, we observe that
\begin{align*}
    \frac{1}{|\mathcal{I}_3|}\sum_{i\in \mathcal{I}_3} (r_i^1 + r_i^2 + r_i^3)^2
    \leq 
    \frac{3}{|\mathcal{I}_3|}\sum_{i\in \mathcal{I}_3} (r_i^1)^2+(r_i^2)^2+(r_i^3)^2
\end{align*}
The last term is handled similarly to $R_3$ in the proof of Theorem~\ref{thm:conditionalAIPWconv},
\begin{align*}
    \frac{1}{|\mathcal{I}_3|}\sum_{i\in \mathcal{I}_3} (r_i^3)^2
    &= 
        \frac{1}{|\mathcal{I}_3|}\sum_{i\in \mathcal{I}_3}
        \left(\widehat{g}_{\hat{\theta}}(t,\bW_i)-g_{\hat{\theta}}(t,\bW_i)\right)^2
        \left(\frac{\one(T_i=t)}{\widehat{m}_{\hat{\theta}}(t \given \bW_i)}-\frac{\one(T_i=t)}{m_{\hat{\theta}}(t \given \bW_i)}\right)^2 \\
    &\leq 
        \frac{1}{|\mathcal{I}_3|}\Big(\sum_{i\in \mathcal{I}_3} (\widehat{g}_{\hat{\theta}}(t,\bW_i)-g_{\hat{\theta}}(t,\bW_i))^2\Big)
        \Big(\sum_{i\in \mathcal{I}_3} \Big(\frac{\one(T_i=t)}{\widehat{m}_{\hat{\theta}}(t \given \bW_i)}-\frac{\one(T_i=t)}{m_{\hat{\theta}}(t \given \bW_i)}\Big)^2\Big) \\
    &\leq  \ell_c^2 n\mathcal{E}_1^{(n)}\mathcal{E}_{2,t}^{(n)} \xrightarrow{P} 0,
\end{align*}
where we have used the naive inequality $\sum a_i^2b_i^2 \leq (\sum a_i^2)(\sum b_i^2)$ rather than Cauchy-Schwarz. From the proof of Theorem \ref{thm:conditionalAIPWconv}, we have that 
\begin{align*}
    \ex[(r_i^1)^2 \given \mathbf{Z}_{\hat{\bullet}}, \hat{\theta}] = \var[r_i^1 \given \mathbf{Z}_{\hat{\bullet}}, \hat{\theta}]\leq \frac{1-c}{c}\left(\widehat{g}_{\hat{\theta}}(t,\bW_i)-g_{\hat{\theta}}(t,\bW_i)\right)^2,
\end{align*}
and hence the conditional Markov's inequality yields
\begin{align*}
    \mathbb{P}\left(\frac{1}{|\mathcal{I}_3|}\sum_{i\in \mathcal{I}_3}(r_i^1)^2>\epsilon \given \mathbf{Z}_{\hat{\bullet}}, \hat{\theta}\right) 
    &\leq \epsilon^{-1}\ex\left[\frac{1}{|\mathcal{I}_3|}\sum_{i\in \mathcal{I}_3}(r_i^1)^2 \given \mathbf{Z}_{\hat{\bullet}}, \hat{\theta}\right] \\
    &= \epsilon^{-1}\frac{1}{n}\sum_{i=1}^n \ex\left[(r_i^1)^2 \given \mathbf{Z}_{\hat{\bullet}}, \hat{\theta}\right] \\
    &\leq \frac{1-c}{\epsilon c}\mathcal{E}_{2,t}^{(n)} \xrightarrow{P} 0.
\end{align*}
Thus we also conclude that $\frac{1}{|\mathcal{I}_3|}\sum_{i\in \mathcal{I}_3}(r_i^1)^2 \xrightarrow{P} 0$.
Analogously, the final remainder $\frac{1}{|\mathcal{I}_3|}\sum_{i\in \mathcal{I}_3} (r_i^2)^2$ can be shown to converge to zero in probability by leveraging the argument for $R_2\xrightarrow{P}0$ in the proof of Theorem \ref{thm:conditionalAIPWconv}.

Combining the arguments above we conclude that 
\begin{align*}
    \widehat{\mathcal{V}}_t(\widehat{g},\widehat{m}) 
    = \frac{1}{|\mathcal{I}_3|}\sum_{i\in \mathcal{I}_3} \widehat{u}_i^2 
        - \left(\frac{1}{|\mathcal{I}_3|}\sum_{i\in \mathcal{I}_3}\widehat{u}_{i}\right)^2 
    = \frac{1}{|\mathcal{I}_3|}\sum_{i\in \mathcal{I}_3} u_i(\hat{\theta})^2
        - \Big(\frac{1}{|\mathcal{I}_3|}
        \sum_{i\in \mathcal{I}_3}u_i(\hat{\theta})\Big)^2
        + o_P(1) 
\end{align*}
As noted in the proof of Theorem~\ref{thm:conditionalAIPWconv},
for each $\theta\in\Theta$ the terms $\{u_i(\theta)\}_{i\in \mathcal{I}_3}$ are i.i.d 
with $\ex[|u_i(\theta)|^{2+\delta}]\leq C'\ex[|Y|^{2+\delta}]<\infty$. 
Hence, the uniform law of large numbers, as stated in \citet[Lem. 19]{shah2020hardness}, implies that
\begin{align*}
    \frac{1}{|\mathcal{I}_3|}\sum_{i\in \mathcal{I}_3} u_i(\theta) 
    - \ex[u_1(\theta)]) \xrightarrow{P/\Theta} 0,
    \qquad
    \frac{1}{|\mathcal{I}_3|}\sum_{i\in \mathcal{I}_3} u_i^2(\theta) 
    - \ex[u_1(\theta)^2] \xrightarrow{P/\Theta} 0,
\end{align*}
where the convergence in probability holds uniformly over $\theta\in\Theta$.
Since $\var(u_1(\theta))=\var(\psi_t(\bZ_\theta))=\mathbb{V}_t(\bZ_\theta)$, this lets us conclude that
\begin{align*}
     S_n(\theta) \coloneq
     \frac{1}{|\mathcal{I}_3|}\sum_{i\in \mathcal{I}_3} {u}_i(\theta)^2
        - \Big(\frac{1}{|\mathcal{I}_3|}\sum_{i\in \mathcal{I}_3} u_{i}(\theta)\Big)^2 
        - \mathbb{V}_t(\bZ_\theta)
        \xrightarrow{P/\Theta} 0.
\end{align*}
We now use that convergence in distribution is equivalent to convergence in probability for deterministic limit variables, see \citet[Cor. B.4.]{christgau2023nonparametric} for a general uniform version of the statement. Since Assumption~\ref{asm:samples} implies that the conditional distribution $S_n(\hat{\theta})\given \hat{\theta}=\theta$ is the same as $S_n(\theta)$, the computation in \eqref{eq:uniformimpliescondtional} with $\Phi(\cdot)=\one(\cdot \geq 0)$ yields that $S_n(\hat{\theta})\xrightarrow{d}0$. 
This lets us conclude that
\begin{equation*}
    \widehat{\mathcal{V}}_t(\widehat{g},\widehat{m}) 
    = S_n(\hat{\theta}) + o_P(1) \xrightarrow{P} 0,
\end{equation*}
which finishes the proof.
\hfill \qedsymbol{}

\subsection{Proof of Theorem \ref{thm:SIregularity}}
We fix $P\in\mathcal{P}$ and suppress it from notation throughout the proof.
Given $\theta\in \real^d$, the single-index model assumption
\eqref{eq:singleindex} implies that
\begin{align*}
    b_t(\bW^\top \theta) 
    = \ex[Y \given T=t, \bW^\top \theta]
    = \ex[h_t(\bW^\top \theta_P) \given T=t, \bW^\top \theta].
\end{align*}
Now since $h_t \in C^1(\mathbb{R})$, we may write
\begin{align*}
    h_t(\bW^\top \theta) 
    &= h_t(\bW^\top \theta_P) + R(\theta,\theta_P),\\
    R(\theta,\theta_P) 
    &=
    h_t'(\bW^\top \theta_P) \bW^\top (\theta - \theta_P)
    + r(\bW^\top(\theta - \theta_P))\bW^\top (\theta - \theta_P)
\end{align*}
for a continuous function $r\colon \mathbb{R} \to \mathbb{R}$ satisfying that $\lim_{\epsilon \to 0}r(\epsilon) = 0$. It follows that
\begin{align*}
    \ex[h_t(\bW^\top \theta_P) \given T=t, \bW^\top \theta] 
    = h_t(\bW^\top \theta) 
    - \ex[R(\theta,\theta_P) \given T=t, \bW^\top \theta],
\end{align*}
The theoretical bias induced by adjusting for $\bW^\top\theta$ instead of $\bW$, or equivalently $\bW^\top \theta_P$, is therefore
\begin{align*}
    \chi_t(\bW^\top \theta;P) - \chi_t(\bW)  
    &= \chi_t(\bW^\top \theta;P) - \chi_t(\bW^\top \theta_P)\\
    &= \ex\left[h_t(\bW^\top \theta) 
        - \ex[R(\theta,\theta_P) \given T=t, \bW^\top \theta]
        - h_t(\bW^\top \theta_P)
        \right] \\
    &= \ex[
        R(\theta,\theta_P)
        -\ex[R(\theta,\theta_P)\given T=t,\bW^\top \theta]] \\
    &= 
    \ex\left[
        R(\theta,\theta_P) \cdot
        \left(
        1
        -
        \frac{\one(T=t)}{\mathbb{P}(T=t\given \bW^\top \theta)}
        \right)
    \right] \\
    &= 
        C(\theta,\theta_P)^\top (\theta - \theta_P),
\end{align*}
where 
\begin{align*}
    C(\theta;\theta_P) = 
    \ex\left[
        \left(
            h_t'(\bW^\top \theta)
            + r(\bW^\top(\theta - \theta_P))
        \right)
        \left(
        1
        -
        \frac{\one(T=t)}{\mathbb{P}(T=t\given \bW^\top \theta)}
        \right)
        \bW
    \right].
\end{align*}
To show that $u$ is differentiable at $\theta = \theta_P$, it therefore suffices to show that the mapping $\theta \mapsto C(\theta;\theta_P)$ is continuous at $\theta_P$. 

We first show that $\mathbb{P}(T=t \given \bW^\top \theta)$ is continuous at $\theta_P$ almost surely, which follows after applying a coordinate change such that $\theta$ becomes a basis vector. To be more precise, we may choose a neighborhood $\mathcal{U}\subset \real^d$ of $\theta_P$ and $d-1$ continuous functions 
\[
    q_2, \ldots, q_d \colon \mathcal{U}\longrightarrow \mathcal{S}^{d-1}
\]
such that $Q(\theta) \coloneq (\lVert\theta \rVert^{-1}\theta, q_2(\theta),\ldots q_d(\theta))$ is an orthogonal matrix for every $\theta\in \mathcal{U}$. For example, the first $d$ vectors of the Gram-Schmidt process applied to $(\theta,\mathbf{e}_1,\ldots, \mathbf{e}_d)$ is continuous and yields an orthonormal basis for $\theta\in \real^d \setminus \operatorname{span}(\mathbf{e}_d)$, so this works in case $\theta_P\neq \pm \mathbf{e}_d$. 
Let $Z(\theta) = Q(\theta)^\top \bW$ and note that $Z(\theta)_1 = \lVert\theta \rVert^{-1} \theta^\top \bW 
=\lVert\theta \rVert^{-1}\bW^\top \theta$.
Then by iterated expectations,
\begin{align*}
    &\mathbb{P}(T=t \given \bW^\top \theta) 
    = \ex[m(t\given \bW)\given \bW^\top \theta] \\
    &= \ex[m(t\given Q(\theta)Z(\theta))\given Z(\theta)_1] \\
    &= 
    \frac{
        \int_{\real^{d-1}} 
        m(t\given \widetilde{\bW}(\theta,z))
        p_{\bW}(\widetilde{\bW}(\theta,z))
        \mathrm{d}z
    }{
        \int_{\real^{d-1}}
        p_{\bW}(\widetilde{\bW}(\theta,z))
        \mathrm{d}z
    },
    \qquad 
    \widetilde{\bW}(\theta,z) \coloneq 
    Q(\theta)
    \begin{pmatrix}
        Z(\theta)_1 \\ z
    \end{pmatrix}
\end{align*}
Each integrand is bounded and continuous over $\theta\in \mathcal{U}$.
From dominated convergence it follows that $\mathbb{P}(T=t \given \bW^\top \theta)$ is continuous at $\theta_P$ almost surely.
It follows that almost surely,
\begin{align*}
    &\left(
            h_t'(\bW^\top \theta)
            + r(\bW^\top(\theta - \theta_P))
    \right)
    \left(
    1
    -
    \frac{\one(T=t)}{\mathbb{P}(T=t\given \bW^\top \theta)}
    \right) 
    \bW \\
    &\longrightarrow
    h_t'(\bW^\top \theta_P)
    \left(
    1
    -
    \frac{\one(T=t)}{\mathbb{P}(T=t\given \bW^\top \theta_P)}
    \right)\bW,
    \qquad \text{for }\theta \to \theta_P. 
\end{align*}
By dominated convergence again we conclude that
\begin{align*}
    C(\theta;\theta_P) \longrightarrow
    \ex\left[
     h_t'(\bW^\top \theta_P)
    \left(
    1
    -
    \frac{\one(T=t)}{\mathbb{P}(T=t\given \bW^\top \theta_P)}
    \right)\bW\right], \qquad \theta \to \theta_P.
\end{align*}
This shows that $\theta \mapsto C(\theta;\theta_P)$ is continuous at $\theta_P$, and hence we conclude that $u$ is differentiable at $\theta=\theta_P$ with gradient
\begin{align*}
    \nabla u (\theta_P) &=
    \ex\left[
     h_t'(\bW^\top \theta_P)
    \left(
    1
    -
    \frac{\one(T=t)}{\mathbb{P}(T=t\given \bW^\top \theta_P)}
    \right)\bW\right] \\
    &=
    \ex\left[
     h_t'(\bW^\top \theta_P)
    \left(
    1
    -
    \frac{\mathbb{P}(T=t\given \bW)}{\mathbb{P}(T=t\given \bW^\top \theta_P)}
    \right)\bW\right].
\end{align*}
\hfill \qedsymbol{}

\supplementarysection{Details of simulation study}\label{sup:simdetails}

Our experiments were conducted in Python \citep{van1995python}. 
The linear and logistic regression was imported from the \texttt{scikit-learn} package \citep{scikit-learn}, and the neural network for the DOPE-IDX was implemented using \texttt{pytorch} \citep{paszke2019pytorch}.

\begin{figure}
        \centering
        \includegraphics[width=0.6\linewidth]{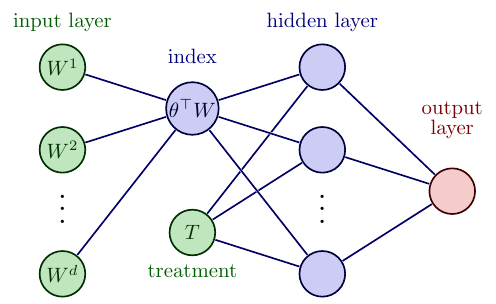}
        \caption{Neural network architecture for single-index model.}
        \label{fig:NNarchitecture}
\end{figure}
The neural network architecture is illustrated in Figure~\ref{fig:NNarchitecture}. 
The network was optimized using MSE loss and the ADAM optimizer with \texttt{lr=1e-3} and \texttt{n\_iter=1200} in the simulation experiment. For the NHANES application the settings were similar, but with BCELoss and \texttt{n\_iter=3000}.
We discovered that the optimization procedure could get stuck in a local minimum with significantly lower training score. To avoid getting stuck in a such a `bad local minimum', it was possible to refit the network with random initialization several times and pick the model with highest training score. The frequency of bad local minimums is roughly $15\%$, with variation between each setting. Thus, refitting the network, say 5 times, is enough to avoid a bad minimum $>99,99\%$ of the time. 
To save the computational cost of refitting many times for each simulation, we initialized the first layer with an initial value $\theta_{init} = \beta + \mathrm{Unif}(-0.1,0.1)^{\otimes d}$. 
The logistic regression for the propensity score was fitted without $\ell_2$-penalty and optimized using the \texttt{lbfgs} optimizer. The propensity score was clipped to the interval $(0.01,0.99)$ for all estimators of the adjusted mean.

\begin{figure}[htb!]
    \centering
    \includegraphics[width=\linewidth]{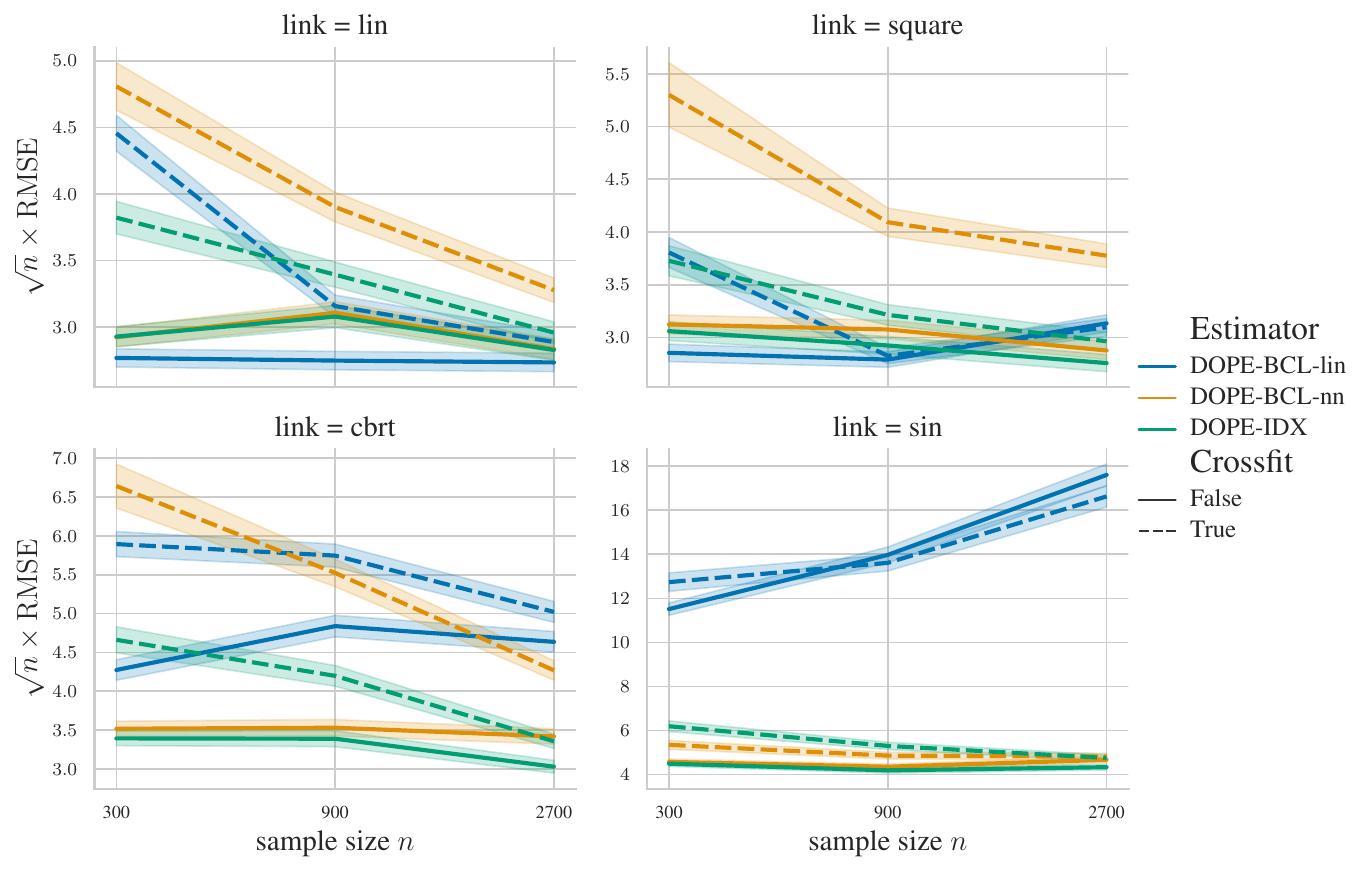}
    \caption{Comparison of DOPE estimators with and without sample splitting.}
    \label{fig:Cross-fitting}
\end{figure}
Figure~\ref{fig:Cross-fitting} shows the cross-fitted DOPE estimators versus their full sample counterparts in the simulation setup of Section~\ref{sec:simulation}. We observe that cross-fitting generally seems to decrease performance for small sample sizes, but the discrepancy diminishes for larger sample sizes.

\begin{table}[htb!]
    \centering
    \begin{tabular}{lrrr}
    \toprule
    Estimator &  Estimate &  BS se &            BS CI \\
    \midrule
    Regr. (NN)          &     0.020 &  0.008 &   (0.002, 0.035) \\
    Regr. (Logistic)    &     0.027 &  0.009 &   (0.012, 0.048) \\
    DOPE-BCL (Logistic) &     0.024 &  0.010 &    (0.004, 0.040) \\
    Naive contrast      &     0.388 &  0.010 &   (0.369, 0.407) \\
    DOPE-BCL (NN)       &     0.018 &  0.010 &  (-0.003, 0.036) \\
    DOPE-IDX (NN)       &     0.023 &  0.012 &  (-0.001, 0.047) \\
    AIPW (NN)           &     0.017 &  0.016 &  (-0.017, 0.046) \\
    AIPW (Logistic)     &     0.022 &  0.016 &  (-0.012, 0.051) \\
    IPW (Logistic)      &    -0.047 &  0.027 &  (-0.119, -0.01) \\
    \bottomrule \\
    \end{tabular}
    \caption{Treatment effect estimates for NHANES dataset, where covariates with more than 50\% missing values have been removed. }
    \label{tab:NHANES_pruned_version}
\end{table}

Table~\ref{tab:NHANES_pruned_version} corresponds to Table~\ref{tab:pulsepressure} in the main manuscript, but where covariates with more than $50\%$ missing data have been removed rather than imputed. Except for the naive contrast (which has dropped two rows) the ordering according to bootstrap variance is the same.

\begin{figure}
    \centering
    \includegraphics[width=\linewidth]{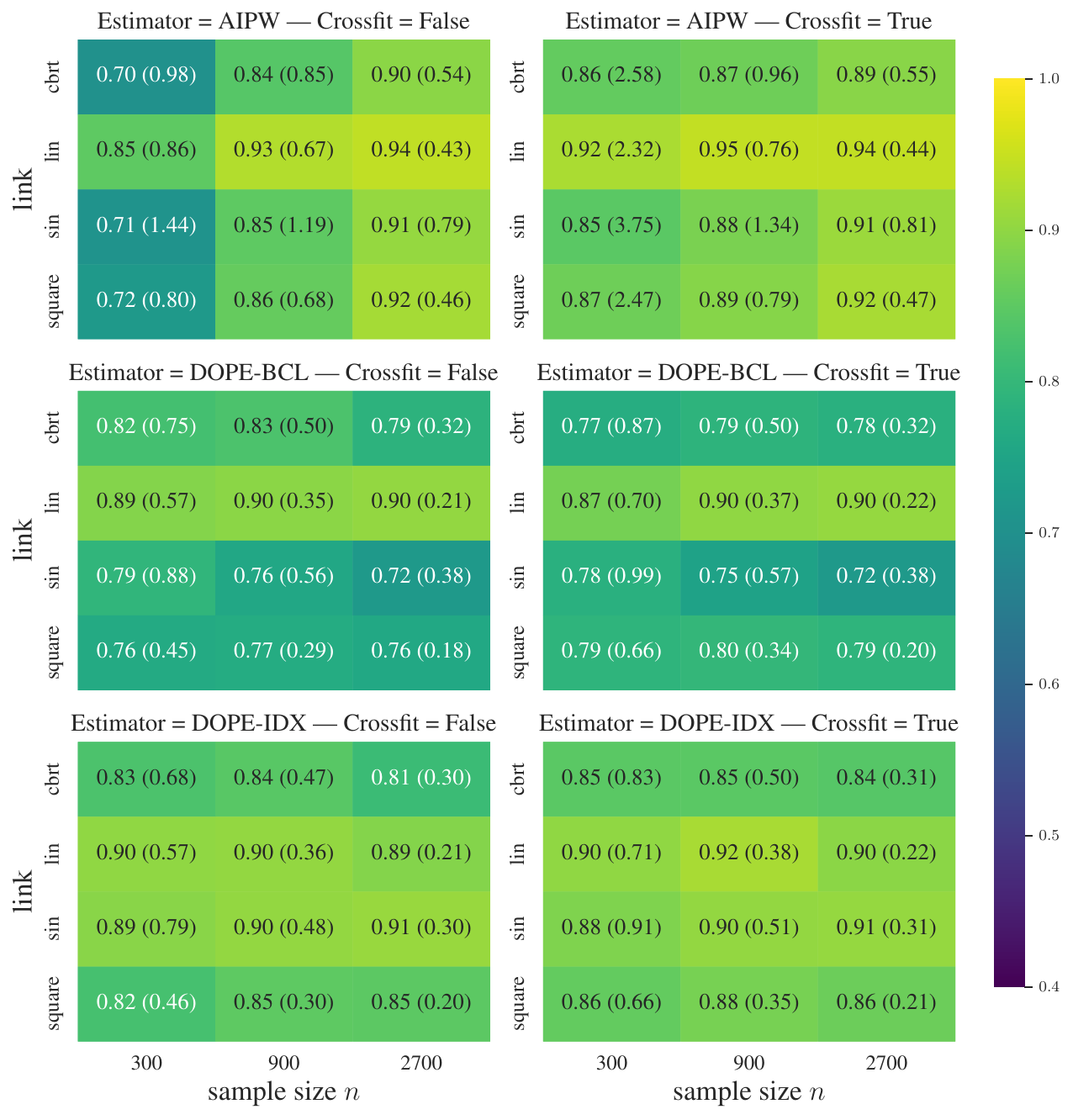}
    \caption{Coverage of asymptotic confidence intervals of the adjusted mean $\chi_1$, aggregated over $d\in\{4,12,36\}$ and with the true $\beta$ given in \eqref{eq:samplescheme}.
    The numbers in the parentheses denote the corresponding average width of the confidence interval.
    } 
    \label{fig:Coverage-with-length}
\end{figure}

Figure~\ref{fig:Coverage-with-length} shows the same coverage values as Figure~\ref{fig:Coverage-cf}, but also includes the average lengths of the confidence intervals in the parentheses. We note that the lengths of the confidence intervals for the DOPE estimators are much smaller than those of the AIPW estimator.

\supplementarysection{Extension to cross-fitting} \label{sup:crossfitting}
\begin{algorithm} \caption{Cross-fitted DOPE} \label{alg:crossfitted}
  \textbf{input}: observations $(T_i,\bW_i,Y_i)_{i\in[n]}$, partition $[n] = J_1\cup \cdots \cup J_K$\;
  \textbf{options}: integer $1\leq m\leq K-2$ and options for Algorithm \ref{alg:generalalg}\;
  \For{$k=1,\ldots,K$}{
    Set $\mathcal{I}_3=J_k$, $\mathcal{I}_1=\bigcup_{l=k+1}^{k+m}J_l$ and 
    $\mathcal{I}_2 = [n]\setminus (\mathcal{I}_1 \cup \mathcal{I}_3)$\;
    
    Compute $\widehat{\chi}_{t,k}^{\mathrm{dope}}$ as the output of Algorithm~\ref{alg:generalalg}\;
    
    Compute variance estimate $\widehat{\mathcal{V}}_{t,k}$ given as in \eqref{eq:varestimator}\;
    }
  \Return{$\widehat{\chi}_{t}^{\mathsf{x}} \coloneq \frac{1}{K} \sum_{k=1}^K \widehat{\chi}_{t,k}^{\mathrm{dope}}$ 
  and $\widehat{\mathcal{V}}_{t}^{\mathsf{x}} \coloneq \frac{1}{K} \sum_{k=1}^K\widehat{\mathcal{V}}_{t,k}$}
\end{algorithm}
A cross-fitting procedure for the DOPE is described in Algorithm~\ref{alg:crossfitted}, which computes both a cross-fitted version of the DOPE and its variance estimator. Here the indices of the folds are understood to cycle modulo $K$ such that $J_{K+1}=J_1$ and so forth. This version of cross-fitting with three index sets has also been referred to as `double cross-fitting' by \citet{zivich2021machine}.

We note that the `standard arguments' for establishing convergence of the cross-fitted estimator cannot be applied directly to our case. This is because, for each fold $k\in[K]$, the corresponding oracle terms $U_{\hat{\theta}_k}^{(n)}$ do not only involve the data indexed by $\mathcal{I}_3^k$, but also depend on $\hat{\theta}_k$ which is estimated from data indexed by $\mathcal{I}_1^k$. Hence the oracle terms are not independent. However, we believe that this dependency should be negligible, and perhaps the convergence can be established under a more refined theoretical analysis.

%% file: paper_acm/main_paper.tex
\begin{abstract}
We propose a novel and model-free estimand to measure the conditional association between a time-to-event and an exposure given covariates. 
The estimand may be interpreted as a weighted hazard difference, and can be viewed as an assumption-lean generalization of the cumulative exposure coefficient in the Aalen additive hazards model. We develop an estimation method based on the principles of double machine learning, which is algorithm-agnostic and incorporates cross-fitting. We prove that our method is doubly robust in the sense that, when the nuisance functions are learned with modest rate conditions, the resulting estimator converges to the true estimand at a $\sqrt{n}$-rate. A simulation study is conducted, showcasing that our estimation procedure yields an efficient and robust estimator of our estimand.
\end{abstract}

\section{Introduction}
Statistical inference for the conditional association between a time-to-event and an exposure is the primary objective in many bio-medical applications. A critical task for such applications is to first define an appropriate measure of association. Formulating a general and interpretable measure is intricate and remains an active research topic. 

The hazard ratio is one of the most popular and reported measures of association in clinical research, with its application being a subject of ongoing discussion since the influential work of \citet{cox1972regression}.
Despite its popularity, the hazard ratio, and the often accompanying Cox model, are subject to several significant drawbacks, as highlighted by \citet{hernan2010hazards,martinussen2013collaps,stensrud2020test}, among others. These drawbacks 
include non-collapsibility of the hazard ratio and that it is unsuitable for quantifying causal effects. The pitfalls of model misspecification is another notable drawback. When used as a conditional association measure within the framework of the 
Cox model, its interpretation is problematic without strong model assumptions such as proportionality.

To address the issue of model misspecification, \citet{vansteelandt2022assumption} propose a more general measure of association that reduces to the log (cumulative) hazard ratio under the Cox model. Generalizing a model-based measure of association with a model-free measure has now become known as assumption-lean inference \citep{berk2019assumption}. This approach subsumes model-based and estimand-based methods, and can be advantageous when each method has its own merits, see \cite{vansteelandt2022assumptionGLM} and the surrounding discussion for more details.

In this work we develop an analogous approach, but based on the hazard difference as an effect measure.
The hazard difference is collapsible for the \textit{additive hazards model}, introduced by Aalen \citep{aalen1980model,aalen1989linear}. As mentioned, this desirable property does not hold for the hazard ratio and the Cox model
\citep{martinussen2013collaps,daniel2021making}.
There are other appealing properties of the hazard difference, but the discussion is subtle, see, for example, \citet{didelez2022logic,martinussen2020subtleties}.  
Regardless, there are several applications where additive hazards are more plausible than proportional ones; we refer to \cite{Breslow1987statistical,lin1994semiparametric,KRAVDAL1997,Lee2019analysis,bischofberger2023smooth}.

\begin{table}[ht]
    \centering
    \resizebox{0.95\textwidth}{!}{
    \begin{tabular}{|p{4cm}|p{5cm}|p{5cm}|}
        \hline
        \textbf{Assumptions on the non-baseline term} 
        & \textbf{Ratio-based target}  & \textbf{Difference-based target} \\ \hline
        Parametric & Cox model \citep{cox1972regression} & Additive risk model in \cite{lin1994semiparametric} \\ \hline
        Semi-parametric & Partial Cox model \citep{sasieni1992information,huang1999efficient,zhong2022deep}
        & Partially additive hazards \citep{mckeague1994partly,dukes2019doubly,hou2023treatment} \\ \hline
         Assumption-lean  & \citep{vansteelandt2022assumption} 
        & This work \\ \hline
    \end{tabular}
    }\vspace{6pt}
    \caption{Overview of selected works on inference for the conditional association between a time-to-event and an exposure, defined using varying levels of model assumptions. All models include a nonparametric baseline term that does not depend on covariates.
    }
    \label{tab:overview}
\end{table}

Despite some of their well-behaved theoretical properties, additive hazards are ``somewhat overlooked in practice'' according to \citet[p. 103]{martinussen2006dynamic}. 
A common criticism of additive hazard models is that they can lead to hazard estimates that are negative and are thus not valid hazards. 
To circumvent this issue, we define a model-free target parameter based on the \textit{Local Covariance Measure (LCM)} introduced by \citet{christgau2023nonparametric}. We show that, within a partially additive hazards model, it coincides with the cumulative effect of the exposure. 
However, our target does not rely on such modeling assumptions and we discuss its interpretations in more general settings. 
Table~\ref{tab:overview} gives an overview of our target parameter in relation to some existing works.
We show how the target may be estimated using flexible machine learning methods, and prove asymptotic normality under suitable rate conditions on the nuisance parameters.

\subsection{A hazard difference estimand} To motivate our general framework 
and results we first present a simplified example. 

\begin{example} \label{ex:simple}
Let $T \in [0, 1]$ denote a survival 
time, let $X \in \{0, 1\}$ denote a binary baseline exposure and let $Z$ denote 
additional baseline covariates. With $\bH_t(X, Z)$ denoting the conditional cumulative 
hazard given $(X, Z)$, the conditional survival function is 
$\mathbb{P}(T > t \given X, Z) = e^{- \bH_t(X, Z)}$, and 
\begin{equation}
    \label{eq:log-surv-ratio}
    \Theta_t(Z) \coloneq \log \frac{\mathbb{P}(T > t \given X = 0, Z)}{\mathbb{P}(T > t \given X = 1, Z)} = \bH_t(1, Z) - \bH_t(0, Z)
\end{equation}
quantifies the, possibly heterogeneous, exposure effect as a function of $Z$. 
Suppose that 
\[
    \bH_t(X, Z) = \int_0^t \bh_s(X, Z) \mathrm{d} s
\]
is given in terms of the conditional hazard $\bh_s(X, Z)$, then 
\[
\Theta_t(Z) = \int_0^t  \underbrace{\bh_s(1, Z) - \bh_s(0, Z)}_{
    \eqcolon \theta_s(Z)} \mathrm{d} s
= \int_0^t \theta_s(Z) \mathrm{d} s
\]
is the integrated difference of conditional hazards. We can always express the 
conditional hazard as 
\begin{equation}
    \label{eq:cond-hazard-baseline}
    \bh_s(X, Z) = X \theta_s(Z) + \bh_s(0, Z),
\end{equation}
and the additive decomposition \eqref{eq:cond-hazard-baseline} does not impose any model restrictions unless we restrict one of the functions $\theta_s(Z)$ or $\bh_s(0, Z)$.

We seek an estimand the summarizes the values of $\theta_s(Z)$ over $s$ and $Z$. 
To this end, let $w_s(Z)$ denote a \emph{weight process} with 
$w_s(Z) \geq 0$ and $\int_0^1 \ex[w_s(Z)] \mathrm{d} s = 1$, 
and consider the \emph{weighted hazard difference estimand}
\[
    \gamma^w_1 \coloneq \int_0^1 \ex[ \theta_{s}(Z) w_{s}(Z)] \mathrm{d} s.
\]
The real-valued parameter $\gamma^w_1$ is the endpoint of 
the functional estimand 
\[
    \gamma^w_t \coloneq \int_0^t \ex[ \theta_{s}(Z) w_{s}(Z)] \mathrm{d} s, \quad t \in [0, 1].
\]
For $w_s(Z) = 1$, the estimand simplifies to the expected cumulative hazard difference $\gamma^1_t = \ex( \Theta_{t}(Z))$, but other weights 
are also of interest. For the \emph{partially additive hazards model}, where 
$\theta_s$ is independent of $Z$, we see that $\gamma_t^w = \int_0^t \theta_s \ex[w_s(Z)] \mathrm{d} s$ is a cumulative weighted hazard difference.
Efficient and (rate) doubly robust estimation was considered for the partially 
additive model in \cite{dukes2019doubly} and \cite{hou2023treatment} using semiparametric estimation theory  -- in particular when $\theta_t = \theta$ is independent of $t$ and 
$\gamma_t^w = \theta t$. In this paper we introduce an assumption-lean 
additive hazard estimand and corresponding estimation theory that, as a special case, applies to model \eqref{eq:cond-hazard-baseline} without any model restrictions.

To give an explicit formula for our assumption-lean estimand for model \eqref{eq:cond-hazard-baseline} within this simplified example,
let $Y_t = \one(T\geq t)$ denote the at-risk indicator, and define what we call 
the \emph{Aalen covariance} weights
\begin{equation}
    \label{eq:AC-weights}
    w_s^{\textsc{ac}}(Z) 
    \coloneq \frac{Y_s \pi_s(Z) (1 - \pi_s(Z)) }{\ex[Y_s \pi_s(Z) (1 - \pi_s(Z))]}, 
\end{equation}
where $\pi_s(Z) = \mathbb{P}(X = 1 \given T \geq s, Z)$. The resulting estimand can be expressed as
\begin{equation} 
    \label{eq:ACM-baseline}
     \gamma_t = \int_0^t \ex\left[\theta_s(Z) w_s^{\textsc{ac}}(Z) \right] \mathrm{d} s
     = \int_0^t \frac{\ex\left[\theta_s(Z)Y_s \pi_s(Z) (1 - \pi_s(Z))\right]}{\ex[Y_s \pi_s(Z) (1 - \pi_s(Z))]} \mathrm{d}s.
\end{equation}
The weights \eqref{eq:AC-weights} and the formula \eqref{eq:ACM-baseline} may at first appear 
unmotivated. In Section~3 in \cite{vansteelandt2022assumptionGLM} 
we find weights similar to \eqref{eq:AC-weights} used to represent assumption-lean GLM regression estimands, and also in Section 2 in \cite{vansteelandt2022assumption} similar 
weights appear in the representation of assumption-lean Cox regression estimands. 
\end{example}

In Example \ref{ex:simple_cont} we show that the Aalen covariance weights and the 
estimand \eqref{eq:ACM-baseline} appear naturally as an assumption-lean Aalen regression estimand.
The estimand $\gamma_t$ given by \eqref{eq:ACM-baseline} is an example 
of the \emph{Aalen covariance measure} (ACM) -- the general estimand we introduce in Definition \ref{dfn:ACM}.
This is, in turn, an example of a \emph{Local Covariance Measure} as introduced by \citet{christgau2023nonparametric}. 
We show in this paper that ACM can be defined and interpreted as an 
assumption-lean estimand for any real-valued and time-dependent $X$-process and time-dependent $Z$-process, thus generalizing model \eqref{eq:cond-hazard-baseline} considerably.

\subsection{Organization of paper}
We organize the rest of the paper as follows.
In Section~\ref{sec:acmsetup}, we first introduce our general setup and define our target estimand. We then discuss interpretations of the estimand based on theoretical properties and computations in concrete examples. 
Conditions on the censoring mechanism are postponed to the end of the section to avoid a technical discussion from the onset.
In Section~\ref{sec:acmestimation}, we discuss an estimation methods for our target estimand, in particular we propose the \textsc{x-acm}, which is based on double machine learning estimator.
In Section~\ref{sec:acmasymptotics}, we prove that the \textsc{x-acm} is asymptotically Gaussian with a $\sqrt{n}$-rate under suitable conditions. The most significant conditions are rate-conditions on nuisance estimators, which are similar to those for other (rate) doubly robust estimators. 
In Section~\ref{sec:acmsimulations}, conduct a simulation study that demonstrates the effect of the \textsc{x-acm}. We conclude with a discussion
in Section~\ref{sec:acmdiscuss} outlining potential future research directions.

The supplementary material consists of Section~\ref{sec:acmproofs}, which contains proofs of all our main results and
related auxiliary lemmas, and Section~\ref{app:multhaz}, which includes a discussion of our framework in the context of multiplicative hazards models.

\section{Target estimand and interpretations}\label{sec:acmsetup}

\subsection{General setup}
We consider a survival time $T^*>0$ censored by a random variable $C\in [0,1]$, such that the observed time is $T \coloneq T^*\wedge C$ together with the indicator $\Delta \coloneq \one(T^*\leq C)$. From these we define the observed counting process $N_t = \one(T\leq t, \Delta = 1)$, and the at-risk indicator $Y_t = \one(T\geq t)$. 
In addition to $(T,\Delta)$, we also observe a covariate process $Z = (Z_t)$ and an exposure process $X=(X_t)$ of particular interest. The covariate process can, in principle, take values in any measurable space, whereas the exposure process is assumed real-valued and left-continuous. 

We postpone an elaborate discussion of assumptions on the censoring mechanism to Section~\ref{sec:acmcensoring}. In summary, we introduce Assumption~\ref{asm:censoring}, which is an instance of \emph{independent right-censoring} in the sense of \citet{AndersenBorganGillKeiding:1993}, and which will remain an assumption throughout the remainder of the manuscript.

For any stochastic process $W = (W_t)$, we let $\ol{W}_t = (W_s)_{s<t}$ denote the history up to time $t>0$, but not including the timepoint 
$t$.\footnote{This convention differs from other works such as \citep{lok2017mimicking,lok2008}, which do include the timepoint $t$ in the history. For any left-continuous process, this convention does not affect the information content of the history. In particular, there is no distinction between the conventions when conditioning on $\oX_t$, see also Remark \ref{rem:filt}}
Note that by left-continuity, $X_t = X_{t-} \coloneq \lim_{s\to t^-} X_s$ is determined by the history $\oX_t$.

The conditional cumulative hazard $\bH_t$, given the processes $X$ and $Z$, is defined by
\begin{align*}
    \bH_t(\oX_t,\oZ_t) \coloneq - \log \mathbb{P}(T^* > t \given  \oX_t,\oZ_t).
\end{align*}
The overall objective is to quantify how $\bH_t(\oX_t,\oZ_t)$ depends on the 
exposure process $\oX_t$. Such a dependence can, in general, be complicated and depend on 
the value of the covariate process $\oZ_t$. Our approach is \emph{assumption-lean} meaning that 
we do not make strong model assumptions on how $\bH_t$ depends on $\oX_t$ and $\oZ_t$. We will 
assume, though, that $\bH_t(\oX_t,\oZ_t)$ is given by a conditional hazard function, denoted by $\bh_t$, 
that is, $\bH_t(\oX_t,\oZ_t)=\int_0^t \bh_s(\oX_s,\oZ_s) \mathrm{d}s$.

\subsection{Aalen Covariance Measure}\label{sec:acmACMdefinition}
Before we define our model-free estimand in the most general case, we need to introduce two auxiliary stochastic processes:
the predictable projection of $X$ and the $Z$-compensated counting process.

A stochastic process is said to be \textit{\lc{}} if its sample paths are left-continuous and have right-limits at all times. We henceforth assume that $X$ is \lc{} and square-integrable.
Then, the \textit{predictable projection} of $X$ onto the histories of $N$ and $Z$ can be defined as the unique\footnote{Technically, it is unique up to \textit{evanescence}.} \lc{} process, denoted by $\Pi=(\Pi_t)$, satisfying that
\begin{equation}\label{eq:predictableprojection}
    \Pi_t = \ex[X_t \given \ol{N}_t, \oZ_t].
\end{equation}
For more details, see for example Corollary 7.6.8 in \citet{cohen2015stochastic}.
We will primarily need the process $\Pi_t$ on the event $(Y_t=1)=(T\geq t)$, where it takes the value $\pi_t(\oZ_t)\coloneq \ex[X_t \given T\geq t, \oZ_t]$.\footnote{We technically define $\pi_t$ via the Doob-Dynkin lemma such that $Y_t\Pi_t = Y_t \pi_t(\oZ_t)$ holds \textit{surely}. Thus we can ensure that $\pi_t(\oZ_t)$ is a \lc{} version of $\ex[X_t \given T\geq t, \oZ_t]$.}


Introducing 
\[
    h_t(\oZ_s) = \ex[\bh_s(\oX_s,\oZ_s) \given T \geq t, \oZ_t], 
\] 
the \textit{innovation theorem}, see, e.g., Theorem II.T14 in \cite{Bremaud:1981},
gives that $h_t$ is the conditional hazard given $Z$ that satisfies the relation
\begin{align*}
 H_t(\oZ_t) \coloneq \int_0^t h_s(\oZ_s)\mathrm{d}s = - \log \mathbb{P}(T^* > t \given \oZ_t).    
\end{align*}

In terms of the conditional hazard given $Z$ we define the $Z$-compensated counting process 
$$
    M_t 
    \coloneq N_t - \int_0^tY_sh_s(\oZ_s)\mathrm{d}s
    = N_t - H_{t\wedge T}(\oZ_{t\wedge T}).
$$
We can now define our target estimand.
\begin{definition}[Aalen Covariance Measure] \label{dfn:ACM}
    For each $t\in [0,1]$ let 
    \begin{align}\label{eq:rho}
      \rho(t) \coloneq \ex[Y_t(X_t-\Pi_t)^2] = \ex[Y_t(X_t-\pi_t(\oZ_t))^2],
    \end{align}
    and define, whenever $\rho(t) > 0$, the     \textit{residual process} $G = (G_t)$ by
    \begin{align}\label{eq:residualprocess}
    G_t \coloneq \frac{Y_t(X_t - \Pi_t)}{\rho(t)}
    =\frac{Y_t(X_t - \pi_t(\oZ_t))}{\rho(t)}.
    \end{align}
    If $\rho(t)>0$ for all $t\in[0,1]$, the \emph{Aalen Covariance Measure (ACM)} is defined as the functional estimand $\gamma = (\gamma_t)$ given by
    \begin{equation}\label{eq:LCMdef}
        \gamma_t = \ex\left[\int_0^t G_s \mathrm{d}M_s\right]
    \end{equation}
    for $t\in [0,1]$.
\end{definition}

We discuss the positivity assumption $\rho(t) > 0$ further in Remark~\ref{rmk:boundedassumptions}. 
To appreciate Definition \ref{dfn:ACM} we compute the ACM for the partially additive hazards model.

\begin{example}[Partially additive hazards model] \label{ex:pahm}
Suppose that 
\begin{equation}\label{eq:additivemodel}
    \bh_t(\oX_t,\oZ_t) = \theta_t X_t + g_t(\oZ_{t}),
\end{equation}
for some $\theta_t\in \real$ and a function $g_t\in L^1(P_{\oZ_t})$, where $P_{\oZ_t}$ denotes the distribution of~$\oZ_t$. The model given by \eqref{eq:additivemodel} is the \emph{partially additive hazards model}. It specifies the direct effect of $X$ in a separate term with time-varying coefficient $\theta_t$. For this model 
the innovation theorem yields the explicit expression
\begin{equation}\label{eq:innovationinadditivemodel}
    h_t(\oZ_t) = \theta_t\pi_t(\oZ_t) + g_t(\oZ_t).
\end{equation}

Introducing the $(X, Z)$-compensated counting process, $\bM_t \coloneq N_t - \bH_{t\wedge T}(\oX_{t\wedge T}, \oZ_{t\wedge T})$
we get by \eqref{eq:innovationinadditivemodel} that
\[
    M_t = \bM_t + \int_0^t \theta_s(X_s - \pi_s(\oZ_s)) \mathrm{d}s.
\]  
Since $\bM_t$ is a mean zero martingale and the residual process is predictable, 
\begin{align} \nonumber
    \gamma_t & =  \ex\left[ \int_0^t G_s \mathrm{d} \bM_s + \int_0^t G_s \theta_s (X_s - \pi_s(\oZ_s)) \mathrm{d}s \right] \\
    & = \ex\left[ \int_0^t  \theta_s  \frac{Y_s(X_s - \pi_s(\oZ_s))^2}{\rho(s)} \mathrm{d}s \right] 
    =  \int_0^t  \theta_s \mathrm{d}s. \label{eq:cumulativeeffect}
\end{align}
Hence, the ACM reduces to the cumulative direct effect $\Theta_t = \int_0^t \theta_s \mathrm{d}s$ for the 
partially additive hazards model.
\end{example}

We note that both \citet{dukes2019doubly,hou2023treatment} consider models that fall under the partially additive model given by \eqref{eq:additivemodel}, and they propose estimation of the cumulative effect $\Theta_t$ with doubly robust estimation methods. 
The example above shows that our ACM estimand reduces to the cumulative effect for the partially additive model, but we will not make this model assumption. We present additional represetations and interpretations of the ACM in Section \ref{sec:acminter}

\begin{remark} \label{rem:filt}
    The ACM is an example of a \textit{Local Covariance Measure} (LCM), which was introduced by 
    \citet{christgau2023nonparametric} in a more general setting of counting processes, where all definitions and results 
    are formulated in terms of filtrations and intensities. Since this work is focused on survival analysis, the results are formulated using histories and conditional hazards to make them more accessible. To describe the connection between 
    the notation used in this paper and in  \cite{christgau2023nonparametric}, introduce the two filtrations 
    $\cF_t \coloneq \sigma(Z_s, N_s; s \leq t)$ and $\cG_t \coloneq \sigma(X_s,Z_s,N_s; s \leq t)$. Then 
    $\cG_{t-} = \sigma(\ol{X}_t, \ol{Z}_t, \ol{N}_t)$ and $\cF_{t-} = \sigma(\ol{Z}_t, \ol{N}_t)$, and 
    $\lambda_t \coloneq Y_t h_t(\oZ_t)$ and $\blambda_t \coloneq Y_t \bh_t(\oX_t,\oZ_t)$ are the
    $\cF_t$- and $\cG_t$-predictable intensities, respectively, for the counting process $N$.
    The martingale argument in Example \ref{ex:pahm} uses that $\bM_t$ is a $\cG_t$-martingale, which relies on the implicit Assumption~\ref{asm:censoring} on the censoring mechanism.
    The supplementary Section~\ref{sec:acmproofs} elaborates on the notation, 
    which is also used in the proofs. 

    In \cite{christgau2023nonparametric}, the LCM was introduced for any $\cG_t$-predictable 
    residual process $G_t$ satisfying $\ex[G_t\given \cF_{t-}] = 0$. The residual process given 
    by \eqref{eq:residualprocess} is a specific example of such a process. The LCM was introduced 
    in \cite{christgau2023nonparametric} to quantify deviations from the hypothesis that $N$ is 
    \emph{locally independent} of $X$ given $\cF_t$. The ACM, being an example of a LCM, 
    enjoys the same interpretation: under local independence,  $\gamma$ is equal to the zero function. 
    The ACM was, however, not considered in \cite{christgau2023nonparametric}, and we proceed to show that the ACM enjoys the interpretation as an effect measure generalizing the cumulative effect in a partially additive hazards model.    
\end{remark}

\subsection{Interpretations of the ACM as an effect measure} \label{sec:acminter}
By Proposition 2.6 in \cite{christgau2023nonparametric}, the 
LCM, and therefore also the ACM, has the following equivalent representation for each $t\in[0,1]$,
\begin{align}\label{eq:LCMisLCM}
    \gamma_t & = \int_0^t \cov\big(G_s,Y_s\big(\bh_s(\oX_s,\oZ_s)-h_s(\oZ_s)\big)\big) \mathrm{d}s \\
    & = \int_0^t \ex\left[G_s \big(\bh_s(\oX_s,\oZ_s)-h_s(\oZ_s)\big)\right] \mathrm{d}s, \label{eq:LCMisLCM2}
\end{align}
where \eqref{eq:LCMisLCM2} follows from $G_s$ having mean zero and $Y_s \in \{0,1\}$.
Example \ref{ex:pahm} showed that the ACM equals the cumulative direct effect under the partially additive hazards model. The following proposition asserts that in general -- when $\eqref{eq:additivemodel}$ may not hold -- the ACM equals the cumulated coefficient in the $L^2$-projection of the conditional hazard $\bh_t$ onto a partially additive model.

%
\begin{prop}\label{prop:LCMisL2projection}
    Consider the minimization problem
    \begin{align}\label{eq:projectionobjective}
        \minimize_{\vartheta, g} \!:& \quad 
            \ex\left[\int_0^1 Y_t\big(\bh_t(\oX_t,\oZ_t) 
                - \vartheta_t X_t - g_t(\oZ_t)\big)^2\mathrm{d}t
                \right]
    \end{align}
    over all measurable functions $\vartheta \colon [0,1]\to \real$ and collections of functions $(g_t)_{t\in[0,1]}$ such that $g_t(\oZ_t)$ is progressively measurable.
    Let $\vartheta^\star=(\vartheta_t^\star)$ be given by 
    $$
    \vartheta_t^\star \coloneq \cov\!\big(G_t,Y_t\big(\bh_t(\oX_t,\oZ_t)-h_t(\oZ_t)\big)\big)
    $$ 
    for $t\in[0,1]$.
    
    Then there exists $g^\star$ such that $(\vartheta^\star, g^\star)$ is a solution to \eqref{eq:projectionobjective}, and $\vartheta^\star$ is unique up to modifications on a null set of $[0,1]$.
    In particular, the ACM can be expressed as 
    $$
        \gamma_t=\int_0^t\vartheta_s\mathrm{d}s, \qquad t\in[0,1],
    $$ 
    for any solution $(\vartheta,g)$ to \eqref{eq:projectionobjective}.
\end{prop}

All proofs are deferred to the supplement.

\begin{example} \label{ex:pahm_cont}
Continuing Example \ref{ex:pahm} we see that for the partially additive hazards model 
given by \eqref{eq:additivemodel}, the minimizer given by Proposition \ref{prop:LCMisL2projection}
is 
\begin{align*}
    \vartheta^\star_t & = \cov\!\big(G_t,Y_t\big(\bh_t(\oX_t,\oZ_t)-h_t(\oZ_t)\big)\big) \\
    & = \cov\!\big(G_t,\theta_t Y_t\big(X_s - \pi_t(\ol{Z}_t)\big)\big) \\
    & = \frac{\theta_t \ex[Y_t(X_t -\pi_t(\ol{Z}_t))^2]}{\rho(s)} = \theta_t.
\end{align*}
Thus, as also established directly in Example \ref{ex:pahm}, for the partially additive hazards model, 
$\gamma_t = \Theta_t = \int_0^t \theta_s \mathrm{d} s$.
\end{example}

\begin{example} \label{ex:simple_cont}
As a generalization of the introductory Example \ref{ex:simple}, suppose that $X_t \in \{0, 1\}$
and that 
\[
    \bh_t(\oX_t,\oZ_t) = X_t \theta_t(\oZ_t) + g_t(\oZ_t).
\]
If $\theta_t(\oZ_t)$ does not depend on $\oZ_t$, this model is also a partially additive hazards model.
In general, we see that 
\[
    h_t(\oZ_t) = \pi_t(\oZ_t) \theta_t(\oZ_t) + g_t(\oZ_t)
\]
where $\pi_t(\oZ_t) = \mathbb{P}(X_t = 1 \given T \geq t, \oZ_t)$. Hence,
\begin{align*}
    \rho(t) & = \ex\left[Y_t(X_t - \pi_t(\oZ_t))^2 \right] \\
    & = \ex\left[Y_t \ex\left[(X_t - \pi_t(\oZ_t))^2 \mid T \geq 1, \oZ_t \right]  \right] \\
    & = \ex\left[Y_t \pi_t(\oZ_t)(1 - \pi_t(\oZ_t)) \right].
\end{align*}
The minimizer given by Proposition \ref{prop:LCMisL2projection} is   
\begin{align*}
    \vartheta^\star_t & = \cov\!\big(G_t,Y_t\big(\bh_t(\oX_t, \oZ_t)-h_t(\oZ_t)\big)\big) \\
    & = \cov\!\big(G_t,\theta_t(\oZ_t) Y_t \big(X_t - \pi_t(\oZ_t)\big)\big) \\
    & = \ex\left[ \theta_t(\oZ_t) \frac{Y_t \ex\left[  \big(X_t - \pi_t(\oZ_t)\big)^2  \given T \geq t, Z \right]}{\rho(t)} \right] \\
    & = \ex\left[ \theta_t(\oZ_t) \frac{Y_t \pi_t(\oZ_t)(1 - \pi_t(\oZ_t))} {\ex\left[Y_t \pi_t(\oZ_t)(1 - \pi_t(\oZ_t)) \right]} \right] \\
    & =  \ex\left[ \theta_t(\oZ_t) w_t^{\textsc{ac}}(\oZ_t) \right],
\end{align*}
where the Aalen covariance weights are 
\begin{equation}
    w_t^{\textsc{ac}}(Z) = \frac{Y_t \pi_t(\oZ_t)(1 - \pi_t(\oZ_t))} {\ex\left[Y_t \pi_t(\oZ_t)(1 - \pi_t(\oZ_t)) \right]}. 
\end{equation}
In the special case with baseline exposure and baseline covariates, as treated 
in Example~\ref{ex:simple}, these weights reduce to the weights in \eqref{eq:AC-weights}, 
and the estimand given by \eqref{eq:ACM-baseline} is, indeed, the ACM.

To give a more specific example of a model that is 
\emph{not} a partially additive model, suppose 
that $Z_t \in \real^d$ and that $\theta_t(\oZ_t) = \theta^T Z_t$ for $\theta \in \real^d$. Then 
    \[
        \gamma_t^w = \theta^T \int_0^t \ex[Z_s w_s^{\textsc{ac}}(\oZ_s)] \mathrm{d} s
        = \theta^T \omega_t^{\textsc{ac}} 
        = \sum_{k=1}^d \theta_k \omega_{k, t}^{\textsc{ac}}
    \]
    is a linear combination of the $\theta_k$-parameters with  
    $\omega_{k, t}^{\textsc{ac}} = \int_0^t  \ex\left[ Z_{k,s} w_s^{\textsc{ac}}(\oZ_s) \right] \mathrm{d} s$.
\end{example}

Proposition~\ref{prop:LCMisL2projection} allows us to interpret the ACM as a cumulative direct effect in the best $L^2$ model-approximation by a partially additive hazards model.
This underlines that the ACM is well-defined and retains some level of interpretability for (small) deviations from the partially additive model. In contrast, the model-based target parameter in \eqref{eq:cumulativeeffect} is \textit{a priori} undefined when the partially additive model is misspecified. 

    

Beyond the partially additive model, as treated in Examples \ref{ex:pahm} and \ref{ex:pahm_cont}, 
we cannot expect the ACM to have a simple closed-form expression. Even for 
the simple baseline model \eqref{eq:cond-hazard-baseline}, the innocent looking representation 
\[
    \gamma_t = \int_0^t \ex\left[ \theta_s(Z) w_s^{\textsc{ac}}(Z) \right] \mathrm{d} s,
\]
that follows from Example \ref{ex:simple_cont}, involves 
the weights and in turn the predictable projections $\pi_t(Z)$. We elaborate on the computation of 
$\pi_t(Z)$ for a multiplicative (baseline) hazards model in supplementary Section~\ref{app:multhaz}, and we discuss the 
representation of the ACM further for the Cox proportional hazards model in 
Example~\ref{ex:LCMinCox}.

%

\subsection{Censoring}\label{sec:acmcensoring}
Let $N_t^*\coloneq \one(T^* \leq t)$ denote the uncensored counting process,  and let $N_t^{\mathrm{c}} = \one(C\leq t)$ denote the counting process for censoring. Recall the definitions of the filtrations $\cF_t$ and $\cG_t$ in Remark \ref{rem:filt},
and define similarly the filtrations $\cF_t^* \coloneq \sigma(Z_s, N_s^*; s \leq t)$ and 
$\cG_t^* \coloneq \sigma(X_s,Z_s,N_s^*; s \leq t)$.

We assume two conditions on the censoring mechanism, both are forms of \emph{independent right-censoring} in the sense of \citet{AndersenBorganGillKeiding:1993}. We refer to \cite{martinussen2006dynamic,AndersenBorganGillKeiding:1993} for interpretations and intuition behind this type of censoring. 
Formulated in terms of local independence as in \citep{roysland2022graphical}, we specifically assume that $N_t^*$ is locally independent of $N_t^{\mathrm{c}}$ given both $\cF_t$ and $\cG_t$. 
Alternatively, the censoring condition may be formulated as in Assumption~\ref{asm:censoring} below.

An equivalent definition of the conditional hazards, $h_s$ and $\bh_s$, is that
\begin{equation*}
    N_t^* - \int_0^t \one(T^*\geq s) h_s(\oZ_s) \mathrm{d}s
    \qquad \text{and} \qquad
    N_t^* - \int_0^t \one(T^*\geq s) \bh_s(\oX_t,\oZ_s) \mathrm{d}s
\end{equation*}
are martingales with respect to the filtrations $\cF_t^*$ and $\cG_t^*$, respectively. Independent right-censoring means that the conditional hazards also yield valid compensators of the observed counting process in the following sense.
\begin{asm}[Independent right-censoring] \label{asm:censoring}
It holds that
\begin{align*}
    N_t - \int_0^t \one(T\geq s) h_s(\oZ_s) \mathrm{d}s
    \qquad \text{and} \qquad
    N_t - \int_0^t \one(T\geq s) \bh_s(\oX_t,\oZ_s) \mathrm{d}s
\end{align*}
are martingales with respect to the filtrations $\cF_t$ and $\cG_t$, respectively.    
\end{asm}
All of our results regarding estimation and asymptotic theory will rely on Assumption~\ref{asm:censoring}, and no other assumptions on censoring.

It is natural to ask if the ACM depends on the censoring distribution. 
For the partially additive hazards model \eqref{eq:additivemodel}, the ACM is determined by the conditional hazard $\bh$
and is thus unaffected by the censoring distribution in this special case. However, the more general interpretations of the ACM in \eqref{eq:LCMisLCM} and in Proposition~\ref{prop:LCMisL2projection} depend \textit{a priori} on the censoring distribution via the at-risk indicator $Y_t$.
In view of \eqref{eq:LCMisLCM}, the ACM is unaffected by censoring if for all $t\in [0,1]$,
\begin{align*}
    \frac{\ex[u(\oX_t,\oZ_t)\given C\geq t,T^*\geq t]}{
    \ex[v(X_t,\oZ_t)\given C\geq t,T^*\geq t]
    }
    =
    \frac{\ex[u(\oX_t,\oZ_t)\given T^*\geq t]}{
    \ex[v(X_t,\oZ_t)\given T^*\geq t]
    },
\end{align*}
where 
\begin{align*}
    u(\oX_t,\oZ_t) 
        &= (X_t-\pi_t(\oZ_t))(\bh_t(\oX_t,\oZ_t)-h_t(\oZ_t)), \\
    v(X_t,\oZ_t) 
        &= (X_t-\pi_t(\ol{Z}_t))^2.
\end{align*}
This holds, for example, when $C\ind X,Z,T^*$. 


\section{Estimation}\label{sec:acmestimation}
We let $[n] = \{1, \ldots, n\}$ for $n \in \mathbb{N}$ and assume that a sample of $n$ observations, $(\bW^{(i)})_{i\in[n]} =(T^{(i)},\Delta^{(i)},X^{(i)},Z^{(i)})_{i\in [n]}$, is available as independent copies of the template observation $\bW \coloneq (T,\Delta,X,Z)$.
For each index $i\in [n]$, we use a superscript notation, ${}^{(i)}$, for the corresponding quantities introduced in Section~\ref{sec:acmsetup}. For example, $M_t^{(i)} = N_t^{(i)} - H_{t\wedge T^{(i)}}(\oZ_{t\wedge T^{(i)}}^{(i)})$ is the $Z^{(i)}$-compensated counting process, where $N_t^{(i)}=\one(T^{(i)}\leq t, \Delta^{(i)}=1)$. 

We will use \textit{cross-fitting} to construct an estimator of the ACM, but for simplicity we first describe our estimator based on single sample split $[n] = \mathcal{I}_1 \cup \mathcal{I}_2$, where $\mathcal{I}_1$ and $\mathcal{I}_2$ are disjoint sets. For convenience we let $n_1 \coloneq |\mathcal{I}_1|$ and $n_2 \coloneq |\mathcal{I}_2|$, and we assume that $\min\{n_1,n_2\}\to \infty$ as $n\to \infty$.
The data indexed by $\mathcal{I}_1$, henceforth called the \textit{training data}, is used to compute estimates $\widehat{\pi}_t$
and $\widehat{h}_t$ of the functions $\pi_t$ and $h_t$. 

The remaining data $(\bW_i)_{i\in \mathcal{I}_2}$ is used to compute the estimate of ACM by the following 
three-step procedure:
\begin{enumerate}[(a)]
    \item  Compute an estimate of $\rho(\cdot)$, as defined by \eqref{eq:rho}. Such an estimate could be
    \begin{equation}\label{eq:empiricalrho}
        \tilde{\rho}(t) = \frac{1}{n_2}\sum_{i\in \mathcal{I}_2} 
            Y_t^{(i)} \big(X_t^{(i)}-\widehat{\pi}_t(\oZ_t^{(i)})\big)^2.
    \end{equation}
    Since we need to divide by an estimate of $\rho(\cdot)$, it will be convenient to consider a clipped estimate of the form $\widehat{\rho}(\cdot) = \max\{\tilde{\rho}(\cdot),c\}$ for some $c>0$. We return to this discussion in both the general theory and for practical considerations.
    
    \item Compute estimates of the residual process and $Z^{(i)}$-compensated counting process by 
    \begin{align}\label{eq:residualestimate}
        \widehat{G}_t^{(i)} &= \widehat{\rho}(t)^{-1}Y_t^{(i)}(X_t^{(i)}-\widehat{\pi}_t(\oZ_t^{(i)})),\\
        \hatM_t^{(i)} &= N_t^{(i)} - \int_0^tY_s\,\widehat{h}_s(\oZ_s^{(i)})\mathrm{d}s, \nonumber
    \end{align}
    for each index $i\in \mathcal{I}_2$.
    
    \item Compute the ACM estimator, denoted $\widehat{\gamma}=(\widehat{\gamma}_t)$, by
    \begin{equation}\label{eq:ACMestimator}
        \widehat{\gamma}_t 
        = \frac{1}{n_2} \sum_{i\in \mathcal{I}_2} \int_0^t \widehat{G}_s^{(i)} \mathrm{d}\hatM_s^{(i)},
        \qquad t\in [0,1].
    \end{equation}
\end{enumerate}

While the formula \eqref{eq:ACMestimator} is the same for the LCM estimator proposed in \cite{christgau2023nonparametric}, there is an important technical distinction:
the residual estimate $\widehat{G}_t^{(i)}$ cannot be determined from the training data and $i$-th observation alone, since the estimate $\widehat{\rho}$ is based on all of the data indexed by $\mathcal{I}_2$.
Thus the residual estimates are not conditionally independent given the training data, which makes the theoretical analysis of the estimator a bit more involved. See the discussion in Section \ref{sec:acmdiscuss} for an elaboration regarding this difference. 

A general procedure for computing the ACM estimator is described in Algorithm~\ref{alg:acm}. There are additional numerical challenges when computing the ACM estimator in practice. The estimator \eqref{eq:ACMestimator} is, for instance, defined as an integral, and its practical evaluation may require numerical integration. We return to this practical discussion in 
Section \ref{sec:acmpractical}.

\begin{algorithm} \caption{Single split estimate of ACM} \label{alg:acm}
  \textbf{input}: sample $(T^{(i)},\Delta^{(i)},X^{(i)},Z^{(i)})_{i \in [n]}$ and partition $\mathcal{I}_1 \cup \mathcal{I}_2 = [n]$\;
  \textbf{options}: regression methods for estimation of $h_t(\cdot)$ and $\pi_t(\cdot)$\;
  \Begin{
    fit conditional hazard $\widehat{h}_t(\cdot)$ by regressing $(T^{(i)},\Delta^{(i)})_{i\in\mathcal{I}_1}$ onto $(\ol{Z}^{(i)}_t)_{i\in\mathcal{I}_1}$\; 
    fit projection $\widehat{\pi}_t(\cdot)$ by regressing $(X^{(i)}_t)_{i\in\mathcal{I}_1}$ onto $(T^{(i)},\Delta^{(i)},\ol{Z}^{(i)}_t)_{i\in\mathcal{I}_1}$\; 
    compute $\widehat{\rho}(\cdot)$ using $\widehat{\pi}_t(\cdot)$ and $(T^{(i)},\Delta^{(i)},X^{(i)},Z^{(i)})_{i\in\mathcal{I}_2}$\;
    compute $(\widehat{G}^{(i)},\hatM^{(i)})_{i\in\mathcal{I}_2}$ according to \eqref{eq:residualestimate}\;
  }
  \textbf{output}: the ACM estimate $\widehat{\gamma}_t = \frac{1}{n_2} \sum_{i\in\mathcal{I}_2} 
    \int_0^t \widehat{G}_s^{(i)}\mathrm{d}\hatM_s^{(i)}$\;
\end{algorithm} 

\subsection{Cross-fitted ACM estimator}\label{sec:acmcross-fitting}
The ACM estimator defined in \eqref{eq:ACMestimator}, and described in Algorithm~\ref{alg:acm}, is based on a single data split $[n]=\mathcal{I}_1\cup \mathcal{I}_2$ and only uses the data indexed by $\mathcal{I}_2$ for estimation of the target. The more sophisticated technique of 2-fold cross-fitting produces an additional estimate by swapping the roles of $\mathcal{I}_1$ and $\mathcal{I}_2$ and aggregates the two estimates by averaging. More generally, $K$-fold cross-fitting partitions the data into $K$ disjoint folds, and cyclically uses each fold for estimation of the ACM. We assume that the folds have sizes 
at least $\floor{n/K}$ and at most $\ceil{n/K}$.
The procedure is described in detail in Algorithm~\ref{alg:acmx}, and the resulting cross-fitted estimator is referred to as the \textsc{x-acm} and is denoted by $\widecheck{\gamma}$. We refer to \cite{chernozhukov2018} and references therein for a more general description of cross-fitting.

\begin{algorithm} \caption{$K$-fold cross-fitted \textsc{x-acm}} \label{alg:acmx}
  \textbf{input}: sample $(T^{(i)},\Delta^{(i)},X^{(i)},Z^{(i)})_{i\in[n]}$ and partition $J_1 \cup \cdots \cup J_K= [n]$\;
  \textbf{options}: options for Algorithm~\ref{alg:acm}\;
  \Begin{
    \For{$k = 1,\ldots, K$}{
        compute $\widehat{\gamma}^{k}$ using Algorithm~\ref{alg:acm} with $\mathcal{I}_1 = \bigcup_{\ell \neq k} J_\ell$ and $\mathcal{I}_2 = J_k$\;
    }
  }
  \textbf{output}: the \textsc{x-acm} estimate $\widecheck{\gamma} = \frac{1}{K} \sum_{k=1}^K \widehat{\gamma}^{k}$
\end{algorithm}

\section{Asymptotic theory}\label{sec:acmasymptotics}
We present in this section asymptotic representations of the estimation errors, 
$\widehat{\gamma} - \gamma$ and $\widecheck{\gamma} - \gamma$, for the sample split and cross-fit 
estimators. We leverage these representations to derive asymptotic Gaussian process limits, 
which can be used for statistical inference. 
The asymptotic properties of the general LCM estimator were analyzed in \cite{christgau2023nonparametric} for the specific purpose of understanding type I and type II errors of the \textit{Local Covariance Test} -- a test of local independence based on the LCM estimator. In \cite{christgau2023nonparametric} it was established, under suitable regularity conditions, that:
\begin{itemize}
    \item The asymptotic limit of the LCM estimator is a Gaussian martingale under the null hypothesis $\gamma = 0$.
    \item The LCM estimator is $\sqrt{n}$-consistent using the \textit{additive residual process} $X_t-\Pi_t$. 
\end{itemize}
These results cannot directly be applied to our proposed ACM estimator, as the residual estimates in \eqref{eq:residualestimate} are constructed with a mutual dependence via $\widehat{\rho}$. However, by tailoring the theoretical analysis to the residual process given by \eqref{eq:residualprocess}, and its corresponding estimates in \eqref{eq:residualestimate}, we establish the asymptotic Gaussian process limit of $\sqrt{n_2}(\widehat{\gamma}-\gamma)$, also under the alternative $\gamma\neq 0$. 

It will be instructive and useful for the general analysis to first understand the asymptotics of the ACM estimator in the simplified scenario where $\rho$ is assumed known. In this case we can largely use the results from \cite{christgau2023nonparametric}, and we can then in Section \ref{sec:acmgeneral} combine these results with a decomposition 
of the full estimation error to arrive at the general results when $\rho$ is estimated.

\subsection{Asymptotics with a \texorpdfstring{$\rho$}{rho}-oracle}

In this section we analyze the asymptotics of the $\rho$-oracle estimator
\begin{align}\label{eq:tgamma}
    \widetilde{\gamma}_t \coloneq \frac{1}{n_2} \sum_{i\in\mathcal{I}_2} \int_0^t \frac{Y_t^{(i)}(X_t^{(i)}-\widehat{\pi}_t(\oZ_t^{(i)}))}{\rho(s)}\mathrm{d}\hatM_s^{(i)}, 
    \qquad t\in[0,1].
\end{align}
That is, we assume that $\rho$ is known and not estimated. 
Since the terms in the sum are conditionally i.i.d. given $(\widehat{\pi}_t,\widehat{h}_t)$ -- and these nuisance estimates are obtained from the training data only -- the estimator $\tgamma_t$ fits into the framework of \cite{christgau2023nonparametric}.
While \cite{christgau2023nonparametric} did not establish an asymptotic limit in the alternative $\gamma\neq 0$, we can follow their initial asymptotic analysis. 
To this end, it will be convenient to first define the processes 
\begin{align}\label{eq:additiveresidual}
     E_t &\coloneq Y_t(X_t-\Pi_t) = \rho(t) G_t, \\
    \widehat{E}_t &\coloneq Y_t(X_t-\widehat{\pi}_t(\oZ_t)) = \widehat{\rho}(t)\widehat{G}_t.
\end{align}

The following assumption is essentially a reformulation of Assumption 4.1 in \citep{christgau2023nonparametric}. 

\begin{asm}[Boundedness]\label{asm:boundedness}
    There exist constants $C_{\bh}, C_E, c_\rho > 0$ such that for each $t\in [0,1]$:
    \begin{enumerate}
        \item[(i)] $Y_t\max\big\{\bh_t(\oX_t,\oZ_t),\ \widehat{h}_t(\oZ_t)\big\} \leq C_{\bh}$.
        \item[(ii)] $\max\{|E_t|,|\widehat{E}_t|\} \leq C_E$.
        \item[(iii)] $\rho(t) \geq c_\rho$.
    \end{enumerate}
\end{asm}
A few remarks regarding Assumption~\ref{asm:boundedness} are in order.
\begin{remark}[Remarks on boundedness assumptions]\label{rmk:boundedassumptions}
    \phantom{} 
    \begin{enumerate}[(a)]
    \item From the innovation theorem it follows that 
    $$
        Y_t\bh_t(\oX_t,\oZ_t) \leq C_{\bh} 
            \implies Y_th_t(\oZ_t) \leq C_{\bh} .
    $$
    Since hazards are non-negative, the assumption $(i)$ also entails the bounds $Y_t|\bh_t(\oX_t,\oZ_t) - h_t(\oZ_t)|\leq C_{\bh}$ and $Y_t |h_t(\oZ_t) - \widehat{h}_t(\oZ_t)|\leq C_{\bh}$.x
    \item If the exposure process takes values in an interval, that is, $X_t\in [a,b]$, then also $\Pi_t\in[a,b]$ and it is reasonable to impose $\widehat{\pi}_t\in[a,b]$. 
    In this case, $(ii)$ is satisfied with $C_E = b-a$. 

    \item Note that
    $$
        \rho(t) = \mathbb{P}(T\geq t)\ex[(X_t-\Pi_t)^2\given T\geq t].
    $$
    Thus Assumption~\ref{asm:boundedness} (iii) is equivalent to the two conditions:
    \mbox{$\mathbb{P}(T\geq 1)>0$} and $\inf_{t\in[0,1]}\ex[(X_t-\Pi_t)^2\given T\geq t]>0$. The first is natural in order to perform statistical inference up to time $t=1$.
    The second condition should also not be surprising: even for the partially additive model we would need 
    $\ex[(X_t-\Pi_t)^2\given T\geq t]>0$ for all $t \in [0, 1]$ to identify the coefficient $\theta_t$.
    
    In case that \textit{(iii)} appears to be violated for the observed data, it is always possible to restrict the observation window and target the ACM over a shorter timespan.
    \end{enumerate}
\end{remark}

To establish convergence it is, of course, also necessary to control the estimation error of the nuisance parameters.

\begin{asm}[Analogous to Assumption 4.2 in \cite{christgau2023nonparametric}]\label{asm:rateconditions}
    With 
    \begin{align*}
        a_n \coloneq \int_0^1 \ex\Big[Y_t(\widehat{\pi}_t(\oZ_t)-\pi_t(\oZ_t))^2\Big] \mathrm{d}t,
        \quad \text{and} \quad
        b_n \coloneq \int_0^1 \ex\Big[Y_t(\widehat{h}_t(\oZ_t)-h_t(\oZ_t))^2\Big] \mathrm{d}t,
    \end{align*}
    it holds that
    \begin{align*}
        \max\{a_n,b_n,n_2\cdot a_nb_n\} \longrightarrow 0
    \end{align*}
    as $n\to\infty$.
\end{asm}

Under these assumptions, it follows that $\tgamma_t$ is asymptotically equivalent to the corresponding oracle term, where the nuisance estimates $(\widehat{\pi}_t,\widehat{h}_t)$ are replaced with $(\pi_t,h_t)$.
\begin{prop}~\label{prop:AisequivtoU}
    Under Assumptions~\ref{asm:boundedness} and \ref{asm:rateconditions}, it holds that
    \begin{equation*}
        \sup_{t\in [0,1]} |\sqrt{n_2}\cdot \tgamma_t-U_t| \xrightarrow{P} 0
    \end{equation*}
    as $n\to \infty$, where $(U_t)$ is the oracle process given by
    \begin{equation*}
        U_t = \frac{1}{\sqrt{n_2}} \sum_{i \in \mathcal{I}_2} 
        \int_0^t G_s^{(i)} \mathrm{d} M_{s}^{(i)},
        \qquad t\in[0,1].
    \end{equation*}
\end{prop}
The proof of Proposition~\ref{prop:AisequivtoU} is largely based on the results of \cite{christgau2023nonparametric}, with a few minor modifications. 
Under Assumptions~\ref{asm:boundedness} and \ref{asm:rateconditions}, it is possible to establish that the process $U-\sqrt{n_2} \gamma$ converges in distribution to a continuous Gaussian process by appealing to a central limit theorem in Skorokhod space. However, as we are primarily interested in the actual ACM estimator, $\widehat{\gamma}$, we return to the asymptotics of the process $(U_t)$ in the proof of Theorem \ref{thm:ACMasymptotics}, see Section \ref{sec:acmproofACMasymptotics} in the supplement, where it is analyzed in combination with another non-vanishing term.

\begin{remark}\label{rem:empiricalrates}
    It should be possible to weaken the rate requirements to the empirical errors. To wit, if we define
    \begin{align*}
        \tilde{a}_n
            &\coloneq \frac{1}{n_2}\sum_{i\in \mathcal{I}_2} \int_0^1 Y_t^{(i)}(\widehat{\pi}_t(\oZ_t^{(i)}) - \pi_t(\oZ^{(i)}))^2 \mathrm{d}t, \\
        \tilde{b}_n &\coloneq \frac{1}{n_2}\sum_{i\in \mathcal{I}_2} \int_0^1 Y_t^{(i)}(\widehat{h}_t(\oZ^{(i)}) - h_t(\oZ^{(i)}))^2 \mathrm{d}t,
    \end{align*}
    then Assumption~\ref{asm:rateconditions} can likely be replaced by the weaker condition that
    $$
        \max\{\tilde{a}_n,\tilde{b}_n,n_2\cdot \tilde{a}_n\tilde{b}_n\} \xrightarrow{P} 0,
        \qquad n\to \infty.
    $$   
    In fact, $(\tilde{a}_n,\tilde{b}_n) \xrightarrow{P}0$ is equivalent with $(a_n,b_n) = (\ex[\tilde{a}_n],\ex[\tilde{b}_n])\to 0$ under Assumption~\ref{asm:boundedness} due to uniform integrability. 
    The scaled product error, $n_2\cdot \tilde{a}_n\tilde{b}_n$, is not \textit{a priori} uniformly integrable, so its convergence does not imply that $n_2\cdot a_nb_n \to 0$. 
\end{remark}

\subsection{Asymptotics when \texorpdfstring{$\rho$}{rho} is estimated} \label{sec:acmgeneral}

To analyze the ACM estimator in the general case, where $\rho$ is estimated, we consider the decomposition
\begin{align}\label{eq:acmdecomposition}
    \sqrt{n_2}\cdot \widehat{\gamma}
        = 
        \sqrt{n_2} \cdot \tgamma + B,
\end{align}
where $\tgamma$ is the $\rho$-oracle estimator defined as in \eqref{eq:tgamma}, and where $B$ is the process given by
\begin{align*}
    B_t = \frac{1}{\sqrt{n_2}} \sum_{i\in\mathcal{I}_2} \int_0^t 
    \left(
    \frac{1}{\widehat{\rho}(s)} 
    -
    \frac{1}{\rho(s)}
    \right)
    \widehat{E}_s^{(i)}
    \mathrm{d}\hatM_s^{(i)}.
\end{align*}
In Proposition~\ref{prop:AisequivtoU} we found an asymptotically equivalent oracle process for $\sqrt{n_2}\cdot \tgamma$. 
We proceed to find an asymptotically equivalent oracle process for $B$, and then we identify the limit distribution of the sum of the oracle processes.
The decompositions of $\tgamma$ and $B$ into oracle terms and remainder terms can be found in the proofs of Proposition \ref{prop:AisequivtoU} and Theorem~\ref{thm:BequivtoV}.

In view of Assumption~\ref{asm:boundedness}~$(iii)$, we will for our general theory, and in the residual estimates $\widehat{G}^{(i)}$, consider the clipped estimator $\widehat{\rho}$ given by
\begin{equation}\label{eq:clippedrho}
    \widehat{\rho}(t) = \max\{c_{\widehat{\rho}}, \tilde{\rho}(t)\},
\end{equation}
for a $c_{\widehat{\rho}} > 0$ and where $\tilde{\rho}$ is the empirical estimate from \eqref{eq:empiricalrho}. We assume that $c_{\widehat{\rho}}$ is chosen such that $\lim_{n\to \infty}c_{\widehat{\rho}} = 0$ and $c_{\widehat{\rho}}^{-1} = o(\sqrt{n_2})$, or in other words, $c_{\widehat{\rho}}$ tends to zero, but slower than $1/\sqrt{n_2}$. In particular, this implies that $c_{\widehat{\rho}} < c_{\rho}$
eventually, and thus the event $(\widehat{\rho}\neq \tilde{\rho})$ has vanishing probability. 

To control the error of $\widehat{\rho}$ with a sufficiently fast rate, we require not only that $\widehat{\pi}_t$ is estimated consistently, but that it can be estimated with a $n^{1/4}$-rate.
\begin{asm}\label{asm:Pirate}
    With $a_n' \coloneq \int_0^1 \ex\big[Y_t(\widehat{\pi}_t(\oZ_t)-\pi_t(\oZ_t))^4\big] \mathrm{d}t,$
    it holds that $n_2 a_n'\to 0$ as $n\to \infty$.
\end{asm}
This assumption is analogous to standard assumptions in the partially linear model, cf. Assumption 4.1(e)(ii) in \citet{chernozhukov2018}.
However, we have assumed the rate on the slightly more restrictive 4-norm.
We believe that this rate requirement could be relaxed to $\sqrt{n_2} a_n\to 0$ under (potentially) additional mild assumptions. The reason that $a_n'$ simplifies the asymptotic analysis is that it controls the estimation error of $\widehat{\rho}$ in terms of the 2-norm.
\begin{prop}\label{prop:rhorate}
    Under Assumptions~\ref{asm:boundedness} and \ref{asm:Pirate}, 
    it holds that 
    \begin{equation}\label{eq:rho2normbound}
       \limsup_{n\to \infty} n_2\int_0^1 \ex\big[(\widehat{\rho}(t) - \rho(t))^2\big] 
       \mathrm{d}t
       \leq \frac{C_E^4}{4} . 
    \end{equation}
\end{prop}

The following result describes the asymptotics of $B$.

\begin{thm}\label{thm:BequivtoV}
    Under Assumptions \ref{asm:boundedness}, \ref{asm:rateconditions}, and \ref{asm:Pirate}, it holds that 
    $$
        \sup_{t\in[0,1]}|B_t - V_t| \xrightarrow{P} 0,
    $$
    where $V=(V_t)$ is the process given by
    \begin{align*}
    V_t = 
    \sqrt{n_2}\gamma_t  -
    \frac{1}{\sqrt{n_2}}\sum_{i\in\mathcal{I}_2}\int_0^t 
    \frac{(E_s^{(i)})^2}{\rho(s)}
    \ex[G_s(\bh_s(\oX_s,\oZ_s)-h_s(\oZ_s))]\mathrm{d}s.
    \end{align*}
\end{thm}
The proof of Theorem~\ref{thm:BequivtoV} does not readily follow from any of the results in \cite{christgau2023nonparametric}, and can be dissected into two novel steps: first control the approximation
\begin{align*}
    \left(
    \frac{1}{\widehat{\rho}(s)} 
    -
    \frac{1}{\rho(s)}
    \right)\!
    \left(X_s^{(i)} -\widehat{\pi}_s(\oZ_s^{(i)})\right)
    \mathrm{d}\hatM_s^{(i)}
    \approx
    \frac{\rho(s)-\widehat{\rho}(s)}{\rho(s)^2}
    (X_s^{(i)} -\Pi_s^{(i)})
    \mathrm{d}M_s^{(i)},
\end{align*}
asymptotically (Proposition~\ref{prop:remainderconvergenceB}), and then use a functional law of large numbers to reduce all stochasticity to the stochastics from $\rho-\widehat{\rho}$ (Proposition~\ref{prop:tildeBequivtoV}).

To summarize, we have considered the decomposition
$\sqrt{n_2}\cdot\widehat{\gamma} = \sqrt{n_2}\cdot\tgamma + B$, and established under Assumptions~\ref{asm:boundedness}, \ref{asm:rateconditions}, and \ref{asm:Pirate} that
\begin{align*}
    \sqrt{n_2}\cdot\tgamma = U + o_P(n^{-1/2})
    \quad \text{and} \quad
    B = V + o_P(n^{-1/2}).
\end{align*}
In view of \eqref{eq:LCMisLCM2} we can regard $\gamma$ as a signed measure on $[0,1]$ with Radon-Nikodym derivative
$$
\frac{\mathrm{d}\gamma_s}{\mathrm{d}s} \coloneq \ex[G_s(\bh_s(\oX_s,\oZ_s)-h_s(\oZ_s))],
$$ 
which gives that
\begin{align}\label{eq:oraclesum}
    U_t + V_t - \sqrt{n_2}\gamma_t
    &= \frac{1}{\sqrt{n_2}} \sum_{i \in \mathcal{I}_2} 
    \left(\int_0^t G_s^{(i)}\mathrm{d}M_{s}^{(i)} 
        - \int_0^t E_s^{(i)} G_s^{(i)} \mathrm{d}\gamma_s\right).
\end{align}
The integrand `$G_s^{(i)}\mathrm{d}M_{s}^{(i)} - E_s^{(i)} G_s^{(i)} \mathrm{d}\gamma_s$' corresponds to the oracle term of a partially linear model on an `infinitesimal scale'. Theorem~\ref{thm:oracleconvergence} in the supplement establishes asymptotic normality of \eqref{eq:oraclesum} based on a central limit theorem in Skorokhod space.

Combining Proposition \ref{prop:AisequivtoU}, Theorem \ref{thm:BequivtoV}, and Theorem \ref{thm:oracleconvergence}, we obtain asymptotic Gaussianity of the ACM estimator and the \textsc{x-acm}, $\widecheck{\gamma}$, described in Section~\ref{sec:acmcross-fitting}.
\begin{thm}\label{thm:ACMasymptotics}
    Let $\Gamma=(\Gamma_t)_{t\in[0,1]}$ be a continuous Gaussian process with mean zero and the same covariance function as the process
    $\big(\int_0^tG_s\mathrm{d}M_s - \int_0^t E_s G_s \mathrm{d}\gamma_s\big)_{t\in[0,1]}$.
    
    Under Assumptions~\ref{asm:boundedness}, \ref{asm:rateconditions}, and \ref{asm:Pirate}, it holds that
    \begin{align*}
        \sqrt{n_2}(\;\! \widehat{\gamma} - \gamma) \xrightarrow{d} \Gamma
        \qquad \text{and} \qquad 
        \sqrt{n}(\;\! \widecheck{\gamma} - \gamma) \xrightarrow{d} \Gamma
    \end{align*}
    with respect to the uniform topology as $n\to \infty$.
\end{thm}
Note that the cross-fitted estimate is scaled with $\sqrt{n}$ rather than $\sqrt{n_2}$, and is thus more efficient than the single-split estimator $\widehat{\gamma}$.

\begin{remark}
Having established asymptotic Gaussianity of the ACM estimator, it is natural to ask if the (co)variance function of $\Gamma$ can be estimated in order to perform statistical inference for $\gamma$. 
Let $\mathcal{V}$ be the variance function of the Gaussian process $\Gamma$. Under the null hypothesis that $\gamma=0$, $\Gamma$ is a martingale and 
\begin{equation*}
    \mathcal{V}(t) = \var\left( \int_0^t G_s \mathrm{d} M_s \right) = 
    \ex\left[ \int_0^t G_s^2 \mathrm{d} N_s \right],
\end{equation*}
which can be estimated by 
\[
    \frac{1}{n_2} \sum_{i\in\mathcal{I}_2} 
    \
        \left(\int_0^t \widehat{G}_s^{(i)}\mathrm{d}\hatM_s^{(i)} 
        \right)^2
        \quad \text{or} \quad 
        \frac{1}{n_2} \sum_{i\in\mathcal{I}_2} 
        \int_0^t (\widehat{G}_s^{(i)})^2\mathrm{d}N_s^{(i)},
\]
cf. Proposition 4.7 in \cite{christgau2023nonparametric}. However, if the main purpose is to test 
the null hypothesis, it suffices to use the LCM with the additive residual process $X_t - \Pi_t$ that omits the scaling by $\rho$ and thus avoids estimation of $\rho$. Then, a test of this null can be carried out using the X-LCT from \cite{christgau2023nonparametric}. However, if a pointwise confidence band for the ACM estimator is desired, we surmise that
\begin{align*}
    \widehat{\mathcal{V}}(t) = 
    \frac{1}{n_2} \sum_{i\in\mathcal{I}_2} 
    \left(
        \int_0^t \widehat{G}_s^{(i)}\mathrm{d}\hatM_s^{(i)}
        -
        \int_0^t \widehat{E}_s^{(i)}\widehat{G}_s^{(i)}\mathrm{d}\widehat{\gamma}_s
    \right)^2
\end{align*}
is a consistent estimator of $\mathcal{V}(t)$, though a proof is outstanding.
\end{remark}

\section{Simulation study}\label{sec:acmsimulations}

Simulations were conducted in the special case with time-independent $X$ and $Z$, and with administrative censoring $C=1$. 
Two different settings were considered for the data generating process:
\begin{enumerate}
    \item[($\mathcal{S}_{\text{lin}}$)] An Aalen additive model with
        \begin{align*}
            Z=(Z_1,\ldots,Z_d) &\sim \mathrm{Unif}(0,1)^{\otimes d}, \\
            X\given Z &\sim \mathrm{Unif}(Z_1,Z_1+1), \\
            \bh_t^{\mathrm{A}}(X,Z) &= 2t (1+X + \beta^\top Z),
        \end{align*}
        where the coefficient $\beta=(\beta_1,\ldots,\beta_d)$ is sampled once and independently for each dataset, with $(\beta_1,\ldots,\beta_4) \sim \mathrm{Dirichlet}(\mathbf{1}_4)$
        and $\beta_i=0$ for $5\leq i\leq d$. \\ \par
        
    \item[($\mathcal{S}_{\text{par}}$)] A partially additive model with
        \begin{align*}
            Z=(Z_1,\ldots,Z_d) &\sim \mathrm{N}(0,1)^{\otimes d}, \\
            X\given Z &\sim \mathrm{Unif}(\phi(Z_1),\phi(Z_1)+1), \\
            \bh_t^{\mathrm{pA}}(X,Z) &= 2t (1+X + \phi(Z_1)),
        \end{align*}
        where $\phi(t)=\exp(-2t^2)$ is a scaled Gaussian density.
        
\end{enumerate}
Note that both settings adhere to the partially additive model \eqref{eq:additivemodel}, and in this case the ACM coincides with the cumulative direct effect $\Theta_t = \int_0^t 2s\mathrm{d}s = t$. 
The uncensored survival time $T^*$ can in this case be simulated by
\begin{equation*}
    T^* \coloneq \frac{\sqrt{E}}{\sqrt{1+X+f(Z)}},
\end{equation*}
where $E\ind X,Z$ is standard exponentially distributed, and where $f(z)$ 
equals $\beta^\top z$ for the setting $\mathcal{S}_{\text{lin}}$, and equals $\phi(z_1)$ for the setting $\mathcal{S}_{\text{par}}$. 

The simulations were performed with covariate dimensions $d\in \{4,16\}$ and datasets of sample sizes $n\in \{200,600,1800\}$.
For each setting $N=500$ datasets were simulated.

\subsection{Estimators and implementation} \label{sec:acmpractical}
The simulation study was implemented in Python.
As a benchmark, an estimator of the exposure coefficient in the Aalen additive model was considered, which was implemented using \texttt{AalenAdditiveFitter} in the \textit{lifelines} package \citep{Davidson-Pilon2019}.
This estimator is referred to as `Aalen estimator' and denoted by $\widehat{A} = (\widehat{A}_t)$.

Estimation of the ACM requires methods for estimating the predictable projection and the conditional hazard. For estimation of $\pi_t$ we implemented the following procedure. 
\begin{enumerate}
    \item Instantiate a grid $\mathcal{T} = (\frac{k}{n_\mathcal{T}-1} \colon k= 0,1,\ldots, n_\mathcal{T}-1)$. Simulations were conducted with $n_\mathcal{T}=20$. 
    
    \item Let $[n]=\mathcal{I}_{1}\cup \mathcal{I}_2$ be a given sample-split, and 
    let $x_{it}\coloneq (t,Y_t^{(i)},Z_{1t}^{(i)},\ldots,Z_{dt}^{(i)})$, where
    $Z_{jt}^{(i)}$ denote the $j$-th covariate for the $i$-th individual (independent of $t$ in this setting). Then format:
    \begin{itemize}
        \item covariate matrices $\mathbb{X}_{\text{train}}\coloneq (x_{it})_{(i,t)\in \mathcal{I}_1 \times \mathcal{T}}$ and
         $\mathbb{X}_{\text{test}}\coloneq (x_{it})_{(i,t)\in\mathcal{I}_2 \times \mathcal{T}}$ of dimensions $(|\mathcal{I}_1|\cdot n_{\mathcal{T}})\times (d+2)$ and $(|\mathcal{I}_2| \cdot n_{\mathcal{T}})\times (d+2)$, respectively. 
        \item a response vector $\mathbb{Y}_{\text{train}} \coloneq (X_t^{(i)})_{(t,i)\in\mathcal{T}\times \mathcal{I}_1}$ of length $n_{\mathcal{T}}\cdot |\mathcal{I}_1|$.
    \end{itemize} 
    
    \item Regress $\mathbb{Y}_{\text{train}}$ onto $\mathbb{X}_{\text{train}}$ using any regression method. For this we considered OLS and \textit{gradient boosting} \citep{friedman2001greedy,hastie2009elements}, which are denoted by a subscript \texttt{lin} and \texttt{gb}, respectively. These were implemented using classes \texttt{LinearRegression} and \texttt{GradientBoostingRegressor} from the Scikit-learn package \citep{scikit-learn}.

    \item Use the fitted methods to predict response values for $\mathbb{X}_{\text{test}}$ to produce estimates of the predictable projection over the grid $\mathcal{T}$ for each test observation.

    \item Interpolate to general time points $t\in[0,1]$ using the estimated value at time $t_{\leq} = \max\{s\in \mathcal{T} \given s\leq t \}$.
\end{enumerate}
It is also possible to subset matrices in step $(2)$ to pairs $(i,t)$ such that $(T^{(i)}>t)$. This corresponds to directly targeting the quantity $\pi_t(Z) = \ex[X\given Y_t = 1, Z]$, rather than targeting $\ex[X\given Y_t,Z_t]$ as done in the steps $(3)$ and $(4)$. Both methods were tried initially, with no noteworthy differences in the analysis result observed. 
For the interpolation in step $(5)$, it is important that the method preserves predictability, which would not be ensured by, e.g., \textit{linear interpolation}.

For conditional hazard estimation, we considered two off-the-shelf implementations: 
\texttt{GradientBoostingSurvivalAnalysis} from the scikit-survival module \citep{sksurv} and the \texttt{AalenAdditiveFitter} from the \textit{lifelines} package \citep{Davidson-Pilon2019}. The two methods are again indicated by a subscript \texttt{gb} and \texttt{lin}, respectively.

The estimate $\widehat{\rho}$ was computed with the clipping value $c_{\widehat{\rho}}=0.005$. However, the estimate was set to `\texttt{nan}' rather than $c_{\widehat{\rho}}$ whenever $\tilde{\rho}(t)<c_{\widehat{\rho}}$. This is similar to how the \texttt{AalenAdditiveFitter} returns \texttt{nan} for times $t$ such that $\sum_{i} Y_t^{(i)} < 3d$.

The \textsc{x-acm} was implemented with 4-fold cross-fitting according to Algorithm~\ref{alg:acmx}. In addition, a version of the ACM estimator without any sample splitting was implemented, and this estimator is referred to as the \textsc{n-acm}. 
Thus, 4 different estimators of the ACM were implemented: 
$\textsc{n-acm}_{\texttt{lin}}$,
$\textsc{x-acm}_{\texttt{lin}}$,
$\textsc{n-acm}_{\texttt{gb}}$, and
$\textsc{x-acm}_{\texttt{gb}}$.

\subsection{Results}
\begin{figure}
    \centering
    \includegraphics[width=\linewidth]{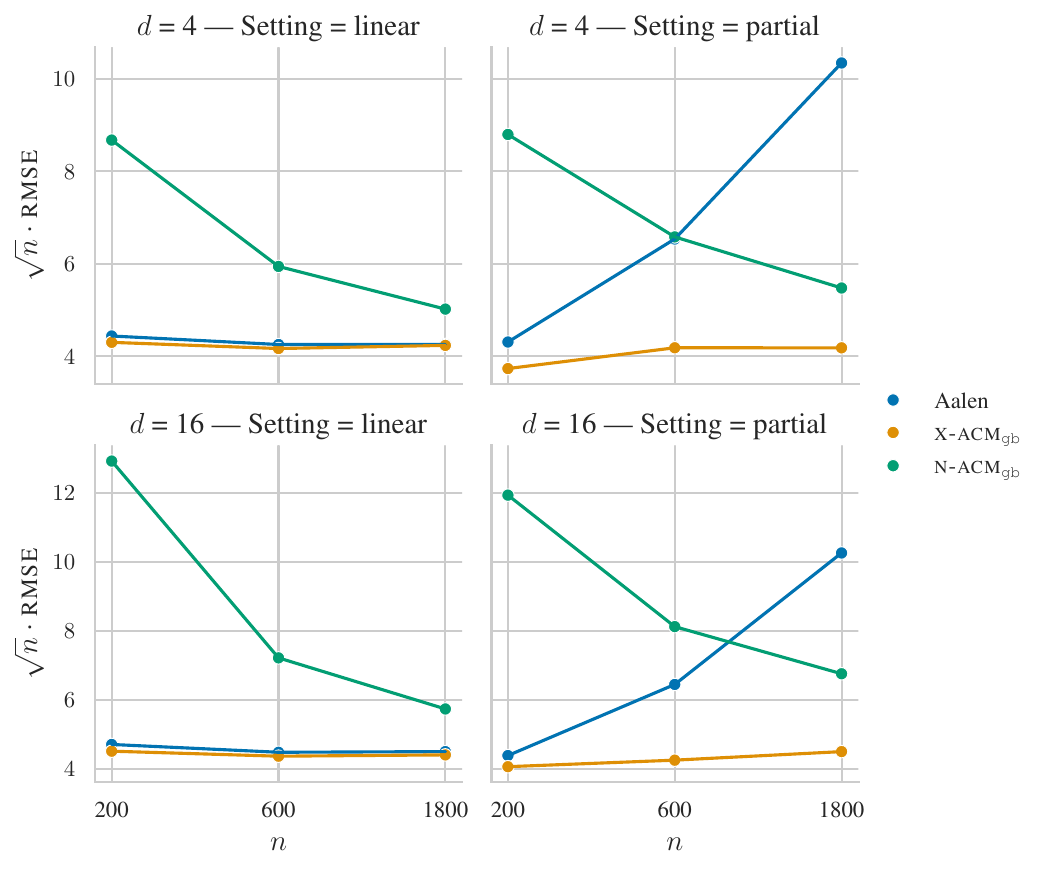}
    \caption{Scaled \textsc{RMSE} for various estimators with respect to the cumulative direct effect.}
    \label{fig:mse_all_settings}
\end{figure}
Figure~\ref{fig:mse_all_settings} shows the scaled root mean square error \textsc{RMSE} for various estimators with respect to the cumulative direct effect at time $t=67/127 \approx 0.622$. This timepoint was the largest, for which all estimators yielded well-defined estimates for all the simulated datasets (the Aalen estimator yielded \texttt{nan} at time $t=68/127$ for a dataset in the additive setting with $n=200$ and $d=16$).
We observed initially that the ACM estimators based on linear regression methods, namely $\textsc{n-acm}_{\texttt{lin}}$ and $\textsc{x-acm}_{\texttt{lin}}$, behaved similarly to the Aalen estimate $\widehat{A}_t$. In view of this fact, and since our focus is on estimation of the ACM using flexible learning methods, the results for linear methods are omitted for more clear visualization.

 In the left column of Figure~\ref{fig:mse_all_settings}, corresponding to the setting $\mathcal{S}_{\text{lin}}$, we observe that the Aalen estimator and $\textsc{x-acm}_{\texttt{gb}}$ have similar performance and appear to be $\sqrt{n}$-consistent. 
It is expected that both these estimators are $\sqrt{n}$-consistent, but it is perhaps surprising that the $\textsc{x-acm}_{\texttt{gb}}$ is as efficient as the Aalen estimator in the linear setting. The small discrepancy might be due to numerical approximations when fitting the estimators to the same common grid, or difference in bias-variance tradeoff because of (implicit) penalization. 

The estimator $\textsc{n-acm}_{\texttt{gb}}$ without sample splitting performs significantly worse, although the discrepancy seems to diminish for larger sample sizes. 
In theory, sample-splitting is used to control the empirical process terms\footnote{
In our analysis, these are the terms denoted by $R^{(1)},R^{(2)},R^{(4)}, r_1, \tilde{R}^{(2)}$, and $\tilde{R}^{(4)}$ found in the supplementary Section~\ref{sec:acmproofs}.}, but these terms can also be handled by Donsker conditions or algorithmic stability, see for example the discussions in \cite{chernozhukov2018,chen2022debiased}. This indicates, that such Donsker conditions are violated in our setup, and that the gradient boosting methods are not sufficiently stable. 
In practice, overfitting can lead to biases that would otherwise not appear with sample-splitting \citep{chernozhukov2018}. 
Since the gradient boosted methods are fitted with the same hyper parameters, and these were tuned towards the larger sample sizes, it is plausible that they overfit in the smaller sample sizes, resulting in a systematic bias.   

We observe that each of the ACM estimators perform similarly across the linear setting $\mathcal{S}_{\text{lin}}$ and in the partial setting $\mathcal{S}_{\text{par}}$. This is different for the Aalen estimator $\widehat{A}_t$, which fails severely in estimating the effect $\Theta_t$. 
This is to be expected, as the Aalen additive model is misspecified, and as a consequence, $\widehat{A}_t$ cannot reasonable account for the nonlinear effect $\phi(Z_1)$. 



In summary, we conclude that $\textsc{x-acm}_{\texttt{gb}}$ performs well across all settings, indicating that it provides an efficient and robust estimator of the cumulative direct effect.

\section{Discussion} \label{sec:acmdiscuss}
The ACM was proposed as a model-free estimand to quantify the conditional relationship between a time-to-event and an exposure, controlling for given covariates.
While it is a specific instance of the more general LCM introduced by \citet{christgau2023nonparametric}, we have highlighted several unique properties and results that apply to the ACM -- and not necessarily to the LCM -- which were not considered in \cite{christgau2023nonparametric}.

The ACM may be understood as a weighted hazard difference estimand, and we showed in Proposition~\ref{prop:LCMisL2projection} that it coincides with the cumulative exposure coefficient in the partially additive hazards model \eqref{eq:additivemodel} that best approximates the true hazard in~$L^2$. As a result, the ACM measures the cumulative direct effect within partially additive hazards model, and in particular within the Aalen additive model.

The asymptotic analysis for our proposed estimator, the \textsc{x-acm}, was to some extent based on results from \cite{christgau2023nonparametric}. However, a novel examination of the term $B$ in \eqref{eq:acmdecomposition} was necessary to account for the estimation error in the function $\rho$ appearing in the denominator. Notably, the remainder term $D^{(2)}$ from \cite{christgau2023nonparametric}, analogous to $R^{(4)}$ in \eqref{eq:Adecomposition}, does not vanish for the ACM when $\rho$ is estimated. Moreover, the decomposition in \eqref{eq:decomposition} differs from typical decompositions of estimators related to the \emph{partially linear model}, considered in, e.g., 
\citep{robinson1988root,lundborg2023perturbation,hines2023optimally}, 
because these decompositions are not feasible when integrating over time.

\textbf{Estimating $\rho$ using the training data}. 
Suppose that the function $\rho$ is estimated using the training data, which corresponds to replacing $\mathcal{I}_2$ with $\mathcal{I}_1$ in \eqref{eq:empiricalrho}. 
Although the residual estimates would then fit within the framework of \cite{christgau2023nonparametric}, the analysis does not simplify as the term $D^{(2)}$ in \cite{christgau2023nonparametric} still remains non-vanishing. Nevertheless, the analysis in Section~\ref{sec:acmgeneral} can be adapted to derive an oracle decomposition similar to \eqref{eq:oraclesum}, but with the integrals $\int_0^t E_s^{(i)} G_s^{(i)} \mathrm{d}\gamma_s$ summed over $i \in \mathcal{I}_1$. This makes it unfeasible to aggregate the cross-fitted estimates, resulting in a smaller effective sample when using this approach. Additionally, controlling the estimation error of $\rho$ in this case requires Donsker class conditions or algorithmic stability, cf. the term $r_1$ in the proof of Proposition~\ref{prop:rhorate}. Our simulations in Section~\ref{sec:acmsimulations} suggest that such conditions may not hold in practice, at least when gradient boosting is used.

\textbf{The cross-fitting scheme}. 
The \textsc{x-acm}, described in Section~\ref{sec:acmcross-fitting}, employs cross-fitting similarly to the \textsc{dml}1 estimator described in \cite{chernozhukov2018}. While it is conceivable that a cross-fitting scheme corresponding to \textsc{dml}2 --~where numerators and denominators are aggregated separately across sample splits~-- could be implemented, the integration over time complicates this approach. If possible, however, such an approach might lead to increased numerical stability in small samples, cf. Remark~3.1 in \cite{chernozhukov2018}.

\textbf{Relaxations of assumptions}. 
We already discussed, in Remark~\ref{rem:empiricalrates}, the possibility of relaxing the rate requirements to hold for the empirical errors instead of the expected errors. 
The assumptions on boundedness, Assumptions~\ref{asm:boundedness} $(i)$ and $(ii)$, 
can be found in similar works, cf. Theorem~2 in \cite{vansteelandt2022assumption} or Assumption~1 in \cite{hou2023treatment}. It is, however, likely that they can be replaced by suitable tail bounds or bounds on conditional variances, but at the cost of more technical proofs. 
Finally, Assumption~\ref{asm:Pirate} states that $\Pi$ can be estimated at an order $o(n^{-1/4})$ in terms of the $4$-norm. While the order $o(n^{-1/4})$ is to be expected, see \citep{chernozhukov2018}, it is conceivable that the condition can be relaxed to hold for the $2$-norm rather than the slightly more restrictive $4$-norm.

\textbf{Alternative estimands.}
A limitation of the ACM is the inverse weighting by the function $\rho$, potentially making it difficult to estimate accurately. An interesting direction for future work is to consider alternative ways of weighting the hazard difference. 
This may lead to estimands that are easier to estimate, but still interpretable. 
Such a compromise was considered by \citet{vansteelandt2022assumption}, who propose a hazard ratio estimand that deliberately avoids inverse weighting by the conditional density of the exposure. However, their estimand still requires inverse weighting by the cumulative hazard, which imposes a limitation similar to the inverse weighting with $\rho$. 

The related works \citet{dukes2019doubly,hou2023treatment}, cf. Table~\ref{tab:overview}, consider estimation of $\Theta_t$ in \eqref{eq:additivemodel}, with time-independent exposure and covariates, based on (efficient) orthogonal score methodology. It would be interesting to investigate how the \textsc{x-acm} performs relatively to their methodology, both for a well-specified and misspecified model. Furthermore, supposing that their methodologies can be shown to target general model-free estimands, it would be interesting to compare these estimands with the ACM.

\section*{Acknowledgments}
We are grateful to Anton Rask Lundborg for helpful discussions. AMC and NRH were supported by a research grant (NNF20OC0062897) from Novo Nordisk Fonden.

%% file: paper_acm/appendix.tex
\clearpage
\section*{Supplement to `Assumption-lean Aalen regression'}

The supplementary material consists of Section~\ref{sec:acmproofs}, which contains proofs of the main results and
related auxiliary lemmas, and Section~\ref{app:multhaz}, which includes a discussion of the ACM in the context of multiplicative hazards models.

\supplementarysection{Auxillary results and proofs}\label{sec:acmproofs}
\subsection{Additional notation for proofs}\label{sec:acmproofnotation}
While the main manuscript is formulated using histories and the conditional hazards $h_t$ and $\bh_t$, it will be convenient to work with filtrations and intensities in the proofs, as also discussed in Remark \ref{rem:filt}. Recall that 
 we define the filtrations $\cF_t \coloneq \sigma(Z_s, N_s; s \leq t)$ and $\cG_t \coloneq \sigma(X_s,Z_s,N_s; s \leq t)$.
 The stochastic processes $\lambda=(\lambda_t)$ and $\blambda = (\blambda_t)$, defined by
\begin{align*}
    \lambda_t \coloneq Y_t h_t(\oZ_t) 
    \qquad \text{and} \qquad
    \blambda_t \coloneq Y_t \bh_t(\oX_t,\oZ_t),
\end{align*}
are then, under independent right-censoring as described in Section~\ref{sec:acmcensoring}, 
the $\cF_t$- and $\cG_t$-intensity, respectively, of $N_t$.
We remark that, by definition, this means that $M_t = N_t - \int_0^t \lambda_s \mathrm{d}s$ is an $\cF_t$-martingale and that
\begin{align*}
    \bM_t \coloneq N_t - \int_0^t \blambda_s \mathrm{d}s
\end{align*}
is a $\cG_t$-martingale. 
By the innovation theorem, it also holds that $\ex[\blambda_t \given \cF_{t-}]= \lambda_t$.
This notation is consistent with \cite{christgau2023nonparametric}.

For the proofs in the asymptotic analysis we will additionally use the following notation:
\begin{align*}
    \widehat{\lambda}_t^{(i)} \coloneq Y_t^{(i)}\widehat{h}_t(\oZ^{(i)}) \\
    \widehat{\Pi}_t^{(i)} \coloneq \widehat{\pi}_t(\oZ^{(i)})
\end{align*}
for $i\in [n]$ and $t\in [0,1]$. 
Finally, we let $\|\cdot\|_p$ denote the norm on $L^p([0,1]\times \Omega, \mathrm{d}t\otimes \mathbb{P})$, so that, for example,
\begin{align*}
    b_n = \int_0^1 \ex\big[(\widehat{\lambda}_t- \lambda_t)^2\big] \mathrm{d}t 
    = \|\widehat{\lambda} - \lambda\|_2^2.
\end{align*}


\subsection{Proof of Proposition \ref{prop:LCMisL2projection}}
\label{sec:acmproofofLCMisL2}
We first consider the minimization objective for a fixed timepoint $t\in[0,1]$.
\begin{lem}\label{lem:fixedtimeL2projection}
    For fixed $t\in[0,1]$, it holds that 
    \begin{align}\label{eq:projectionsolutions}
    \begin{cases}
        \vartheta_t^\star = \cov(G_t,\blambda_t-\lambda_t) \\
        g_t^\star = h_t -\cov(G_t, \blambda_t-\lambda_t) \pi_t
    \end{cases}
    \end{align}
    are solutions to
    \begin{align*}
        \minimize_{\vartheta_t\in \real, g_t \in L^0(P_{\oZ_t})}& \quad 
            \ex\left[Y_t\left(\bh_t(\oX_t,\oZ_t)
            - (\vartheta_t X_t + g_t(\oZ_t))\right)^2\right] 
    \end{align*}
    where $P_{\oZ_t}$ denotes the distribution of $\oZ_t$.
    Moreover, the solution $\vartheta_t^\star$ is unique.
\end{lem}
\begin{proof}
    Recall that on the event $(Y_t=1)=(T\geq t)$, the predictable projection $\Pi_t$ takes the value $\pi_t(\oZ_t) = \ex[X_t\given T\geq t, \oZ_t]$. 
    Since $\pi_t, h_t \in L^0(P_{\oZ_t})$, the minimization problem is equivalent to minimizing 
    \begin{align*}
        u(\vartheta_t,\ell_t)
        &\coloneq
        \ex\Big[Y_t\big(\bh_t(\oX_t,\oZ_t) - h_t(\oZ_t) 
            - \vartheta_t (X_t-\pi_t(\oZ_t))-\ell_t(\oZ_t)\big)^2\Big] \\
        &=
        \ex\Big[\big(\blambda_t - \lambda_t 
            - \vartheta_t Y_t(X_t-\Pi_t)-Y_t\ell_t(\oZ_t)\big)^2\Big],
    \end{align*}
    which corresponds to the substitution 
    $g_t =\ell_t + h_t - \vartheta_t \pi_t$. Note that
    $$
        \ex[\blambda_t - \lambda_t \given \cF_{t-}] = 0 
        =\ex[X_t-\Pi_t\given \cF_{t-}].
    $$
    Now, expanding the square in $u(\vartheta_t,\ell_t)$ and using that $\ex[\ell_t(\oZ_t)^2]\geq 0$, we obtain that
    \begin{align*}
        u(\vartheta_t,\ell_t)
            &\geq \ex[(\blambda_t-\lambda_t)^2] 
                + \vartheta_t^2\ex[Y_t(X_t-\Pi_t)^2] 
            - 2\vartheta_t \ex[(\blambda_t- \lambda_t)(X_t-\Pi_t)].
    \end{align*}
    Equality is attained when $\ell_t =0$, which is equivalent to 
    $g_t = h_t  - \vartheta_t \pi_t \in L^0(P_{\oZ_t})$. The right-hand side is a quadratic in $\vartheta_t$ and its unique minimizer is given by
    \begin{align*}
        \frac{\ex[(\blambda_t-\lambda_t)(X_t-\Pi_t)]}{\ex[Y_t(X_t-\Pi_t)^2]}
        =
        \cov(G_t, \blambda_t-\lambda_t).
    \end{align*}
    This shows that $(\vartheta_t^\star,g_t^\star)$ in \eqref{eq:projectionsolutions} are indeed minimizers and that $\vartheta^\star$ is unique.
\end{proof}

Returning to the proof of Proposition~\ref{prop:LCMisL2projection}, we let $(\vartheta_t^{\star}, g_t^{\star})$ be the pointwise minimizers from \eqref{eq:projectionsolutions} for each $t\in [0,1]$.
For any $(\vartheta,g)$ as in Proposition~\ref{prop:LCMisL2projection}, it follows directly from Tonelli's theorem and Lemma~\ref{lem:fixedtimeL2projection} that
\begin{align}\label{eq:minimizationprofiling}
     \textnormal{minimization objective}
     &= \int_0^1 \ex\left[Y_t\big(\bh_t(\oX_t,\oZ_t) 
                - \vartheta_t X_t - g_t(\oZ_t)\big)^2
                \right]\mathrm{d}t \nonumber\\
    &\geq 
     \int_0^1 \ex\left[ Y_t\big(\bh_t(\oX_t,\oZ_t) 
                - \vartheta_t^\star X_t - g_t^\star(\oZ_t)\big)^2
        \right] \mathrm{d}t.
\end{align}
We conclude that $(\vartheta^\star,g^\star)$ is indeed a minimizer to \eqref{eq:projectionobjective}, since $\vartheta_t$ is measurable and $g_t^\star \in L^0(P_{\oZ_t})$ for each $t\in[0,1]$. 
Note that equality occurs in \eqref{eq:minimizationprofiling} 
if and only if
\begin{align*}
    \ex\left[Y_t\big(\bh_t(\oX_t,\oZ_t) 
                - \vartheta_t X_t - g_t(\oZ_t)\big)^2\right]
        = \ex\left[ Y_t\big(\bh_t(\oX_t,\oZ_t) 
                - \vartheta_t^\star X_t - g_t^\star(\oZ_t)\big)^2
                \right]
\end{align*}
for almost all $t\in [0,1]$. By the uniqueness of $\vartheta^\star$ in Lemma~\ref{lem:fixedtimeL2projection}, this is equivalent to $\vartheta_t = \vartheta_t^\star$ for almost all $t\in [0,1]$. The last part of the proposition now follows since $\gamma_t = \int_0^t \vartheta_s^\star \mathrm{d}s$, according to \eqref{eq:LCMisLCM}.
\hfill $\Box$

\subsection{Proof of Proposition~\ref{prop:AisequivtoU}}
We consider a decomposition similar to that of \cite{christgau2023nonparametric}, but with their non-vanishing processes combined. To this end, define the $\rho$-oracle residual estimates
$$
    \widetilde{G}_t \coloneq \frac{\widehat{E}_t}{\rho(t)} 
    = \frac{Y_t(X_t-\widehat{\Pi}_t)}{\rho(t)},
    \qquad t\in [0,1].
$$
Then note that
\begin{align}\label{eq:Adecomposition}
    \sqrt{n_2}\cdot \tgamma = U + R^{(1)} + R^{(2)} + R^{(3)} + R^{(4)},
\end{align}
where the processes $U$, $R^{(1)}$, $R^{(2)}$, $R^{(3)}$, and $R^{(4)}$ are given by
    \begin{align*}
    U_t & = \frac{1}{\sqrt{n_2}} \sum_{i \in \mathcal{I}_2} 
    \int_0^t G_s^{(i)} \mathrm{d} M_{s}^{(i)}, 
    \\
    R^{(1)}_{t} & = \frac{1}{\sqrt{n_2}} \sum_{i \in \mathcal{I}_2}  \int_0^t
    G_s^{(i)}
    \big( \lambda_{s}^{(i)} - \widehat{\lambda}_{s}^{(i)}\big) 
    \mathrm{d}s,
    \\
    R^{(2)}_{t} & = \frac{1}{\sqrt{n_2}} \sum_{i \in \mathcal{I}_2}  \int_0^t 
    \big( \widetilde{G}_{s}^{(i)} - G_s^{(i)} \big) \mathrm{d} \bM_{s}^{(i)},
    \\
    R^{(3)}_{t} & = \frac{1}{\sqrt{n_2}} \sum_{i \in \mathcal{I}_2}  \int_0^t 
    \big( \widetilde{G}_{s}^{(i)} - G_s^{(i)} \big) 
    \big( \lambda_{s}^{(i)} - \widehat{\lambda}_{s}^{(i)} \big) 
    \mathrm{d}s,
    \\
    R_t^{(4)} & = \frac{1}{\sqrt{n_2}} \sum_{i \in \mathcal{I}_2}  \int_0^t 
    \big(\widetilde{G}_{s}^{(i)} - G_s^{(i)}\big)
    \big(\blambda_{s}^{(i)} - \lambda_{s}^{(i)}\big) 
    \mathrm{d}s.
    \end{align*}
In view of this decomposition, it suffices to show that 
$$
    \sup_{t\in [0,1]}|R_t^{(\ell)}|\xrightarrow{P} 0
$$ 
as $n\to\infty$ for $\ell = 1,\ldots, 4$.
    To this end, note that Assumption~\ref{asm:boundedness} ensures that Assumption~4.1 in \cite{christgau2023nonparametric} holds with: 
    \begin{align*}
        \max\{|G_t|,|\widetilde{G}_t|\} \leq c_{\rho}^{-1} C_E, 
        \qquad \text{and} \qquad
        \max\{\blambda_t,\widehat{\lambda}_t\} \leq C_{\bh}.
    \end{align*}
    Assumption~\ref{asm:rateconditions} can be translated directly into Assumption~4.2 in \cite{christgau2023nonparametric}.
    We therefore conclude that $R^{(1)}$, $R^{(2)}$, and $R^{(3)}$ converge (uniformly) to zero in probability by Proposition 4.4 in \cite{christgau2023nonparametric}. 
    For the last remainder process we have that
    \begin{align*}
        \sup_{t\in [0,1]}|R_t^{(4)}| \leq 
        c_{\rho}^{-1}\sup_{t\in [0,1]}\Big\lvert\frac{1}{\sqrt{n_2}} \sum_{i \in \mathcal{I}_2}  \int_0^t 
        \big(\widehat{E}_{s}^{(i)} - E_s^{(i)}\big)
        \big(\blambda_{s}^{(i)} - \lambda_{s}^{(i)}\big) 
        \mathrm{d}s\Big\rvert \xrightarrow{P} 0,
    \end{align*}
    where the convergence in probability is established in Lemma~A.10 in \cite{christgau2023nonparametric}.
\hfill $\Box$

\subsection{Proof of Proposition~\ref{prop:rhorate}}
    For $n$ sufficiently large it holds that $c_{\widehat{\rho}}<c_{\rho}$, in which case Assumption~\ref{asm:boundedness}$(iii)$ implies that $|\widehat{\rho}(t)-\rho(t)|\leq |\tilde{\rho}(t)-\rho(t)|$ for all $t\in[0,1]$.
    Hence it suffices to establish the bound for the empirical estimator $\tilde{\rho}$ from \eqref{eq:empiricalrho}. Note also that 
    \[
        \tilde{\rho}(t) =  \frac{1}{n_2}\sum_{i\in \mathcal{I}_2} 
         \big(\widehat{E}_t^{(i)}\big)^2,
    \]
    and consider the following decomposition
    \begin{align*}
        \sqrt{n_2} \cdot (\tilde{\rho}(t)-\rho(t))
        = \sqrt{n_2} \cdot (\overline{\rho}-\rho) + 2r_1(t) + r_2(t)
    \end{align*}
    where 
    \begin{align*}
        \overline{\rho}(t) &\coloneq \frac{1}{n_2} \sum_{i\in \mathcal{I}_2} (E_t^{(i)})^2, \\
        r_1(t) &\coloneq
            \frac{1}{\sqrt{n_2}} \sum_{i\in \mathcal{I}_2} 
            E_t^{(i)}(\widehat{E}_t^{(i)} - E_t^{(i)}), \\
        r_2(t) &\coloneq
            \frac{1}{\sqrt{n_2}} \sum_{i\in \mathcal{I}_2} (\widehat{E}_t^{(i)}-E_t^{(i)})^2.
    \end{align*}
    We first note that the terms in $r_1$ are conditionally i.i.d. given $\mathcal{D}_1\coloneq \sigma(\bW^{(i)}\colon i\in\mathcal{I}_1)$ with conditional mean zero. 
    It follows that
    \begin{align*}
        \ex\big[r_1(t)^2\big] 
        &= \ex\big[\,\var(r_1(t)\given \mathcal{D}_1)\big] \\
        &= \ex\Big[\frac{1}{n_2} \sum_{i\in \mathcal{I}_2}\var\big(
        E_t^{(i)}(\widehat{E}_t^{(i)} - E_t^{(i)})\given \mathcal{D}_1\big)\Big]\\
        &= \ex\big[(E_t^{(i)})^2(\widehat{E}_t^{(i)} - E_t^{(i)})^2\big] 
        \leq C_E^2\ex[(\Pi_t-\widehat{\Pi}_t)^2],
    \end{align*}
    where we have used that $\widehat{E}_t^{(i)} - E_t^{(i)} = Y_t(X_t^{(i)} - \widehat{\Pi}_t^{(i)} - X_t^{(i)} + \Pi_t^{(i)}) =  
    Y_t(\Pi_t^{(i)} - \widehat{\Pi}_t^{(i)})$. 
    By integrating over time we obtain 
    \begin{equation}\label{eq:rhorater1}
         \| r_1 \|_2\leq C_E \|\widehat{\Pi} - \Pi\|_2 \to 0.
    \end{equation}
    For the other remainder, Cauchy-Schwarz yields
    \begin{align*}
        \ex[r_2(t)^2] 
            = \frac{1}{n_2} \sum_{i,j\in\mathcal{I}_2} 
                \ex[(\widehat{E}_t^{(i)}-E_t^{(i)})^2
                (\widehat{E}_t^{(j)}-E_t^{(j)})^2] 
                \leq n_2 \ex\big[(\widehat{E}_t-E_t)^4\big].
    \end{align*}
    By Assumption~\ref{asm:Pirate} we conclude that 
    \begin{equation}\label{eq:rhorater2}
        \|r_2\|_2^2 \leq n_2 \|\widehat{E}-E\|_4^4
        = n_2\|\widehat{\Pi}-\Pi\|_4^4\longrightarrow 0,
        \qquad n\to \infty.
    \end{equation}
    For the oracle sum, $\overline{\rho}$, we note that $\ex[\overline{\rho}(t)] = \rho(t)$, and that $\overline{\rho}(t)$ is a sum of i.i.d. random variables with the same distribution as $E_t^2 \in [0, C_E^2]$. 
    Invoking Popoviciu's inequality on variances \citep{popoviciu1935equations}, we obtain
    \begin{align}\label{eq:rhorateoraclerho}
        \int_0^1\ex\big[(\overline{\rho}(t)-\rho(t))^2\big]\mathrm{d}t
        = \int_0^1 \var\big(\overline{\rho}(t)\big)\mathrm{d}t
        = \frac{1}{n_2}\int_0^1 \var\big(E_t^2\big)\mathrm{d}t 
        \leq \frac{C_E^4}{4n_2}.
    \end{align}
    Combining \eqref{eq:rhorater1}, \eqref{eq:rhorater2}, and \eqref{eq:rhorateoraclerho} via Minkowski's inequality we obtain \eqref{eq:rho2normbound}. \par
\hfill \qedsymbol{}

\subsection{Proof of Theorem~\ref{thm:BequivtoV}}
We first note that 
$$
    \frac{1}{\widehat{\rho}(s)}
    -\frac{1}{\rho(s)}
    =
    \frac{\rho(s)-\widehat{\rho}(s)}{\rho(s)\widehat{\rho}(s)}
    =
        \underbrace{\frac{\rho(s)-\widehat{\rho}(s)}{\rho(s)^2}
        }_{\eqcolon \zeta_s}
        +
        \underbrace{
        \frac{(\rho(s)-\widehat{\rho}(s))^2}{\widehat{\rho}(s)\rho(s)^2}
        }_{\eqcolon \xi_s}.
$$
Now consider a decomposition of $B$ similar to that of $\sqrt{n_2}\cdot\tgamma$,
\begin{align*}
    B = \tilde{B}_t + 
    \sum_{\ell=1}^5 \tilde{R}^{(\ell)}
\end{align*}
where
    \begin{align*}
    \tilde{B}_t & = \frac{1}{\sqrt{n_2}} \sum_{i \in \mathcal{I}_2} 
    \int_0^t\zeta_s E_s^{(i)}(\blambda_s^{(i)}-\lambda_s^{(i)}) \mathrm{d}s, 
    \\
    \tilde R_t^{(1)} & = \frac{1}{\sqrt{n_2}} \sum_{i \in \mathcal{I}_2}  \int_0^t 
    \zeta_s(\widehat{E}_{s}^{(i)} - E_s^{(i)})
    (\blambda_{s}^{(i)} - \lambda_{s}^{(i)}) 
    \mathrm{d}s 
    \\
    \tilde R^{(2)}_{t} & = \frac{1}{\sqrt{n_2}} \sum_{i \in \mathcal{I}_2}  \int_0^t 
    \zeta_s\widehat{E}_{s}^{(i)} \mathrm{d} \bM_{s}^{(i)},
    \\
    \tilde R^{(3)}_{t} & = \frac{1}{\sqrt{n_2}} \sum_{i \in \mathcal{I}_2}  \int_0^t
    \zeta_s \widehat{E}_s^{(i)}
    \left( \lambda_{s}^{(i)} - \widehat{\lambda}_{s}^{(i)}\right) 
    \mathrm{d}s,
    \\
    \tilde R_t^{(4)} & = \frac{1}{\sqrt{n_2}} \sum_{i \in \mathcal{I}_2}  \int_0^t 
    \xi_s \widehat{E}_s^{(i)} \mathrm{d}\bM_s^{(i)}  
    \\
    \tilde R_t^{(5)} & = \frac{1}{\sqrt{n_2}} \sum_{i \in \mathcal{I}_2}  \int_0^t 
    \xi_s \widehat{E}_s^{(i)}(\blambda_s^{(i)}-\widehat{\lambda}_s^{(i)}) \mathrm{d}s
    \end{align*}
The decomposition may seem peculiar at first glance: some of the remainder processes could be decomposed further into terms that would be more manageable. However, crude bounds on these error terms will suffice since $\zeta$ and $\xi$ also decay at an order of~$n_2^{-1/2}$. 
Theorem~\ref{thm:BequivtoV} follows from combining Proposition~\ref{prop:remainderconvergenceB} and Proposition~\ref{prop:tildeBequivtoV}, stated below. \hfill $\Box$

\begin{prop}\label{prop:remainderconvergenceB}
    Under Assumptions \ref{asm:boundedness}, \ref{asm:rateconditions}, and \ref{asm:Pirate}, it holds that 
    \begin{equation}\label{eq:Rtildeconvergence}
        \sup_{t\in [0,1]}|\tilde{R}_t^{(\ell)}|\xrightarrow{P} 0
    \end{equation}
    for $\ell \in \{1,2,3,4,5\}$.
\end{prop}

\begin{proof}
    Let $\langle\cdot\,, \cdot \rangle_{2}$ denote the inner product on $L^2([0,1]\times \Omega, \mathrm{d}t\otimes \mathbb{P})$.    
    The triangle inequality and Cauchy-Schwarz can be used to bound the expectation of the first remainder term as follows:
    \begin{align}\label{eq:proofofBremainder1}
        \ex\Big[\sup_{t\in [0,1]}|\tilde{R}_t^{(1)}|\Big] 
        & \leq \frac{1}{\sqrt{n_2}} \sum_{i \in \mathcal{I}_2}  
        \big\langle \, |\zeta|, \, |\widehat{E}^{(i)}-E^{(i)}|\cdot |\blambda^{(i)} - \lambda^{(i)}|
        \big\rangle_2 \nonumber \\
        &\leq \sqrt{n_2} \cdot \|\zeta\|_2\cdot  
        \big\|(\widehat{E}-E)\cdot (\blambda 
        -\lambda) \big\|_2 \nonumber \\
        &\leq c_{\rho}^{-2} C_{\bh}\sqrt{n_2} \cdot \big\|\widehat{\rho}-\rho\big\|_2
            \cdot \big\|\widehat{E}-E\big\|_2,
    \end{align}
    where we have used Assumption~\ref{asm:boundedness} for the last inequality. 
    By a similar argument we obtain
    \begin{align}\label{eq:proofofBremainder3}
        \ex\Big[\sup_{t\in [0,1]}|\tilde{R}_t^{(3)}|\Big] 
        \leq c_{\rho}^{-2} C_E\sqrt{n_2} \cdot \big\|\widehat{\rho}-\rho\big\|_2
            \cdot \big\|\widehat{\lambda} -\lambda \big\|_2.
    \end{align}
    Now, combining Assumption~\ref{asm:rateconditions} with Proposition~\ref{prop:rhorate}, we see that the right-hand sides of \eqref{eq:proofofBremainder1} and \eqref{eq:proofofBremainder3} vanish as $n\to\infty$. Hence we conclude that the desired convergence, \eqref{eq:Rtildeconvergence}, holds for $\ell=1$ and $\ell=3$.
    
    For the second and fourth remainder, let $\mathcal{D}_1 \coloneq \sigma(\bW_i \colon i \in \mathcal{I}_1)$ denote the information contained in the training data. Let further $\mathsf{G}_t = \mathcal{D}_1 \vee \sigma(\cG_t^{(i)} \colon i \in \mathcal{I}_2)$ denote the filtration containing the histories of the processes $(N^{(i)},X^{(i)},Z^{(i)})_{i\in \mathcal{I}_2}$
    together with the information $\mathcal{D}_1$.
    Then note that: 
    \begin{itemize}
        \item since the observations are independent, the compensated $\cG_t^{(i)}$-martingale $\bM_t^{(i)}$ is in fact also a $\mathsf{G}_t$-martingale
        for each $i\in \mathcal{I}_2$. 

        \item each term in the empirical estimator $\tilde{\rho}$, defined in \eqref{eq:empiricalrho}, is $\mathsf{G}_t$-predictable, and hence $\widehat{\rho}(t) = \max\{c_{\widehat{\rho}},\tilde{\rho}(t)\}$ is $\mathsf{G}_t$-predictable.
        As a consequence, both $\zeta_t$ and $\xi_t$ are also $\mathsf{G}_t$-predictable. 

        \item Assumption~\ref{asm:boundedness} implies that 
        $|\zeta_t|\leq c_{\rho}^{-2}C_E^2$ and that $|\xi_t|\leq c_{\widehat{\rho}}^{-1}c_{\rho}^{-2} C_E^4$. 
        Hence both $\int_0^1 \zeta_t^2 \blambda_t \mathrm{d}t$ and $\int_0^1 \xi_t^2 \blambda_t \mathrm{d}t$ have finite expectations. 
    \end{itemize}
    We conclude that the terms in $\tilde{R}^{(2)}$ are mean-zero $\mathsf{G}_t$-martingales, cf. Lemma~A.1 in \cite{christgau2023nonparametric}, and hence $\tilde{R}^{(2)}$ is also a $\mathsf{G}_t$-martingale. 
    By the same argument, $\tilde{R}^{(4)}$ is a $\mathsf{G}_t$-martingale.
    Therefore we can apply Doob's submartingale inequality, which yields that
    \begin{align}\label{eq:Doobsmg}
        \mathbb{P}\Big(\sup_{t\in [0,1]}|\tilde{R}_t^{(\ell)}|>\varepsilon\Big)
        \leq \frac{1}{\varepsilon^2}\ex\big[(\tilde{R}_1^{(\ell)})^2\,\big]
    \end{align}
    for $\ell \in \{2,4\}$.
    Now using the calculus of the quadratic characteristic for compensated counting processes 
    (see, for example, Proposition II.4.1 in \citep{AndersenBorganGillKeiding:1993}) we obtain
    \begin{align*}
        \ex\Big[(\tilde{R}_1^{(2)})^2\Big]
        =
        \ex\Big[\langle\tilde{R}_1^{(2)}\rangle\Big]
        &= \ex\Big[\int_0^1 \zeta_t^2 \widehat{E}_t^2\blambda_t \mathrm{d}t\Big] \\
        &= C_{\bh} C_E^2
            \int_0^1 \frac{(\rho(t)-\widehat{\rho}(t))^2}{\rho(t)^4} \mathrm{d}t 
        \leq c_{\rho}^{-4}C_{\bh} C_E^2\|\rho-\widehat{\rho}\|_2^2 \to 0.
    \end{align*}
    Similarly, observe that
    \begin{align*}
        \ex\big[(\tilde{R}_1^{(4)})^2\big]
        =
        \ex\Big[\langle\tilde{R}_1^{(4)}\rangle\Big]
        &= \ex\Big[\int_0^1 \xi_t^2 \widehat{E}_t^2\blambda_t \mathrm{d}t\Big] \\
        &\leq C_{\bh} C_E^2 \int_0^1 \frac{(\rho(t)-\widehat{\rho}(t))^4}{\widehat{\rho}(t)^2\rho(t)^4} \mathrm{d}t 
        \leq c_{\widehat{\rho}}^{-2} c_{\rho}^{-4} C_{\bh}C_E^6  \|\rho - \widehat{\rho}\|_2^2 \to 0,
    \end{align*}
    where the convergence follows from Proposition~\ref{prop:rhorate} and the fact that $c_{\widehat{\rho}}^{-2} = o(n)$.
    Combined with \eqref{eq:Doobsmg}, we conclude that \eqref{eq:Rtildeconvergence} holds for $\ell=2$  and $\ell=4$.
    
    For the final term, a direct bound yields
    \begin{align*}
        \ex\Big[\sup_{t\in [0,1]}|\tilde{R}_t^{(5)}|\Big] 
        \leq C_EC_{\bh}\sqrt{n_2} \|\xi\|_1
        \leq c_{\rho}^{-2}c_{\widehat{\rho}}^{-1}C_EC_{\bh}\sqrt{n_2} \|\rho-\widehat{\rho}\|_2^2 \to 0.
    \end{align*}
    Thus we have established \eqref{eq:Rtildeconvergence} for $\ell\in\{1,2,3,4,5\}$.
\end{proof}

\begin{prop}\label{prop:tildeBequivtoV}
    Under Assumptions \ref{asm:boundedness} and \ref{asm:Pirate}, it holds that 
    $$
        \sup_{t\in[0,1]}|\tilde{B}_t - V_t| \xrightarrow{P} 0,
    $$ 
    where $V=(V_t)$ is the process given by
    \begin{align*}
    V_t = \frac{1}{\sqrt{n_2}}\sum_{i\in\mathcal{I}_2}\int_0^t \frac{(\rho(s)-(E_s^{(i)})^2)}{\rho(s)^2} \ex[E_s(\blambda_s-\lambda_s)]\mathrm{d}s,
    \qquad t\in[0,1].
    \end{align*}
\end{prop}
\begin{proof}
The results follows from combining Lemmas~\ref{lem:BtoVhat} and \ref{lem:VhattoV}.
\end{proof}

\begin{lem}\label{lem:BtoVhat}
    Under Assumptions \ref{asm:boundedness} and \ref{asm:Pirate}, it holds that $\sup_{t\in[0,1]}|\tilde{B}_t - \widehat{V}_t| \xrightarrow{P}0$, where $\widehat{V}=(\widehat{V}_t)$ is the process given by
    \begin{align*}
    \widehat{V}_t &= \sqrt{n_2}\int_0^t \zeta_s \ex[E_s(\blambda_s-\lambda_s)]\mathrm{d}s, \qquad t\in[0,1].
    \end{align*}
\end{lem}
\begin{proof}
Letting $\|\cdot\|_{[0,1]}$ denote the norm on $L^2([0,1],\mathrm{d}t)$, an application of Cauchy-Schwarz yields that
\begin{align}\label{eq:Vequivproof}
    \sup_{t}|\tilde{B}_t - \widehat{V}_t |
    &=\sqrt{n_2} \int_0^1 |\zeta_s|
     \Big\lvert \frac{1}{n_2} \sum_{i \in \mathcal{I}_2}
        E_s^{(i)}(\blambda_s^{(i)}-\lambda_s^{(i)})- \ex[E_s(\blambda_s-\lambda_s)]
     \Big\rvert\mathrm{d}s \nonumber \\
     &=
     \Big\| \frac{1}{n_2} \sum_{i \in \mathcal{I}_2}
        E^{(i)}(\blambda^{(i)}-\lambda^{(i)})- \ex[E(\blambda-\lambda)]
     \Big\|_{[0,1]}
     \cdot 
     \sqrt{n_2} \|\zeta\|_{[0,1]}
\end{align}
Since
$
    \ex\|E(\blambda - \lambda)\|_{[0,1]} \leq C_E C_{\bh},
$
the law of large numbers -- in the space $L^2([0,1])$ -- 
gives that the first factor of \eqref{eq:Vequivproof} converges to zero almost surely, and in particular also in probability. See for example \citet[Cor. 7.10]{ledoux1991probability} for a general formulation of the law of large numbers on separable Banach spaces. 

For the second factor, we note that 
$$
    \sqrt{n_2}\|\zeta\|_{[0,1]} 
    \leq c_{\rho}^{-2}\sqrt{n_2}\|\widehat{\rho} - \rho\|_{[0,1]}.
$$
The right-hand side is bounded in expectation by \eqref{eq:rho2normbound}, and in particular it is bounded in probability.
Combined we conclude that the right-hand side of \eqref{eq:Vequivproof} converges to zero in probability.
\end{proof}

\begin{lem}\label{lem:VhattoV}
    Under Assumptions \ref{asm:boundedness} and \ref{asm:Pirate}, it holds that 
    $$
        \ex\Big[\sup_{t\in[0,1]}|\widehat{V}_t - V_t|\Big] \xrightarrow{P}0
    $$
\end{lem}
\begin{proof}
Let $\tilde{V} = (\tilde{V}_t)$ be the process given by
\begin{align*}
    \tilde{V}_t \coloneq \sqrt{n_2}\int_0^t \varsigma_s \ex[E_s(\blambda_s-\lambda_s)]\mathrm{d}s,
    \qquad \text{where} \quad
    \varsigma_s \coloneq \frac{\rho(s)-\tilde{\rho}(s)}{\rho(s)^2}.
\end{align*}
In other words, $\varsigma_s$ corresponds to $\zeta_s$, but where the clipped estimate, $\widehat{\rho}$, is replaced by the empirical estimate $\tilde{\rho}$ from \eqref{eq:empiricalrho}.
As such, we have equality of the events $(\varsigma_s \neq \zeta_s) = (\tilde{\rho}(s)<c_{\widehat{\rho}})$.
Let $n$ be large enough such that $c_{\widehat{\rho}}<c_{\rho}<\rho(s)$ for all $s\in[0,1]$,
and then we observe that
\begin{align*}
    \sup_{t\in[0,1]}|\widehat{V}_t-\tilde{V}_t|
    &\leq C_EC_{\bh}\sqrt{n_2}
        \int_0^1|\zeta_s-\varsigma_s| \mathrm{d}s\\
    &\leq c_{\widehat{\rho}}c_{\rho}^{-2}C_EC_{\bh}\sqrt{n_2}
        \int_0^1 \one\left(\tilde{\rho}(s)<c_{\widehat{\rho}}\right) \mathrm{d}s \\
    &\leq c_{\widehat{\rho}}c_{\rho}^{-2}C_EC_{\bh}\sqrt{n_2}
        \int_0^1 \one\left((\tilde{\rho}(s)-\rho(s))^2<(\rho(s)-c_{\widehat{\rho}})^2\right) \mathrm{d}s.
\end{align*}
Applying Tonelli's theorem and using Markov's inequality to the above we obtain
\begin{align*}
    \ex\Big[\sup_{t\in[0,1]}|\widehat{V}_t-\tilde{V}_t|\Big]
    &\lesssim
    \sqrt{n_2} \int_0^1 \mathbb{P}\left((\tilde{\rho}(s)-\rho(s))^2<(\rho(s)-c_{\widehat{\rho}})^2\right) \mathrm{d}s \\
    &\leq 
    \sqrt{n_2} \int_0^1 \frac{\ex\big[(\tilde{\rho}(s)-\rho(s))^2\big]
    }{(\rho(s)-c_{\widehat{\rho}})^2} \mathrm{d}s \\
    &\leq (c_{\rho}-c_{\widehat{\rho}})^{-2}\sqrt{n_2}\|\tilde{\rho}-\rho\|_2^2 \to 0,
\end{align*}
where `$\lesssim$' is an inequality up to the constant $c_{\widehat{\rho}}c_{\rho}^{-2}C_EC_{\bh}$, and where the last convergence follows from Proposition~\ref{prop:rhorate}.

It now suffices to show that $\ex\Big[\sup_{t\in[0,1]}|\tilde{V}_t - V_t| \Big]\xrightarrow{P}0$. Using the decomposition $\tilde{\rho} = \overline{\rho} + r_1 + r_2$ from the proof of Proposition~\ref{prop:rhorate}, we first see that
\begin{align*}
    V_t = \sqrt{n_2}\int_0^t \frac{\rho(s)-\overline{\rho}(s)}{\rho(s)^2} \ex[E_s(\blambda_s-\lambda_s)]\mathrm{d}s,
\end{align*}
so that
\begin{align*}
    \sup_{t\in[0,1]} |\tilde{V}_t - V_t|
    \leq \frac{C_EC_{\bh}}{c_{\rho}^2}\sqrt{n_2}
        \int_0^1 |r_1(s)| + |r_2(s)| \mathrm{d}s.
\end{align*}
Using the intermediate results \eqref{eq:rhorater1} and \eqref{eq:rhorater2} from the proof of Proposition \ref{prop:rhorate}, it follows that
\begin{align*}
    \ex\Big[\sup_{t\in[0,1]} |\tilde{V}_t - V_t|\Big]
    = \frac{C_EC_{\bh}}{c_{\rho}^2}\sqrt{n_2}(\|r_1\|_1+\|r_2\|_1)\to 0.
\end{align*}
\end{proof}

\subsection{Proof of Theorem~\ref{thm:ACMasymptotics}} \label{sec:acmproofACMasymptotics}
We first establish asymptotic Gaussianity of the oracle process.
\begin{thm}\label{thm:oracleconvergence}
    Let $S=(S_t)_{t\in[0,1]}$ be the process given by
    $S_t = U_t + V_t - \sqrt{n_2}\gamma_t$,
    and let $\Gamma=(\Gamma_t)_{t\in[0,1]}$ be the process from Theorem~\ref{thm:ACMasymptotics}.

    Under Assumption~\ref{asm:boundedness}, it holds that
    $
        S \xrightarrow{d} \Gamma
    $
    with respect to the uniform topology as $n\to \infty$.
\end{thm}

\begin{proof}
    For any $0\leq x<y \leq 1$, define the integrals
    \begin{align*}
        I_{xy} \coloneq \int_x^y G_z \mathrm{d}M_z,
        \qquad \text{and} \qquad
        J_{xy} \coloneq \int_x^y E_zG_z\mathrm{d}\gamma_z.
    \end{align*}
    The process $S_t$ is a sum of i.i.d. processes with same distribution as $I_{0t}-J_{0t}$, and note that this process has mean zero: it holds that $\ex[I_{0t}]=\gamma_t$ by the definition of the LCM, and that $\ex[J_{0t}]=\gamma_t$ since $\ex[E_sG_s] = 1$.  
    Thus we can apply the central limit theorem in Skorokhod space derived in 
    \citet[Example 1.14.24]{vanDerVaartWellner2023}.
    In view of this CLT, it suffices to show that
    there exist continuous, strictly increasing functions $F_1,F_2 \colon [0,1] \to [0,\infty)$ and constants $\epsilon_1,\epsilon_2>0$, such that for every $0\leq s<t<u\leq 1$,
    \begin{align}
        \ex[(I_{st}+J_{st})^2] &\leq |F_1(t)-F_1(s)|^{1/2 + \epsilon_1}, 
        \label{eq:chainingconditionC}\\
        \ex[(I_{st}+J_{st})^2(I_{tu}+J_{tu})^2] &\leq |F_2(u)-F_2(s)|^{1 + \epsilon_2}.
        \label{eq:chainingconditionD}
    \end{align}
    To this end, let $0\leq s<t<u\leq 1$ be given.
    In the following use $x\lesssim y$ to denote that the inequality $x\leq K\cdot y$ holds for a constant $K>0$ independent of $(s,t,u)$. In fact, in each case the constant will be a polynomial over the constants $C_{\bh}, C_E$, and $c_{\rho}^{-1}$ from Assumption~\ref{asm:boundedness}.

     We first note that 
    \begin{equation}\label{eq:Jbound}
        J_{st} \lesssim t-s,
        \quad \text{and} \quad
        J_{tu} \lesssim u-t,
    \end{equation}
    under Assumption \ref{asm:boundedness}. Using the quadratic characteristic for compensated counting processes, we also have
    \begin{align*}
        \ex[I_{st}^2] &\lesssim 
            \ex\Big[\Big(\int_s^t G_r\mathrm{d}\bM_r\Big)^2\Big]
            +\ex\Big[\Big(\int_s^t 
                G_r(\blambda_r-\lambda_r)\mathrm{d}r\Big)^2\Big] \\
            &= \ex\Big[\int_s^t G_r^2\blambda \mathrm{d}r\Big]
            +\ex\Big[\Big(\int_s^t 
                G_r(\blambda_r-\lambda_r)\mathrm{d}r\Big)^2\Big] \\
            & \lesssim (t-s) + (t-s)^2 \leq 2(t-s)
    \end{align*}
    Thus we conclude that
    \begin{align*}
        \ex[(I_{st}+J_{st})^2] 
        \lesssim \ex[I_{st}^2]+\ex[J_{st}^2] 
        \lesssim t-s.
    \end{align*}
    This shows that the condition \eqref{eq:chainingconditionC} holds with $\epsilon_1=\frac{1}{2}$ and $F_1(r)=C_1 r$ for a suitably large constant $C_1>0$.

    For the second condition we have that 
    \begin{align*}
        \ex[(I_{st}+J_{st})^2(I_{tu}+J_{tu})^2]
        &\leq 4(
            \ex[I_{st}^2I_{tu}^2]
            +\ex[I_{st}^2J_{tu}^2]
            +\ex[J_{st}^2I_{tu}^2]
            +\ex[J_{st}^2J_{tu}^2]
            ).
    \end{align*}
    By same argument as $\ex[I_{st}^2]\lesssim t-s$, we also have that $\ex[I_{tu}^2]\lesssim t-u$.
    Combining this with the inequality \eqref{eq:Jbound} we obtain
    \begin{align*}
        \ex[I_{st}^2J_{tu}^2] + \ex[J_{st}^2I_{tu}^2] + \ex[J_{st}^2J_{tu}^2]
        \lesssim (t-s)(u-t)
    \end{align*}
    For the term $\ex[I_{st}^2I_{tu}^2]$, we note that if $T\in (s,t]$ then $M$ is constant on the interval $(t,u]$, in which case $I_{tu}=0$.
    Therefore
    \begin{align*}
        \ex[I_{st}^2I_{tu}^2]
        = \ex\left[
        \left(\int_s^t G_r \lambda_r\mathrm{d}r\right)^2 
        I_{tu}^2
        \right]
        \lesssim  (t-s)^2 
        \ex\left[
        I_{tu}^2
        \right]
        \lesssim  (t-s)^2(u-t)
    \end{align*}
    Combined we conclude that
    \begin{equation*}
        \ex[(I_{st}+J_{st})^2(I_{tu}+J_{tu})^2]
        \lesssim (t-s)(u-t) \leq (u-s)^2.
    \end{equation*}
    This shows that the second condition \eqref{eq:chainingconditionD} is satisfied with $\epsilon=1$ and $F_2(r)=C_2r$ for a suitably large constant $C_2>0$.
\end{proof}

Returning to the proof of Theorem~\ref{thm:ACMasymptotics}, the limit distribution of the ACM estimator is now a simple consequence of the continuous mapping theorem.\footnote{Applied to the map $\mathbb{D}^3 \ni (x,y,z) \mapsto x+y+z \in \mathbb{D}$, where $\mathbb{D}$ is the space of \lc{} functions on $[0,1]$ endowed with the uniform norm.} 
Using the limits establised in Proposition~\ref{prop:AisequivtoU}, Theorem~\ref{thm:BequivtoV}, and Theorem~\ref{thm:oracleconvergence}, we conclude that
\begin{equation}\label{eq:ACMctsmapping}
    \sqrt{n_2}(\;\! \widehat{\gamma} - \gamma)
    = \underbrace{S}_{\xrightarrow{d}\Gamma} 
        + \underbrace{(\sqrt{n_2}\cdot \tgamma - U)}_{\xrightarrow{p}0} 
        + \underbrace{(B - V)}_{\xrightarrow{p}0}
    \: \xrightarrow{\: d\:} \: \Gamma
\end{equation}
in the uniform topology as $n\to \infty$.

Turning to the cross-fitted ACM estimator, let $\widehat{\gamma}^{k}$ denote ACM estimator fitted with $\mathcal{I}_1 = [n]\setminus J_k$ and $\mathcal{I}_2=J_k$ for each $k \in [K]$. Let also 
$$
    S_t^k = \frac{1}{\sqrt{|J_k|}}\sum_{i\in J_k} s_t^{(i)},
    \qquad \text{where}\quad 
    s_t^{(i)}\coloneq
    \int_0^t G_s^{(i)}\mathrm{d}M_s^{(i)} + \int_0^t E_s^{(i)}G_s^{(i)}\mathrm{d}\gamma_s,
$$
denote the corresponding oracle processes.

In the following computation, we use $x\sim y$ as a shorthand for $\sup_{t\in[0,1]}|x_t-y_t|\xrightarrow{p}0$.
Using that folds are assumed to have sizes $|J_k| \in \{\floor{n/K},\ceil{n/K}\}$, 
Proposition~\ref{prop:AisequivtoU} and Theorem~\ref{thm:BequivtoV} applied to each fold yields that
\begin{align*}
    \sqrt{n}(\widecheck{\gamma}-\gamma)
    \sim 
    \frac{\sqrt{n}}{K} \sum_{k=1}^K \frac{S^k}{\sqrt{|J_k|}}
    = 
    \frac{\sqrt{n}}{K} \sum_{k=1}^K \frac{1}{|J_k|} \sum_{i\in J_k} s^{(i)} 
    \sim \frac{1}{\sqrt{n}}  \sum_{i\in [n]} s^{(i)} 
\end{align*}
where last equivalence follows from Lemma~\ref{lem:cf}.
By the continuous mapping theorem, applied as in \eqref{eq:ACMctsmapping}, it suffices to establish that $\frac{1}{\sqrt{n}}  \sum_{i\in [n]} s^{(i)} \xrightarrow{d} \Gamma$.
However, this is simply the CLT established in Theorem~\ref{thm:oracleconvergence} with $\mathcal{I}_2 = [n]$. \hfill $\Box$

The following lemma is a Banach space generalization of one of the standard steps for proving validity of cross-fitting, cf. Lemma S1 in \cite{lundborg2023perturbation}.
\begin{lem}\label{lem:cf}
    Let $W$ be a random element in a Banach space $(\mathbb{D},\|\cdot\|)$, satisfying that $\ex[\|W\|]<\infty$, and let $(W_i)_{i\in \mathbb{N}}$ be i.i.d. copies of $W$. 
    For a fixed $K\in \mathbb{N}$, let $[n]=J_1\cup \cdots \cup J_K$ be a disjoint partition of $[n]$ with 
    $|J_k| \in \{\floor{n/K}, \ceil{n/K}\}$ for each $k\in [K]$.
    Finally, let $S = \frac{1}{n}\sum_{i\in [n]} W_i$ and $S_k = \frac{1}{|J_k|} \sum_{i\in J_k} W_i$ for each $k\in [K]$. 
    
    Then, for any $r<1$,
    \begin{align*}
        n^r \cdot \Big(\frac{1}{K} \sum_{k=1}^K S_k - S\Big) \xrightarrow{L^1} 0
    \end{align*}
    as $n\to \infty$.
\end{lem}
\begin{proof}
Note that
$\frac{1}{n+K} < \frac{1}{K|J_k|}  < \frac{1}{n-K}$
and hence 
$|\frac{1}{K|J_k|} -\frac{1}{n} | \leq \frac{1}{n-K} - \frac{1}{n} = \frac{K}{(n-K)n}$.
Since
\begin{align*}
    \frac{1}{K} \sum_{k=1}^K S_k 
        - S
    = \sum_{k=1}^K \sum_{i\in J_k}
        \left(\frac{1}{K|J_k|} -\frac{1}{n}\right) W_i,
\end{align*}
we conclude that 
\begin{align*}
    n^r \cdot \ex\Big\|\frac{1}{K} \sum_{k=1}^K S_k - S\Big\|
        \leq \frac{Kn^r}{n-K} \ex\|W\| \longrightarrow 0
\end{align*}
as $n\to \infty$.
\end{proof}

\supplementarysection{Multiplicative hazards models} \label{app:multhaz}

The predictable projection, $\pi_t(\overline{Z}_t) = \ex[X_t \given T\geq t, \overline{Z}_t]$, is generally 
difficult to compute -- even with baseline covariates $Z$ and a binary baseline exposure $X$, 
where $\pi_t(Z)=\mathbb{P}(X=1\given T \geq t, Z)$. However, for a multiplicative hazards 
model, and under a strengthened censoring assumption, we give a fairly explicit 
representation of $\pi_t(Z)$ in Proposition~\ref{prop:binarycomputations} below. 
Example \ref{ex:LCMinCox} elaborates further on the ACM in the Cox proportional hazards model. 

\begin{prop}[Predictable projection for binary exposure]\label{prop:binarycomputations}
    Suppose that $X\in \{0,1\}$ is a binary baseline exposure, that $Z$ is a baseline covariate and that the 
    model is a nonparametric time-varying proportional hazards model: 
    \begin{align*}
        \bh_t(X,Z) = (X \theta_t + 1)\phi_t(Z).
    \end{align*}
    If $C \ind (T^*, X) \given Z$, then the predictable projection is given by
    \begin{align*}
        \pi_t(Z) = \frac{\pi_0(Z)}{\pi_0(Z) + (1-\pi_0(Z))e^{\mathsf{I}_t(Z)}},
    \end{align*}
    where $\pi_0(Z) = \mathbb{P}(X=1\given Z)$ is the (baseline) propensity score and 
    \begin{align*}
        \mathsf{I}_t(Z)\coloneq \int_0^t \theta_s \phi_s(Z)\mathrm{d}s.
    \end{align*}
    The predictable projection can be expressed as the logistic model:
    $$
        \mathrm{logit}(\pi_t(Z)) =  \mathrm{logit}(\pi_0(Z)) -  \mathsf{I}_t(Z).
    $$
\end{prop}
\begin{proof}
    Recall that $Y_t = \one(T \geq t) = \one(C \geq t)\one(T^* \geq t)$. Using the 
    definition of $\pi_t(Z)$ and the conditional independence assumption regarding censoring, 
    we find that 
    \begin{align*}
        \pi_t(Z) 
            & = \ex[X\given T\geq t, Z] 
            = \frac{\ex[XY_t\given Z]}{\mathbb{P}(T\geq t \given Z)} \\
            & = \frac{\ex[X\one(T^* \geq t) \given Z]\mathbb{P}(C \geq t \mid Z)}{\mathbb{P}(T^* \geq t \given Z)\mathbb{P}(C \geq t \given Z)} \\
            & = \frac{\ex[X\one(T^* \geq t) \given Z]}{\mathbb{P}(T^* \geq t \given Z)}.
    \end{align*}
    Using the correspondence between survival time and cumulative hazard we obtain that
    \begin{align*}
    \ex[X\one(T^* \geq t)\given Z]
        &= \ex[X\cdot \mathbb{P}(T^* \geq t \given X,Z)\given Z] \\
        &= \ex\Big[X\exp\Big(-X\int_0^t \theta_s\phi_s(Z)\mathrm{d}s 
            - \int_0^t\phi_s(Z)\mathrm{d}s\Big)\given Z\Big] \\
        &= \mathbb{P}(X=1\given Z)\exp\Big(-\int_0^t \theta_s\phi_s(Z)\mathrm{d}s 
            -\int_0^t \phi_s(Z)\mathrm{d}s\Big) \\
        &= \pi_0(Z) e^{-\mathsf{I}_t(Z) - \mathsf{II}_t(Z)},
    \end{align*}
    where $\mathsf{I}_t(z)\coloneq \int_0^t \theta_s\phi_s(z)\mathrm{d}s$ 
    and $\mathsf{II}_t(z) \coloneq \int_0^t \phi_s(z)\mathrm{d}s$.
    Similarly the denominator is given by
    \begin{align*}
        \mathbb{P}(T^* \geq t \given Z) 
            & = \ex\left[ \exp\left(- X \int_0^t \theta_s \phi_s(Z)\mathrm{d}s
                -\int_0^t b_s\phi_s(Z)\mathrm{d}s \right)  \given Z \right] \\
            & = ((1 - \pi_0(Z)) + e^{-\mathsf{I}_t(Z)} \pi_0(Z)) e^{-\mathsf{II}_t(Z)}.
    \end{align*}

    From this we obtain that
    \begin{align*}
        \pi_t(Z) & = \frac{\pi_0(Z) e^{-\mathsf{I}_t(Z)}}{(1 - \pi_0(Z)) + e^{-\mathsf{I}_t(Z)} \pi_0(Z)} \\
         & = \frac{\pi_0(Z)}{\pi_0(Z) + (1-\pi_0(Z))e^{\mathsf{I}_t(Z)}} \\ 
         & = \frac{1}{1 + e^{-\mathrm{logit}(\pi_0(Z)) } e^{\mathsf{I}_t(Z)}} 
         = \mathrm{expit}(\mathrm{logit}(\pi_0(Z))  - \mathsf{I}_t(Z)).
    \end{align*}
\end{proof}

\begin{example}[ACM in a Cox proportional hazards model]\label{ex:LCMinCox}
    With baseline covariates $Z\in \real^d$, in addition to a binary exposure $X\in\{0,1\}$, the Cox 
    proportional hazards model asserts that:
    \begin{align}\label{eq:Coxmodel}
        \bh_t(X,Z)
            = h^0(t)\exp(\beta_X X + \beta_Z^\top Z) = (X (e^{\beta_X} - 1) + 1) h^0(t) \exp(\beta_Z^\top Z).
    \end{align}
    Letting $\varphi(Z) = (e^{\beta_X}-1)e^{\beta_Z^\top Z}$ we see that this model is a special case of Example \ref{ex:simple} with 
    $\theta_t(Z) =  h^0(t) \varphi(Z)$. It follows by \eqref{eq:ACM-baseline} 
    that the ACM can be expressed as 
    \begin{align*}
        \gamma_t & = \int_0^t \ex\left[h^0(s) \varphi(Z) w_s^{\textsc{ac}}(Z) \right] \mathrm{d} s \\
        & = \ex\Big[ \varphi(Z) \underbrace{\int_0^t h^0(s)  w_s^{\textsc{ac}}(Z) \mathrm{d} s}_{\eqcolon W_t^{\textsc{ac}}(Z) } \Big] 
        = \ex\left[ \varphi(Z)  W_t^{\textsc{ac}}(Z) \right] 
        = \ex\left[ e^{\beta_Z^\top Z}  W_t^{\textsc{ac}}(Z) \right] (e^{\beta_X}-1). 
    \end{align*}
    Thus $\gamma_t = c(t) (e^{\beta_X}-1) $ is proportional to $e^{\beta_X}-1$ with a $t$-dependent proportionality constant  $c(t) = \ex\left[ e^{\beta_Z^\top Z}  W_t^{\textsc{ac}}(Z) \right] \geq 0$.

    Using the innovation theorem we see that 
    \begin{align*}
        h_t(Z)
            &= \big(\pi_t(Z) (e^{\beta_X}-1) + 1\big)
                h^0(t)\exp(\beta_Z^\top Z),
    \end{align*}
    and we could also arrive at the representation of the ACM directly from the identity \eqref{eq:LCMisLCM2}.
    Indeed, $\bh_s(X,Z) - h_s(Z) = (e^{\beta_X}-1)(X - \pi_s(Z))h^0(s)\exp(\beta_Z^\top Z)$, and 
    \eqref{eq:LCMisLCM2} gives 
    \begin{align*}
        \gamma_t & = (e^{\beta_X}-1) \int_0^t h^0(s) \ex\left[ e^{\beta_Z^\top Z} \frac{Y_s(X - \pi_s(Z))^2}{\rho(s)} \right] \mathrm{d} s
    \end{align*}
    The proportionality constant must, of course, be $c(t)$, but this can also be verified directly using 
    the same argument as in Example \eqref{ex:simple_cont}. We may note that
     $h_t(\cdot)$ is not generally of the form of a Cox proportional hazards model, which illustrates
     the fact that the Cox model is not closed under marginalization. 

    To compute the proportionality constant we need the weights 
    \[   w_t^{\textsc{ac}}(Z) 
    = \frac{Y_t \pi_t(Z) (1 - \pi_t(Z)) }{\ex[Y_t \pi_t(Z) (1 - \pi_t(Z))]},
    \] 
    and to this end we need to compute the predictable projection $\pi_t(Z)$.    
    Under the censoring assumption $C \ind (T^*, X) \given Z$, it follows 
    from Proposition~\ref{prop:binarycomputations} that   
    \begin{align}\label{eq:coxPI}
        \mathrm{logit}(\pi_t(Z)) = \mathrm{logit}(\pi_0(Z)) - \varphi(Z) H^0(t) ,
    \end{align}
    where  $H^0(t) = \int_0^t h^0(s) \mathrm{d}s$.
    This shows that $\pi_t(Z)$ evolves in time with a logistic growth rate precomposed with $H^0$. 
    Formula \eqref{eq:coxPI} makes it straightforward to compute $\pi_t(z) (1 - \pi_t(z))$
    for any given $t$ and $z$ (provided that we can compute $\pi_0(z)$), but it does not 
    appear to yield a simple analytic expression of neither the weights $w_t^{\textsc{ac}}(Z)$ 
    nor the proportionality constant $c(t)$. It is, however, useful for Monte Carlo computation 
    of the ACM.

\end{example}